\pdfoutput=1

\documentclass[
  onecolumn,  
  nofootinbib,showpacs,superscriptaddress,preprintnumbers]{revtex4-2}

\usepackage[utf8]{inputenc}
\usepackage{graphicx, color}

\usepackage[svgnames,psnames]{xcolor}
\usepackage[colorlinks,citecolor=DarkGreen,linkcolor=FireBrick,linktocpage]{hyperref}

\usepackage{dcolumn}
\usepackage{bm}
\usepackage{amsmath}
\usepackage{amsthm}
\usepackage{ascmac}
\usepackage{xparse} 
\usepackage{mathtools} 
\usepackage{blkarray} 
\usepackage{amscd,verbatim}
\usepackage{amssymb}
\usepackage[all]{xy}
\usepackage{tikz}
\usepackage{tikz-cd}
\usepackage{pgf,pgfplots}
\usepackage{enumitem}
\usepackage{caption}
\usepackage{subcaption}
\usepackage{floatrow}
\usepackage{slashed}

\usetikzlibrary{decorations.pathmorphing}

\newcommand\coker{\mathop{{\rm coker}}}

\newcommand\p{\partial}

\newcommand\wt{\widetilde}

\usepackage[normalem]{ulem}  

\renewcommand\sout{\bgroup \color{red} \ULdepth=-.5ex \ULset}

\definecolor{mydarkred}{RGB}{233,20,35}
\definecolor{mypurple}{RGB}{120, 35, 160}
\definecolor{mydarkpurple}{RGB}{128, 100, 162}
\definecolor{mybrown}{RGB}{255, 195, 0}
\definecolor{myaqua}{RGB}{29, 153, 168}
\definecolor{myblue}{RGB}{91, 129, 184}  
\definecolor{mygreen}{RGB}{84, 174, 50}  
\definecolor{mybrightblue}{RGB}{0, 140, 255}  

\tikzstyle{species}=[
  circle,draw=black!100, thick, inner sep=0pt,minimum size=6mm
  ] 
\tikzstyle{reaction}=[rectangle,draw=black!100,fill=black!15,thick, inner sep=0pt,minimum size=5mm]

\theoremstyle{remark}

\newtheorem{lemma}{Lemma}

\newtheorem{theorem}{Theorem}

\newtheorem*{thm*}{Theorem}

\newtheorem{proposition}[lemma]{Proposition}

\newtheorem{cor}[lemma]{Corollary}

\newtheorem{definition}[lemma]{Definition}

\newtheorem{assm}[lemma]{Assumption}

\begin{document}
\preprint{KUNS-2973}

\title{ Complete characterization of robust perfect adaptation 
\\
in biochemical reaction networks }

\author{Yuji~Hirono}
\email{yuji.hirono@gmail.com}
\affiliation{Department of Physics, Kyoto University, Kyoto 606-8502, Japan}
\affiliation{RIKEN iTHEMS, RIKEN, Wako 351-0198, Japan}

\author{Ankit~Gupta}
\email{ankit.gupta@bsse.ethz.ch}
\affiliation{
Department of Biosystems Science and Engineering, ETH Zurich, 4058 Basel, Switzerland
}

\author{Mustafa~Khammash}
\email{mustafa.khammash@bsse.ethz.ch}
\affiliation{
Department of Biosystems Science and Engineering, ETH Zurich, 4058 Basel, Switzerland
}

\date{\today}

\begin{abstract}
Perfect adaptation is a phenomenon whereby the output variables of a system can reach and maintain certain values despite external disturbances.
Robust perfect adaptation (RPA) refers to an adaptation property that does not require fine-tuning of system parameters.
RPA plays a vital role for the survival of living systems in unpredictable environments, and there are numerous examples of biological implementations of this feature.
However, complex interaction patterns among components in biochemical systems pose a significant challenge in identifying RPA properties and the associated regulatory mechanisms.
The goal of this paper is to present a novel approach for identifying all the RPA properties that are realized for a generic choice of kinetics for general deterministic chemical reaction systems. 
This is accomplished by proving that an RPA property with respect to a system parameter can be represented by a subnetwork with certain topological features. 
This connection is then exploited to show that these special structures generate \emph{all} kinetics-independent RPA properties, allowing us to systematically identify all such RPA properties by enumerating these subnetworks. An efficient method is developed to carry out this enumeration, and we provide a computational package for this purpose.
We pinpoint the integral feedback controllers that work in concert to realize each RPA property, casting our results into the familiar control-theoretic paradigm of the Internal Model Principle.
Furthermore, we generalize the regulation problem to the multi-output scenario where the target values belong to a robust manifold of nonzero dimension, and provide a sufficient topological condition for this to happen.
We call the emergence of this phenomenon as manifold RPA.
The present work significantly advances our understanding of regulatory mechanisms that lead to RPA in endogenous biochemical systems, and it also provides rational design principles for synthetic controllers. 
We demonstrate these results through illustrative examples as well as biological ones.
The present results indicate that an RPA property is essentially equivalent to the existence of a ``topological invariant", which is an instance of what we coin as the ``{\bf R}obust {\bf A}daptation is {\bf T}opological'' (RAT) principle. 
\end{abstract}

\maketitle

\tableofcontents

\section{Introduction} 

\subsection{ Context and motivation }

Maintaining stability in a variable environment is a crucial issue for biological systems~\cite{stelling2004robustness,kitano2004biological,kitano2007towards}. 
One strategy adopted by living cells to achieve this is perfect adaptation, which is a property of a system to maintain the level of certain quantities by countering the effects of disturbances within biochemical reaction networks~\cite{10.1093/nar/28.1.27,jeong2000large,doi:10.1126/science.1073374}.
Perfect adaptation is said to be robust when no fine-tuning of system parameters is needed to achieve the adaptation.
Having this property of robust perfect adaptation (RPA)
contributes to the stability of a system,
and indeed there are numerous biological examples of RPA such as chemotaxis
of prokaryotes~\cite{barkai1997robustness} and eukaryotes~\cite{LEVCHENKO200250}, 
calcium homeostasis~\cite{el2002calcium},
glucose uptake of cancer cells~\cite{Tveit2020.01.02.892729},
yeast osmoregulation~\cite{MUZZEY2009160}, 
cell signaling~\cite{FERRELL201662}, 
scaling of morphogen gradient~\cite{ben2010scaling}, and so on. 
Revealing how cells implement RPA in molecular networks is important not only for understanding the origin of robustness in living systems but also for designing synthetic biomolecular systems. Recent advances in synthetic biology techniques have afforded scientists an unprecedented ability to engineer biomolecular controllers, utilizing genetic components, and transport them into living cells, where they can accomplish novel functions, like RPA. This emerging field, known as Cybergenetics~\cite{9779327}, focuses on the analysis and design of genetic control systems and holds tremendous potential for various domains, including industrial biotechnology and medical therapy.

Most existing work on RPA in biological reaction systems borrows ideas from control theory and views RPA as the property of an input-output system to robustly reject constant-in-time disturbances to the input variable by ensuring that the steady-state level of the output remains unaffected~\cite{RevModPhys.88.035006,KHAMMASH2021509}. Typically, the output is the concentration of some chemical species (e.g.\ a protein of interest) with a pre-defined steady-state value (called the \emph{set-point}), and the input is the concentration of some other species (e.g. enzyme, chemical inducer etc.) or some exogenous variable. For such single-input single-output (SISO) reaction systems, it is known that only two types of RPA topologies emerge: the incoherent feedforward (IFF) loops and the negative feedback (NFB) networks. This was first discovered computationally through an exhaustive search over three-node networks~\cite{ma2009defining} (see also Ref.~\cite{TANG2016274}), but it has recently been proved mathematically that any arbitrarily-sized RPA network must essentially be composed of IFF and NFB modules ~\cite{araujo2018topological,wang2021structure}. The reason for such strict structural requirements for RPA networks is linked to the famous Internal Model Principle (IMP)~\cite{bin2022internal,SONTAG2003119} of control theory, which mandates that any RPA system must be internally organized (after possibly a change of coordinates) into two distinct components -- an \emph{internal model} (IM) that computes the time-integral of the deviation of the output and its set-point, and the \emph{rest of the network} (RoN) which receives both input disturbances as well as counteracting signals from the IM to reject these disturbances (see Fig. \ref{fig:imp} in Section~\ref{sec:rpa_int_ctr}). This argument shows that an IM serves as an embedded controller that generates RPA by implementing the famous integral feedback mechanism. This mechanism not only resides at the forefront of modern control engineering but has also been identified in numerous endogenous biochemical RPA networks~\cite{barkai1997robustness,alon1999robustness,ANG2010723,FERRELL201662,aoki2019universal,8619101,KHAMMASH2021509}. A specific type of SISO RPA is called Absolute Concentration Robustness (ACR) which refers to the robustness of the concentration of output species when the initial state of the system is perturbed. Traditionally, the study of ACR networks has relied on techniques from Chemical Reaction Network Theory (CRNT) \cite{feinberg2019foundations}, starting from the seminal paper by Shinar and Feinberg in 2010 \cite{shinar2010structural}. A recent development has established a connection between ACR and integral controllers \cite{cappelletti2020hidden}, resulting in a deeper understanding of the mechanisms underlying ACR.

Although robustness to a specific single input may be relevant in some scenarios, the notion of a single input source for disturbances appears to be generally limiting within biological contexts, where disturbances manifest as parameter perturbations that can arise due to multitude of reasons, such as temperature changes~\cite{ni2009control}, onset of stress conditions~\cite{toni2011qualitative}, resource competition and burden effects \cite{qian2017resource}, etc. 
As many of these disturbances are often simultaneously present, it is important to consider RPA networks with multiple inputs, and in such a setting, conditions have been identified that completely characterize biochemical networks that exhibit a maximal form of RPA (called maxRPA), whereby the set-point of one output species is robust to the most number of reaction rate parameters, i.e.\ it depends on the least number of parameters~\cite{gupta2022universal}. 
The networks exhibiting maxRPA are examples of multiple-input single-output (MISO) RPA systems. 
The characterization of more general multiple-input multiple-output (MIMO) RPA systems has been largely unexplored so far (see however Ref.~\cite{alexis2022regulation} for the two-output case). 
A major challenge for finding MIMO RPA structures in biological systems arises from the absence of a presumptive separation between a controlled system and a controller, which contrasts with control engineering, where these components are given from the outset.

Inspired by the connection between RPA and integral actions, recently an algebraic approach has been devised, that identifies all RPA properties, along with the hidden integral actions that generate this property~\cite{araujo2023universal}. This method assumes mass-action kinetics (see Eq.~\eqref{defn:mass_action_kin}) to express the reaction rates as a function of reacting species' concentrations. 
This ensures that the steady-states of the system satisfy a system of polynomial equations, and by investigating the ideal generated by these polynomials, through resource-intensive Gr\"{o}bner basis computations, all RPA properties can in principle be found. 
While the law of mass-action works remarkably well in \emph{in vitro} (i.e.\ `test tube' ) conditions, where the reacting species are sufficiently dilute and well-mixed \cite{guldberg1864studies}, many experimental and computational studies have demonstrated that this law breaks down within living cells (see Ref.~\cite{schnell2004reaction} and the references therein). 
Several factors contribute to this breakdown, including macromolecular crowding in the cytoplasmic soup where reactions occur \cite{hall2003macromolecular} and spatial compartmentalisation within the cells \cite{bauermann2022chemical}. 
These factors significantly impede molecular diffusion compared to \emph{in vitro} conditions \cite{elowitz1999protein}, resulting in the loss of well-mixed characteristics. 
Additionally, reaction networks serve as simplifications of the intricate reality found within cells. Kinetics-altering factors like conformational changes in protein structures~\cite{bu2011proteins}, allosteric perturbations~\cite{perica2021systems}, resource competition~\cite{qian2017resource}, and context dependence~\cite{cardinale2012contextualizing} are deliberately disregarded for the sake of model simplicity and manageability. 
In such cases, the assumption of mass-action kinetics becomes tenuous. One could hypothesize that if RPA is crucial for a cell's survival, it would make evolutionary sense to select RPA-generating strategies that are independent of kinetics, meaning the RPA property remains invariant regardless of the chosen kinetics. 
This notion holds significance not only in understanding living systems but also in the field of synthetic biology where one may aim to develop molecular controllers that can achieve RPA when interfaced with natural intracellular networks, irrespective of the specific kinetics employed~\cite{gupta2022universal}. 
It is important to note that while the structure of a reaction network is determined by the chemical properties of the reacting species, which can be reliably determined through \emph{in vitro} experiments, knowledge of reaction kinetics does not readily transfer from \emph{in vitro} studies to the actual conditions within living cells~\cite{minton2001influence}. 

Distinctly from control-theoretic ideas, there exist alternative methods to characterize the response of reaction systems to parameter perturbations, based on the topological analysis of subnetworks containing the perturbed parameters~\cite{PhysRevLett.117.048101,PhysRevE.96.022322,PhysRevE.98.012417,PhysRevResearch.3.043123,https://doi.org/10.48550/arxiv.2302.01270}.
In such methods, a crucial role is played by the integers assigned to a given subnetwork based on its topological characteristics.
It has been shown that if a subnetwork is output-complete (meaning that all reactions that have a reactant species within the subnetwork are included in the subnetwork) and its \emph{influence index} is zero, the steady-state values of the species-concentrations and reaction fluxes outside this subnetwork are insensitive to perturbations of parameters inside the subnetwork. 
Namely, these special subnetworks confine the effect of perturbations inside them, and hence such subnetworks are called \emph{buffering structures}. 
This result can be interpreted as a sufficient condition for RPA: the concentrations and reaction rates outside the subnetwork exhibit RPA, if the subnetwork is output-complete and has a zero influence index. 
More recently, a slightly different index called the \emph{flux influence index} was introduced to identify parameters under the perturbation of which \emph{all} the reaction rates (fluxes) exhibit RPA~\cite{https://doi.org/10.48550/arxiv.2302.01270}. 
A basic idea behind these topological approaches is that we can estimate the impact of parameter perturbations based on indices determined from network topology.
Although these indices provide us with sufficient conditions for RPA, whether all RPA properties can be explained in this way or not has been an open question.
The underlying regulatory mechanisms responsible for realizing these RPA properties and the relation of these results to control-theoretical approaches also remain unresolved.

\begin{figure}[tb]
  \centering
  \includegraphics
  [clip, trim=0cm 6cm 0cm 0cm, scale=0.45]
  {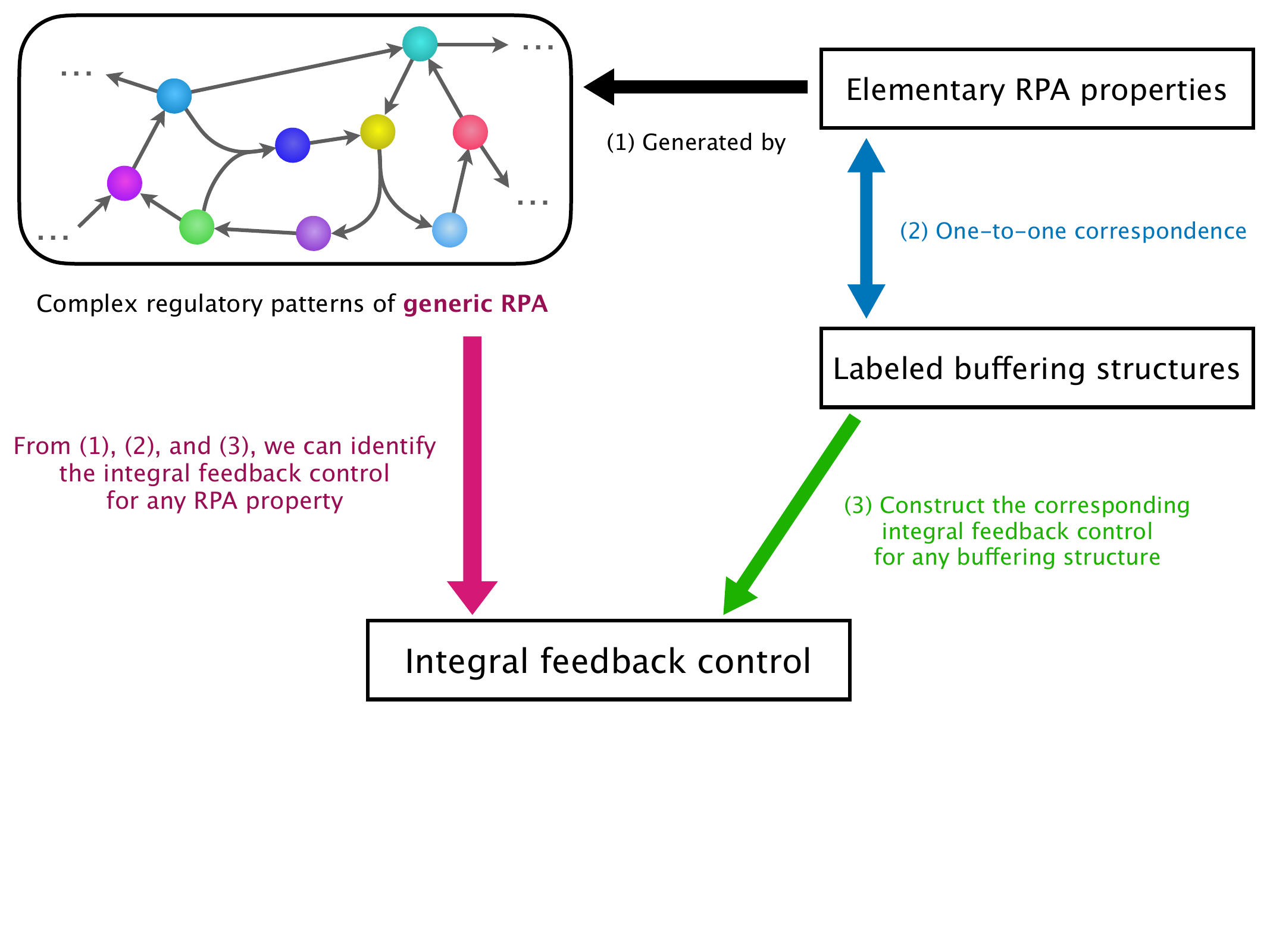}
 \caption{ Outline of the main results of the present work. A generic RPA property in a deterministic chemical reaction system can be decomposed into elementary RPA properties, each of which represents an RPA property with respect to a system parameter. We show that every elementary RPA property can be represented by a subnetwork with specific topological characteristics, that we call a labeled buffering structure. Then, we construct integral feedback control corresponding to any buffering structure.
 Consequently, we identify integral feedback control for every RPA property.
  }
  \label{fig:fig1} 
\end{figure}

\subsection{ Main results and structure of the paper }\label{sec:introd_main_results}

Motivated by these developments, the objective of this paper is to develop a systematic approach to identify all RPA properties that do not require a specific choice of kinetics for general deterministic chemical reaction systems (see Fig.~\ref{fig:fig1} for the sketch of our strategy).
Existing works on RPA define this notion as robustness of one or more quantities of interest (like steady-state species concentrations, reaction rates) with respect to certain system parameters (e.g.\ rate constants, values of conserved quantities, etc.). 
To characterize regulatory patterns in biochemical system, 
we shall regard a \emph{generic} RPA property as a triplet $(\mathcal{V}, \mathcal{E}, \bm p)$ of subsets of species, reactions and system parameters respectively, which signifies that the steady-state concentrations of species in $\mathcal{V}$ and rates of reactions in $\mathcal{E}$ exhibit RPA to perturbations of any parameter in $\bm p$. 
It is easy to see that any generic RPA property can be viewed as a combination of \emph{elementary} RPA properties that are each defined as subsets of species and reactions that exhibit RPA with respect to a single parameter $p$\footnote{See Sec.~\ref{sec:assumptions} for definitions of the elementary and generic RPA properties}. 
The main result of the paper proves a one-to-one correspondence between each elementary RPA property in a deterministic reaction system and a buffering structure with additional information, which we call a \emph{labeled buffering structure} (Theorem~\ref{thm:rpa-bs}).
Therefore, to characterize and systematically identify all generic RPA properties, it suffices to enumerate all the labeled buffering structures. We have devised an efficient algorithm to accomplish this enumeration and implemented it in \emph{Mathematica} as an open-source computational package named {\bf RPAFinder}~\cite{RPAFinder}.
This package can be used to identify the degrees of freedom with a desired RPA property, which is important for designing RPA-achieving synthetic controllers~\cite{9779327} and for finding yield-optimizing perturbations in metabolic pathways (see Sec.~\ref{sec:yeast}). 
Furthermore, we uncover the hidden integral feedback control responsible for achieving a generic RPA property by explicitly constructing the Internal Model along with the set of integrators it generates.
Therefore, we discover a novel Internal Model Principle for kinetics-independent RPA networks, which is of independent interest in the control-theory community and has deep connections with Cybergenetic applications~\cite{gupta2023internal}.
In a number of examples, we find that the set of integrators forces the dynamics of multiple output species towards a manifold with nonzero dimension.
This phenomenon, which we call \emph{manifold RPA} (see Fig.~\ref{fig:manifold-rpa-schematic} for a schematic illustration), implies that while the steady-state concentrations of the output species is sensitive to certain parameters, the relationship among the species is robustly maintained. 
We show that any output-complete subnetwork gives rise to manifold RPA (Theorem~\ref{thm:law-of-manifold-localization}), and the dimension of the target manifold is given by the influence index. This can be seen as a natural generalization the law of localization.
This form of multi-output control generalizes conventional approaches and has important implications for rational design in synthetic biology~\cite{alexis2022regulation}.

At this point, we would like to compare our results with the characterization results obtained in Ref.~\cite{araujo2023universal} for RPA networks with the mass-action kinetics. Unlike the approach in Ref.~\cite{araujo2023universal}, our analysis does not rely on any computationally-expensive procedures. Instead, we employ elementary linear-algebraic computations to fully characterize and identify all RPA properties. This approach is particularly well-suited for application to large-scale networks, such as those encountered in metabolic engineering (refer to Section \ref{sec:yeast}). Furthermore, our analysis yields a conventional IMP for RPA networks with a single IM that generates a global set of integral actions. This is in contrast to the non-standard IMP formulation presented in Ref.~\cite{araujo2023universal}, which considers each RPA network as a topologically organized collection of subnetworks, with each subnetwork possessing its own independent IM and localized integral action. Generally, these localized integral actions cannot be synthesized to form global integral actions, as required by the standard IMP. Nonetheless, the findings in Ref.~\cite{araujo2023universal} offer intriguing insights into the various ways in which RPA can arise due to intricate polynomial factorizations. Our results show that many of these ways `drop out' if deviations from mass-action kinetics are anticipated, revealing a linear structure that we meticulously unravel in this paper.

\begin{figure}[tb]
  \centering
  \includegraphics
  [clip, trim=0cm 9cm 0cm 0cm, scale=0.4]
  {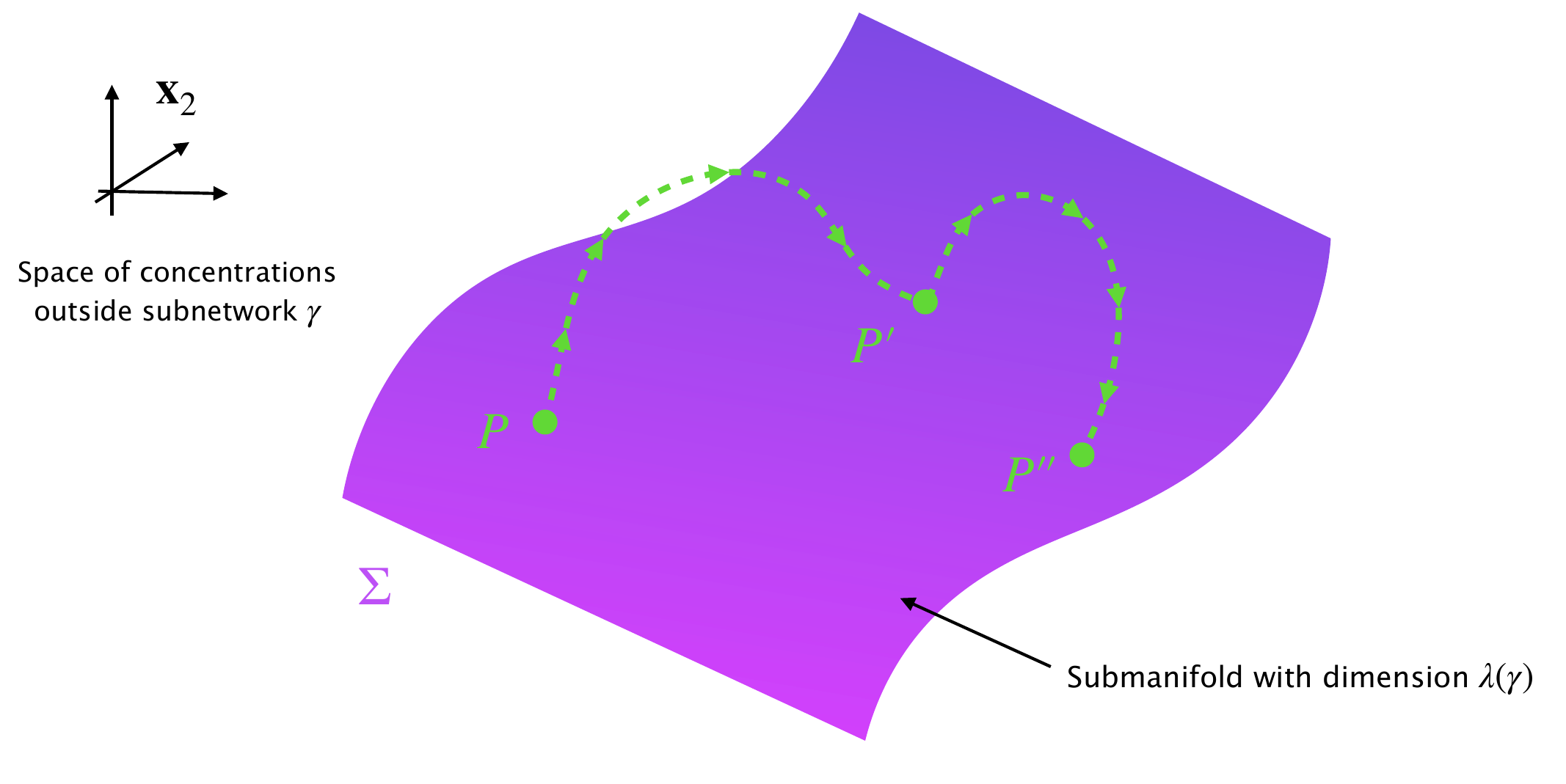}
 \caption{ Schematic of manifold RPA. For a chosen output-complete subnetwork $\gamma$, the concentrations $\bm x_2$ of the species outside the subnetwork $\gamma$ are regulated to a submanifold $\Sigma$, whose dimension is given by the influence index of the subnetwork, $\lambda(\gamma)$. For example, suppose that $\bm x_2$ is initially at point $P$. If we perturb a parameter inside $\gamma$, the state deviates from that point, but eventually end up in another point $P'$ inside $\Sigma$. 
 If we perturb another parameter in $\gamma$, the state deviate from $\Sigma$ temporarily, but it comes back to another point in $\Sigma$. See Sec.~\ref{sec:manifold-rpa} for details.
 }
  \label{fig:manifold-rpa-schematic} 
\end{figure}

If we view the phenomenon of RPA in a broader context, the fact that a system with an RPA property is characterized by a ``topological invariant'' should not be restricted to deterministic chemical reaction systems and would also be true for more generic systems. 
Since the ``robustness'' of adaptation does not allow the fine-tuning of reaction rates to achieve adaptation, the RPA conditions we obtain should be insensitive to deformations of the system, such as changes in parameters or modifications of rate functions.
Indeed, the condition of a zero index (or output-completeness) is \emph{topological}, in the sense that it is insensitive to these deformations of reaction systems. 
Generically, one can expect that the class of systems
with an RPA property, which does not necessarily have to be deterministic chemical reaction systems, should be of this feature. 
Let us summarize this claim as the \emph{``Robust Adaptation is Topological'' (RAT) principle},
by which we mean that the class of dynamical systems that exhibit robust adaptation is characterized by a topological invariant.
The RAT principle itself is not a concrete theorem, 
but rather a template of theorems, 
and different technical assumptions should be made 
depending on the nature of the systems under consideration. 
Namely, depending on technical details such as
the class of ``dynamical systems'' to be considered
or the choice of output variables, 
the corresponding ``topological invariant''
should be defined appropriately.
The law of localization~\cite{PhysRevLett.117.048101,PhysRevE.96.022322} 
and flux RPA~\cite{https://doi.org/10.48550/arxiv.2302.01270}
are particular realizations of the RAT principle for deterministic chemical reaction systems. 
The one-to-one correspondence of elementary RPA properties and labeled buffering structures indicates that topological characteristics can, in fact, exhaust all the RPA properties in these systems. 
This leads us to hypothesize that topological characterization is crucial for other classes of dynamical systems that exhibit RPA properties beyond deterministic chemical reaction systems.

Let us describe the organization of the paper. 
In Sec.~\ref{sec:crs}, after introducing basic notions on the description of deterministic chemical reaction systems, we specify the technical setting, including the definitions of generic and elementary RPA properties. 
In Sec.~\ref{sec:LoL}, we introduce a formalism to describe the response of steady states to parameter perturbations, and review topological approaches toward RPA.
In Sec.~\ref{sec:control}, we give a review of a control theoretical approach toward RPA, and we review the theory of maxRPA. 
In Sec.~\ref{sec:rpa-lbs}, we show that all the elementary RPA properties in a network can be represented by labeled buffering structures, and generic RPA properties are generated by them. We also give a numerical algorithm for the identification of all the labeled buffering structures for a given chemical reaction network. We illustrate the method through several examples. 
In Sec.~\ref{sec:equivalence}, we discuss the topological characterization of kinetics-independent maxRPA, and interpret this result via the law of localization. 
Hinted by this connection, in Sec.~\ref{sec:integrator-general}, we describe how to construct integral feedback control for a given buffering structure.
In Sec.~\ref{sec:manifold-rpa}, we discuss the phenomenon of manifold RPA, and provide a sufficient topological condition for this to be realized.
In Sec.~\ref{sec:rat}, we discuss the naturality of a topological characterization of systems with RPA properties and propose the RAT principle. 
Section~\ref{sec:conclusion} is devoted to a conclusion and future outlook. 

\section{ Chemical reaction systems and robust perfect adaptation }\label{sec:crs}

Various processes within biological cells can be viewed and modeled as chemical reaction networks.
In this section, we introduce the description of generic deterministic chemical reaction systems. 
Then, we introduce notions for discussing the phenomenon of robust perfect adaptation along with basic technical assumptions.

\subsection{ Chemical reaction systems } \label{sec:crns_intro}

We consider a chemical reaction network 
$\Gamma= (V, E)$, which consists of a set $V$ of chemical species 
and a set $E$ of chemical reactions. 
A reaction $e_A \in E$ can be specified as
\begin{equation}
  e_A :
 \sum_i s_{iA} v_i \to 
 \sum_i t_{iA} v_i , 
\end{equation}
where $v_i \in V$, and $s_{iA},t_{iA} \in \mathbb Z_{\ge 0}$ are stoichiometric constants of species $v_i$ for reaction $e_A$. 
We denote the stoichiometric matrix by 
$S$, whose components are given by $S_{iA} \coloneqq t_{iA} - s_{iA}$. 
The total numbers of chemical species and reactions will be denoted by
\begin{equation}
M \coloneqq |V| , 
\quad 
N \coloneqq |E| ,
\end{equation}
where $|\ldots|$ of a set indicates its cardinality.

When we consider a dynamical system based on a given reaction network $\Gamma$, there are different levels of description.
If the number of molecules is large, we can model the system
deterministically, and in that case, the system variables
are the concentrations $x_i$ of chemical species $v_i \in V$.
On the other hand, if the number of molecules is small, stochastic nature of the system becomes important, and then the system has to be modeled as a stochastic reaction system \cite{anderson2015stochastic} governed by a chemical master equation.
In the present paper, we focus on the former situation and consider deterministic dynamics.
The time evolution of the concentrations is governed by the rate equation\footnote{We use bold fonts to indicate vectors. Here, the components of a vector $\bm x$ are given by $x_i$.}, 
\begin{equation}
  \frac{d}{dt} \bm x (t) =  S \bm r , 
  \label{eq:rate}
\end{equation}
where $r_A$ is the reaction rate of reaction $e_A \in E$.
To solve the rate equations, we need to express the reaction rates, $r_A$, as functions of reactant concentrations, $\bm x$, and parameters, $k_A$, i.e., $r_A = r_A (\bm x; k_A)$. Such a choice is called {\it kinetics}. 
For example, in the case of mass-action kinetics, the rate function of a reaction $e_A$ is proportional to the product of reactant concentrations,
\begin{equation}
    r_A ( \bm x; k_A) = k_A \prod_i x_i^{s_{iA}}. 
    \label{defn:mass_action_kin}
\end{equation}
In the present work, we {\it do not assume any specific kinetics} unless otherwise stated. Rather, we seek properties that do not depend on the choice of kinetics.

Since we are interested in RPA, we consider a situation where the system reaches an asymptotically-stable steady state in the long-time limit as we discuss in more detail in Sec.~\ref{sec:assumptions}. The steady-state solution can be obtained by solving the following equations,\footnote{
Here we use different characters for 
indices of 
chemical species ($i, j, \ldots$), 
reactions ($A,B,\ldots$), 
the basis of the kernel of $S$ ($\alpha, \beta, \ldots$), and that of cokernel of $S$ ($\bar\alpha, \bar\beta, \ldots$).
We use the notation where $|i|$ indicates the number of values the index $i$ takes. 
Thus, $|i|$ indicates the number of chemical species,
$|A|$ indicates the number of reactions,
$|\alpha|$ denotes the dimension of $\ker S$, 
and $|\bar\alpha|$ denotes the dimension of $\coker S$. 
}
\begin{align}
\sum_A  S_{iA}  \bar r_A ( \bar{\bm x} (\bm k, \bm \ell); k_A ) &= 0, \label{eq:sr} \\
 \sum_i d^{(\bar \alpha)}_i \bar x_i (\bm k, \bm \ell)
 &= \ell^{\bar\alpha},   \label{eq:l-dx} 
\end{align} 
where $\bar x_i$ and $\bar r_A$ indicate the steady-state values of 
concentration $x_i$ and reaction rate $r_A$, respectively.
Note that we here allow the stoichiometric matrix to have left null vectors, and the system can have conserved quantities.
The set of vectors $\{\bm d^{(\bar\alpha)} \}_{\bar\alpha=1,\ldots,|\bar\alpha|}$ is a basis of $\coker S$, 
and we need to specify the values of conserved quantities 
as in Eq.~\eqref{eq:l-dx} to obtain the steady-state solution (if $\coker S$ is nontrivial).
The set $(\bm k, \bm \ell) \in K \times L$ specifies the parameters of this chemical reaction system, where $\bm k$ is the vector of all reaction parameters, $\bm \ell$ is the vector of the values of all the conserved quantities, and $K$ and $L$ are open sets of all admissible parameter vectors and values of conserved quantities, respectively.

\subsection{ Robust perfect adaptation } \label{sec:assumptions}

Let us now give the definition of robust perfect adaptation and present basic assumptions. 

We say that a deterministic chemical reaction system exhibits {\it perfect adaptation} in species $v_i \in V$ with respect to a certain parameter $p \in (\bm k, \bm \ell)$
when $x_i(t)$ asymptotically reaches a steady-state value that is independent of $p$, i.e.
\begin{equation}
\lim _{t \to \infty} x_i(t) = \bar x_i ,
\label{eq:lim-x-bar-x}
\end{equation}
where $\bar x_i$ is independent of $p$.
Similarly, we say that the system exhibits perfect adaptation in reaction $e_A$ with respect to $p$ if the corresponding reaction rate $r_A$ reaches a value independent of $p$ asymptotically.
A perfect adaptation property is said to be {\it robust}, if the adaptation occurs without fine-tuning of system parameters.
In this paper, we are concerned with such {\it robust perfect adaptation (RPA)} realized in deterministic chemical reaction systems. 
If the system exhibits RPA in $v_i \in V$ with respect to $p$, we will also say that the system has an RPA property in $v_i$ with respect to $p$. 
Observe that, if the system exhibits perfect adaptation in $v_i \in V$, 
Eq.~\eqref{eq:lim-x-bar-x} implies that the steady-state concentration 
is unaffected by constant-in-time disturbances of parameter $p$, 
\begin{equation}
   p \mapsto p + \delta p . 
\end{equation}
Namely, the system has the ability to reject the effect of disturbances in $p$ on $x_i$.\footnote{ We will discuss the control-theoretical aspects of this phenomenon from Sec.~\ref{sec:control}. }

As a most general case of RPA, a subset $\mathcal V \subset V$ of species and a subset $\mathcal E \subset E$ of reactions can exhibit RPA with respect to a subset of system parameters $\bm p \subset (\bm k, \bm \ell)$. Thus, a {\it generic RPA property} of a system is characterized by these three ingredients, $(\mathcal V, \mathcal E, \bm p)$.

How can we characterize all such generic RPA properties in a given chemical reaction system? 
In a deterministic chemical reaction system, 
the property of a steady state is specified by the values of concentrations and reaction rates at the steady state, $\bar{x}_i$ and $\bar{r}_A$.
In general, the steady-state concentration $\bar x_i$ of species $v_i \in V$
is a function of a subset of system parameters $(\bm k, \bm \ell)$. 
Namely, concentration $x_i$ exhibits RPA under the perturbation of the parameters on which $\bar x_i$ does not depend.
For example, if the steady-state concentration of $v_1$ depends only on $k_1$ and $k_3$, we can write it as $\bar x_1 (k_1, k_3)$, and hence it will exhibit RPA under the perturbation of $k_2, k_4$, or other system parameters except for $k_1$ and $k_3$. 
If we identify the dependencies of all the concentrations and reaction rates on all the parameters, we will have a complete characterization of the RPA properties existing in the system. 

Equivalently, we can perform the perturbation $p \mapsto p + \delta p$ of a chosen parameter $p \in (\bm k, \bm \ell)$, and identify the concentrations and reaction rates whose steady-state values are affected and unaffected by this perturbation.
The unaffected ones exhibit RPA with respect to $p$. 
We call this separation of degrees of freedom whose steady-state values are affected and unaffected by $p$ as {\it an elementary RPA property with respect to parameter $p$}. 
Once we identify elementary RPA properties with respect to all the system parameters, we have a complete characterization of the RPA properties in the system, since a generic RPA property can be obtained by combinations of elementary ones.

In this work, we allow the reaction kinetics to be generic. While mass-action kinetics has often been assumed in prior studies, there exist multiple factors that can undermine this assumption. Moreover, experimental verification of its validity within cellular environments is highly challenging. Accordingly, the RPA properties that we discuss here are of the following nature: 
\begin{definition}[Kinetics-independent RPA]
An RPA property is said to be {\it kinetics-independent} if 
the adaptation occurs regardless of the choice of kinetics 
as long as the stability of the system is maintained. 
\end{definition}\label{def:kinetics-indep-rpa}
For example, even if we observe an RPA property with a system with mass-action kinetics, it may disappear once we deform the kinetics 
away from mass-action kinetics. 
The RPA properties we discuss here are those that survive such deformations of a system.

Let us comment on the nature of perturbations under consideration. 
In this paper, we consider the constant-in-time disturbances of 
reaction parameters and the values of conserved quantities, 
\begin{equation}
k_A 
\mapsto 
k_A + \delta k_A, 
\quad 
\ell_{\bar\alpha} 
\mapsto
\ell_{\bar\alpha} + \delta \ell_{\bar\alpha} . 
\label{eq:pert-k-l}
\end{equation}
When a reaction rate has multiple parameters, 
the perturbation $k_A \mapsto k_A + \delta k_A$ indicates that one of them is disturbed.
Since we allow the kinetics to be a generic one, the perturbation of a reaction parameter, $ k_A \mapsto k_A + \delta k_A$, can be regarded as an arbitrary one-parameter deformation of the reaction rate function.
There can be situations where parameters of several reactions are simultaneously perturbed, for example, through the change of temperatures. 
Once we know the response to the change of each parameter, 
we can use this information to compute the response
under, for example, a temperature variation
by considering the temperature dependence of each parameter, $k_A (T)$.

Let us summarize the assumptions we make throughout the paper.
As we are interested in the phenomenon of RPA, we are here concerned with the situation where a steady state is realized eventually, after adding a constant-in-time perturbation to the system. 
Accordingly, we make the following assumption:
\begin{assm}[Stability]\label{assm:stability}
The reaction system has an asymptotically stable steady state. 
In other words, the Jacobian matrix\footnote{
Here, we mean the Jacobian matrix of a reaction system where the values of conserved quantities are fixed. Thus, there is no zero eigenvalue associated with conserved quantities.
} at the steady state is non-singular and every eigenvalue has a strictly negative real part.
\end{assm}
We also assume the following,
\begin{assm}
The steady-state concentration is nonzero, $\bar x_i (\bm k, \bm \ell) > 0$, 
for any $v_i \in V$, and 
\begin{equation}
\left(
\frac{\p r_A (\bm x ; k_A)}{\p x_i} 
\right)_{\bm x = \bar{\bm x}(\bm k, \bm \ell)} 
\begin{cases}
\neq 0 & \quad \text{if $v_i \in V$ is a reactant of reaction $e_A$} \\
= 0 & \quad \text{otherwise}
\end{cases}. 
\label{eq:drdx-neq-0}
\end{equation}
\end{assm}

\section{ RPA from a topological perspective }  \label{sec:LoL}

We here introduce a basic formalism for characterizing the responses 
of steady states to parameter perturbations. 
Then we introduce the law of  localization~\cite{PhysRevLett.117.048101,PhysRevE.96.022322}, 
that gives us a sufficient condition for RPA based on the network topology.

\subsection{ Notations } 

In this subsection, let us summarize the notations we will use in the following. 
Generically, we will use overlines to indicate quantities at steady state, 
such as $\bar {\bm x} (\bm k, \bm \ell)$ and $\bar{\bm r} ( \bar{\bm x} (\bm k, \bm \ell) ; k_A)$. 
The arguments of a function may be omitted for simplicity. 
The derivative of a steady-state concentration with respect to 
rate parameter $k_B$ 
and the value $\ell^{\bar\alpha}$ of a conserved quantity 
will be denoted as 
\begin{equation}
\p_B \bar x_i
= 
\bar x_{i, B} 
\coloneqq 
\frac{\p}{\p k^B} \bar x_i (\bm k, \bm \ell)
, 
\quad \quad 
\p_{\bar\alpha} \bar x_i
= 
\bar x_{i, {\bar\alpha}} 
\coloneqq 
\frac{\p}{\p \ell^{\bar\alpha} } \bar x_i (\bm k, \bm \ell)
\end{equation}
As for reactions, note that each steady-state reaction rate,
$\bar r_A (\bar{\bm  x} (\bm k, \bm \ell) ; k_A)$,
has an implicit dependence on parameters $(\bm k, \bm \ell)$ through $\bar {\bm x}$, as well as an explicit dependence on a particular parameter $k_A$. 
In the following expressions, we mean the derivative of steady-state reaction rates with respect to parameter $k_B$ 
including the both contributions of implicit and explicit dependence, 
\begin{equation}
\p_{B} \bar r_A 
= 
\bar r_{A,B} 
\coloneqq 
\frac{\p}{\p k_B} 
\bar r_A (\bar{\bm  x} (\bm k, \bm \ell) ; k_A) ,
\end{equation}
and similarly for the derivative with respect to $\ell^{\bar\alpha}$. 
On the other hand, the following expression means the derivative of 
the steady-state reaction rate with respect to $k_B$ 
only through the explicit dependence, 
\begin{equation}
\p_{B} r_A
=
r_{A,B}
\coloneqq 
\left( 
\frac{\p}{\p k_B} r_A (\bm x  ; k_A) 
\right)_{\bm x = \bar{\bm  x}(\bm k, \bm \ell)} . 
\end{equation}
Note that this quantity is nonzero only when $A=B$, 
i.e., $r_{A,B} \propto \delta_{AB}$, where $\delta_{AB}$ 
denotes the Kronecker delta. 
The following expressions indicate the derivative of 
the rate function with respect to a reactant concentration, 
that is evaluated at steady state, 
\begin{equation}
\frac{\p r_A}{\p x_i}
= 
\p_i r_A
=
r_{A,i} 
\coloneqq 
\left( 
\frac{\p
}{\p x_i}
r_A(\bm x ; k_A)
\right)_{\bm x = \bar{\bm x}(\bm k, \bm \ell)}. 
\end{equation}

In the following, we will often choose a subnetwork $\gamma = (V_\gamma, E_\gamma)\subset \Gamma = (V, E)$. 
The degrees of freedom inside $\gamma$ will be denoted with indices with stars, 
while those outside $\gamma$ will be denoted with indices with primes. 
For example, $x_{i^\star}$ 
denotes the concentration of a species inside $\gamma$ 
(i.e., $v_{i^\star} \in V_\gamma$),
while 
$x_{i'}$ denotes the concentration of a species outside $\gamma$, 
$v_{i'} \in V \setminus V_\gamma$. 
Similarly, $r_{A^\star}$ denotes 
the reaction flux of 
a reaction inside $\gamma$, 
$e_{A^\star} \in E_\gamma$, 
while 
$r_{A'}$ denotes the reaction flux of 
a reaction outside $\gamma$, 
$e_{A'} \in E \setminus E_\gamma$.

\subsection{ Response of steady states }

Let us now discuss the response of steady states against parameter perturbations. 
Equation~\eqref{eq:sr} indicates that the steady-state reaction rates are in the kernel of the stoichiometric matrix $S$. Hence, the rates can be written as 
\begin{equation}
  \bar r_A (\bar{\bm x} (\bm k, \bm \ell); k_A ) 
  = \sum_\alpha \mu_\alpha (\bm k, \bm \ell) c^{(\alpha)}_A ,
  \label{eq:r-mu-c}
\end{equation}
where $\{\bm c^{(\alpha)}\}_{\alpha=1,\ldots ,|\alpha|}$ 
is a basis of $\ker S$. 
Taking the derivative of 
Eqs.~(\ref{eq:r-mu-c}) and (\ref{eq:l-dx}) 
with respect to $k_B$ and $\ell^{\bar\beta}$, 
we have
\begin{align}
 \sum_i \frac{\p r_A}{\p x_i} 
  \frac{\p \bar x_i}{\p k_B} 
  + 
  \frac{\p r_A}{\p k_B}
  &=  \sum_\alpha \frac{\p \mu_\alpha }{\p k_B} c^{(\alpha)}_A  , 
  \label{eq:der-1}
  \\
  \sum_i \frac{ \p r_A }  {\p x_i } 
 \frac{\p \bar x_i}{\p \ell^{\bar\beta }} 
 &=  
 \sum_\alpha  \frac{ \p \mu_\alpha }  {\p \ell^{\bar\beta }} c^{(\alpha)}_A , \\
\sum_i d_i^{(\bar\alpha)} \frac{\p \bar x_i}{\p k_B} &= 0,  
 \\
\sum_i  d_i^{(\bar\alpha)} \frac{\p \bar x_i}{\p \ell^{\bar\beta}}
 &=  \delta^{\bar\alpha \bar\beta} . 
\label{eq:der-4}
\end{align}
As noted earlier, steady-state reaction rates $\bar r_A (\bar {\bm x} (\bm k, \bm \ell); k_A)$ have an explicit dependence on $k_A$, as well as an implicit dependence on $\bm k$ and $\bm \ell$ through steady-state concentrations,
$\bar x_i (\bm k, \bm \ell)$, 
and $\frac{\p r_A}{\p k_B}$ indicates the derivative with respect to the {\it explicit} dependence. 
Note also that $\frac{\p r_A}{\p x_i}$ is evaluated at steady state. 
Equations~\eqref{eq:der-1} -- \eqref{eq:der-4} 
can be simplified by introducing a matrix defined by\footnote{
The {\bf A}-matrix is first introduced in Ref.~\cite{MOCHIZUKI2015189} for 
the case $\coker S = \bm 0$ 
and it was extended to the case $\coker S \neq \bm 0$ 
in Ref~\cite{PhysRevE.96.022322}. 
}
\begin{equation}
  {\bf A} 
  \coloneqq 
  \left[
  \begin{array}{c|c}
    \p_i r_A 
    & - c^{(\alpha)}_A \\ \hline 
    d^{(\bar\alpha)}_i & \bm 0_{|\bar\alpha| \times |\alpha|}
  \end{array}
  \right]
  ,
  \label{eq:mat-a-def}
\end{equation}
where $\p_i \coloneqq \p / \p x_i$. 
The matrix ${\bf A}$ is $(|A| + |\bar\alpha|) \times (|i| + |\alpha|)$ dimensional, and 
because of the Fredholm's theorem, $|A| + |\bar \alpha| = |i| + |\alpha|$, it is square.
Using this, Eqs.~\eqref{eq:der-1}--\eqref{eq:der-4} are summarized compactly in the matrix form, 
\begin{equation}
  {\bf A} \,
  \p_B
  \begin{pmatrix}
  \bar {\bm x} \\
  {\bm \mu}
  \end{pmatrix}
  = - 
  \begin{pmatrix}
  \p_B  \bm r  \\
    \bm 0
\end{pmatrix}
, 
\quad 
{\bf A} \, 
\p_{\bar\beta}
\begin{pmatrix}
 \bar{\bm x} \\
 \bm \mu 
\end{pmatrix}
= 
\begin{pmatrix}
 \bm 0 \\
 \p_{\bar\alpha} \bm \ell 
\end{pmatrix}, 
\label{eq:sensitivity-1}
\end{equation}
where, as mentioned earlier,
$\p_B \coloneqq \partial /\partial k^B$ 
and 
$\p_{\bar\beta} \coloneqq \partial /\partial \ell^{\bar\beta}$.

We note that the matrix ${\bf A}$ is invertible when the stability condition \ref{assm:stability} is satisfied.
For the case $\coker S = \bm 0$, 
the invertibility of ${\bf A}$ results 
from the one-to-one correspondence of 
the eigenspectrum of the Jacobian matrix
and the generalized-eigenspectrum of ${\bf A}$.
In Appendix~\ref{sec:spectral-correspondence}, 
we extend this correspondence to the case $\coker S \neq \bm 0$, 
from which the invertibility of ${\bf A}$ follows even in the presence of nontrivial $\coker S$, as long as the stability condition \ref{assm:stability} is satisfied (see Corollary~\ref{cor:invertivility-of-a}). 
Thus, by multiplying ${\bf A}^{-1}$ on Eq.~\eqref{eq:sensitivity-1}, we have
\begin{equation}
 \p_B
 \begin{pmatrix}
  \bar{\bm x} \\
  \bm \mu
  \end{pmatrix}
  = - 
  {\bf A}^{-1}
  \begin{pmatrix}
  \p_B  \bm r \\
    \bm 0
\end{pmatrix}
, 
\quad 
\p_{\bar\beta}
\begin{pmatrix}
\bar{\bm x} \\
\bm \mu
\end{pmatrix}
=  {\bf A}^{-1}
\begin{pmatrix}
 \bm 0 \\
 \p_{\bar\alpha} \bm \ell 
\end{pmatrix}.  
\label{eq:sensitivity-2}
\end{equation}
Note that $\p_B r_A$ is a diagonal matrix, i.e., $\p_B r_A \propto \delta_{BA}$. 
If we partition ${\bf A}^{-1}$ as 
\begin{equation}
  {\bf A}^{-1} = 
  \begin{pmatrix}
    ({\bf A}^{-1})_{iA} & ({\bf A}^{-1})_{i \bar\alpha  } \\
    ({\bf A}^{-1})_{ \alpha A } & ({\bf A}^{-1})_{ \alpha \bar\alpha }  
  \end{pmatrix}, 
\end{equation}
the responses of steady-state concentrations and reaction rates to the perturbations of $k_B$ and $\ell^{\bar\beta}$ are 
proportional to the following components, 
\begin{equation}
  \p_B \bar x_i \propto  ({\bf A}^{-1})_{iB} , 
  \quad 
  \p_{\bar\beta} \bar x_i \propto  ({\bf A}^{-1})_{i \bar\beta } , 
  \quad 
  \p_{B} \mu_{\alpha} \propto ({\bf A}^{-1})_{\alpha B} ,
  \quad 
  \p_{\bar\beta} \mu_{\alpha} \propto ({\bf A}^{-1})_{\alpha \bar\beta }.
  \label{eq:response}
\end{equation}

Let us introduce notations with which 
the sensitivity can be expressed in a concise manner. 
We organize the chemical concentrations and $\mu^\alpha$ 
as a single vector\footnote{We use $\nu,\rho,\sigma, \ldots$ to denote the index of the vector $\bar y$. }, 
\begin{equation}
\bar y^{\nu} \coloneqq
  \begin{pmatrix}
   \bar x^i \\
   \mu^\alpha
  \end{pmatrix}. 
  \label{eq:y-def}
\end{equation}
We also denote the parameters $k^A$
and the values $\ell^{\bar\alpha}$ 
of conserved quantities collectively, 
\begin{equation}
 q^{\nu} \coloneqq 
  \begin{pmatrix}
   k^A \\
   \ell^{\bar \alpha}  
  \end{pmatrix}. 
\label{eq:q-def}
\end{equation}
Note that the vectors $\bar y^\mu$ and $q^\mu$ have the same dimensions, $|\nu| = |A| + |\bar \alpha| = |i| + |\alpha|$. 
Let us also introduce 
\begin{equation}
 L^\nu
  \coloneqq
  \begin{pmatrix}
   r^A \\
   -\ell^{\bar \alpha} 
   \end{pmatrix} . 
\end{equation}
The vector $L^{\nu}$ depends on both variables $y$
and parameters $q$, 
$L^{\nu} = L^{\nu }( y, q)$. 
Its derivative with respect to the parameters 
is represented as a matrix of the the following form,
\begin{equation}
 \frac{ \p L^{\nu} }{\p q^{\rho }}
  =
  \begin{pmatrix}
   \frac{\p r^A}{\p k^B} & \bm 0 \\
   \bm 0 &- \delta^{\bar \alpha \bar \beta}
  \end{pmatrix} . 
\end{equation}
With these notations, 
we can express Eq.~\eqref{eq:sensitivity-1} as a single equation, 
\begin{equation}
\sum_{\rho} 
{\bf  A}_{\nu \rho } \, \frac{\p \bar y^{\rho}}{\p q^{\sigma}} 
 = 
 - \frac{\p L^{\nu}}{\p q^{\sigma}} . 
 \label{eq:a-yq-r} 
\end{equation}

\subsection{ Law of localization }

In this subsection, we introduce the law of localization~\cite{PhysRevLett.117.048101,PhysRevE.96.022322}, 
which gives us a sufficient topological condition 
so that the effect of perturbations inside a subnetwork is localized inside it. 

Let us choose a subnetwork $\gamma$, 
which is specified by subsets of chemical species and reactions, 
$\gamma = (V_\gamma, E_\gamma)$. 
A subnetwork $\gamma$ is called {\it output-complete} 
if $E_\gamma$ includes all the reactions 
whose reactants are in $V_\gamma$. 
For a given output-complete subnetwork $\gamma$, 
its {\it influence index} is defined by 
\begin{equation}
  \lambda(\gamma)
  \coloneqq 
-  |V_\gamma |
 + 
  |E_\gamma |
  - 
  |(\ker S)_{{\rm supp\,} \gamma }|
  + 
  | P^0_\gamma (\coker S) | . 
  \label{eq:index-def} 
\end{equation}
The definitions of the spaces 
that appear in the influence index 
are given as follows: 
\begin{eqnarray}
 (\ker S)_{{\rm supp}\,\gamma}
 &\coloneqq &
 \left\{
  \bm c \in \ker S
  \, \middle| \,
  P^1_\gamma \bm c = \bm c 
  \right\}, 
   \\
   P^0_\gamma (\coker S) 
   &\coloneqq &
   \left\{
   P^0_\gamma \bm d 
   \, \middle| \,
   \bm d \in \coker S 
   \right\}, 
\end{eqnarray}  
where $S$ is the stoichiometric matrix, 
$P^0_\gamma$ and $P^1_\gamma$
are the projection matrices to $\gamma$
in the space of chemical species and reactions, respectively. 
Namely, $(\ker S)_{{\rm supp}\,\gamma}$
is the space of vectors of $\ker S$ 
supported inside $\gamma$, 
and $P^0_\gamma (\coker S)$
is the projection of $\coker S$ to $\gamma$.
We note that the influence index of an output-complete subnetwork is nonnegative under the assumption of stability~\eqref{assm:stability}.

The statement of the law of localization is as follows:\footnote{
In Ref.~\cite{PhysRevLett.117.048101}, $\coker S = \bm 0$ is assumed. The analysis is generalized to the case with $\coker S \neq \bm 0$ in Ref.~\cite{PhysRevE.96.022322}. 
}
\begin{theorem}[Law of localization]  \label{th:lol}
Let $\gamma \subset \Gamma$ be an output-complete subnetwork of 
a deterministic chemical reaction system
satisfying the assumptions in Sec.~\ref{sec:assumptions}. 
When $\gamma$ satisfies $\lambda(\gamma)=0$, 
the steady-state values of 
chemical concentrations and reaction rates 
outside $\gamma$ do not change 
under the perturbation of rate parameters or 
conserved quantities with nonzero support\footnote{
Let us clarify precise meaning of this expression. 
For a given subnetwork, we first take basis $\{ \bm d^{(\bar\alpha')} \}_{\bar \alpha' = 1, \ldots, |\bar\alpha'|}$ of those vectors in $\coker S$ that have support only in $V \setminus V_\gamma$. We then extend this to a basis for the whole $\coker S$ by including vectors $\{ \bm d^{(\bar\alpha^\star)} \}_{\bar \alpha^\star = 1, \ldots, |\bar\alpha^\star|}$ that have nonzero support inside $\gamma$. Therefore if $P_\gamma^0$ denotes the projection matrix to $\gamma$, 
we have 
$P_\gamma^0 \bm d^{(\bar\alpha^\star)} \neq \bm 0$ 
while 
$P_\gamma^0 \bm d^{(\bar\alpha')} = \bm 0$. Note that the vectors $\{
P^0_\gamma \bm d^{\bar\alpha^\star}
\}_{
\bar\alpha^\star = 1, \ldots, |\bar\alpha^\star|
}$
will be linearly independent by construction.
The perturbed parameter $\ell^{\bar\alpha^\star}$ is associated with any basis constructed this way. 
}
inside $\gamma$.
Namely, 
we have 
$\p_{B^\star} \bar x_{i'} = 0$, 
$\p_{B^\star} \bar r_{A'} = 0$,
$\p_{\bar\alpha^\star} \bar x_{i'} = 0$, 
and 
$\p_{\bar\alpha^\star} \bar r_{A'} = 0$
for any 
$e_{B^\star} \in E_\gamma$, 
$v_{i'} \in V \setminus V_\gamma$, 
$e_{A'} \in E \setminus E_\gamma$, 
and 
$\p_{\bar\alpha^\star} 
= \frac{\p}{\p \ell^{\bar\alpha^\star}} 
$
is the derivative with respect to the value of a conserved quantity with nonzero support in $\gamma$. 
\end{theorem}
For the proof of the theorem, see Refs.~\cite{PhysRevLett.117.048101,PhysRevE.96.022322,PhysRevResearch.3.043123}.

This theorem tells us that the effect of perturbation is confined inside 
a subnetwork with $\lambda (\gamma)=0$ for steady states,
and for this reason, a subnetwork with a zero influence index is called a {\it buffering structure}. 

We note that the law of localization can in fact be interpreted as a form of RPA: the concentrations and reaction rates outside a buffering structure $\gamma$ exhibit RPA with respect to the perturbation of parameters inside $\gamma$. In other words, the law of localization provides us with a sufficient condition for RPA.

In a similar spirit, the structural condition for the reaction fluxes to exhibit RPA has been identified. 
It was shown that, by using a slightly different index, which is called the {\it flux influence index}, $\lambda_{\rm f}(\gamma)$, we can find reaction parameters
under the perturbation of which {\it all} the reaction fluxes 
exhibit RPA~\cite{https://doi.org/10.48550/arxiv.2302.01270}. 
A subnetwork with $\lambda_{\rm f}(\gamma) = 0$ is called 
a {\it strong buffering structure}. 
The flux influence is related to the influence index by 
$\lambda_{\rm f}(\gamma) = 
\lambda(\gamma) + |(\ker S)_{{\rm supp}\,\gamma}|.$
A strong buffering structure is always a buffering structure, 
i.e.\ 
$\lambda_{\rm f}(\gamma) =0$ implies $\lambda(\gamma)=0$, 
which is obvious from the relation between them.

Buffering structures are closed under union and intersection: 
namely, if $\gamma_1$ and $\gamma_2$ are buffering structures, 
so are $\gamma_1 \cap \gamma_2$ and $\gamma_1 \cup \gamma_2$.
This can be proven easily using the submodularity 
of the influence index~\cite{PhysRevResearch.3.043123}. 
The strong buffering structures also are closed
under union and intersection~\cite{https://doi.org/10.48550/arxiv.2302.01270}.

\section{RPA from a control-theoretic perspective}\label{sec:control} 

In this section, we discuss RPA from a control-theoretic perspective.
After introducing a general idea, we discuss maxRPA networks in which one distinguished species exhibit RPA with respect to maximally many parameters in a reaction system.
Then, we discuss how a buffering structure naturally arises for kinetics-independent maxRPA networks. In fact, the existence of such a buffering structure will be shown to be equivalent to the conditions for kinetics-independent maxRPA. 
As we prove in Section \ref{sec:rpa-lbs}, this equivalence goes beyond maxRPA networks, and it extends to all generic RPA properties.

\subsection{RPA and integral feedback control} \label{sec:rpa_int_ctr}

The biological notion of RPA is essentially equivalent to the notion of robust steady-state tracking, that is well known in control theory. This allows us to borrow control-theoretic concepts to understand the structural conditions for networks to exhibit RPA. One such concept is the famous Internal Model Principle (IMP)~\cite{bin2022internal}, that sheds light into the organization of RPA networks. 
In particular, it says that an RPA network $\Gamma$ must contain a subsystem called the \emph{internal model} (IM), that can generate the class of disturbances to which the RPA network adapts.
Furthermore, the IM generates the disturbance using only the regulated output variable (i.e. concentration of $X$) as input, and then passes restorative signals to the \emph{rest of the network} (RoN) in order to eliminate the effect of the disturbance. 
When the output species $X$ belongs to RoN, this creates a natural ``feedback'' between IM and RoN (see Fig.~\ref{fig:imp}). 
Since we only consider constant-in-time perturbations of system parameters as disturbances, the function realized by the IM must be an ``integrator'' which computes the time-integral of the deviation of the output species concentration from its set-point (see Ref.~\cite{gupta2023internal} for more details). 

There can be several instantiations of the IMP depending on the type of disturbances allowed and the nature of the robustness desired, and each version needs to be independently proven. So far the IMP has been shown in complete generality for linear systems \cite{francis1975internal,francis1976internal} and for nonlinear systems that can be decomposed as affine functions of the disturbances~\cite{SONTAG2003119} (see Ref.~\cite{bin2022internal} for more details). It is important to note that typically RPA systems are not naturally structured into the IMP-mandated form (i.e.\ IM and RoN in feedback), but IMP asserts that they can be brought in this form after a nonlinear coordinate transformation (see the incoherent feedforward (IFF) example in Ref.~\cite{gupta2023internal}).

\begin{figure}[t]
  \centering
  \includegraphics
  [clip, trim=4cm 14cm 4cm 0cm, scale=0.54]
  {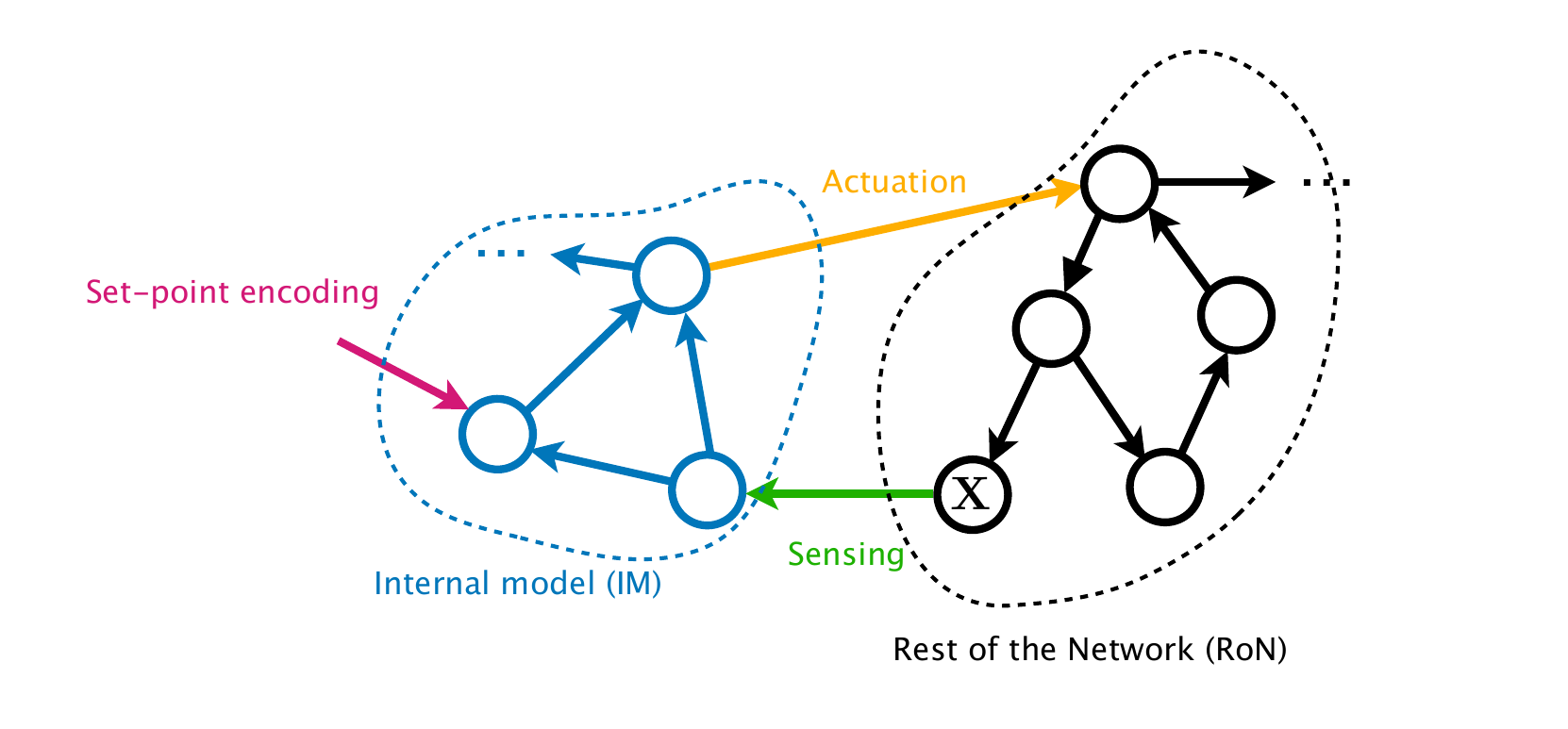} 
 \caption{The RPA network can be decomposed based on the Internal Model Principle (IMP), employing a potential change of coordinates. This decomposition results in two components: an Internal Model (IM) and the Rest of the Network (RoN). By utilizing a feedback mechanism, robustness against disturbances is achieved, with the IM playing a crucial role. The IM actively measures disturbances by evaluating the deviation of the regulated output $\mathbf{X}$ from the desired set-point, and subsequently transmits corrective signals to the RoN, effectively countering the disturbances.}
  \label{fig:imp} 
\end{figure}

\subsection{ Characterization of maxRPA networks}\label{sec:maxrpa_charac}

Here, we discuss maxRPA networks~\cite{gupta2022universal} in detail and present their precise mathematical characterization, which generalizes the characterization result in Ref.~\cite{gupta2022universal}. 

Consider a reaction network $\Gamma =(V, E)$ with $V= \{v_1,\dots,v_M\}$ consisting of $M$ species and $E = \{e_1,\dots, e_N\}$ consisting of $N$ reactions. Suppose that this network regulates its last species $v_M = X$, which we shall call the output species. In particular, there is a set-point at which the concentration of $X$ is maintained.
We say that a network satisfies the maxRPA property if the steady-state concentration of the output species $v_M = X$ only depends on the parameters $k_{\bar{1}}$ and $k_{\bar{2}}$ of the last two reactions $e_{\bar 1} \coloneqq e_{N-1}$ 
and 
$e_{\bar 2} \coloneqq e_{N}$. This implies that there is a function $\phi_{\rm out} $ such that
\begin{equation}
\label{max_rpa_defn}
\lim _{t \to \infty} x(t) = 
\phi_{\rm out} \left(k_{\bar 1}, k_{\bar 2} \right)
\quad \text{for any} \quad (\bm k, \bm \ell) \in K \times L,
\end{equation}
where $x(t)$ denotes the concentration of $X$ at time $t$.
Here the function $\phi_{\rm out}$ provides an encoding between parameters $k_{\bar 1}$ and $k_{\bar 2}$, and the set-point for the maxRPA species. We assume that this function depends non-trivially on both $k_{\bar 1}$ and $k_{\bar 2}$, i.e.\ its partial derivative with respect to both parameters is nonzero. 

When the last two reaction have mass-action kinetics (see Eq.~\eqref{defn:mass_action_kin}), it can be shown that $\phi_{\rm out} $ can only depend on the ratio of $k_{\bar 1}$ and $k_{\bar 2}$ (see Ref.~\cite{gupta2022universal}), but we do not assume mass-action kinetics here and so $\phi_{\rm out}$ is arbitrary. Observe that Eq.~\eqref{max_rpa_defn} implies that the steady-state output concentration is unaffected by constant-in-time disturbances that perturb parameters in $(\bm k, \bm \ell)$, except $k_{\bar 1}$ and $k_{\bar 2}$. 

The characterization result for maxRPA networks in Ref.~\cite{gupta2022universal} assumed that the cokernel of the stoichiometric matrix $S$ is trivial, i.e.\ $\coker S = \bm 0$.
The extension of this characterization result, which we now state, will relax this assumption along with the assumption that reactions $e_{\bar 1}$ and $e_{\bar 2}$ have the mass-action kinetics.
\begin{theorem}[maxRPA characterization result]
\label{thm:maxRPA_characterization}
Assuming that the chemical reaction system $\Gamma =(V, E)$ is stable (see Assumption~\ref{assm:stability}), it exhibits maxRPA for the output species $v_M = X$ if and only if the following two conditions are satisfied:
\begin{enumerate}
\item 
There exists a vector $\bm q \in \mathbb R^{M}$ and a positive $\kappa$ such that 
\begin{equation}
\bm q^\top S =  
\begin{bmatrix}
0, & \cdots,  0, & \kappa, & -1
\end{bmatrix}.
\label{eq:qs}
\end{equation}
If such a pair $({\bm q}, \kappa)$ exists then the value of $\kappa$ is unique and the value of vector $\bm q$ is unique up to addition of vectors in $\coker S$. 

\item The ratio of the rate functions $r_{\bar 1}(\bm x; k_{\bar 1})$ and $r_{\bar 2}(\bm x; k_{\bar 2})$ for the last two reactions, depends on only the output species concentration $x_M$ along with parameters $k_{\bar 1}$ and $k_{\bar 2}$, i.e.\ there exists a function $\Phi(x_M, k_{\bar 1}, k_{\bar 2})$ such that
\begin{align}
\label{max_rpa_ratio_cond}
\frac{r_{\bar 2}(\bm x; k_{\bar 2})}{r_{\bar 1}(\bm x; k_{\bar 1})} = \Phi(x_M, k_{\bar 1}, k_{\bar 2}).
\end{align}
\end{enumerate}
Moreover if these conditions hold then the set-point $\bar{x}_M $ for the output species is uniquely determined by the implicit relation
\begin{align}
\label{max_rpa_gen_set_point}
 \Phi(\bar{x}_M , k_{\bar 1}, k_{\bar 2}) = \kappa
 \end{align}
where $\kappa$ is the same constant as in Eq.~\eqref{eq:qs}. In other words, there is a unique set-point encoding function $\phi_{\rm out} $ such that
\begin{align}
\label{max_rpa_gen_set_point_2}
 \Phi(\phi_{\rm out} \left(k_{\bar 1}, k_{\bar 2} \right), k_{\bar 1}, k_{\bar 2}) = \kappa.
 \end{align}

\end{theorem}
We present a proof in Appendix \ref{proof:maxrpageneralizattion}.

Note that the first condition~\eqref{eq:qs} is exactly the same as in the case with $\coker S = \bm 0$. 
When $\coker S \neq \bm 0$, the vector $\bm q$ (if it exists) will not be unique, as for any conservation relation 
${\bm d} \in \coker S$, ${\bm d}^\top S = {\bm 0}$ and so $(\bm q+\bm d)^\top S= \bm q^\top S$. However, as will be shown in the proof of Theorem \ref{thm:maxRPA_characterization}, the value of $\kappa$ is always unique, and the value of ${\bm q}$ in unique in the quotient subspace $\mathbb R^{M} / \coker S$.
Therefore, a vector $\bm q \in (\coker S)^{\perp}$ such that $(\bm q, \kappa)$ satisfies Eq.~\eqref{eq:qs} is going to be unique\footnote{
Let us comment on how such a vector $\bm q$ can be identified, along with $\kappa$.
Let us pick a basis $\{\bm d^{(\bar\alpha)} \}_{\bar\alpha=1,\ldots,|\bar\alpha|}$ for $\coker S$, and define the $M \times |\bar\alpha|$ matrix $D$ by horizontally stacking these basis vectors as columns. The condition $\bm q \in (\coker S)^{\perp}$ is equivalent to $\bm q^{\top} D = \bm 0$, and hence $(\bm q, \kappa)$ satisfies the augmented linear system
\begin{equation}
\bm q^\top \begin{bmatrix}S & D\end{bmatrix} =  
\begin{bmatrix}
0, & \cdots,  0, & \kappa, & -1,& 0, & \cdots,  0
\end{bmatrix},
\label{eq:qs2}
\end{equation}
where $\begin{bmatrix}S & D\end{bmatrix}$ is the $M \times (N+ |\bar\alpha|)$ augmented matrix formed by horizontally stacking $S$ and $D$, and the vector on the right has $(N-2)$ zeros at the start and $|\bar\alpha|$ zeros at the end. Observe that the rows of this augmented matrix $\begin{bmatrix}S & D\end{bmatrix}$ are independent because if $\bm d$ in any vector in $\coker  \begin{bmatrix}S & D\end{bmatrix}$ then we must have $\bm d^\top S = \bm 0$ and $\bm d^\top D = \bm 0$, which means that $\bm d$ is present in both $\coker S$ and $(\coker S)^{\perp}$, and that is only possible when $\bm d =0$. Let $\bm u_j$ be the $N$-dimensional vector whose $j$-th component is $1$ and the rest are zeros. Then Eq.~\eqref{eq:qs2} can be equivalently written as the linear system
\begin{equation}
\begin{bmatrix}
S^\top  \ & \ -\bm u_{N-1} \\
D^{\top} \ & \ {\bm 0}
\end{bmatrix}
\begin{bmatrix}
\bm q \\
\kappa
\end{bmatrix} = -\begin{bmatrix} \bm u_{N} \\ \bm 0 \end{bmatrix}.
\label{eq:qs3}
\end{equation}
The pair $(\bm q, \kappa)$ can be obtained by solving this linear equation.
}, where $(\coker S)^{\perp}$ denotes the orthogonal complement of $\coker S$. 

Observe that if we assume the mass-action kinetics for the last two reactions, then they can be written as
\begin{equation}
r_{\bar 1} (\bm x;k_{\bar 1} )= k_{\bar 1} m_1(\bm x_{\bar{M}} )x^{\nu_{\bar{1}}}_M  \quad \textnormal{and} \quad     r_{\bar 2}(\bm x;k_{\bar 2} ) = k_{\bar 2} m_2(\bm x_{\bar{M}} ) x^{\nu_{\bar{2}}}_M,    
\end{equation}
where $\bm x_{\bar{M}}$ is the concentration vector for all the species except $X$, $m_1(\bm x_{\bar{M}} )$ and $m_2(\bm x_{\bar{M}} )$ are monomials (see the definition of mass-action kinetics \eqref{defn:mass_action_kin}), and $\nu_{\bar 1}$ and $\nu_{\bar 2}$ are the numbers reactant molecules of $X$ in reactions $e_{\bar 1}$ and $e_{\bar 2}$. Therefore Eq.~\eqref{max_rpa_ratio_cond} implies that $m_1 = m_2$ and so reactions $e_{\bar 1}$ and $e_{\bar 2}$ have an equal number of all the species (except the output species $X$) as reactants. Setting $\nu \coloneqq \nu_{\bar 2} - \nu_{\bar 1}$, Eq.~\eqref{max_rpa_gen_set_point} becomes equivalent to
\begin{equation}
\bar{x}_M = \left( \kappa \frac{k_{\bar 1} }{ k_{\bar 2}} \right)^{\frac{1}{\nu}},
\end{equation}
which is the set-point encoding function in Ref.~\cite{gupta2022universal}.

We end this section with a couple of important definitions.
\begin{definition}
\label{defn:anti_homo_maxrpa}
A maxRPA network is called \emph{homothetic} is its associated $\bm q$ vector has all its nonzero components of the same sign. Otherwise, the maxRPA network is called \emph{antithetic}.
\end{definition}

Let us define a subnetwork $\bar \gamma \subset \Gamma$ consisting of the output species $v_M = X$ and the last two reactions $e_{\bar 1}$ and $e_{\bar 2}$
\begin{equation}
\label{defn_gamma_comp}
\bar \gamma = (\{X\}, \{e_{\bar 1}, e_{\bar 2} \}). 
\end{equation}
The complement of this subnetwork, 
\begin{equation}
\label{defn_gamma}
\gamma \coloneqq \Gamma \setminus \bar \gamma,
\end{equation}
consists of all the species except $X$ and all the reactions except $e_{\bar 1}$ and $e_{\bar 2}$. Condition \eqref{max_rpa_ratio_cond} shows that for maxRPA to occur, the kinetics of the last two reactions must be \emph{fine-tuned} to match the dependence on reactants other than the output species $X$. The only way to have maxRPA without fine-tuning is that the last two reactions do not have any species (other than $X$) as reactants, which is what we shall call \emph{kinetics-independent} maxRPA.
\begin{definition}
\label{defn:kinetics_indep_maxrpa}
A maxRPA network is said to be \emph{kinetics-independent} if the reactions rates of the two set-point determining reactions (i.e.\ $e_{\bar 1}$ and $e_{\bar 2}$) do not depend on any species except the output species $X$. In other words, only $X$ can be a reactant for these two reactants.
\end{definition}
Note that a maxRPA network becomes kinetics-independent, precisely in the scenario where the subnetwork $\gamma$ is output-complete, because all reactions involving species in $\gamma$ are in $\gamma$. In fact, Theorem \ref{thm:maxrpa_equivalence} will show that $\gamma$ is a buffering structure, i.e. its influence index is zero, $\lambda(\gamma) = 0$.

\subsection{An Internal Model Principle for maxRPA networks} \label{maxrpa_imp}

We now discuss how an Internal Model Principle can be formulated for maxRPA networks based on the unique vector $\bm q$ that characterizes them. Consider a maxRPA network with stoichiometric matrix $S$, and a pair $(\bm q, \kappa)$ satisfying Eq.~\eqref{eq:qs} with $\bm q \in (\coker S)^\perp$ and $\kappa > 0$. Recall that $\Phi(x_M, k_{\bar 1}, k_{\bar 2})$ is the ratio of rate functions $r_{\bar 2}(\bm x; k_{\bar{2}})/r_{\bar 1}(\bm x; k_{\bar{1}})$ which can only depend on the output species concentration $x_M$ as per Condition 2 of Theorem \ref{thm:maxRPA_characterization}. Defining 
\begin{equation}
z \coloneqq \bm q \cdot \bm x,
\end{equation}
we see that this is an \emph{integrator} for the maxRPA networks as its time-derivative is proportional to the ``error" in Eq.~\eqref{max_rpa_gen_set_point}, 
\begin{align*}
\dot z = \bm q^\top S \ \bm r(\bm x, \bm k) = 
\begin{bmatrix}
0 & \cdots & 0 & \kappa & -1
\end{bmatrix} {\bm r(\bm x, \bm k) } &= \kappa \ r_{\bar 1}(\bm x; k_{\bar{1}}) - r_{\bar 2}(\bm x; k_{\bar 2}) = \ r_{\bar 1}(\bm x; k_{\bar{1}}) (\kappa  - \Phi(x_M, k_{\bar 1}, k_{\bar 2})).
\end{align*}
Since the network dynamics is stable, the presence of this integrator ensures that the dynamics is driven to a steady-state where Eq.~\eqref{max_rpa_gen_set_point} holds, and hence the set-point $\bar{x}_M$ for the output species is only a function of $k_{\bar 1}$ and $k_{\bar 2}$. 

The Internal Model for the maxRPA network consists of species that form the support of the vector
$\bm q = \begin{bmatrix} q_1 & \cdots & q_M \end{bmatrix}^\top$, 
namely, 
\begin{equation}
V_{\rm IM} = \left\{ v_i \in V \, \middle| \, q_i \neq 0 \right\}.
\end{equation}
In order to satisfy Condition 2 of Theorem \ref{thm:maxRPA_characterization}, at least one of the last two reactions $e_{\bar{1}}$ and $e_{\bar{2}}$ must have the output species $X$ as a reactant. Without loss of generality, we may assume that $e_{\bar{2}}$ has $X$ as a reactant, and we call it the output \emph{sensing} reaction, while we refer to the other reaction $e_{\bar 1}$ as the \emph{set-point encoding} reaction. 
In a number of situations, the output species $v_M = X$ does not belong to $V_{\rm IM}$
\footnote{For example, this would happen when there is a combination of first $(N-2)$ reactions and conservation relations in $\coker S$ that only modifies the output species $v_M$, i.e. there exists a vector $\bm y$ such that
$\begin{bmatrix}\bar{S} & D \end{bmatrix} \bm y = \bm u_M,$
where $\bar{S}$ is the matrix formed by removing the last two columns from $S$, $D$ is the matrix with basis vectors of $\coker S$ as columns and $\bm u_M$ is a $M$-dimensional vector whose last component is $1$ and the rest are zeros. Note that since $\bm q$ satisfies Eq.~\eqref{eq:qs2} we would have $ q_M = {\bm q}^\top u_m =    {\bm q}^\top \begin{bmatrix} \bar{S} &  D\end{bmatrix} \bm y=  0$. 
}(i.e.\ $q_M = 0$) and so $X$ belongs to the rest of the network, giving rise to the IMP decomposition shown in Fig.~\ref{fig:schematic-equivalence}. The \emph{actuation} reactions from IM to RoN are necessary to complete the feedback loop and ensure network stability, but the form of these reactions can be arbitrary.

The characterization result for maxRPA networks (Theorem \ref{thm:maxRPA_characterization}) has two conditions, in which the first one is structural (i.e. Eq.~\eqref{eq:qs}) while the second condition (i.e.\ Eq.~\eqref{max_rpa_ratio_cond}) depends on the kinetics of the last two reactions. 
If we consider kinetics-independent maxRPA networks (see Definition \ref{defn:kinetics_indep_maxrpa}) then the second condition also becomes structural, and in this situation the two conditions become equivalent to the subnetwork $\gamma$, defined by Eq.~\eqref{defn_gamma}, being a buffering structure.
The next theorem states this equivalence and one can find its schematic illustration in Fig.~\ref{fig:schematic-equivalence}.  

\begin{theorem}
\label{thm:maxrpa_equivalence}
Consider a network $\Gamma$ and define its subnetwork $\gamma$ by Eq.~\eqref{defn_gamma}. Then, $\Gamma$ is a kinetics-independent maxRPA network if and only if the subnetwork $\gamma$ is a buffering structure, i.e.\ it is output-complete with a zero influence index $\lambda(\gamma) = 0$.
\end{theorem}
The proof of this theorem will be given in Sec.~\ref{sec:proof-th1}.

\begin{figure}[tb]
  \centering
  \includegraphics
  [clip, trim=6cm 9cm 6cm 3cm, scale=0.58]
  {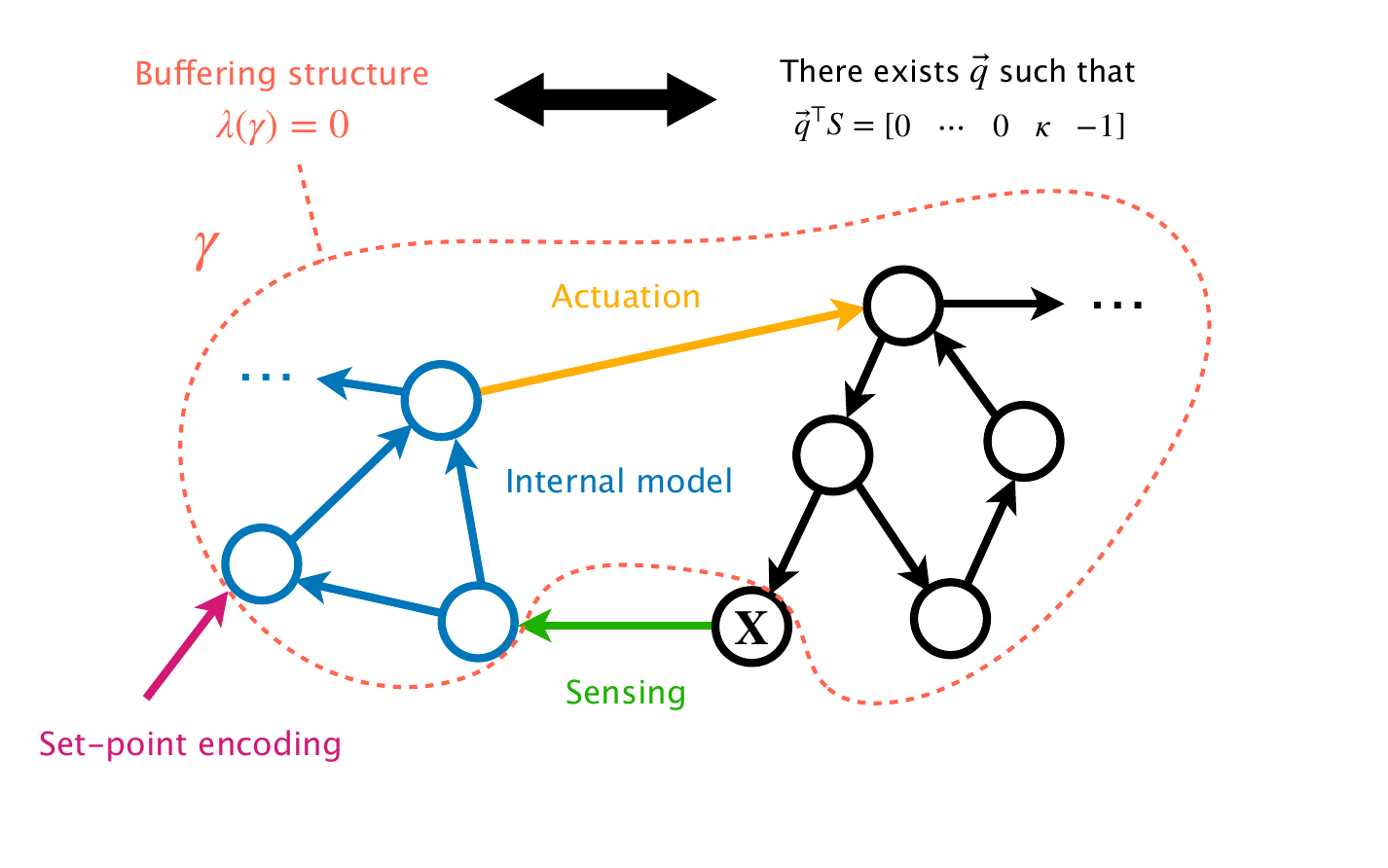}
  \caption{
  Schematic figure representing the equivalence. 
  The stoichiometric condition~\eqref{eq:qs} for the system to exhibit maxRPA is equivalent to the condition the part of the network encircled by a dashed line has a zero influence index, $\lambda(\gamma) = 0$.}
  \label{fig:schematic-equivalence} 
\end{figure}

We now briefly discuss how this topological characterization of the kinetics-independent maxRPA property is connected to the more general results in the next section. As per the terminology introduced in Sec.~\ref{sec:assumptions}, we can view maxRPA as a generic RPA property $(\mathcal{V}, \mathcal{E}, \bm p)$ where $\mathcal{V} = \{X\}$, $\mathcal{E} = \{e_{\bar 1}, e_{\bar 2}\}$ and $\bm p$ is the set of all system parameters except $k_{\bar 1}$ and $k_{\bar 2}$ which appear in the kinetics of reactions in $\mathcal{E}$. Note that $\gamma= \Gamma \setminus (\mathcal{V}, \mathcal{E})$ is the subnetwork consisting of all the species except $X$, and all the reactions except $e_{\bar 1}$ and $e_{\bar 2}$. It will be shown later that kinetics-independence and network stability, together imply that at the steady-state, the concentrations and rates of all the species and reactions in $\gamma$ are \emph{sensitive} to at least one of the parameters in $\bm p$. This ensures that if for each parameter $p_\mu$ in $\bm p$, if $\gamma_\mu$ is the subnetwork of all the species and reactions that are sensitive to $p_\mu$, then $\gamma$ is in fact the union of all such $\gamma_\mu$-s. This is made mathematically precise in the next section, where each $\gamma_\mu$ is identified as a \emph{labelled buffering structure}. As buffering structures are closed under the union operation, it follows that $\gamma$ is itself a buffering structure, as asserted by Theorem \ref{thm:maxrpa_equivalence}, and it shows that the maxRPA property is generated by all the labelled buffering structures corresponding to parameters in $\bm p$. The central goal of this paper, is to demonstrate that these arguments can be extended to any generic kinetics-independent RPA property (not just maxRPA) and hence the identification of all labelled buffering structures is sufficient to completely characterise and systematically identify all such generic RPA properties.

\section{ Characterization of all RPA properties }\label{sec:rpa-lbs}

The law of localization discussed in Sec.~\ref{sec:LoL} gives us a sufficient condition for a generic RPA, and the condition is topological, meaning that it is determined by network topology.
A natural question arises as to whether any RPA property can be captured with topological criteria.  
In this section, we show that 
{\it all} the elementary RPA properties that are kinetics-independent 
can be indeed characterized/found by labeled buffering structures,
which are buffering structures annotated with supplementary information. 
Since generic RPA properties can be constructed from elementary ones, we thus have a complete characterization of all RPA properties.
We first introduce the transitivity of influence 
in Sec.~\ref{sec:transitivity}, that will be used later
when we show the one-to-one correspondence of the RPA properties 
and labeled buffering structures in Sec.~\ref{sec:rpa-lbs-thm}.
We then present a computational method to identify labeled buffering structures in Sec.~\ref{sec:num-alg} and discuss examples in Sec.~\ref{sec:lbs-examples}.

\subsection{ Transitivity of influence }\label{sec:transitivity}

Here we introduce the transitivity of influence. 
The transitivity of influence for reaction perturbations is established first for the monomolecular case~\cite{doi:10.1002/mma.3436,https://doi.org/10.1002/mma.4557} and then for the multimolecular case~\cite{doi:10.1002/mma.4668} (for the case $\coker S = \bm 0$).

Let us prepare notations. 
For reactions $e_A, e_B \in E$, if $\p_A \bar r_B \neq 0$ holds,
we say ``$e_A$ influences $e_B$'' and 
represent this in symbols as
\begin{equation}
e_A \leadsto e_B . 
\end{equation}
On the other hand, if $\p_A \bar r_B = 0$, 
we say that ``$e_A$ does not influence $e_B$'' and 
we write this as $e_A \not \leadsto e_B$.
Similarly, if $\p_A \bar x_i \neq 0$
for some $v_i \in V$ and $e_A \in E$, 
we express this as 
\begin{equation}
e_A \leadsto v_i. 
\end{equation}
When $v_i \in V$ is a reactant of reaction $e_A$ (i.e.\ $s_{iA} \neq 0$), 
we write 
\begin{equation}
   v_i \vdash e_A . 
\end{equation}
We may also express this with the corresponding indices 
as $i \,  \vdash A$. 
While the absence of conserved quantities is assumed in 
Ref.~\cite{doi:10.1002/mma.4668}, 
we here allow them and thus consider the perturbation of conserved quantities as well as reaction parameters.
When the perturbation of $\ell_{\bar\alpha}$ affects 
the flux of $e_A$ (i.e., $\p_{\bar\alpha} \bar r_{A} \neq 0$),
we express this as
\begin{equation}
\ell_{\bar\alpha}  \leadsto e_A. 
\end{equation}

The transitivity of influence 
is stated as follows~\cite{doi:10.1002/mma.4668}: 
\begin{theorem}[Transitivity of influence]
If reaction $e_A$ influences $e_B$ and
$e_B$ influences $e_C$, then $e_A$ influences $e_C$. 
In symbols, 
\begin{equation}
e_A \leadsto e_B \leadsto e_C 
\implies 
e_A \leadsto e_C. 
\label{eq:a-b-c-to-a-c}
\end{equation}
\end{theorem}
\begin{proof}
The proof we describe here is slightly different from the one 
given in Ref.~\cite{doi:10.1002/mma.4668}. 
While $\coker S=\bm 0$ is assumed in Ref.~\cite{doi:10.1002/mma.4668},
the following proof is valid for $\coker S \neq \bm 0$.

We note that, to prove Eq.~\eqref{eq:a-b-c-to-a-c}, 
it suffices to show that, 
if reaction $e_A$ influences $v_i$
and $v_i$ is a reactant of $e_B$ and 
$e_B$ influences $e_C$, 
then $e_A$ influences $e_C$. 
In symbols, 
\begin{equation}
e_A \leadsto v_i \vdash e_B \leadsto e_C 
\implies 
e_A \leadsto e_C. \label{eq:a-i-b-c-to-a-c}
\end{equation}

Let us first show that Eq.~\eqref{eq:a-i-b-c-to-a-c} implies 
Eq.~\eqref{eq:a-b-c-to-a-c}. 
When $e_A = e_B$ or $e_B = e_C$ is true, 
Eq.~\eqref{eq:a-b-c-to-a-c} is always satisfied, 
so it is sufficient to consider the case where 
$e_A = e_C \neq e_B$ or
$e_A$, $e_B$, $e_C$ are all different. 
Since $e_A \leadsto e_B$, 
\begin{equation}
0 \neq 
\bar r_{B, A} 
= 
\sum_{i \, \vdash B } r_{B,i} \, \bar x_{i,A}, 
\end{equation}
where the summation is over the species that are reactants of $e_B$.
For this to be nonzero, there must be a species $v_i \vdash e_B$  
such that $e_A \leadsto v_i$. 
By Eq.~\eqref{eq:a-i-b-c-to-a-c}, 
this implies $e_A \leadsto e_C$, and Eq.~\eqref{eq:a-b-c-to-a-c} holds.

Now let us prove Eq.~\eqref{eq:a-i-b-c-to-a-c}. 
We here use the formula for 
the second-order response of steady-state fluxes\footnote{
Note that this formula is correct when
$e_A$, $e_B$, $e_C$ are all different or 
either of $e_A \neq e_B$ is equal to $e_C$. 
}, 
which can be derived straightforwardly from Eq.~\eqref{eq:sensitivity-2}
(see Appendix~\ref{sec:second-order-formula} for derivation), 
\begin{equation}
\bar r_{C,AB}
=     
\bar r_{C,A} F_{A,B}
+
\bar r_{C,B} F_{B,A} 
+ 
\sum_{D} 
\sum_{j \, \vdash D} \sum_{k \, \vdash D}
\bar r_{C,D} \,
\frac{r_{D,jk} }{r_{D,D}}
\,
\bar x_{j,A} \,
\bar x_{k,B} , 
\end{equation}
where 
$
\bar r_{C,AB} \coloneqq 
\frac{\p^2}{\p k^A \p k^B} 
\bar r_{C} (\bar{\bm x}(\bm k, \bm \ell); k_C) 
$, 
$
r_{D,jk} \coloneqq 
(\frac{\p^2}{\p x^j \p x^k} r_{D} (\bm x ; k_D) )_{\bm x = \bar{ \bm x}(\bm k, \bm \ell)}
$,
$
r_{A,A} \coloneqq 
(\frac{\p}{\p k_A} r_{A}(\bm x ; k_A) )_{\bm x = \bar{\bm x}(\bm k, \bm \ell)} 
$, 
$F_{A} \coloneqq \ln r_{A,A}$ 
and 
$F_{A,B} \coloneqq 
\frac{\p}{\p k_B} 
\left[ 
(\ln r_{A,A})_{\bm x = \bar{ \bm x}(\bm k, \bm \ell)}
\right]
$. 
From the second term, we have the following 
contribution, 
\begin{equation}
\bar r_{C,AB}
=     
\bar r_{C,B}
\sum_{i \,\vdash B} 
(\ln  r_{B,B} )_{,i} \, \bar x_{i,A}
+ 
\cdots .
\end{equation}
Thus, when we shift the value of parameter
as $k_B \mapsto k_B + \delta k_B$, 
$\bar r_{C,A}$ as a function of $k_B$, 
$\bar r_{C,A} (k_B)$, is modified as 
(we here only denote $k_B$ dependence)
\begin{equation}
\bar r_{C,A} (k_B + \delta k_B) 
=     
\bar r_{C,A} (k_B)
+ 
\bar r_{C,B}
\sum_{j \,\vdash B} 
(\ln r_{B,B})_{,j} \, \bar x_{j,A}
\delta k_B 
+ \cdots .
\end{equation}
From the assumption, we have 
$\bar r_{C,B} \neq 0$ 
and 
$\bar x_{i,A} \neq 0$. 
Thus, an infinitesimal shift of $k_B$ induces nontrivial changes in the value of $\bar r_{C,A}$. 
This indicates that, without fine-tuning, 
the value of $\bar r_{C,A}$ is nonzero. 
This proves Eq.~\eqref{eq:a-i-b-c-to-a-c}. 

\end{proof}

\subsection{ RPA and labeled buffering structures }\label{sec:rpa-lbs-thm}

\begin{figure}[tb]
  \centering
  \includegraphics
  [clip, trim=1cm 6cm 1cm 6cm, scale=0.45]
  {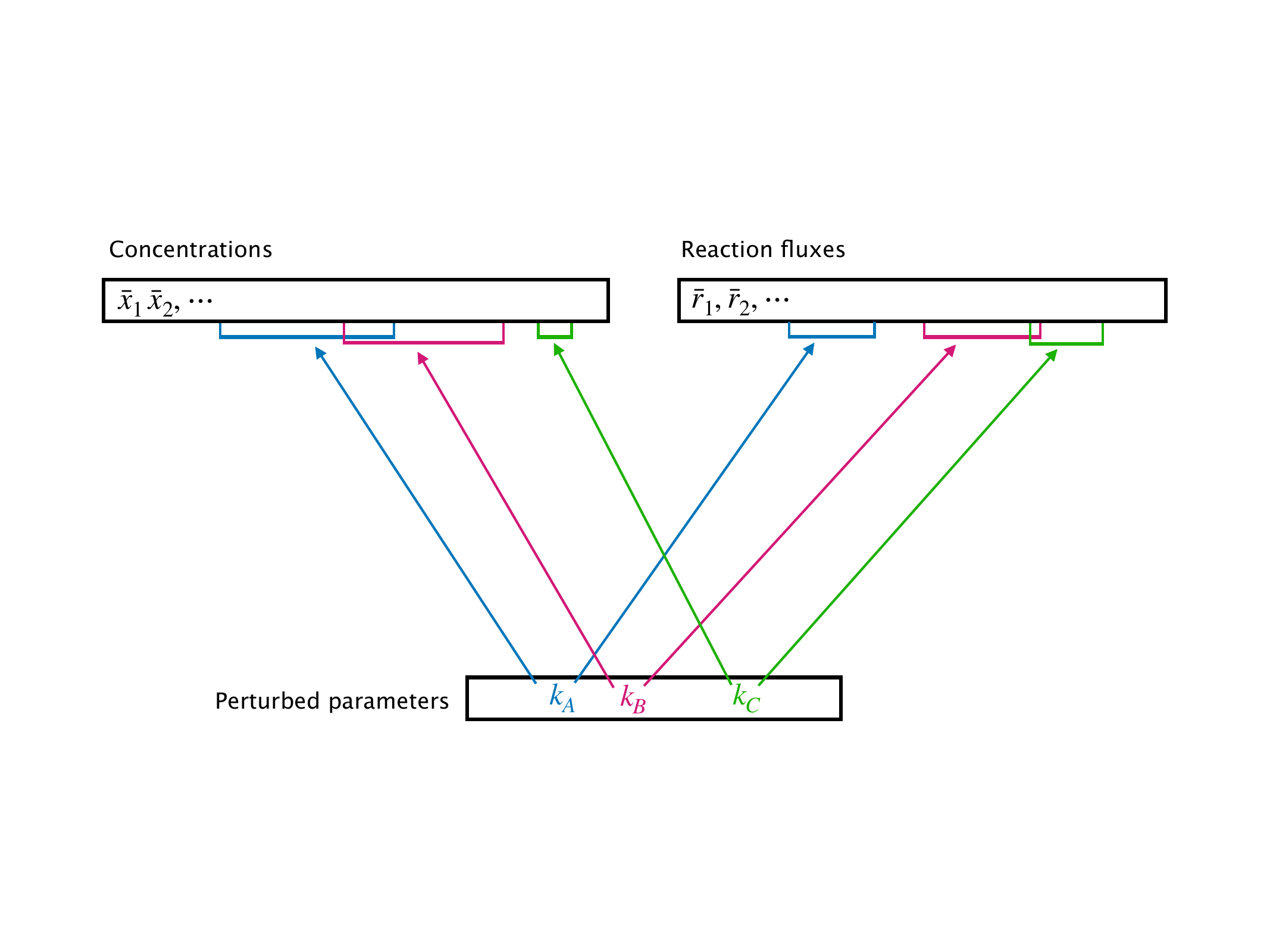} 
  \caption{
Schematic of the characterization of RPA properties in a chemical reaction system. 
  }
  \label{fig:rpa-schematic} 
\end{figure}

Suppose we would like to find all RPA properties for a given reaction system. 
For each parameter of the system, we can hypothetically 
perform its perturbation and observe the response. 
In general, some of the steady-state concentrations and reaction rates 
show nonzero response while others do not respond, and the latter exhibit perfect adaptation.
To test the robustness of adaptation, 
we can change the system parameters (or the form of rate functions) 
and repeat the hypothetical experiment.
Doing this procedure for all the parameters, 
we can obtain all elementary RPA properties of the system (see Fig.~\ref{fig:rpa-schematic})\footnote{Such a perturbation-based approach for systematically studying the steady-state response is similar to the one proposed in Ref.~\cite{giordano2016computing} for computing the \emph{influence matrix} and also to the method described in Ref.~\cite{doi:10.1002/mma.4668} for deriving the \emph{influence graph}.
A crucial difference is that we connect RPA properties to topological characteristics of subnetworks (via the one-to-one correspondence shown in this section), that facilitates the identification of integral feedback control for each RPA property.
}.

Here, we show that an elementary RPA property with respect to a parameter 
can be represented by a buffering structure with supplementary information,
by proving the following statement: 
\begin{theorem}[One-to-one correspondence of elementary RPA properties and labeled buffering structures] \label{thm:rpa-bs}
For a chosen reaction $e_B \in E$, 
we collect the species and reactions influenced by $e_B$, 
\begin{equation}
e_B \leadsto ( V_B, E_B ) . 
\end{equation}
The subnetwork $\gamma_0 \coloneqq ( V_B, E_B )$ can be 
made output-complete by adding reactions $\mathcal E_B \subset E \setminus E_B$ that have reactants in $V_B$.
The resulting output-complete subnetwork 
$\gamma \coloneqq ( V_B, E_B \cup \mathcal E_B )$ 
is a buffering structure, i.e., $\lambda(\gamma)=0$.     
\end{theorem}

We note that a similar statement holds for the perturbation of conserved quantities. 
Namely, if we find species and reactions that are affected by the perturbation of a conserved quantity as 
$\ell_{\bar\alpha} \leadsto ( V_{\bar\alpha}, E_{\bar\alpha})$, 
then the corresponding output-complete subnetwork
$\gamma \coloneqq ( V_{\bar\alpha}, E_{\bar\alpha} \cup \mathcal E_{\bar\alpha})$
is a buffering structure.
In the following, we present the proof of the theorem for the case of perturbations of reaction parameters. The proof for conserved-quantity perturbations is completely analogous. 

\begin{proof}
We first show that, 
although $\gamma_0$ is identified as the subnetwork within which 
the perturbation of the parameter of $e_B$ is confined, 
the effect of the perturbation of 
{\it any} reaction $e_{C^\star}$ in $E_B$ 
is in fact confined inside $\gamma_0$. 
Namely, for any $e_{C^\star} \in E_{B}$, we have 
\begin{equation}
    e_{C^\star} \,\, \slashed{\leadsto} \,\,
    \{ V \setminus V_B, E \setminus E_B \}. 
\end{equation}
Suppose that this is not the case,
then there exists reaction $e_{D'} \in E \setminus E_B$ 
which is influenced by a reaction inside $E_B$. 
Then, because of the transitivity, 
we should have $e_B \leadsto e_{D'} \not \in E_B$. 
This contradicts with the fact that 
$E_B$ includes all the reactions influenced by $e_B$.

Let us denote the set of reactions 
in $E \setminus E_B$ 
that have a species of $V_B$ as a reactant by $\mathcal E_{B}$. 
Namely,
\begin{equation}
\mathcal E_B = 
\left\{ 
e_C \in E \setminus E_B
\, \middle| \,
\exists \, v_{i^\star} \in V_B \text{ such that } v_{i^\star} \vdash e_C 
\right\} . 
\label{eq:curly-e-def}
\end{equation}
The subnetwork $\gamma = (V_B, E_B \cup \mathcal E_B)$ is 
obviously output-complete.

We show that the influence of the perturbations 
of the reaction in $\mathcal E_B$ is 
also localized in $\gamma$. 
If $e_C \in \mathcal E_B$ is equal to $e_B$, 
this is trivially true, so we consider the case $e_C \neq e_B$ in the following. 
Since $\bar x_{i',B}$ with  $v_{i'} \in V \setminus V_B$, 
should vanish regardless of parameters, 
the derivative of $\bar x_{i',B}$ with respect to the parameter 
of $e_C \in \mathcal E_B$ should vanish as well, 
\begin{equation}
0 
= \bar x_{i',BC}
=  
\bar x_{i',C} F_{C,B} + \bar x_{i',B} F_{B,C}
+ 
\sum_{D} \sum_{j \, \vdash D} \sum_{k \, \vdash D}
\frac{r_{D,jk} }{r_{D,D}}
\,
\bar x_{i',D} \, 
\bar x_{j,B} \,
\bar x_{k,C} . 
\end{equation}
The second term on the RHS is zero because $\bar x_{i',B}=0$. 
As for the third term, 
we can divide the summation over all the reactions to 
those inside/outside $E_\gamma \coloneqq E_B \cup \mathcal E_B$ as
$\sum_D \cdots = \sum_{D^\star} \cdots + \sum_{D'} \cdots$,
and noting that $\bar x_{i',D^\star} = 0$ and $\bar x_{j',B}=0$,
the third term is now written as
\begin{equation}
\begin{split}  
\text{(third term)}
&= 
\sum_{D'} 
\sum_{j^\star \, \vdash D'}
\sum_{k \, \vdash D'} 
\frac{r_{D',j^\star k} }{r_{D',D'}}
\,
\bar x_{i',D'} \, 
\bar x_{j^\star,B} \,
\bar x_{k,C} .
\end{split}
\end{equation}
Since there is no reaction outside $\gamma$ 
whose reactant is inside $\gamma$ by construction, 
the summation over $j^\star$ is empty and 
this term vanishes as well. 
Thus, we have 
\begin{equation}
0 
= 
\bar x_{i', C} 
\sum_{i^\star \vdash C}
(\ln  r_{C,C} )_{,i^\star} \bar x_{i^\star, B}. 
\label{eq:x-ip-c-sum-ln-pcrc}
\end{equation}
There is at least one 
species $v_{i^\star}$ 
that is a reactant of $e_C$, by construction, 
and $(\ln r_{C,C})_{,i^\star} \neq 0$. 
Since $\bar x_{i^\star, B} \neq 0$ by definition, 
if the RHS of Eq.~\eqref{eq:x-ip-c-sum-ln-pcrc} is to vanish 
without fine-tuning, we should have $\bar x_{i',C} = 0$. 
Thus, there are no concentrations outside $\gamma$ 
that are influenced by $e_C \in \mathcal E_B$. 
As a result, no reactions outside $\gamma$ are 
influenced by $e_C \in \mathcal E_B$. 
Therefore, we have shown that the influence by any reaction $e_{C^\star}$ inside $\gamma$ is localized inside $\gamma$.

Let us here choose a basis of $\ker S$ in the following manner.
Given a subnetwork $\gamma$, we pick a basis 
$\{\bm c^{(\alpha^\star)} \}_{\alpha^\star=1,\ldots,|\alpha^\star|}$ 
of $(\ker S)_{{\rm supp}\,\gamma}$, and 
we arrange the basis of $\ker S$ so that $\bar r_A$ is written as
\begin{equation}
\bar r_A = 
\sum_{\alpha^\star} \mu_{\alpha^\star} c^{(\alpha^\star)}_A 
+ 
\sum_{\alpha'} \mu_{\alpha'} c^{(\alpha')}_A ,
\label{eq:r-bar-expansion}
\end{equation}
where $\{\bm c^{(\alpha')} \}_{\alpha'=1,\ldots,|\alpha'|}$ are basis vectors 
with nonzero support in $E \setminus E_\gamma$. 
Note that $c^{(\alpha^\star)}_{A'}=0$, since it is supported in $\gamma$. 
We employ the expansion the steady-state fluxes
in the form~\eqref{eq:r-bar-expansion}.
From the response localization, 
for $e_{C^\star} \in E_\gamma$ 
and 
$e_{A'} \in E \setminus E_\gamma$,
we have 
\begin{equation}
0 = \bar r_{A', C^\star} 
= 
\sum_{\alpha'} 
\mu_{\alpha',C^\star} \, c^{(\alpha')}_{A'} , 
\label{p-b-r-0}
\end{equation}
where we used $c^{(\alpha^\star)}_{A'}=0$. 
Let us denote a basis vector as 
$\bm c^{(\alpha)} = 
\begin{bmatrix}
\bm c_1^{(\alpha)} \\
\bm c_2^{(\alpha)} 
\end{bmatrix}
$, 
where 
$\bm c_1^{(\alpha)}$ and $\bm c_2^{(\alpha)}$ are components 
inside and outside of $\gamma$, respectively.
With the current choice of basis, the vectors 
$\{ \bm c^{(\alpha')}_2 \}_{\alpha'=1,\ldots,|\alpha'|}$ 
are linearly independent\footnote{
Suppose that the vectors $\{ \bm c^{(\alpha')}_2 \}_{\alpha'=1,\ldots,|\alpha'|}$ are linearly dependent. 
Then, by taking a linear combination of $\{\bm c^{\alpha'} \}_{\alpha'=1,\ldots,|\alpha'|}$, 
we can make a vector supported inside $\gamma$, 
\begin{equation}
\bm c' = \begin{bmatrix}
    \bm c'_1 \\
    \bm 0
\end{bmatrix},
\end{equation}
which is an element of $(\ker S)_{{\rm supp}\,\gamma}$.
Since the vector $\{ \bm c^{(\alpha)} \}_{\alpha=1,\ldots,|\alpha|}$ 
are linearly independent, 
$\bm c'_1$ should be linearly independent of 
$\{ \bm c_1^{(\alpha^\star)} \}_{\alpha^\star=1,\ldots,|\alpha^\star|}$. 
This contradicts the fact that 
$\{ \bm c^{(\alpha^\star)} \}_{\alpha^\star=1,\ldots,|\alpha^\star|}$
is a basis vector of $(\ker S)_{{\rm supp}\,\gamma}$. 
}, and 
Eq.~\eqref{p-b-r-0} implies that
\begin{equation}
\mu_{\alpha', C^\star} = 0. 
\end{equation}
The response of the steady-state concentrations and reaction fluxes to the perturbation of parameter $k_{C^\star}$ with 
$e_{C^\star} \in E_\gamma$ 
can be characterized by the following equation 
using the ${\bf A}$-matrix, 
%
The response is determined by the following equation,
\begin{equation}
\left[ 
\begin{array}{c|c}
{\bf A}_{11} & {\bf A}_{12} \\ \hline 
{\bf A}_{21} & {\bf A}_{22}
\end{array}
\right] 
\left[ 
\begin{array}{c}
\bm y_{1,C^\star}
\\ \hline 
\bm y_{2,C^\star}
\end{array}
\right]
= 
- 
\left[ 
\begin{array}{c}
\p_{C^\star} \bm r_1 
\\ \hline 
\bm 0
\end{array}
\right],
\label{eq:part-a-y-d}
\end{equation}
where we have defined 
\begin{equation}
\bm y_{1,C^\star} 
\coloneqq 
\p_{C^\star}
\begin{bmatrix}
{\bar x}_{i^{\star}} 
\\
\mu_{\alpha^\star} 
\end{bmatrix}
, 
\quad 
\bm y_{2,C^\star} 
\coloneqq 
\p_{C^\star}
\begin{bmatrix}
\bar x_{i'}
\\
\mu_{\alpha'}     
\end{bmatrix}. 
\end{equation}
We have $\bm y_{2, C^\star} = \bm 0$ 
for any $e_{C^\star} \in E_\gamma$ 
from the response localization.

\begin{figure}[tb]
  \centering
  \includegraphics
  [clip, trim=1cm 10cm 1cm 8cm, scale=0.4]
  {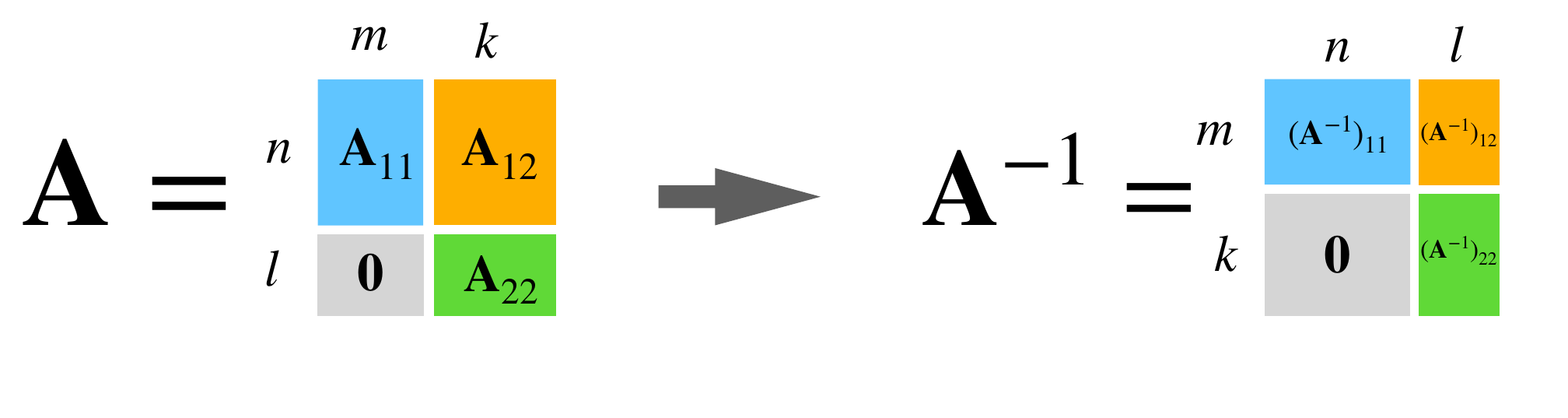} 
  \caption{
  Structures of ${\bf A}$-matrix and its inverse. 
  }
  \label{fig:a-ainv} 
\end{figure}

Suppose that ${\bf A}_{11}$ is an $n \times m$ matrix. 
Since ${\bf A}_{21}$ is a zero matrix, 
the invertibility of ${\bf A}$ requires that 
${\bf A}_{11}$ should be square or vertically long, 
and we have $n \ge m$.
On the other hand, 
from the localization of responses, 
its inverse ${\bf A}^{-1}$ should have the structure 
shown in Fig.~\ref{fig:a-ainv}, 
with zero matrix in the lower-left sector, 
and the upper-left part is an $m \times n$ matrix. 
So that ${\bf A}^{-1}$ be invertible, 
$({\bf A}^{-1})_{11}$ should be square or vertically long, 
and we need $m \ge n$. 
Thus, the only possibility is $n=m$, 
and ${\bf A}_{11}$ is a square matrix. 

Recall that 
the set of vectors 
$\{ {\bm c}^{(\alpha^{\star})}\}_{\alpha^{\star}=1,\ldots,|\alpha^\star|}$ 
is a basis of $(\ker S)_{{\rm supp}\,\gamma}$. 
The fact that ${\bf A}_{11}$ is a square matrix 
implies 
$|E_\gamma| + |P^0_\gamma (\coker S)| 
- |V_\gamma| - |(\ker S)_{{\rm supp}\,\gamma}|
=0$, 
meaning that $\lambda(\gamma)=0$. 
Thus, we have the claim. 

\end{proof}

Motivated by Theorem~\ref{thm:rpa-bs}, 
we introduce {\it labeled buffering structures} 
to characterize the elementary RPA properties in a reaction network. 
Specifically, for each reaction $e_A$, we identify the
corresponding buffering structure 
as in Theorem~\ref{thm:rpa-bs}, 
\begin{equation}
\gamma_A 
= 
(V_A, E_A \cup \mathcal E_A) . 
\label{eq:lbs-def}
\end{equation}
The concentrations of $V\setminus V_A$ and 
reaction fluxes of $E \setminus E_A$ exhibit 
RPA with respect to the perturbation of $e_A$. 
We call Eq.~\eqref{eq:lbs-def} as a {\it labeled buffering structure}. 
The name comes from the fact that, 
if we regard a labeled buffering structure 
$(V_A, E_A \cup \mathcal E_A)$ as a pair of sets and forget about the labels of parameters (and distinction of $E_A$ and $\mathcal E_A$), 
it reduces to an ordinary buffering structure. 
Similarly, we also define a labeled buffering structure 
associated with the perturbation of conserved quantity $\ell_{\bar\alpha}$
\begin{equation}
\gamma_{\bar\alpha} 
= 
(V_{\bar\alpha}, E_{\bar\alpha} \cup \mathcal E_{\bar\alpha}) . 
\label{eq:lbs-def-cons}
\end{equation}
By enumerating labeled buffering structures for all the parameters, we can identify all the RPA properties in a given reaction network, 
via the one-to-one correspondence, 
\begin{equation}
\{
\text{ elementary RPA properties } 
\}  
\longleftrightarrow 
\{ \text{ labeled buffering structures } \}. 
\end{equation}
There is a possibility that the RPA properties with respect to two or more reactions are exactly the same. 
We use multiple indices in this case, such as $\gamma_{A,B,C}$ 
when the RPA properties with respect to reaction $e_A,e_B,$ and $e_B$ are the same.

As we stated in Sec.~ \ref{sec:assumptions}, we allow the kinetics to be generic,
and RPA properties that can be detected by labeled buffering structures are kinetics-independent ones. 
For a specific choice of kinetics, such as mass-action kinetics, there can be further RPA properties that are robust only within the chosen kinetics. 
The RPA properties represented by labeled buffering structures are also robust under the change of kinetics (as long as stability is not jeopardized).

We note that a generic RPA property can be generated by elementary RPA properties (see Fig.~\ref{fig:generic-rpa}).
Let us consider, a subset of parameters $\bm p \subset (\bm k, \bm \ell)$. 
Correspondingly to each element $p_\mu \in \bm p$, 
we have a labeled buffering structure 
$\gamma_\mu = (V_\mu, E_\mu \cup \mathcal E_\mu)$. 
Let us define 
$
V_{\bm p} \coloneqq 
\bigcup_{\mu} V_\mu
$
and 
$
E_{\bm p} \coloneqq 
\bigcup_{\mu} E_\mu
$. 
Then, the subnetwork $\gamma_{\bm p}$ defined by 
\begin{equation}
\gamma_{\bm p} \coloneqq 
\left( V_{\bm p}, E_{\bm p} \cup \mathcal E_{\bm p} \right) 
\end{equation}
is a buffering structure, where 
$\mathcal E_{\bm p}$ is the minimal set of reactions to make 
$\left( V_{\bm p}, E_{\bm p} \cup \mathcal E_{\bm p} \right)$ output-complete.
Since $\gamma_{\bm p}$ is the union of $\gamma_\mu$ with $p_\mu \in \bm p$ (as a pair of sets), the vanishing of the index $\lambda(\gamma_{\bm p})$ follows from the closure property of buffering structures under union~\cite{PhysRevResearch.3.043123}.
This buffering structure $\gamma_{\bm p}$ is a minimal one with respect to $\bm p$ in the sense that for any species in $V_{\bm p}$ and reactions in $E_{\bm p}$ there exists a parameter $p_\mu \in \bm p$ to which it is sensitive to. 
Then, if we define  
\begin{equation}
\mathcal V \coloneqq V \setminus V_{\bm p}
\quad 
\mathcal E \coloneqq E \setminus E_{\bm p} 
\end{equation}
then the triple $(\mathcal V, \mathcal E, \bm p)$ exhibits a generic RPA property, i.e.\ concentrations of the species in $\mathcal V$ 
and reaction rates of the reactions in $\mathcal E$ exhibit RPA with respect to any parameter in $\bm p$. 
The fact that all the generic RPA properties can be captured 
by this construction 
follows via the same logic of the proof for elementary RPA properties. 
Let us summarize the argument of this paragraph in the following proposition: 
\begin{proposition}
For any generic RPA property $(\mathcal V, \mathcal E, \bm p)$, 
there exists an associated buffering structure. 
\end{proposition}

\begin{figure}[tb]
  \centering
  \includegraphics
  [clip, trim=1cm 7cm 1cm 1cm, scale=0.43]
  {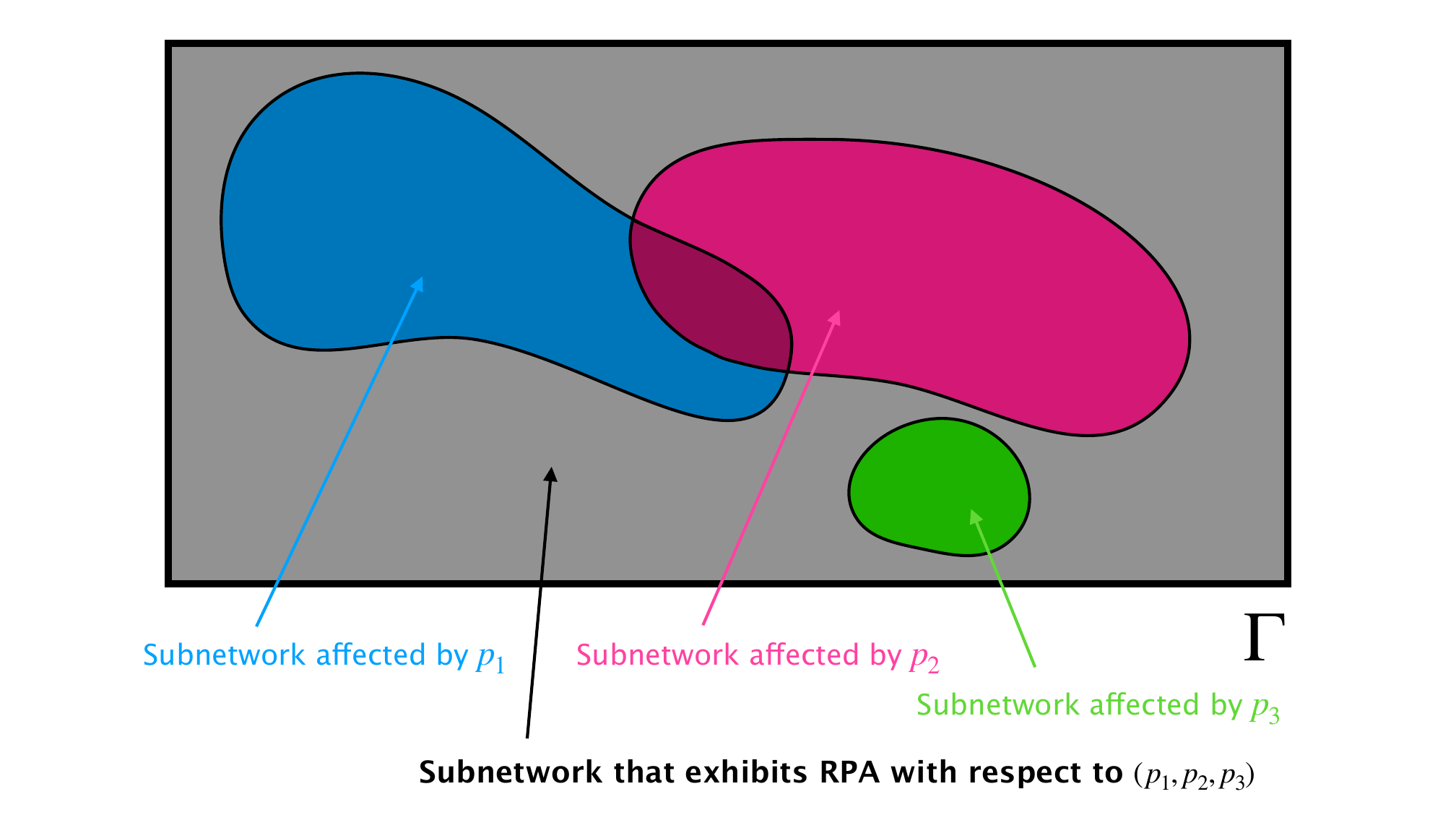}
  \caption{ 
  Schematic illustration of a generation of a generic RPA property from elementary RPA properties. Each labeled buffering structure associated with $p_1$, $p_2$ and $p_3$ identifies the corresponding subnetwork affected by each parameter. 
  With these data, we can identify the subnetwork (gray region) that exhibits RPA with respect to a set $(p_1,p_2,p_3)$ of parameters. 
  }
  \label{fig:generic-rpa} 
\end{figure}

Once all the labeled buffering structures for a given reaction network are identified, we have the full characterization of dependencies of all the steady-state concentrations and reaction fluxes on system parameters, which can be used for various analyses.\footnote{
As we discuss in Sec.~\ref{sec:num-alg}, 
the enumeration of all the labeled buffering structures can be done in polynomial time.
} 
Thus, labeled buffering structures can be used to find species and reactions with more complicated patterns of dependencies.
For example, if we are interested in 
concentrations and reaction rates that are affected 
by $e_A$ 
and not affected by $e_B$, 
such species and reactions are identified as 
$(V_A \cap (V\setminus V_B) , E_A \cap (E \setminus E_B))$. 
As another example, suppose that we have several target species 
of interest $V_{\rm out} \subset V$.
If we would like to identify the parameters which affect the concentrations of the species in $V_{\rm out}$, 
this can be achieved by finding the largest buffering structure that does not include species in $V_{\rm out}$.
All the buffering structures in a reaction network can be 
generated from the identified labeled buffering structures, 
and for buffering structures, we can perform bifurcation analysis~\cite{PhysRevE.98.012417,PhysRevE.103.062212}, 
or we can reduce the network to a smaller one without affecting 
steady-state properties by eliminating buffering structures~\cite{PhysRevResearch.3.043123}. 
Although we do not consider addition of new reactions or new species as a part of perturbations, we can study its effect on the dependencies of concentrations and reaction rates on system parameters, by comparing 
the sets of labeled buffering structures before and after the addition of new reactions and species. 

In this way, identified labeled buffering structures form a basis for various analyses. We shall refer to the identification of RPA properties through buffering structures as \emph{topological analysis}, emphasizing the fact that these subnetworks are identified through topological characteristics. This name is also motivated from a more abstract viewpoint: robust adaptation is a topological phenomenon, as we discuss in Sec.~\ref{sec:rat}. 

\subsection{ Algorithm to enumerate labeled buffering structures }\label{sec:num-alg}

Let us present an algorithm to identify all 
the labeled buffering structures for a given deterministic reaction system. 

As we discussed so far, the qualitative responses (i.e., whether the response is zero or nonzero) of a reaction system against parameter perturbations 
are captured by the inverse of the ${\bf A}$-matrix. 
The response of concentrations is given by 
$\bar x_{i,A} \propto ({\bf A}^{-1})_{i A}$. 
As for reaction rates, 
noting that $\bar r_B = \sum_\alpha \mu_\alpha c^{(\alpha)}_B$, 
the response of a reaction rate can be written as 
\begin{equation}
\p_A \bar r_B = 
\sum_{\alpha,C} c^{(\alpha)}_B 
({\bf A}^{-1})_{\alpha C}
\p_A r_C .
\end{equation}
Since $\p_A r_C \propto \delta_{AC}$, we have 
\begin{equation}
\p_A \bar r_B
\propto 
\sum_{\alpha} c^{(\alpha)}_B 
({\bf A}^{-1})_{\alpha A} . 
\end{equation}
We define the species-sensitivity matrix and 
reaction-sensitivity matrix as 
\begin{equation}
 X_{iA} \coloneqq ({\bf A}^{-1})_{i A} , 
\quad 
\quad 
R_{BA} \coloneqq 
\sum_{\alpha} c^{(\alpha)}_B
({\bf A}^{-1})_{\alpha A}. 
\label{eq:x-r-mat}
\end{equation}

The numerical procedure goes as follows. 
For a given chemical reaction network, 
we first construct the ${\bf A}$-matrix. 
In doing to, we assign random numbers for 
the entries with nonzero $r_{A,i}$.
Namely, we put a random number if $v_i \vdash e_A$ and 
zero otherwise. 
Then, numerically compute its inverse, ${\bf A}^{-1}$,
with which we can obtain the species-sensitivity and reaction-sensitivity matrices \eqref{eq:x-r-mat}. 
Based on these matrices, 
for each reaction $e_A$ (and conserved quantity $\ell_{\bar\alpha}$), 
we identify chemical species and reactions 
that are affected by the perturbation of $e_A$ ($\ell_{\bar\alpha}$). 
Namely, we identify $(V_A, E_A)$ for each reaction $e_A$ 
that are defined by 
\begin{equation}
V_A \coloneqq 
\left\{ 
v_i \in V 
\,\, \middle| \,\,
\text{$X_{iA}$ is nonzero}
\right\}, 
\quad 
\quad 
E_A \coloneqq 
\left\{ 
e_B \in E
\,\, \middle| \,\,
\text{$R_{BA}$ is nonzero}  
\right\}. 
\end{equation}
The subnetwork $(V_A,E_A)$ is not necessarily output-complete. 
It can be made output-complete by adding a set of reactions 
$\mathcal E_A$ 
defined in Eq.~\eqref{eq:curly-e-def}. 
Then, we obtain the labeled buffering structure 
associated with reaction $e_A$, 
$\gamma_A = (V_A, E_A \cup \mathcal E_A)$. 
By repeating this for all the reactions, 
we obtain all the labeled buffering structures within a given reaction system. 

The bottleneck of this enumeration process is the inversion of matrix ${\bf A}$, whose computational complexity is $O({\mathcal N}^3)$ 
in the case of the Gaussian elimination, where ${\mathcal N}$ is the dimension of ${\bf A}$, which is roughly the same as the number of reactions. 

According to the algorithm explained above, we have implement a method to enumerate all the labeled buffering structures for a given reaction network in {\bf RPAFinder}, which is available on Github~\cite{RPAFinder}.

\subsection{ Examples }\label{sec:lbs-examples}

Let us here discuss some examples of chemical reaction systems 
to illustrate the use of topological analysis 
based on the identification of labeled buffering structures.

\subsubsection{ A simple example }

We consider a reaction network 
$\Gamma \coloneqq (\{v_1,v_2,v_3 \}, \{e_1,e_2,e_3,e_4,e_5 \})$
whose reactions are given by 
\begin{align}
e_1 &: \emptyset \to v_1, \nonumber \\ 
e_2 &: v_1 \to v_2, \nonumber  \\ 
e_3 &: v_2 \to v_3, \\ 
e_4 &: v_3 \to v_1, \nonumber  \\ 
e_5 &: v_2 \to \emptyset.  \nonumber 
\end{align}
The network can be visualized as 
\begin{equation}\begin{tikzpicture} 
    \node[species] (v1) at (1.25,0) {$v_1$}; 
    \node[species] (v2) at (3.75,0) {$v_2$};
    \node[species] (v3) at (2.5,1.5) {$v_3$}; 
    \node (d1) at (-0.5,0) {}; 
    \node (d2) at (5.5,0) {}; 
 
    \draw [-latex,draw,line width=0.5mm] (d1) edge node[below]{$e_1$} (v1);
    \draw [-latex,draw,line width=0.5mm] (v1) edge node[below]{$e_2$} (v2);
    \draw [-latex,draw,line width=0.5mm] (v2) edge node[above right]{$e_3$} (v3);
    \draw [-latex,draw,line width=0.5mm] (v3) edge node[above left]{$e_4$} (v1);
    \draw [-latex,draw,line width=0.5mm] (v2) edge node[below]{$e_5$} (d2);

\end{tikzpicture}
\end{equation}
Corresponding to the perturbation of the parameter of each reaction, 
we can identify the following labeled buffering structures, 
\begin{align}
\gamma_1 &= (\{v_1,v_2,v_3 \}, \{ e_1,e_2, e_3,e_4,e_5\} \cup \{ \}) = \Gamma,  \nonumber \\
\gamma^\ast_2 &=  (\{v_1\}, \{ \} \cup \{e_2 \}),  \nonumber \\ 
\gamma_3 &=  (\{v_1,v_3\}, \{e_2,e_3,e_4 \}\cup \{ \} ),\\ 
\gamma^\ast_4 &= (\{v_3 \}, \{ \} \cup \{e_4 \}),  \nonumber \\ 
\gamma_5 &= (\{v_1,v_2,v_3 \},  \{e_2, e_3,e_4 \} \cup \{ e_5 \}) ,  \nonumber 
\end{align}
where those with $\ast$ are strong buffering structures~\cite{https://doi.org/10.48550/arxiv.2302.01270}.
Here, $\gamma_1$ and $\gamma_3$ are already output-complete without adding any additional reactions (i.e.\ $\mathcal E_1$ and $\mathcal E_3$ are empty).
For $\gamma^\ast_2, \gamma^\ast_4,$ and $\gamma_5$, 
we supplemented them with reactions so that they are output-complete. 
The parameter of $e_1$ affects all the species and reactions, 
and the corresponding labeled buffering structure is the entire network. 
In the present example, all parameters 
give rise to different labeled buffering structures. 

From the identified labeled buffering structures, we can conclude the following:
\begin{itemize}
\item From $\gamma_3$, we learn that $\bar x_2, \bar r_1,$ and $\bar r_5$ exhibit RPA with respect to $e_3$.
\item The strong buffering structure $\gamma_2^\ast \cup \gamma_4^\ast$ tells us that all the reaction rates $\bar{\bm r}$ exhibit RPA with respect to $e_2$ and $e_4$.
\item $\bar x_2, \bar x_3$ exhibit RPA with respect to $e_2$ (from $\gamma_2^\ast$).
\item $\bar x_1, \bar x_2$ exhibit RPA with respect to $e_4$ (from $\gamma_4^\ast$).
\item No species or reactions exhibit RPA with respect to $e_1$.
\end{itemize}

If we use mass-action kinetics, 
the steady-state concentrations and rates are given by ($k_A$ is the rate constant for $e_A$)
\begin{equation}
\bar x_1 = \frac{k_1}{k_3 k_5} (k_3+k_5), 
\quad 
\bar x_2 = \frac{k_1}{k_5}, 
\quad 
\bar x_3 = \frac{k_1 k_3}{k_4 k_5}, 
\quad 
\bar {\bm r} = k_1 \bm c^{(1)} + \frac{k_1 k_3}{k_5} \bm c^{(2)} ,
\end{equation}
where 
$
\bm c^{(1)} 
\coloneqq 
\begin{bmatrix}
1 & 1 & 0 & 0 & 1    
\end{bmatrix}^\top 
$
and 
$
\bm c^{(2)} 
\coloneqq 
\begin{bmatrix}
0 & 1 & 1 & 1 & 0
\end{bmatrix}^\top 
$. 
These steady-state solutions are indeed compatible with 
the predictions based on the labeled buffering structures.

\subsubsection{ Bacterial chemotaxis }\label{subsec:bacterial_chemo}

We discuss an example describing bacterial chemotaxis~\cite{barkai1997robustness}. 
Bacterial chemotaxis is the ability of bacteria to move towards or away from a chemical stimulus in their environment. 
{\it E.~coli} are able to detect a variety of attractants and use a biased random walk to navigate towards them~\cite{berg1972chemotaxis}.
Their motion is comprised of runs, during which the bacterium moves in a straight line, and tumbles, during which it stops and changes direction.
When a bacterium senses the increase of attractants 
along its path, it reduces the frequency of tumbles, allowing it to move up gradients of attractants. 
When an attractant is added uniformly in space, 
the frequency of tumbles initially decreases
but eventually returns to its original level.
This adaptation is perfect, and the tumbling frequency at the steady state is independent of the attractant concentration.

We here use a simplified model~\cite{alon2006introduction} that captures this adaptation behavior.
Each receptor is bound to a protein kinase CheA, and we will denote the complex by $X$. 
When a complex of a methylated receptor and a kinase is in an active state, $X^\star_m$, 
it phosphorylates a response regulator, which binds to the flagellar motor and changes the rotational direction, enhancing tumbling. 
Thus, $X^\star_m$ determines the tumbling probability and it is the target output variable of the analysis here. 
The set of reaction we consider is 
\begin{align}
    e_1 &: X_m^\star \to X_m \notag \\
    e_2 &: X_m \to X_m^\star, \notag \\
    e_3 &: X \to X_m ,  \\
    e_4 &: X \to X_m^\star , \notag \\
    e_5 &: X_m^\star \to X . \notag 
\end{align}
Here, $e_3$ and $e_4$ are methylation reactions catalyzed by CheR. The demethylation is mediated by CheB. 
When a ligand is bound to a receptor, 
the probability of deactivation process $e_1$ increases, 
which means that the parameter $k_1$ is an increasing function of the ligand concentration $a$, $k_1 = k_1 (a)$. 
Perfect adaptation is realized when 
the following two conditions are met~\cite{barkai1997robustness}: 
\begin{itemize}
\item The demethylation by CheB only acts on active receptors, $X_m^\ast$. 
\item The rate of methylation by CheR 
does not depend on the substrate concentration, $X$.
\end{itemize}
Under these assumptions, the set of reactions above is 
equivalently described by the following reaction network\footnote{
Here, we use $\emptyset$ to indicate the null species, whose abundance is assumed to be constant and its dynamics is not considered.
}, 
\begin{align}
    e_1 &: X_m^\star \to X_m \notag \\
    e_2 &: X_m \to X_m^\star, \notag \\
    e'_3 &:\emptyset  \to X_m , \label{eq:chemo-reactions-2} \\
    e'_4 &:\emptyset   \to X_m^\star , \notag \\
    e'_5 &: X_m^\star \to \emptyset . \notag 
\end{align}

\begin{figure}[tb]
\centering
\resizebox{6cm}{!}{%
    \begin{tikzpicture}[bend angle=45] 
    \node[species,color=mybrightblue] (v1) at (0,0) {$X_m^\star$}; 
    \node[species,color=mybrightblue] (v2) at (2,0) {$X_m$}; 
 
    \node (d1) at (0,1.65) {};
    \node (d2) at (0,-1.65) {};
    \node (d3) at (3.6,0) {};
    
    \draw[-latex,color=mydarkred, line width=0.5mm] (d1) edge node[left]  {$e_4$} (v1); 
    \draw[-latex,line width=0.5mm,out=30,in=150] (v1) edge node[above] {$e_1$} (v2); 
    \draw[-latex,line width=0.5mm,out=210,in=-30] (v2) edge node[below] {$e_2$} (v1);

    \draw[-latex,color=mydarkred, line width=0.5mm] (d3) edge node[above]  {$e_3$} (v2);
  
    \draw[-latex,color=mygreen,line width=0.5mm] (v1) edge node[left] {$e_5$} (d2); 

     \node at (2,1) { \scalebox{0.8} { \color{mydarkred}Set-point encoding}};
     \node at (0.7,-1) { \scalebox{0.8} { \color{mygreen}Sensing}};
     \node at (2.5,-1) { \scalebox{0.8} { \color{mybrightblue}Internal model}};
 
\end{tikzpicture}
}
\caption{
Effective reaction network of the example of bacterial chemotaxis. 
} 
\end{figure}
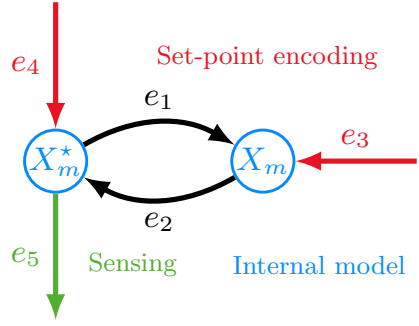

The time evolution of the concentrations of $X_m$ 
and $X_m^\star$ is governed by\footnote{
Here, we use the same symbol to denote 
the species and its concentration 
to simplify notations. 
}
\begin{align}
\dot X_m & = - k_2 X_m + k_1 X^\star_m  + k_3 , \\
\dot X^\star_m & = k_2 X_m 
- k_1 X^\star_m - k_5 X_m^\star + k_4 .
\end{align}
Note that each rate constant 
is proportional to the corresponding enzyme concentration, 
$k_5 \propto B$, and $k_3, k_4 \propto R$,
which are treated as constants here. 
At the steady state, the concentrations and reactions rates are given by 
\begin{equation}
\bar X_m^\star = \frac{k_3 + k_4}{k_5}, 
\quad 
\bar X_m = 
\frac{1}{k_2}
\left( 
k_3 + \frac{k_1 (k_3 + k_4)}{k_5}
\right)
, 
\end{equation}
\begin{equation}
\bar{\bm r} = 
\frac{k_1 (k_3 + k_4)}{k_5} \bm c^{(1)}
+ k_3 \, \bm c^{(2)}
+ k_4 \, \bm c^{(3)},
\end{equation}
where we have taken the basis of $\ker S$ as 
$\bm c^{(1)} \coloneqq 
\begin{bmatrix}
  1 & 1 & 0 & 0 & 0
\end{bmatrix}^\top
$,
$\bm c^{(2)} \coloneqq 
\begin{bmatrix}
  0 & 1 & 1 & 0 & 1
\end{bmatrix}^\top
$,
and 
$\bm c^{(3)} \coloneqq 
\begin{bmatrix}
  0 & 0 & 0 & 1 & 1
\end{bmatrix}^\top
$.
Notably, $\bar X^\star_m$ does not depend on $k_1 (a)$, which means that 
it is insensitive to the ligand concentration $a$. 
This indicates that $X^\star_m$ exhibits RPA with respect to the ligand concentration. 
As pointed out in Ref.~\cite{yi2000robust}, RPA in this system is realized through an integral feedback control. 
The corresponding integrator is identified as 
\begin{equation} 
\frac{d}{dt} 
\left( X^\star_m + X_m \right)
= - k_5 X_m^\star + k_3 + k_4
= - k_5 \delta X_m^\star , 
\end{equation}
where $\delta X_m^\star \coloneqq X_m^\star - \bar X_m^\star$ is the 
deviation of $X_m^\star$ from its steady-state value. 
This equation drives the value of $X_m^\star$ to $\bar X^\star_m$ 
regardless of the value of $k_1$ and $k_2$. 

Let us try to understand this phenomenon using the topological analysis.
The labeled buffering structures in the reaction network~\eqref{eq:chemo-reactions-2} 
are identified as 
\begin{align}
\gamma_1 &= (\{ X_m \}, \{ e_1,e_2 \} \cup \{ \}), \nonumber \\ 
\gamma^\ast_2 &= 
(\{ X_m \}, \{ \} \cup \{ e_2 \}), \nonumber \\ 
\gamma_3 &= 
(\{ X_m, X^\star_m \}, \{ e_1,e_2, e_3,e_5 \} \cup \{ \} ), \\
\gamma_4 &= 
(\{ X_m, X^\star_m \}, \{ e_1,e_2, e_4, e_5 \} \cup \{ \} ),\nonumber \\ 
\gamma_5 &= 
(\{ X_m, X^\star_m \}, \{ e_1,e_2 \} \cup \{ e_5 \}). \nonumber
\end{align}
We can conclude the following from the identified buffering structures: 
\begin{itemize}
    \item The independence of $\bar X_m^\star$ on $k_1$ and $k_2$ is understood from 
    the buffering structure $\gamma_1$. 
    Since $X^\star_m$ is outside $\gamma_1$, the parameters $k_1, k_2$ do not affect the steady-state value of $X^\star_m$. 
    \item The independence of all the reaction fluxes on $k_2$ 
    is understood from the strong buffering structure $\gamma_2^\ast$. 
\end{itemize}
The following two points are also consistent with the identified buffering structure, 
although these are trivial (since the rates of $e_3$ and $e_4$ are given externally): 
\begin{itemize}
    \item The rate $\bar r_4$ depends only on $k_4$, which can be explained by $\gamma_3$ (and partially by $\gamma_5$). 
    \item The rate $\bar r_3$ depends only on $k_3$, which can be explained by $\gamma_4$ (and partially by $\gamma_5$). 
\end{itemize}

\subsubsection{ Sniffer system } \label{sec:sniffer_system}

We note that, the RPA properties found through the identification of buffering structures are kinetics-independent ones. 
For example, 
the perfect adaptation realized through 
the {\it incoherent feedforward} (IFF) structure~\cite{https://doi.org/10.1038/msb.2011.49}
cannot be detected through buffering structures, 
since the adaptation here depends on the choice of kinetics 
and will disappear once the kinetics deviate from it. 
In this sense, the perfect adaptation by IFF is not so robust. 
As a example, let us consider the {\it sniffer system}~\cite{TYSON2003221,10041993}, 
\begin{align}
    e_1 &: v_3 \to v_1 + v_3 ,   \notag \\
    e_2 &: v_1 + v_2 \to v_2, \notag \\
    e_3 &: v_3 \to v_2 + v_3 ,\\
    e_4 &: v_2 \to \emptyset , \notag \\
    e_5 &: \emptyset \to v_3 , \notag \\
    e_6 &: v_3 \to \emptyset . \notag
\end{align}
The rate equations in mass action kinetics are given by 
\begin{align}
\dot x_1 &=  k_1 x_3 - k_2 x_1 x_2 ,\label{eq:ex-iff-x1} \\
\dot x_2 &= k_3 x_3 - k_4 x_2 , \label{eq:ex-iff-x2}\\
\dot x_3 &= k_5 - k_6 x_3 . 
\end{align}
The effect of $v_3$ on $v_1$ enters in two ways: 
directly via $e_1$ and indirectly via $e_3$ and $e_2$. 
Adaptation occurs due to the cancellation of the two effects.
The steady-state concentrations and reaction rates are 
\begin{equation}
\bar x_1 = \frac{k_1 k_4}{k_2 k_3}, 
\quad 
\bar x_2 = \frac{k_3 k_5}{k_4 k_6}, 
\quad 
\bar x_3 = \frac{k_5}{k_6},
\quad
\bar{\bm r}    
= 
\frac{k_1 k_5}{k_6} \bm c^{(1)}
+ 
\frac{k_3 k_5}{k_6} \bm c^{(2)}
+
k_5 \, \bm c^{(3)} .
\end{equation}
where we took a basis of $\ker S$ as 
$\bm c^{(1)}\coloneqq 
    \begin{bmatrix}
      1 & 1 & 0 & 0 & 0 & 0
    \end{bmatrix}^\top$, 
$\bm c^{(2)}\coloneqq 
    \begin{bmatrix}
      0 & 0 & 1 & 1 & 0 & 0
    \end{bmatrix}^\top
$, 
and 
$\bm c^{(3)}\coloneqq\begin{bmatrix} 0 & 0 & 0 & 0 & 1 & 1   \end{bmatrix}^\top
$.
The steady-state concentration $\bar x_1$ is independent of $k_5$ and $k_6$, 
which determine the value of $\bar x_3$. 
However, this independence is due to the fact that
$x_3$ enters linearly in Eqs.~\eqref{eq:ex-iff-x1} and \eqref{eq:ex-iff-x2}. 
Once the kinetics of $r_1$ or $r_2$ deviates from the linear form, 
$\bar x_1$ becomes dependent on $k_5$ and $k_6$. 
In this sense, this adaptation is a kinetics-dependent feature
and hence is not quite robust.

The labeled buffering structures are identified as 
\begin{align}
\gamma_1 &= (\{v_1 \},\{e_1,e_2 \} \cup \{ \}),\nonumber \\
\gamma^\ast_2 &= (\{v_1 \}, \{ \} \cup \{e_2 \} ),\nonumber \\
\gamma_3 &= (\{v_1,v_2 \},\{e_3,e_4 \} \cup \{ e_2 \} ), \\
\gamma^\ast_4 &= (\{v_1,v_2 \},\{ \} \cup \{e_2,e_4 \} ),\nonumber \\
\gamma_5 &= (\{v_1,v_2,v_3 \},\{e_1,e_2,e_3,e_4,e_5, e_6\} \cup \{ \} ) ,\nonumber
\\
\gamma_6 &= (\{v_1,v_2,v_3 \},\{e_1,e_2,e_3,e_4\}\cup \{ e_6\} ) .\nonumber
\end{align}
Indeed, all the buffering structures include $v_1$, 
and hence $\bar x_1$ can depend on every parameter for a generic choice of kinetics. 
We can infer the following RPA properties from the labeled buffering structures:
\begin{itemize}
\item $\bar x_3 \perp k_2, k_3, k_4$ is explained by $\gamma_3$, where $\perp$ means ``is independent of.''
\item $\bar x_2, \bar x_3 \perp k_1, k_2$ is explained by $\gamma_1$. 
\item $\bar r_1, \bar r_5, \bar r_6 \perp k_2, k_3, k_4$ is explained by $\gamma_3$. 
\item $\bar r_5, \bar r_6 \perp k_1,k_2,k_3,k_4$ is explained by $\gamma_1 \cup \gamma_3$. 
\item $\bar {\bm r} \perp k_2, k_4$, because of the strong buffering structure $\gamma_4^\ast$. 
\end{itemize}
The RPA properties above deduced from the buffering structures are kinetics-independent, 
unlike the property $\bar x_1 \perp k_5, k_6$, which is only true for mass-action kinetics and thus is not captured by buffering structures.

\subsubsection{ Ethanol production by yeasts }\label{sec:yeast}

\begin{figure}[tb]
  \centering
  \includegraphics
  [clip, trim=0cm 0cm 0cm 0cm, scale=0.40]
  {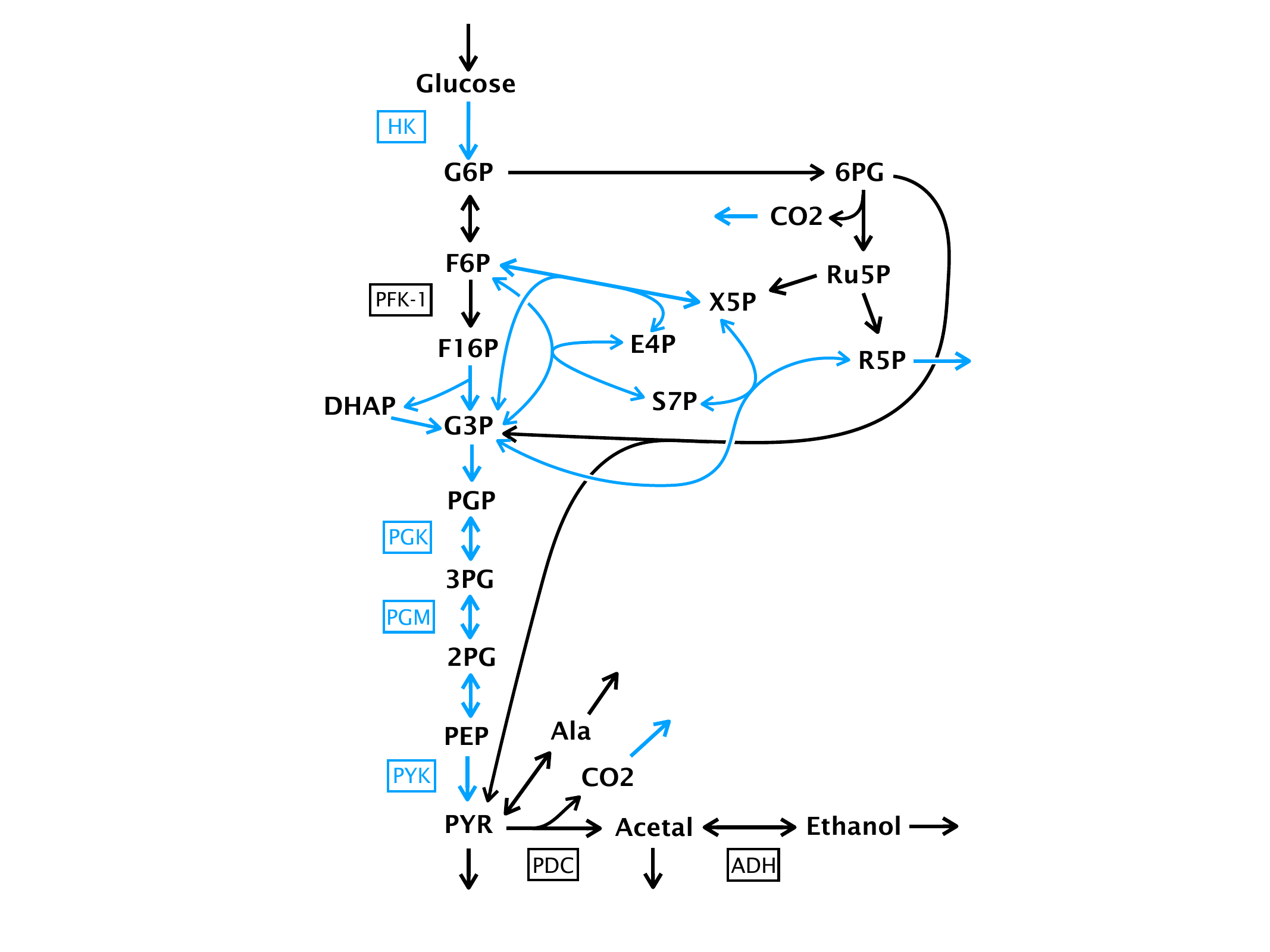} 
  \caption{
Glycolytic pathways of {\it S. cerevisiae}. 
Boxes indicate enzymes catalyzing the corresponding reactions. 
Abbreviations for metabolites: 
2PG, 2-phosphoglycerate;
3PG, 3-phosphoglycerate;
6PG, 6-phosphogluconate;
Acetal, acetaldehyde;
Ala, alanine;
DHAP, dihydroxyacetone phosphate;
E4P, erythrose 4-phosphate;
F6P, fructose 6-phosphate;
F16P, fructose 1,6-bisphosphate;
G3P, glyceraldehyde 3-phosphate;
G6P, glucose 6-phosphate;
PEP, phosphoenolpyruvate;
PGP, 1,3-bisphospho glycerate;
PYR, pyruvate;
R5P, ribose 5-phosphate;
Ru5P, ribulose 5-phosphate;
S7P, sedoheptulose 7-phosphate;
X5P, xylulose 5-phosphate.
Abbreviations for enzymes: 
ADH, alcohol dehydrogenase;
HK, hex-kinase;
PDC, pyruvate decarboxylase;
PFK-1, phosphofructokinase-1;
PGK, phosphoglycerate kinase;
PGM, phosphoglycerate mutase;
PYK, pyruvate kinase. }
  \label{fig:s-c-crn} 
\end{figure}

As a realistic example, let us discuss overexpression experiments 
of metabolic enzymes of the yeast {\it Saccharomyces cerevisiae}, 
which has been extensively used 
in metabolic engineering~\cite{doi:10.1128/MMBR.64.1.34-50.2000}. 
The quantity of interest here is the ethanol production. 
Overexpression experiments 
of glycolytic enzymes were performed 
for the purpose of increasing the flux to ethanol, but they did not lead to a significant improvement.
We revisit the results of these experiments 
from a topological perspective.

We employed the pathways shown in 
Fig.~\ref{fig:s-c-crn},
which consists of glycolysis and the pentose phosphate pathway (PPP) (see Sec.~\ref{sec:sc-reactions} for the list of reactions).
We identified all the labeled buffering structures in this network, which we list in Appendix~\ref{sec:sc-lbs}. 
Based on this, 
we selected the reactions that are included 
in buffering structures which do not include the ethanol production reaction (27 in the list
in Appendix~\ref{sec:sc-reactions}).
Namely, the perturbations of the parameters 
of the blue reactions do not affect ethanol production. 
For example, 
the overexpression of hexokinase (HK)
by 13.9 fold, 
which catalyzes the phosphorylation 
Glucose $\rightarrow$ G6P, 
has led to the improvement of only 7 \% 
compared to the wild type (see Table~\ref{table:yeast-ex-results}). 
Similarly, the overexpressions of 
phosphoglycerate kinase (PGK), 
phosphoglycerate mutase (PGM),
or 
pyruvate kinase (PYK) 
did not result in a significant increase in ethanol production. 
Although the amount of enzymes are increased severalfold, 
the ethanol production shows strong resistance to these changes.
These results are consistent with the expectation from buffering structures, 
since the reactions catalyzed 
by HK, PGK, PGM, PYK are 
inside buffering structures which do not include the flux to ethanol. 
Namely, the strong resilience of 
{\it S. cerevisiae} against these overexpressions can be understood as a result of integral feedback control associated with buffering structures, meaning that it is a consequence of the network topology. 
On the other hand, 
the black reactions can in general affect ethanol production. 
Thus, the overexpressions of 
phosphofructokinase-1 (PFK-1),
pyruvate decarboxylase (PDC),
or 
alcohol dehydrogenase (ADH)
can potentially increase ethanol production 
(see the lower part of Table~\ref{table:yeast-ex-results}
for the corresponding experimental results). 
Although we cannot make a definite judgement 
as to whether the experimental results support this, 
the identification of buffering structures allows us to narrow down the candidates for overexpression of enzymes.

\begingroup
\setlength{\tabcolsep}{10pt} 
\renewcommand{\arraystretch}{1.5} 

\begin{table}[tb]
\centering
\begin{tabular}{| c | c | c| c|} 
 \hline
 Enzyme
  & 
Overexpression fold 
&
Flux to ethanol (\% WT) 
&
Refs. 
\\ 
\hline\hline
 HK
 & 
 13.9
 &
 107
 &
 \cite{https://doi.org/10.1002/yea.320050408} 
  \\ 
 \hline 
 PGK
 &
 7.5
 &
 97
 &
 \cite{https://doi.org/10.1002/yea.320050408}
 \\ 
 \hline 
 PGM
 &
 12.2
 &
 107
 &
 \cite{https://doi.org/10.1002/yea.320050408}
 \\
 \hline
 PYK
 & 
 8.6
 &
 107
 &
 \cite{https://doi.org/10.1002/yea.320050408} 
  \\
 \hline
 \hline 
 PFK-1 (anaerobic)
 & 
 4.6
 &
 106
 &
 \cite{davies1992effects}
 \\ 
 \hline 
 PFK-1 (aerobic) 
 &
 4.6
 &
 130
 &
 \cite{davies1992effects}
 \\
 \hline
 PDC
 & 
 3.7
 &
 85
 &
 \cite{https://doi.org/10.1002/yea.320050408} 
 \\
 \hline
 ADH
 & 
 4.8
 &
 89
 &
 \cite{https://doi.org/10.1002/yea.320050408} 
 \\
 \hline
\end{tabular}
\caption{
Overexpression experiments of 
glycolytic enzymes of {\it S. cerevisiae}. 
}
\label{table:yeast-ex-results}
\end{table}
\endgroup

\section{Topological characterization of kinetics-independent maxRPA networks}\label{sec:equivalence}

So far we have discussed two distinct approaches (topological and control-theoretical) toward the study of RPA properties in chemical reaction networks. A natural question is how these approaches are related. In this section, we provide this connection by proving the equivalence of the two approaches in the maxRPA setting, which is going to give us an insight into the construction of integrators for a generic RPA property.

In the following, 
we first introduce the the decomposition of the influence index in Sec.~\ref{sec:decomposition},
which will be used in the proof. 
In Sec.~\ref{sec:example}, we describe examples 
to highlight the relation of the two approaches and gain an insight into the strategy to prove the equivalence. 
The proof will be completed in Sec.~\ref{sec:proof-th1}.

\subsection{ Decomposition of the influence index }\label{sec:decomposition}

To prove the equivalence, a crucial role is played by 
a decomposition of the influence index shown in Ref.~\cite{PhysRevResearch.3.043123}, which we introduce in this subsection.
The decomposition plays important roles in later analyses as well when we discuss integrators for generic RPA properties. 

In the following, we often choose a subnetwork,
and let us first introduce the notations associated with such a choice.
A subnetwork is specified by subsets of species and reactions, 
$\gamma = (V_\gamma, E_\gamma)$, 
with $V_\gamma \subset V$ and $E_\gamma \subset E$. 
We refer to the chemical species 
and reactions inside $\gamma$ 
as {\it internal}, and 
those in $\Gamma \setminus \gamma$ as {\it external}. 
Accordingly, the stoichiometric matrix $S$ 
can be partitioned into four blocks,
\begin{equation}
  S = 
\begin{bmatrix}
  S_{11} & S_{12} \\
  S_{21} & S_{22} 
\end{bmatrix}. 
\label{eq:s-sepa}
\end{equation}
where $1$ and $2$ correspond to 
internal and external degrees of freedom, respectively.

Now let us introduce the decomposition
of the influence index~\cite{PhysRevResearch.3.043123},
\begin{equation}
  \lambda (\gamma)
  = \widetilde c (\gamma) + d_l (\gamma)  - \widetilde d (\gamma) . 
  \label{eq:lambda-decom}
\end{equation}
The definition and meaning of each term is as follows: 
\begin{itemize}
    \item 
The first term is defined as     
$\widetilde{c}(\gamma)\coloneqq |\widetilde{\rm C}(\gamma)|$,
which is the dimension of the following space,
\begin{equation}
 {\rm \widetilde C}(\gamma)  
\coloneqq 
 \ker S_{11}\, / \,
 (\ker S)_{{\rm supp }\gamma} . 
\end{equation}
The term $\widetilde c(\gamma)$ 
represents the number of {\it emergent cycles} in $\gamma$. 
Intuitively, $\widetilde c(\gamma)$ is the number of cycles in subnetwork $\gamma$, 
that are not cycles in the whole network $\Gamma$. 
\item The second term is 
defined as $d_l (\gamma) \coloneqq |D_l (\gamma)|$
with
\begin{align}
D_l (\gamma) & \coloneqq (\coker S) / X(\gamma) , \label{eq:def-Dl}
\\
X (\gamma) 
&\coloneqq 
  \left\{ 
    \begin{bmatrix}
      \bm d_1 \\
      \bm d_2
    \end{bmatrix}
    \in \coker S 
    \, \middle| \,
    \bm d_1 \in \coker S_{11}
    \right\}. 
 \label{defn:xgamma}
\end{align}
This counts the number of conserved quantities in $\Gamma$ 
whose projections to $\gamma$ are not conserved inside $\gamma$. 
We call such conserved quantities as {\it lost conserved quantities}.
\item 
The third term 
$\widetilde{d}(\gamma) \coloneqq |\widetilde D(\gamma)|$
is the dimension of the space $\widetilde D(\gamma)$, 
which is defined as
\begin{align}
  \widetilde D(\gamma)
  &\coloneqq \coker S_{11} / D_{11}(\gamma) ,
\\
\label{defn:d11gamma}
  D_{11}(\gamma)
  &\coloneqq 
  \left\{ 
    \bm d_1 \in \coker S_{11} 
  \, \middle| \,
  \exists \, \bm d_2 {\text{ such that }} 
  \begin{bmatrix}
    \bm d_1 \\
    \bm d_2
  \end{bmatrix}
  \in \coker S 
  \right\} .
\end{align}
This counts the number of conserved quantities in $\gamma$ 
that cannot be extended to conserved quantities in $\Gamma$. 
We call such them as 
{\it emergent conserved quantities} in $\gamma$. 
\end{itemize}
Note that 
the integers $\widetilde c(\gamma),d_l(\gamma),$ and $\widetilde d(\gamma)$ are all nonnegative by definition.
We describe a linear-algebraic procedure to obtain 
the bases of $\wt{D}(\gamma)$, $\wt{\rm C}(\gamma)$
in Appendix~\ref{sec:find-basis}, that will be used later to construct integrator equations for a generic RPA property and also used in the discussion of manifold RPA. 
{\bf RPAFinder}~\cite{RPAFinder} implements methods to obtain bases of these spaces.

For later purposes, let us also introduce the decomposition of $\coker S$ (see Ref.~\cite{PhysRevResearch.3.043123} for derivation),
\begin{equation}
\coker S \simeq D_{11}(\gamma) \oplus D_l (\gamma) \oplus \bar D'(\gamma) , 
\label{eq:decom-coker-s}
\end{equation}
where $\bar D'(\gamma)$ denotes the space of conserved quantities of $\Gamma$ supported in $\Gamma \setminus \gamma$,
\begin{equation}
\label{defn:dprime_gamma}
\bar D' (\gamma) 
\coloneqq 
  \left\{ 
    \begin{bmatrix}
      \bm 0 \\
      \bm d_2
    \end{bmatrix}
    \in \coker S 
    \right\}
.
\end{equation}

We find that these quantities are in fact relevant in the description of maxRPA. 
In the maxRPA setting, the subnetwork $\gamma$ 
contains all the species and reactions other than $X$ and $\{e_{\bar 1}, e_{\bar 2}\}$, 
and the $M \times N$ dimensional stoichiometric matrix $S$ is partitioned into four blocks where $S_{11}$, $S_{12}$, $S_{21}$ and $S_{22}$ 
are of dimensions 
$(M-1) \times (N-2)$, 
$(M-1) \times 2$, 
$1 \times (N-2)$ and 
$1 \times 2$ respectively.
When $q_M = 0$, we can write vector $\bm q$ as 
$\bm q = 
\begin{bmatrix} \bm q_1 & 0 \end{bmatrix}
$, 
and due to Eq.~\eqref{eq:qs2} the $(M-1)$ dimensional vector $\bm q_1$ must satisfy
\begin{align}
\label{maxrpa_vectowq_1}
\bm q_1^\top S_{11} = \bm 0, \quad  \bm q_1^\top S_{12} 
= 
\begin{bmatrix}
\kappa & -1 
\end{bmatrix}
\quad \textnormal{and} \quad \bm q_1^\top \bar{D} = \bm 0,
\end{align}
where $\bar{D}$ is the $(M-1) \times |\bar \alpha|$ matrix formed by the first $(M-1)$ rows of $D$. This shows that $\bm q_1 \in \coker S_{11}$ and $\bm q_1$ cannot be extended to a vector in $\coker S$ by adding a component\footnote{This is because if there is a scalar $q_2$ such that $\begin{bmatrix}
    \bm q_1 \\
    q_2
\end{bmatrix} \in \coker S$, then since the columns of $D$ span $\coker S$, there exists a vector $y$ such that $D y = \begin{bmatrix}
    \bm q_1 \\
    q_2
\end{bmatrix}.$ This shows that 
$\|\bm q_1\|^2 = \begin{bmatrix} \bm q_1^\top & 0 \end{bmatrix}
\begin{bmatrix}
    \bm q_1 \\
    q_2
\end{bmatrix} = 
\begin{bmatrix}    
\bm q_1^\top & 0
\end{bmatrix}
D y = \bm q_1^\top \bar{D} y = 0$ (due to the last relation in Eq.~\eqref{maxrpa_vectowq_1}) and hence $\bm q_1 = \bm 0$ which is a contradiction as the second relation in Eq.~\eqref{maxrpa_vectowq_1} must hold.}. Hence, the vector $\bm q_1$ is an \emph{emergent conserved quantity} for the subnetwork $\gamma$. 
If the maxRPA network is kinetics-independent (see Definition \ref{defn:kinetics_indep_maxrpa}), 
then reactions $e_{\bar 1}$ and $e_{\bar 2}$ can only have the output species $v_M = X$ as a reactant which implies that all nonzero entries of $S_{12}$ must be positive. 
However, the second relation in Eq.~\eqref{maxrpa_vectowq_1} can hold only when $\bm q$ has both positive and negative components. 
Hence, a kinetics-independent maxRPA network can never be homothetic and it must be antithetic if $q_M=0$.
Interestingly, the same conclusion can be drawn for stochastic maxRPA networks (see Ref.~\cite{gupta2022universal}). Note that, as the example in Sec.~\ref{subsec:maxrpa_homo} shows, kinetics-independent maxRPA networks can be homothetic when $q_M \neq 0$.

The following proposition summarises the discussion in the previous paragraph.
\begin{proposition}
\label{prop:maxrpa}
Suppose $\Gamma$ is a kinetics-independent maxRPA network characterized by a pair $(\bm q, \kappa)$ that satisfies Eq.~\eqref{eq:qs2}. Suppose that the last component of $\bm q$ is $0$, and let $\bm q_1$ be the vector obtained by removing this last component from $\bm q$. Then this network must be antithetic and $\bm q_1$ is an emergent conserved quantity for the subnetwork $\gamma$ given by Eq.~\eqref{defn_gamma}.
\end{proposition}

\subsection{ Examples }\label{sec:example}

In this subsection, we compare the 
control-theoretical and topological approaches through simple examples that exhibit maxRPA to highlight the roles of the concepts introduced in the previous subsection. 
The intuition from these examples is going to be straightforwardly generalized to give the proof of the equivalence in Sec.~\ref{sec:proof-th1}.

\subsubsection{ Example: Homothetic case }\label{subsec:maxrpa_homo}

Let us discuss an example which corresponds a homothetic maxRPA network. 
The network consists of two species $\{x,y\}$ and 
three reactions,
\begin{align}
e_1&: \emptyset \to y ,\notag  \\
e_2&: y \to x, \\
e_3&: x \to \emptyset . \notag 
\end{align}
The stoichiometric matrix reads 
\begin{equation}
S =
\begin{blockarray}{cccc}
&& \\
\begin{block}{c[ccc]}
{y} \quad\,
& -1 & 1 & 0  \\
{x} \quad\, 
& 1 & 0 & - 1  \\
\end{block}
& e_2 & e_1 & e_3 
\end{blockarray} \,\,.    
\end{equation}
If we use mass-action kinetics, the rate equations are 
\begin{eqnarray}
\dot y &= k_1 - k_2 y ,
\\
\dot x &= k_2 y - k_3 x .
\end{eqnarray}
The steady-state concentrations are given by 
\begin{equation}
\bar x = \frac{k_1}{k_3}, 
\quad 
\bar y = \frac{k_1}{k_2}. 
\end{equation}
This is a maxRPA network for species $x$:
the $\bm q$ vector is identified as 
\begin{equation}
\bm q = 
\begin{bmatrix}  1 & 1 \end{bmatrix}^\top . 
\label{eq:q-vector-ex-homothetic}
\end{equation}
Indeed, it satisfies the condition \eqref{eq:qs}, 
\begin{equation}
\bm q^\top S = 
\begin{bmatrix}  0 & 1 & -1 \end{bmatrix}. 
\end{equation}
Since the components of $\bm q$ are all positive, 
this is a homothetic maxRPA network. 
The integrator is given by 
\begin{equation}
\frac{d}{dt} \bm q \cdot \bm x 
= 
\frac{d}{dt} (x + y)
= k_1 - k_3 x 
= -k_3 \delta x
, 
\end{equation}
where $\delta x \coloneqq x - \bar x$. 
This equation drives $x$ toward its steady-state value. 
In this example, 
$e_3$ is the sensing reaction, and $e_1$ is the set-point encoding reaction. 
The maxRPA property of this network corresponds to the fact that 
the target value of $x$, $\bar x$, is independent of $k_2$.

Let us examine this example with the topological analysis.
Labeled buffering structures in this system are 
\begin{align}
\gamma_1 &= (\{x,y\},\{e_1, e_2,e_3 \} \cup \{ \} ) (=\Gamma ) ,\nonumber \\
\gamma^\ast_2 &= (\{y\},\{ \} \cup \{ e_2 \} ), \\
\gamma^\ast_3 &= (\{x\},\{ \} \cup \{ e_3 \} ).\nonumber 
\end{align}
The subnetworks $\gamma_2^\ast$ and $\gamma_3^\ast$ 
are strong buffering structures.
The fact that $\bar x$ does not depend on $k_2$ 
is explained by $\gamma_2^\ast$, so this subnetwork is responsible for the maxRPA behavior. 
The decomposition of $\lambda (\gamma_2^\ast)$ reads 
\begin{equation}
\lambda (\gamma_2^\ast) 
= -1 + 1 - 0 + 0
=
\underset{=\widetilde c(\gamma_2^\ast)}{ 0 }
+
\underset{=d_l(\gamma_2^\ast)}{ 0 }
- 
\underset{=\widetilde d(\gamma_2^\ast)}{ 0 }
= 0.
\end{equation}
In the stoichiometric matrix, 
the subnetwork $\gamma_2^\ast$ can be highlighted as 
\begin{equation}
\begin{tikzpicture}
\node at (0, 0) {
$
S =
\begin{blockarray}{cccc}
&& \\
\begin{block}{c[ccc]}
{y} \quad\,
& -1 & 1 & 0  \\
{x} \quad\, 
& 1 & 0 & - 1  \\
\end{block}
& e_2 & e_1 & e_3 
\end{blockarray} \,\,.
$
}; 
\draw[mydarkred, dashed,line width=1] 
 (-0.4,0.2) rectangle (0.4, 0.7);
\node at (0,1) {\color{mydarkred}$S_{11}$} ;

\end{tikzpicture} 
\end{equation}
As we discuss later generically, 
in this case, the $\bm q$ vector is identified as 
\begin{equation}
\bm q = 
\begin{bmatrix}
    - S_{21}S^+_{11} & 1 
\end{bmatrix}^\top, 
\label{eq:q-vec-case1-ex-homothetic}
\end{equation}
where $S_{11}^+$ denote the Moore-Penrose inverse of $S_{11}$.
In this example, 
$S_{11}=\begin{bmatrix} - 1\end{bmatrix}$
and 
$S_{21}=\begin{bmatrix} 1\end{bmatrix}$, 
which are $1 \times 1$ matrices. 
Here, $S_{11}$ is invertible, 
and $S_{11}^{+} = S_{11}^{-1} = \begin{bmatrix}  -1 \end{bmatrix}$. 
Thus, we can see that 
Eq.~\eqref{eq:q-vec-case1-ex-homothetic} reproduces 
Eq.~\eqref{eq:q-vector-ex-homothetic}. 

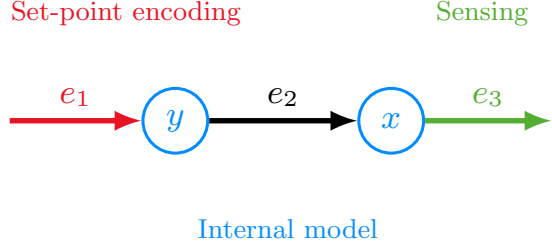
\begin{figure}[tb]
\centering
\resizebox{8cm}{!}{%
    \begin{tikzpicture}

    \node[species,color=mybrightblue] (v1) at (0,0) {$y$}; 
    \node[species,color=mybrightblue] (v2) at (2,0) {$x$}; 
 
    \node (d1) at (-1.65,0) {};
    \node (d2) at (0,-1.65) {};
    \node (d3) at (3.6,0) {};
    
    \draw[-latex,color=mydarkred, line width=0.5mm] (d1) edge node[above] {$e_1$} (v1); 
    \draw[-latex,line width=0.5mm] (v1) edge node[above] {$e_2$} (v2); 

    \draw[-latex,color=mygreen, line width=0.5mm] (v2) edge node[above]  {$e_3$} (d3);
  
     \node at (-0.5,1) { \scalebox{0.8} { \color{mydarkred}Set-point encoding}};
     \node at (2.8,1) { \scalebox{0.8} { \color{mygreen}Sensing}};
     \node at (1,-1) { \scalebox{0.8} { \color{mybrightblue}Internal model}};
 
\end{tikzpicture}
}
\caption{ Example of a homothetic maxRPA network.  } 
\end{figure}

\subsubsection{ Example: Antithetic case }\label{sec:ex-antithetic}

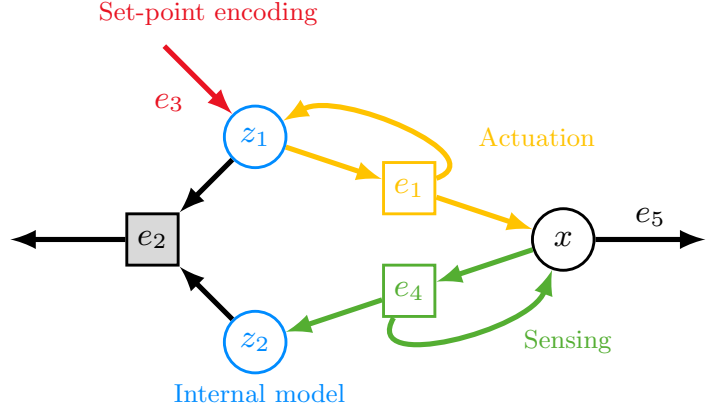
\begin{figure}[tb]
\centering
\resizebox{10cm}{!}{%
    \begin{tikzpicture}[bend angle=45] 
    \node[species,color=mybrightblue] (z1) at (0,2) {$z_1$}; 
    \node[species,color=mybrightblue] (z2) at (0,0) {$z_2$};
    \node[species] (x) at (3,1) {$x$}; 

    \node[rectangle,draw=black!100,thick, inner sep=0pt,minimum size=5mm,color=mybrown] (e1) at (1.5, 1.5) {$e_1$}; 
    
    \node[rectangle,draw=black!100,thick, inner sep=0pt,minimum size=5mm,color=mygreen] (e4) at (1.5, 0.5) {$e_4$}; 

    \node[reaction] (e2) at (-1, 1) {$e_2$}; 

    \node (d1) at (-2.5,1) {};    
    \node (d2) at (4.5,1) {};
    \node (d3) at (-1,3) {};

    \draw [-latex,line width=0.5mm] (x) edge node[above]{$e_5$} (d2);
    \draw [-latex,line width=0.5mm] (e2) -- (d1);
    \draw[-latex, line width=0.5mm, color=mydarkred] (d3) edge node[below left]{$e_3$} (z1);

    \draw[-latex, line width=0.5mm] (z1) -- (e2);  
    \draw[-latex, line width=0.5mm] (z2) -- (e2);

    \draw[-latex, line width=0.5mm,color=mybrown] (z1) -- (e1);
    \draw[-latex, line width=0.5mm,color=mybrown] (e1) -- (x);
    \draw[-latex,line width=0.5mm,out=20,in=30,color=mybrown] (e1) edge (z1); 

    \draw[-latex, line width=0.5mm,color=mygreen] (x) -- (e4);
    \draw[-latex, line width=0.5mm,color=mygreen] (e4) -- (z2);
    \draw[-latex,line width=0.5mm,out=240,in=250,color=mygreen] (e4) edge (x); 

     \node at (-0.5,3.2) { \scalebox{0.8} { \color{mydarkred}Set-point encoding}};
     \node at (3,0) { \scalebox{0.8} { \color{mygreen}Sensing}};
     \node at (0,-0.5) { \scalebox{0.8} { \color{mybrightblue}Internal model}};
     \node at (2.7,2) { \scalebox{0.8} { \color{mybrown}Actuation}};

\end{tikzpicture}
}
\caption{ Example of an antithetic maxRPA network. }
\label{fig:ex-x-z1-z2}
\end{figure}

We consider a reaction network consisting of the following set of reactions:
\begin{align}
    e_1 &: z_1 \to x + z_1,  \notag \\
    e_2 &: z_1 + z_2 \to \emptyset, \notag \\
    e_3 &: \emptyset \to z_1 ,\\
    e_4 &: x \to x + z_2 , \notag \\
    e_5 &: x \to \emptyset . \notag
\end{align}    
The structure of the network is shown in Fig.~\ref{fig:ex-x-z1-z2}. 
The corresponding rate equation is 
\begin{align}
\dot x &= k_1 z_1 - k_5 x, \\
\dot z_1 &= - k_2 z_1 z_2 + k_3 ,\\
\dot z_2 &= - k_2 z_1 z_2 + k_4 x. 
\end{align}
In this network, $e_3$ is a set-point encoding reaction, 
and $e_4$ is a sensing reaction.
Reaction $e_1$ is an actuation from the internal model to the 
the rest of the network for countering the disturbances. 
We here employed mass-action kinetics for definiteness, 
but the conclusion does not depend on the choice of kinetics.
The stoichiometric matrix is given by 
\begin{equation}
S =
\begin{blockarray}{cccccc}
&& \\
\begin{block}{c[ccccc]}
{z_2} \quad\,
& -1 & 0 & 0 & 1 & 0  \\
{z_1} \quad\, 
& -1 & 0 & 0 & 0 & 1  \\
{x} \quad\, 
& 0 & 1 & -1 & 0  & 0  \\
\end{block}
& e_2 & e_1 & e_5 & e_4 & e_3
\end{blockarray} \,\,,    
\label{eq:s-anti-ex}
\end{equation}
where rows and columns are arranged 
for later convenience. 
The steady-state concentrations are given by 
\begin{equation}
\bar x = \frac{k_3}{k_4}, 
\quad
\bar z_1 = \frac{k_3 k_5}{k_1 k_4}, 
\quad 
\bar z_2 = \frac{k_1 k_4}{k_2 k_5},
\end{equation}
and the steady-state fluxes are 
\begin{equation}
\bar{ \bm r} = \frac{k_3k_5}{k_4} \bm c^{(1)} + k_3 \, \bm c^{(2)},
\end{equation}
where we have taken the basis vectors of $\ker S$ as 
\begin{equation}    
    \ker S = {\rm span\,} \{ 
    \begin{bmatrix}
      0 & 1 & 1 & 0 & 0
    \end{bmatrix}^\top,
    \begin{bmatrix}
      1 & 0 & 0 & 1 & 1
    \end{bmatrix}^\top
    \}
    \eqqcolon 
{\rm span} \left\{ \bm c^{(1)}, \bm c^{(2)} \right\} .
\end{equation}
Here, the components are arranged in the same order as the column of 
$S$ in Eq.~\eqref{eq:s-anti-ex}.

The value of $\bar x$ is determined only by $k_3$ and $k_4$, and it does not depend on other parameters. 
Namely, this is an example of a maxRPA network.
The vector $\bm q$ in this example is given by 
\begin{equation}
\bm q = 
\begin{bmatrix}
1 & -1 & 0 
\end{bmatrix}^\top .
\label{eq:ex-anti-q-vec}
\end{equation}
Indeed, $\bm q$ satisfies the condition for maxRPA,
\begin{equation}
\bm q^\top S 
= 
\begin{bmatrix}
0 & 0 & 0 & 1 & -1 
\end{bmatrix}. 
\end{equation}
With this $\bm q$ vector, one can identify the integrator as 
\begin{equation}
\frac{d}{dt}
(z_1 - z_2) 
= k_3 - k_4 x 
= - k_4 \delta x, 
\end{equation}
where $\delta x \coloneqq x - \bar x$.

Let us analyze this example based on the topological analysis. 
We can identify the labeled buffering structures of this system as 
\begin{align}
\gamma^\ast_1 &= ( \{ z_1,z_2 \}, \{ \} \cup \{ e_1,e_2 \}) ,\nonumber \\ 
\gamma^\ast_2 &= ( \{ z_2 \}, \{ \} \cup \{ e_2 \}) ,\nonumber \\ 
\gamma_3 &= ( \{x,z_1, z_2 \}, \{ e_1,e_2,e_3,e_4,e_5 \} \cup  \{ \}) (=\Gamma), \\ 
\gamma_4 &= ( \{x,z_1, z_2 \}, \{e_1,e_5 \} \cup \{ e_2,e_4 \}) ,\nonumber \\ 
\gamma_5 &= ( \{z_1, z_2 \}, \{e_1,e_5 \} \cup \{ e_2 \}) . \nonumber
\end{align}
From these buffering structures, 
we can, for example, infer the following: 
\begin{enumerate}
    \item $\bar x \perp k_1,k_2,k_5$ is understood from $\gamma_5$. 
    \item $\bar z_1 \perp k_2$ is explained by $\gamma^\ast_2$. 
    \item All the steady-state reaction fluxes satisfy
    $\bar{\bm r} \perp k_1, k_2$, can be understood from the strong buffering structure $\gamma^\ast_1$. 
\end{enumerate}
Let us focus on $\gamma_5$, which is responsible for 
the maxRPA property of $\bar x$. 
The subnetwork $\gamma_5$ 
corresponds to the following part of 
$S$, 
\begin{equation}
\begin{tikzpicture}
\node at (0, 0) {
$
S =
\begin{blockarray}{cccccc}
&& \\
\begin{block}{c[ccccc]}
{z_1} \quad\, 
& 0 & -1 & 0 & 1 & 0
\\
{z_2} \quad\, 
& 0 & -1 & 0 & 0 & 1 
\\
{x} \quad\,
& 1 & 0 & -1 & 0 & 0 
\\
\end{block}
 & e_1 & e_2 & e_5
 & e_3 & e_4 
\end{blockarray} \,\, ,
$
}; 
\draw[mydarkred, dashed,line width=1] 
 (-0.75,0) rectangle (1.1, 0.9);
\node at (-0.1, 1.25) {\color{mydarkred}$S_{11}$} ;

\end{tikzpicture} 
\end{equation}
where $\gamma_5$ is indicated by a red rectangle. 
In this example, the internal part of $\bm q$ turns out to be an emergent conserved quantity, as we show below.  
The spaces $\ker S_{11}$ and $\coker S_{11}$ are given by
\begin{align}
\ker S_{11}  
&= 
{\rm span}
\{ 
\begin{bmatrix}
1 & 0 & 1    
\end{bmatrix}^\top,
{
\color{mydarkred} 
\begin{bmatrix}
1 & 0 & -1    
\end{bmatrix}^\top
}
\}, \\
\coker S_{11}    
&= 
{\rm span}
\{ 
{
\color{mydarkred} 
\begin{bmatrix}
    1 & -1 
\end{bmatrix}^\top 
}
\} ,
\label{eq:ex-anti-coker-s11}
\end{align}
where colored vectors are emergent ones (note that $\begin{bmatrix}
1 & -1    
\end{bmatrix} \in \coker S_{11}$, but it cannot be extended to an element in $\coker S$, which is trivial in this example).
We can see that the internal part of the $\bm q$ vector \eqref{eq:ex-anti-q-vec} is the same as the emergent conserved quantity in Eq.~\eqref{eq:ex-anti-coker-s11}. 
The network $\gamma_5$ 
has one emergent conserved quantity 
and one emergent cycle,
and the influence index of $\gamma_5$ is decomposed as 
\begin{equation}
\lambda (\gamma_5) 
= -2 + 3 - 1 + 0
=
\underset{=\widetilde c(\gamma_5)}{ 1 }
+
\underset{=d_l(\gamma_5)}{ 0 }
- 
\underset{=\widetilde d(\gamma_5)}{ 1 }
= 0. 
\end{equation}
This example suggests that, when the subnetwork has an emergent conserved quantity, $\wt{d}(\gamma)>0$, we can use it to construct the integrator. 
We will show that this observation is true for generic cases in 
the proof of Theorem~\ref{thm:maxrpa_equivalence}.

\subsection{ Proof of Theorem~\ref{thm:maxrpa_equivalence} } \label{sec:proof-th1}

From the examples in the previous subsection, we have learned that 
the integrators (or equivalently the $\bm q$ vectors satisfying Eq.~\eqref{eq:qs}) are constructed 
in different ways, depending on whether $\wt{d}(\gamma)=0$ or not. 
As we see below, this observation is true for a generic chemical reaction system, and we can now describe the proof of Theorem~\ref{thm:maxrpa_equivalence}.

\begin{proof}[Proof of Theorem~\ref{thm:maxrpa_equivalence}]
As we discussed at the end of Section \ref{sec:maxrpa_charac}, a network $\Gamma$ is kinetics-independent maxRPA if the last two reactions do not have a species other than the output species $X$ as a reactant. 
This restriction of the kinetics of reactions in subnetwork $\bar{\gamma}$ is equivalent to the subnetwork $\gamma$ being output-complete. 
Hence, to prove Theorem~\ref{thm:maxrpa_equivalence}, it suffices to prove that the zero influence index condition (i.e\ $\lambda(\gamma) = 0$) is equivalent to the existence of a pair $(\bm q, \kappa)$ satisfying Eq.~\eqref{eq:qs}.

Let us first assume that $\gamma$ is a buffering structure. 
We note that when $\lambda(\gamma)=0$ is satisfied, 
$d_l (\gamma)=0$ is always true
from the assumption of the existence of a steady state,
as we show later near Eq.~\eqref{eq:num-var-minus-num-eqs-x2}.
Thus, $\wt{c}(\gamma)=\wt{d}(\gamma)$ holds under the current assumption. 
We consider the following two cases,  
$\widetilde c(\gamma)=0$ 
and 
$\widetilde c(\gamma)\neq0$,
separately. 

\begin{itemize}
\item 
Suppose that $\widetilde c (\gamma)=0$
and $\gamma$ does not have an emergent cycle. 
Since $\lambda(\gamma)=0$
by assumption, 
we also have $\widetilde d(\gamma)=0$. 
As we show in Appendix~\ref{sec:iso}, 
when $\widetilde c(\gamma)=\wt{d}(\gamma)=0$, 
we have the following isomorphisms,
\begin{align}
\label{eq:first_iso}
\ker S / \ker S_{11}
&\simeq \ker S', \\
\label{eq:second_iso}
\coker S / \coker S_{11} 
&\simeq \coker S', 
\end{align}
where 
$S' \coloneqq S_{22} - S_{21} S^+_{11} S_{12}$ is the generalized Schur complement. 
Note that $S'$ is $1 \times 2$ matrix in the current setting. 
We can decompose $\coker S$ as in Eq.~\eqref{eq:decom-coker-s},
and since $d_l(\gamma)=|D_l(\gamma)|=0$, we have 
$\coker S \simeq D_{11}(\gamma) \oplus \bar D'(\gamma)$.
Since $\wt{d}(\gamma)=0$, we have 
$\coker S_{11} \simeq D_{11}(\gamma)$ and
\begin{equation}
\coker S / \coker S_{11} \simeq \bar D' (\gamma). 
\end{equation}
Recall that $\bar D'(\gamma)$ is the space of conserved quantities of $\Gamma$ whose support is in $\bar \gamma = \Gamma \setminus \gamma$, which contains only $X$ as species. 
Since $X$ itself does not constitute a conserved quantity\footnote{ If that is the case, its concentration is solely determined by the value of the conserved quantity, which contradicts the assumption that the steady-state concentration of $X$ depends nontrivially on $k_{\bar 1}$ and $k_{\bar 2}$ (see the text near Eq.~\eqref{max_rpa_defn}).
}, 
we should have $\coker S' = \bm 0$.
Thus, $S'$ is of rank 1 and is not 
a zero matrix. 
In this situation, 
we can pick a vector $\bm q$ by 
\begin{equation}
\bm q = 
c
\begin{bmatrix}
    - S_{21}S^+_{11} & 1 
\end{bmatrix}^\top,
\label{eq:q-vec-case1}
\end{equation}
where $c$ is an overall constant. 
Indeed, 
\begin{equation}
\begin{split}
\bm q^\top S 
&= 
c
\begin{bmatrix}
    - S_{21}S^+_{11} & 1 
\end{bmatrix}
\begin{bmatrix}
S_{11} & S_{12} \\
S_{21} & S_{22} 
\end{bmatrix}    
\\
&= 
c 
\begin{bmatrix}
S_{21}
\left(1 - S_{11}^+ S_{11} \right) 
& S'
\end{bmatrix}. 
\end{split}
\end{equation}
Note that $1 - S_{11}^+ S_{11}$ is the projection matrix to $\ker S_{11}$. 
Since here we have 
$\widetilde c(\gamma)=0$,
which is equivalent to $\ker S_{11} \subset \ker S_{21}$, 
we can write $\bm q^\top S$ as 
\begin{equation}
\bm q^\top S 
= 
c 
\begin{bmatrix}
\bm 0^\top 
& S'
\end{bmatrix} . 
\end{equation}
By choosing the constant $c$ appropriately, 
we can define $\bm q$ satisfying Eq.~\eqref{eq:qs}. 

\item Suppose that $\widetilde c(\gamma) \neq 0$. 
Since 
$\widetilde c(\gamma) = \widetilde d(\gamma)$,
there exists an emergent conserved quantity, $\bm q^\top = 
\begin{bmatrix}\bm q_1^\top & 0\end{bmatrix}$,
which satisfies 
$\bm q_1^\top S_{11} = \bm 0 ^\top$. 
By multiplying $\bm q$ on $S$ from the left, 
\begin{equation}
\bm q^\top S 
= 
\begin{bmatrix}
\bm 0^\top & \bm q_1^\top S_{12}
\end{bmatrix}, 
\end{equation}
where $\bm q_1^\top S_{12}$ is not a zero matrix. 
\end{itemize}
In both cases, we can construct a pair 
$(\bm q, \kappa)$ satisfying Eq.~\eqref{eq:qs} 
(note that 
we assume the existence of a steady state solution, and 
$\bm q^\top S \bm r = 0$ should be able to determine $\bar x$, 
which necessitates that the components of $\bm q^\top S$ should have opposite signs and both should be nonzero, since 
 reaction fluxes are positive). 
Thus, we have shown that if $\gamma$ is a buffering structure then kinetics-independent maxRPA also holds. 
\\

Conversely, let us assume that conditions for kinetics-independent maxRPA are satisfied and 
prove that $\gamma$ is a buffering structure. Since $\gamma$ is output-complete (by definition of kinetics-independent maxRPA), we just need to show that its influence index vanishes, $\lambda(\gamma) = 0$. We first consider the case $d_l (\gamma) \neq 0$. As $\bar{\gamma}$ contains only one species, $d_l (\gamma)=1$, and there exists a conserved quantity $\begin{bmatrix}
  \bm d_1 \\
  d_2
\end{bmatrix} \in \coker S$
with $d_2 \neq 0$. This implies that $\bm d_1^\top S_{11} + d_2 S_{21} = \bm 0$. So if we pick any $\bm c_1 \in \ker S_{11}$, then we have $d_2 S_{21} \bm c_1 = \bm 0$ and as $d_2 \neq 0$ we must have $\bm c_1 \in \ker S_{21}$. Since $\bm c_1$ is an arbitrary element of $\ker S_{11}$, 
we have $\ker S_{11} \subset \ker S_{21}$,
which is equivalent to 
$\widetilde c(\gamma)=0$. 
This implies $\lambda(\gamma)= 1 - \widetilde d(\gamma)$. As $\lambda(\gamma)$ must be nonnegative if the system is stable, in order to show that $\lambda(\gamma) = 0$ it suffices to prove that $\widetilde d(\gamma) \geq 1$. We denote the vector $\bm q$ satisfying Eq.~\eqref{eq:qs} as
\begin{equation}
\label{conv_maxrpa_q_decomp}
\bm q^\top = 
\begin{bmatrix}
    \bm q_1^\top & q_2
\end{bmatrix}. 
\end{equation}
Letting $\bm q' = \bm q - \frac{q_2}{d_2} \begin{bmatrix}
  \bm d_1 \\
  d_2
\end{bmatrix}$
we see that $\bm q'$ also satisfies Eq.~\eqref{eq:qs} and its last component is $0$. Then the first $(M-1)$ components of $\bm q'$ form an emergent conserved quantity. This is because 
\begin{equation}
\bm q'^\top S = 
\begin{bmatrix} \bm q'^\top_1 & 0 \end{bmatrix} 
\begin{bmatrix}
    S_{11} & S_{12} \\ 
    S_{21} & S_{22}
\end{bmatrix}
= 
\begin{bmatrix}
  \bm q'^\top_1 S_{11} &   \bm q'^\top_1 S_{12}
\end{bmatrix}
= 
\begin{bmatrix}
\bm 0^\top & \kappa & -1 
\end{bmatrix}, 
\end{equation}
which implies $\bm q'_1 \in \coker S_{11}$ while $\bm q' \notin \coker S$. 
Thus, there must be at least one emergent conserved quantity $\widetilde d(\gamma) \geq 1$, which proves $\lambda(\gamma) = 0$ in the case $d_l (\gamma)=1$\footnote{
In fact, as we discuss in Sec.~\ref{sec:gen-int}, $\lambda(\gamma)=0$ implies $d_l (\gamma)=0$. So the situation $d_l(\gamma)=1$ in fact does not happen when maxRPA is realized, although this does not affect the current proof. 
}.

Now we come to the case $d_l (\gamma)=0$ where the influence index can be written a $\lambda(\gamma)= \wt{c}(\gamma) - \wt{d}(\gamma)$. 
We denote the vector $\bm q$ satisfying Eq.~\eqref{eq:qs} as in Eq.~\eqref{conv_maxrpa_q_decomp}. From the assumption, 
\begin{equation}
\bm q^\top S 
= 
\begin{bmatrix}
\bm q_1^\top S_{11} + q_2 S_{21} & 
\bm q_1^\top S_{12} + q_2 S_{22} 
\end{bmatrix}
= 
\begin{bmatrix}
\bm 0^\top &  \cdots 
\end{bmatrix}. 
\end{equation}
Thus, we have 
\begin{equation}
\bm q_1^\top S_{11} + q_2 S_{21} = \bm 0^\top  .
\label{eq:q1s1+p2s2}
\end{equation}
When $\ker S_{11}$ is trivial, we have $\wt{c}(\gamma)=0$ 
and 
this implies the vanishing of the influence index (note that 
$\lambda(\gamma)$ and $\widetilde d(\gamma)$  are nonnegative). 
When $\ker S_{11}$ is nontrivial, 
let us pick a nonzero element $\bm c_1 \in \ker S_{11}$. 
By multiplying $\bm c_1$ on Eq.~\eqref{eq:q1s1+p2s2} from the right, we have 
\begin{equation}
q_2 S_{21} \bm c_1 = 0  . 
\end{equation}
Suppose that 
$q_2 \neq 0$.
Since $\bm c_1$ is an arbitrary element of $\ker S_{11}$, 
we have $\ker S_{11} \subset \ker S_{21}$,
which is equivalent to 
$\widetilde c(\gamma)=0$. 
This implies $\lambda(\gamma)=0$. 
Thus, $\gamma$ is a buffering structure. 

Let us consider the case $q_2=0$. 
Then, $\bm q_1$ is an emergent conserved quantity,
and $\widetilde d (\gamma)=1$ (note that $0 \le \wt{d}(\gamma) \le |V \setminus V_\gamma|$). 
Since $\lambda(\gamma) = \widetilde c(\gamma)
- 
\widetilde d(\gamma)
$ is nonnegative for an asymptotically stable system, 
$\widetilde c(\gamma) \ge 1$. 
Thus, there exists at least one emergent cycle, 
meaning that we have $\bm c_1$ such that\footnote{
This equation means that susceptible reactions can independently affect the concentration $x$ of the target species.
}
\begin{equation}
S 
\begin{bmatrix}
    \bm c_1 \\
    \bm 0
\end{bmatrix}
= 
\begin{bmatrix}
    \bm 0 \\
    v
\end{bmatrix},
\label{eq:sc1-v}
\end{equation}
with $v \neq 0$. 
In fact, there is only one emergent cycle, $\wt{c}(\gamma)=0$.
To see this, suppose that there exists $\bm c_2$ with the same property. 
It can be normalized to satisfy 
\begin{equation}
S 
\begin{bmatrix}
    \bm c_2 \\
    \bm 0
\end{bmatrix}
= 
\begin{bmatrix}
    \bm 0 \\
    v
\end{bmatrix}. 
\label{eq:sc2-v}
\end{equation}
By taking the difference 
of Eqs.~\eqref{eq:sc1-v} and \eqref{eq:sc2-v}, 
\begin{equation}
S
\begin{bmatrix}
 \bm c_1 - \bm c_2 \\
 \bm 0
\end{bmatrix}
= 
\bm 0. 
\end{equation}
This means that 
$\bm c_1 - \bm c_2 \in (\ker S)_{{\rm supp\,}\gamma}$. 
Indeed, the choice of $\bm c_1$ is unique
up to an element in $(\ker S)_{{\rm supp\,}\gamma}$, 
which means that $\widetilde c(\gamma)=1$. 
Therefore, we have $\lambda(\gamma)=\widetilde c(\gamma) - \widetilde d(\gamma)=0$. 
Thus, we have shown that, when the conditions for maxRPA are true, 
$\gamma$ is a buffering structure. 
This concludes the proof of Theorem~\ref{thm:maxrpa_equivalence}. 
\end{proof}

Let us make a comment.
In the proof, we have separated the cases depending on whether $\wt{d}(\gamma)=0$ or $1$. 
In the case $\lambda(\gamma) = \wt{d}(\gamma)=\wt c(\gamma) = 0$, the vector $\bm q$ is given by Eq.~\eqref{eq:q-vec-case1}. 
This vector is nothing but the matrix representation of 
a {\it reduction morphism}~\cite{PhysRevResearch.3.043123}:  
in Ref.~\cite{PhysRevResearch.3.043123}, 
mappings between reaction networks are considered, 
a reduction morphism is a map between reaction networks 
under which all the complexes in $V_\gamma$ 
are mapped to a single complex in $\Gamma \setminus \gamma$.
Indeed, the expression \eqref{eq:q-vec-case1} corresponds 
to a special case of Eq.~(141) in Ref.~\cite{PhysRevResearch.3.043123}
(in the current situation, $\Gamma \setminus \gamma$ 
contains only one species and two reactions).

\section{ Integral feedback control for generic RPA } \label{sec:integrator-general}

In Sec.~\ref{sec:rpa-lbs}, we have shown that any generic RPA property in a deterministic reaction system can be represented by a buffering structure.
In this section, we construct the integral feedback controller corresponding to a given buffering structure. 
This allows us to identify the integral control mechanism for \emph{any} generic RPA property in a deterministic chemical reaction system.

\subsection{ Construction of integrators for a buffering structure }\label{sec:gen-int}

According to the IMP, we can expect that RPA with respect to constant-in-time disturbances is realized through integral control, and this expectation turns out to be correct. 
To find the integrator equations, we reformulate the reaction system to an equivalent form that is better suited for this purpose. 
For a given subnetwork, $\gamma$,
which we take to be a buffering structure, 
meaning that $\gamma$ is output-complete and $\lambda(\gamma)=0$\footnote{
In fact, the present formulation is applicable to the case $\lambda(\gamma) >0$. This leads to the phenomenon of manifold RPA, that we discuss in detail in Sec.~\ref{sec:manifold-rpa}.
}, 
we separate the chemical concentrations and reaction rates as 
\begin{equation}
    \bm x = 
    \begin{bmatrix}
      \bm x_1 \\
      \bm x_2
    \end{bmatrix}, 
    \quad 
    \bm r = 
    \begin{bmatrix}
      \bm r_1 \\
      \bm r_2
    \end{bmatrix}.
\end{equation}
With the separation of internal and external degrees of freedom (to $\gamma$),
the rate equations of the whole reaction system is written as 
\begin{equation}
  \frac{d}{dt}
  \begin{bmatrix}
    \bm x_1 \\
    \bm x_2
  \end{bmatrix}
  = 
  \begin{bmatrix}
    S_{11} & S_{12} \\
    S_{21} & S_{22} 
  \end{bmatrix} 
  \begin{bmatrix}
    \bm r_1 \\
    \bm r_2
  \end{bmatrix}
  = 
  \begin{bmatrix}
    S_{11} \bm r_1 + S_{12} \bm r_2 \\    
    S_{21} \bm r_1 + S_{22} \bm r_2 
  \end{bmatrix}.  
  \label{eq:rate-eq-block}
\end{equation}
While the internal reaction rates 
$\bm r_1 = \bm r_1 (\bm x_1, \bm x_2)$ 
in general depend on both the internal and external
chemical concentrations, 
the external reaction rates are functions of only the concentrations of external species, $\bm r_2 = \bm r_2 (\bm x_2)$
because $\gamma$ is chosen to be output-complete.
The first equation 
of Eq.~(\ref{eq:rate-eq-block}) can be solved for 
$\bm r_1$ as 
\begin{equation}
 \bm r_1 = S_{11}^+ \dot{\bm x}_1- S_{11}^+ S_{12} \bm r_2 + \bm c_{11},
 \label{eq:r1-equal}
\end{equation}
where $S_{11}^+$ is the Moore-Penrose inverse of $S_{11}$, 
and $\bm c_{11}$ is an arbitrary element in $\ker S_{11}$.
Substituting this to the second equation of Eq.~(\ref{eq:rate-eq-block}), 
\begin{equation}
\frac{d}{dt}
\left( 
\bm x_2 - S_{21} S_{11}^+ \bm x_1
\right) 
= 
S' \bm r_2 (\bm x_2) + S_{21} \bm c_{11} ,
\label{eq:sp-r2+s21-c11}
\end{equation}
where $S'$ is the generalized Schur complement,
\begin{eqnarray}
S' \coloneqq S_{22}  - S_{21} S_{11}^+ S_{12}. 
\label{eq:def-sp} 
\end{eqnarray}
The second term of the RHS of Eq.~\eqref{eq:sp-r2+s21-c11}
vanishes if and only if the following condition is satisfied, 
\begin{equation}
  \ker S_{11} \subset \ker S_{21}  , 
  \label{eq:s11-s21-cond}
\end{equation}
which is equivalent the absence of emergent cycles,
$\wt{c}(\gamma)=0$\footnote{
When $\wt{c}(\gamma)=0$ is satisfied, 
the second term of Eq.~\eqref{eq:sp-r2+s21-c11} vanishes, and 
we have 
\begin{equation}
  \frac{d}{dt}
  \left(
 {\bm x}_2 - S_{21} S^+_{11} {\bm x}_1
 \right)
 = 
S' \bm r_2 (\bm x_2), 
\label{eq:rate-red-int}
\end{equation}
This motivates us to consider the subnetwork 
$(\bm x_2, \bm r_2)$
whose rate equation is given by 
\begin{equation}
  \frac{d}{dt} {\bm x}_2 
   = S' \bm r_2 (\bm x_2).  
\end{equation}  
Namely, as long as steady states are concerned, 
the subnetwork $(\bm x_2, \bm r_2)$ 
satisfies the rate equation 
whose stoichiometric matrix is $S'$. 
Based on this observation, 
Ref.~\cite{PhysRevResearch.3.043123} proposed 
to use the generalized Schur complement as the stoichiometric matrix of the reduced system. 
In particular, if $\gamma$ is a buffering structure and $\wt{d}(\gamma)=0$, the reduced system is guaranteed to have the same steady-state solution for $\bm x_2$. 
}. 
In general, $\wt{c}(\gamma)$ can be nonzero, and 
the second term on the RHS of Eq.~\eqref{eq:sp-r2+s21-c11}
cannot be dropped. 
To account for this ambiguity,
there appear $\wt{c}(\gamma)$ undetermined variables.
Let us pick a basis for the space of emergent cycles, 
\begin{equation}
\wt{\rm C}(\gamma) = \ker S_{11} / (\ker S )_{{\rm supp}\,\gamma}
\simeq \ker S_{11} \cap (\ker S_{21})^\perp = {\rm span\,} \{ \wt{\bm c}^{(\mathfrak c)} \}_{\mathfrak c = 1, \ldots, |\mathfrak c|}.    
\end{equation}
We introduce new variables
$\{ \wt{w}_{\mathfrak c} \}_{c = 1, \ldots, |\mathfrak c|}$ for $\wt{\rm C}(\gamma)$
and $\{ w_{\alpha^\star} \}_{\alpha^\star=1,\ldots,|\alpha^\star|}$ for $(\ker S)_{{\rm supp}\,\gamma}$ and parametrize $\bm c_{11} \in \ker S_{11}$ by these variables as
\begin{equation}
\bm c_{11} 
=
\sum_{\alpha^\star} w_{\alpha^\star} \bm c_1^{(\alpha^\star)} 
+ 
\sum_{\mathfrak c} \wt{w}_{\mathfrak c} \wt{\bm c}^{(\mathfrak c)}. 
\label{eq:c11-param}
\end{equation}
Although we have solved the first equation of Eq.~\eqref{eq:rate-eq-block} for $\bm r_1$. 
However, this does not work for certain combinations of $\dot{\bm x}_1$ for which 
$\bm r_1$ does not appear on the RHS.
As we see below, this part can be captured by emergent conserved quantities. 
We denote the set of linearly independent emergent conserved quantities by 
$\{\wt{\bm d}^{(\bar{\mathfrak a})}_1 \}_{\bar{\mathfrak a} = 1 \ldots |\bar{\mathfrak a}|}$.
We extend each of them by 
\begin{equation}
\wt{\bm d}^{(\bar{\mathfrak a})} \coloneqq 
\begin{bmatrix}
 \wt{\bm d}^{(\bar{\mathfrak a})}_1
  \\
 \bm 0
\end{bmatrix}. 
\end{equation} 
The vector $\wt{\bm d}^{(\bar{\mathfrak a})}$ satisfies 
\begin{equation}
 \wt{\bm d}^{(\bar{\mathfrak a})\top} S = 
 \begin{bmatrix}
     \bm 0 \\
     \bm c_2^{(\bar{\mathfrak a})} 
 \end{bmatrix}^\top . 
\end{equation}
We note that the vectors 
$\{ \bm c_2^{(\bar{\mathfrak a})}\}_{\bar{\mathfrak a} = 1, \ldots, |\bar{\mathfrak a}|}$ 
are linearly independent\footnote{
Suppose they are not independent. 
Then, there is a certain linear combination such that 
\begin{equation}
\sum_{\bar{\mathfrak a}} b_{\bar{\mathfrak a}} \, \bm c_2^{(\bar{\mathfrak a})} = \bm 0,
\end{equation}
with $b_{\bar{\mathfrak a}} \in \mathbb R$ 
and not all of $b_{\bar{\mathfrak a}}$ are zero. 
This implies that 
\begin{equation}
 \sum_{\bar{\mathfrak a}} b_{\bar{\mathfrak a}} \,  \wt{\bm d}^{(\bar{\mathfrak a})}  \in \coker S. 
\end{equation}
Namely, a certain linear combination of $\wt{\bm d}^{(\bar{\mathfrak a})}$ is in fact in $\coker S$. 
This contradicts the assumption that 
$\wt{\bm d}^{(\bar{\mathfrak a})}$ are independent emergent conserved quantities. 
}.
Taking the time derivative of 
the linear combinations 
$
\wt{\bm d}^{(\bar{\mathfrak a})} \cdot \bm x
= \wt{\bm d}^{(\bar{\mathfrak a})}_1 \cdot {\bm x}_1
$, 
\begin{equation}
\frac{d}{dt }
 \wt{\bm d}^{(\bar{\mathfrak a})} \cdot \bm x
  = 
 \bm c_2^{(\bar{\mathfrak a})}  \cdot \bm r_2 (\bm x_2) .
 \label{eq:d-dx-c2-rw}
\end{equation}

Thus far, we have introduced a number equations as well as new variables $(\bm w, \wt{\bm w})$.
In fact, the set of these equations with additional variables is a equivalent description of the original system. 
Namely, we have the following equivalence: 
\begin{proposition}
The reaction system with variables $(\bm x_1, \bm x_2)$ 
under Eq.~\eqref{eq:rate-eq-block}
and the system with variables $(\bm x_1, \bm w, \bm x_2, \wt{\bm w})$ 
under Eqs.~\eqref{eq:r1-equal}, \eqref{eq:sp-r2+s21-c11}, and \eqref{eq:d-dx-c2-rw} (with parametrization \eqref{eq:c11-param}) are equivalent as a dynamical system, meaning that they have the same solution for $(\bm x_1 (t), \bm x_2 (t))$.
\end{proposition}
The equivalence can be checked by a straightforward computation.

Based on the reformulated description, we can identify the integrator equations realizing the RPA property represented by a buffering structure $\gamma$.
Equation~\eqref{eq:d-dx-c2-rw} gives us one set of integrator equations. 
The other set can be obtained from Eq.~\eqref{eq:sp-r2+s21-c11} as follows. 
Equation~\eqref{eq:sp-r2+s21-c11} can be written as
\begin{equation}
\dot {\bm x}_2 - S_{21} S_{11}^+ \dot{\bm x}_1 
= 
S' \bm r_2 (\bm x_2) + 
\sum_{\mathfrak c} \wt{w}_{\mathfrak c} \bm u^{(\mathfrak c)} 
,
\label{eq:sp-r2+s21-c11-2} 
\end{equation}
where $\bm u^{(\mathfrak c)} \coloneqq S_{21}\wt{\bm c}^{(\mathfrak c)}$. 
Note that $\{ \bm u^{(\mathfrak c)} \}_{\mathfrak c = 1, \ldots, |\mathfrak c|}$ are linearly independent, and we define 
$U \coloneqq 
{\rm span\,} 
\{ \bm u^{(\mathfrak c)} \}_{\mathfrak c = 1, \ldots, |\mathfrak c|}
$.
Since the new variables $\{ \wt{w}_{\mathfrak c} \}_{c = 1, \ldots, |\mathfrak c|}$ enter linearly, we can eliminate by projecting the dynamics on the space of original variables. 
For this, we multiply Eq.~\eqref{eq:sp-r2+s21-c11-2} by the projection matrix $\mathcal P$ to the subspace $U^\perp \subset \mathbb R^{|V \setminus V_\gamma|}$, to obtain the second set of integrator equations as
\begin{equation}
\mathcal P
( \dot {\bm x}_2 - S_{21} S_{11}^+ \dot{\bm x}_1 )
= 
\mathcal P S' \bm r_2 (\bm x_2). 
\label{integrator:gen_buffering}
\end{equation}
The combination of 
Eqs.~\eqref{eq:d-dx-c2-rw} and \eqref{integrator:gen_buffering} and constitute the integrator equations associated with a buffering structure.

To see that these equations actually have the ability to realize the corresponding RPA property, let us show that the steady-state values of $\bm x_2$ can be 
determined from these equations and the obtained solution of $\bm x_2$ is independent of the parameters inside $\gamma$. 
To find the steady-state, we have to specify the values of conserved quantities. Recall that $\coker S$ is decomposed as
\begin{equation}
\coker S \simeq D_{11}(\gamma) \oplus D_l (\gamma) \oplus \bar D'(\gamma) , 
\end{equation}
where $D_{11}(\gamma)$, $D_l(\gamma)$,
and $\bar D'(\gamma)$ are defined 
in Eq.~\eqref{defn:d11gamma}, Eq.~\eqref{eq:def-Dl}, and Eq.~\eqref{defn:dprime_gamma} respectively.
The steady-state solution should satisfy the following equations, 
\begin{align}
S' \bm r_2 (\bm x_2) + S_{21} \bm c_{11} &= \bm 0 , \label{eq:sp-r2-s21-c11}
\\
\bm c_2^{(\bar{\mathfrak a})}  \cdot \bm r_2 (\bm x_2)  &= 0 ,  \label{eq:c2-dot-r2}
\\
{\bm d}_2^{(\bar\alpha')} \cdot \bm x_2 &= \ell^{\bar\alpha'} , 
 \label{eq:d2-x2-l} 
 \\
 \bm d_1^{(\bar{\alpha}_\gamma)} \cdot \bm x_1 
+ 
\bm d_2^{(\bar{\alpha}_\gamma)} \cdot \bm x_2
&= 
\ell^{\bar{\alpha}_\gamma},
\label{eq:dl1-dl2-l}
\\ 
\bm d_1^{(\bar{\alpha}^\star)} \cdot \bm x_1 
&= 
\ell^{\bar{\alpha}^\star},
\label{eq:d1x1-l}
\end{align} 
where we further decomposed $D_{11} (\gamma) \simeq \bar D(\gamma) \oplus D_{11}(\gamma) / \bar D(\gamma)$ with
$\bar D(\gamma) \coloneqq (\coker S)_{{\rm supp\,}\gamma}$
and 
$\left\{
\begin{bmatrix}
\bm d_1^{(\bar{\alpha}_\gamma)}  \\ 
\bm d_2^{(\bar{\alpha}_\gamma)} 
\end{bmatrix}
\right\}_{\bar{\alpha}_\gamma = 1, \ldots, |\bar{\alpha}_\gamma|}   
$
are a basis of $D_{11}(\gamma) / \bar D (\gamma) \oplus D_l (\gamma)$.
Namely, the vectors indexed by $\bar\alpha_\gamma$ are conserved quantities of $\Gamma$ with nonzero support 
both in $\gamma$ and $\Gamma \setminus \gamma$. 
We took the basis of $\bar D'(\gamma)$ as
\begin{equation}
\bar D' (\gamma) 
=
\text{span}
\left\{ 
\begin{bmatrix}
\bm 0 \\ 
\bm d_2^{(\bar\alpha')}
\end{bmatrix}
\right\}_{\bar\alpha' = 1, \ldots, |\bar\alpha'|}
.
\end{equation}

Here, we count the number of constraints on the variables $(\bm x_2, \wt{\bm w})$. The number of these variables is given by 
\begin{equation}
\text{(\# of variables $(\bm x_2, \wt{\bm w})$)} 
= |V \setminus V_\gamma| + \wt{c}(\gamma) . 
\label{eq:num-boundary-variables}
\end{equation}
In the presence of conserved quantities, not all the rows of Eq.~\eqref{eq:sp-r2-s21-c11} are independent.
The number of independent equations in Eq.~\eqref{eq:sp-r2-s21-c11} 
is counted as 
\begin{equation}
|V \setminus V_\gamma| - \bar d' (\gamma) - d_l (\gamma),
\label{eq:rate-eq-indep}
\end{equation}
where we defined $\bar d' (\gamma) \coloneqq |\bar D'(\gamma)|$.
We give the derivation of Eq.~\eqref{eq:rate-eq-indep} in Appendix~\ref{app:derivation-number}. 
Thus, the number of independent equations among 
Eqs.~\eqref{eq:sp-r2-s21-c11} -- \eqref{eq:dl1-dl2-l}
that involve $(\bm x_2, \wt{\bm w})$ is given by
\begin{equation}
\begin{split}    
\text{(\# of independent eqs. that involve $(\bm x_2, \wt{\bm w})$)} 
&=
\underbrace{
|V \setminus V_\gamma| - {\bar d}'(\gamma) 
- d_l (\gamma)
}_{\text{Eq. \eqref{eq:sp-r2-s21-c11}}}
+ 
\underbrace{
{\bar d}'(\gamma) 
}_{\text{Eq. \eqref{eq:d2-x2-l}}}
+
\underbrace{
\wt{d} (\gamma)
}_{\text{Eq. \eqref{eq:c2-dot-r2}}} 
+
\underbrace{
d(\gamma) - d^\star(\gamma) + d_l(\gamma)
}_{\text{Eq. \eqref{eq:dl1-dl2-l}}}
\\
&= 
|V \setminus V_\gamma| + \wt{d} (\gamma) + d(\gamma) - d^\star(\gamma),
\end{split}
\end{equation}
where we defined 
$d^\star (\gamma) \coloneqq | \bar D(\gamma) |$
and $d(\gamma)\coloneqq |D_{11}(\gamma)|$. 

The difference between the number of variables $(\bm x_2, \wt{\bm w})$
and the number of independent equations that involve $(\bm x_2, \wt{\bm w})$ is 
\begin{equation}
\text{(\# of variables $(\bm x_2, \wt{\bm w})$)} 
- 
\text{(\# of independent eqs. that involve $(\bm x_2, \wt{\bm w}$)}
= \lambda(\gamma) - d_l (\gamma) - 
(d(\gamma) - d^\star(\gamma)) ,
\label{eq:num-var-minus-num-eqs-x2}
\end{equation}
where we used the decomposition~\eqref{eq:lambda-decom}
of the influence index.
Note that $d(\gamma) - d^\star(\gamma) \ge 0$, since 
$\bar D(\gamma) \subset D_{11}(\gamma)$.
Thus, the relation \eqref{eq:num-var-minus-num-eqs-x2} implies that, 
when $\lambda(\gamma)=0$, 
we should have $d_l(\gamma)=0$ and $d(\gamma) = d^\star(\gamma)$. 
This means that there is no equation of the form \eqref{eq:dl1-dl2-l} when $\lambda(\gamma)=0$.
This conclusion is consistent with the intuition that, if such a conserved quantity exists, the steady-state values of external concentrations seem affected by the change of internal parameters through this conserved quantity. When $\gamma$ is a buffering structure, such a possibility is excluded.

Therefore, Eqs.~\eqref{eq:sp-r2-s21-c11}, \eqref{eq:c2-dot-r2}, 
and \eqref{eq:d2-x2-l} 
completely specify the steady-state values of $\bm x_2$ when $\lambda(\gamma)=0$. 
Since these equations do not involve any parameter in $\gamma$, 
$\bm x_2$ is independent of them. 
On the assumption of the existence and stability of steady state, 
$\bm x_2$ is driven to values that are independent of the parameters in $\gamma$ 
by the action of integrator equations~\eqref{eq:d-dx-c2-rw} and \eqref{integrator:gen_buffering}.

\subsection{An Internal Model Principle for kinetics-independent RPA}

We saw in Section \ref{maxrpa_imp} that for maxRPA networks, a linear integrator for the dynamics can be constructed, based on the vector $\bm q$ satisfying Eq.~\eqref{eq:qs}, and this provides us with an Internal Model Principle (IMP) for such networks in the case where the output species is not in the support of $\bm q$. 
As argued in the proof in Sec.~\ref{sec:proof-th1} of the equivalence of maxRPA and the law of localization, for such cases, the vector $\bm q$ can be viewed as an emergent conserved quantity associated with a buffering structure (see Theorem \ref{thm:maxrpa_equivalence}).
This argument extends to a general kinetics-independent RPA property characterized by a buffering structure $\gamma$, and it gives us the first set of integrators constructed in Sec.~\ref{sec:gen-int} (see Eq.~\eqref{eq:d-dx-c2-rw}).
Observe that in this $\wt{d}(\gamma)$-dimensional system of integrator equations, the LHS does not involve the species external to $\gamma$ (i.e.\ variables $\bm x_2$) as mandated by the IMP. Defining the Internal Model (IM) as the set of species in $\gamma$ that form the support of these emergent conserved quantities establishes the IMP decomposition shown in Fig~\ref{fig:imp}. Note that all the species external to $\gamma$ are part of the Rest of the Network.

In many RPA examples (see Sections \ref{subsec:maxrpa_homo} and \ref{subsec:bacterial_chemo}), there exist integrators that do not belong to the first set, but rather they belong to the second set given by Eq.~\eqref{integrator:gen_buffering}. 
While integrators belonging to this second set do not conform to the standard IMP (as their LHS involves the variables $\bm x_2$), they do establish integral mechanisms that play a role in leading the concentrations of species external to $\gamma$ to some manifold that is insensitive to parameters inside $\gamma$.
Note that while the integrators in the second set do not conform to the standard IMP in the natural coordinates of the system, it is possible that they become IMP-conformant under a suitably devised coordinate transformation.

\subsection{ Example }\label{sec:ex-manifold-rpa}

Let us look at the construction of integrators
for a simple example with
$0 < \wt{d}(\gamma) < |V \setminus V_\gamma |$.
We consider a reaction network consisting
four species $\{x,y,z_1,z_2\}$ and the following set of reactions:
\begin{align}
    e_1 &: z_1 \to  z_1+x,  \notag \\
    e_2 &: z_1 + z_2 \to \emptyset , \notag \\
    e_3 &: \emptyset \to z_1 , \notag \\
    e_4 &: 2x \to 2x + z_2 ,  \\
    e_5 &: x \to \emptyset , \notag  \\
    e_6 &: y \to y+z_1 , \notag \\
    e_7 &: \emptyset \to y, \notag \\ 
    e_8 &: y \to \emptyset . \notag
\end{align}
The rate equations under mass-action kinetics are 
\begin{align}
\dot {z_1}  &= k_3 + k_6 y  - k_2 z_1 z_2 , \\
\dot {z_2}  &= k_4 x^2 - k_2 z_1 z_2 , \\
\dot x      &= k_1 z_1 - k_5 x , \\
\dot y      &= k_7 - k_8 y . 
\end{align}
The steady-state concentrations are given by
\begin{equation}
\bar x = 
\sqrt{
\frac{1} {k_4}
\left( 
k_3 + \frac{k_6 k_7}{k_8}
\right)
}
, 
\quad 
\bar y = \frac{k_7}{k_8},
\quad 
\bar z_1 = \frac{k_5} {k_1 }
\sqrt{
\frac{1} {k_4}
\left( 
k_3 + \frac{k_6 k_7}{k_8}
\right)
}
,
\quad 
\bar z_2 = 
\frac{k_1}{k_2 k_5}
\sqrt{
k_4
\left( 
k_3 + \frac{k_6 k_7}{k_8}
\right)
}
. 
\end{equation}
The stoichiometric matrix is 
\begin{equation}
S =
\begin{blockarray}{ccccccccc}
&& \\
\begin{block}{c[cccccccc]}
{z_2} \quad\,
& -1 & 0 & 0 & 1 & 0 & 0 & 0 & 0\\
{z_1} \quad\, 
& -1 & 0 & 0 & 0 & 1 & 1 & 0 & 0 \\
{x} \quad\, 
& 0 & 1 & -1 & 0  & 0 & 0 & 0 & 0 \\
{y} \quad\, 
& 0 & 0 & 0 & 0  & 0 & 0 & -1 & 1\\
\end{block}
& e_2 & e_1 & e_5 & e_4 & e_3 & e_6 & e_8 & e_7
\end{blockarray} \,\,.     
\end{equation}
The labeled buffering structures in this network are 
\begin{align}
\gamma_1^\ast &=  (\{ z_1, z_2  \},\{\}\cup  \{e_1, e_2 \}), \nonumber \\
\gamma_2^\ast &=  (\{ z_2  \},\{ \} \cup \{ e_2 \}), \nonumber \\
\gamma_3 &=  (\{ x, z_1,z_2 \}, \{e_1,e_2,e_3,e_4,e_5\} \cup \{ \}),\nonumber \\ 
\gamma_4 &=  (\{ x, z_1,z_2 \}, \{e_1,e_5\} \cup \{e_2,e_4 \}),  \\
\gamma_5 &=  (\{ z_1, z_2  \},\{e_1,e_5 \} \cup \{ e_2 \} ), \nonumber \\
\gamma_6 &= (\{ x, z_1,z_2 \}, \{e_1,e_2,e_4,e_5,e_6 \} \cup \{ \} ),\nonumber \\
\gamma_7 &= (\{ x,y, z_1,z_2 \}, \{e_1,e_2,e_4,e_5,e_6,e_7,e_8\} \cup \{ \}),\nonumber \\
\gamma_8 &= (\{ x,y, z_1,z_2 \}, \{e_1,e_2,e_4,e_5,e_6\} \cup \{e_8\}). \nonumber
\end{align}
We here consider RPA associated with 
with $\gamma_5$. 
The spaces $\ker S_{11}$ and $\coker S_{11}$ are spanned by 
\begin{equation}
\ker S_{11}  = 
\text{span}
\{  
\begin{bmatrix}  0 & 1 & 1 \end{bmatrix}^\top, 
{
\color{mydarkred} 
\begin{bmatrix}  0 & 1 & -1 \end{bmatrix}^\top
}
\}, 
\quad 
\coker S_{11} = 
\text{span} 
\{ 
{
\color{mydarkred} 
\begin{bmatrix} 1 & -1 \end{bmatrix}^\top 
}
\}, 
\label{eq:ex-man-ker-s11-coker-s11-}
\end{equation}
where colored vectors are emergent. 
The influence index of $\gamma_5$ is decomposed as 
\begin{equation}
\lambda (\gamma_5)    
= -2 + 3 - 1 + 0
=
\underset{=\widetilde c(\gamma_5)}{ 1 }
+
\underset{=d_l (\gamma_5)}{ 0 }
- 
\underset{=\widetilde d(\gamma_5)}{ 1 }
= 0. 
\end{equation}

Let us look at Eq.~\eqref{eq:sp-r2+s21-c11} of this example.
The generalized Schur complement 
of $S$ with respect to this subnetwork is 
\begin{equation}
S' = 
\begin{blockarray}{cccccc}
&& \\
\begin{block}{c[ccccc]}
{x} \quad\, 
& 0 & 0 & 0 & 0 & 0 \\
{y} \quad\, 
& 0 & 0 & 0 & -1 & 1 \\
\end{block}
& e_4 & e_3 & e_6 & e_8 & e_7 
\end{blockarray} \,\,.  
\end{equation}
The combination of the LHS 
of Eq.~\eqref{eq:sp-r2+s21-c11} is 
\begin{equation}
\bm x_2 - S_{21} S_{11}^+ \bm x_1 
= 
\begin{bmatrix}
    x \\
    y
\end{bmatrix}.  
\end{equation}
We can pick an element of $\ker S_{11}$ proportional to the emergent one,
$
\bm c_{11} = 
\wt{w} 
\begin{bmatrix}
 0 & 1 & -1 
\end{bmatrix}^\top
$, where $\wt{w} = \wt{w}(t)$ is an arbitrary function of time.
The RHS of Eq.~\eqref{eq:sp-r2+s21-c11} 
is given by 
\begin{equation}
S' \bm r_2 + S_{21} \bm c_{11} 
= 
\begin{bmatrix}
0 & 0 & 0 & 0 & 0 \\
0 & 0 & 0 & -1 & 1
\end{bmatrix}
\begin{bmatrix}
    k_4 x^2 \\
    k_3 \\
    k_6 y \\
    k_8 y \\
    k_7
\end{bmatrix}
+ 
\wt{w}
\begin{bmatrix}
    2 \\ 
    0 
\end{bmatrix}
= 
\begin{bmatrix}
2 \wt{w}  \\
- k_8 y + k_7 
\end{bmatrix}.
\end{equation}
Thus we have 
\begin{equation}
\frac{d}{dt} 
\begin{bmatrix}
x  \\
y 
\end{bmatrix}
= 
\begin{bmatrix}
2 \wt{w} \\
- k_8 y +  k_7 
\end{bmatrix}.     
\label{eq:xyz1z2-sp}
\end{equation}
The second line of this equations gives us one integrator equation (the multiplication of the projection matrix in Eq.~\eqref{integrator:gen_buffering} amounts to picking 
the second line of Eq.~\eqref{eq:xyz1z2-sp}). 
We obtain another integrator equation 
from the emergent conserved quantity 
$\wt{\bm d}_1 = 
\begin{bmatrix}
    1 & -1 
\end{bmatrix}^\top
$ in Eq.~\eqref{eq:ex-man-ker-s11-coker-s11-}. 
By taking the time derivative of $\wt{\bm d}_1 \cdot \bm x_1$, we have
\begin{equation}
\frac{d}{dt}
\wt{\bm d}_1 \cdot \bm x_1 
= 
\frac{d}{dt} (z_2 - z_1) 
=
\wt{\bm d}^{\,\top} S \bm r 
= 
r_4 (x) - r_3 - r_6 (y) 
= 
k_4 x^2 - k_3 - k_6 y . 
\label{eq:xyz1z2-dt}
\end{equation}
Equation~\eqref{eq:xyz1z2-dt}
and the second line of Eq.~\eqref{eq:xyz1z2-sp}
constitute the integrators for the RPA associated with $\gamma_5$. 
They drive the values of $x$ and $y$ 
to their steady-state values that are independent of
the reaction parameters ($k_2, k_1, k_5$) inside $\gamma_5$.

\section{ Regulation to manifolds }\label{sec:manifold-rpa}

In this section, we consider a generalization of a regulation problem
to situations where the target values of the output variables are in a manifold with nonzero dimension. 
We shall call the emergence of this property as \emph{manifold RPA}.
This problem can be treated naturally using the formulation developed in Sec.~\ref{sec:integrator-general}: we have reformulated the reaction system to an equivalent form to find the integrator equations for a given buffering structure, and this procedure is applicable even when $\lambda(\gamma)\neq 0$ as long as $\gamma$ is output-complete. 
This observation leads us to a natural generalization of the law of localization to manifold RPA, where the influence index turns out to give the dimension of the target manifold to which $\bm x_2$ is regulated to.

\subsection{ Law of manifold localization }\label{subsec:man_local}

We have the following theorem: 
\begin{theorem}[The law of manifold localization]
Let $\gamma \subset \Gamma$ be an output-complete subnetwork of 
a deterministic chemical reaction system
satisfying the assumptions in Sec.~\ref{sec:assumptions}. 
Suppose that $\gamma = (V_\gamma, E_\gamma)$ 
is an output-complete subnetwork of $\Gamma$
whose influence index is given by $\lambda(\gamma)$. 
Then, the steady-state values of the concentrations of the species
outside $\gamma$ are located in a $\lambda(\gamma)$-dimensional manifold,
which is invariant under the change of parameters in $\gamma$.
\label{thm:law-of-manifold-localization}
\end{theorem}

\begin{figure}[tb]
  \centering
  \includegraphics
  [clip, trim=0cm 3.5cm 0cm 0cm, scale=0.37]
  {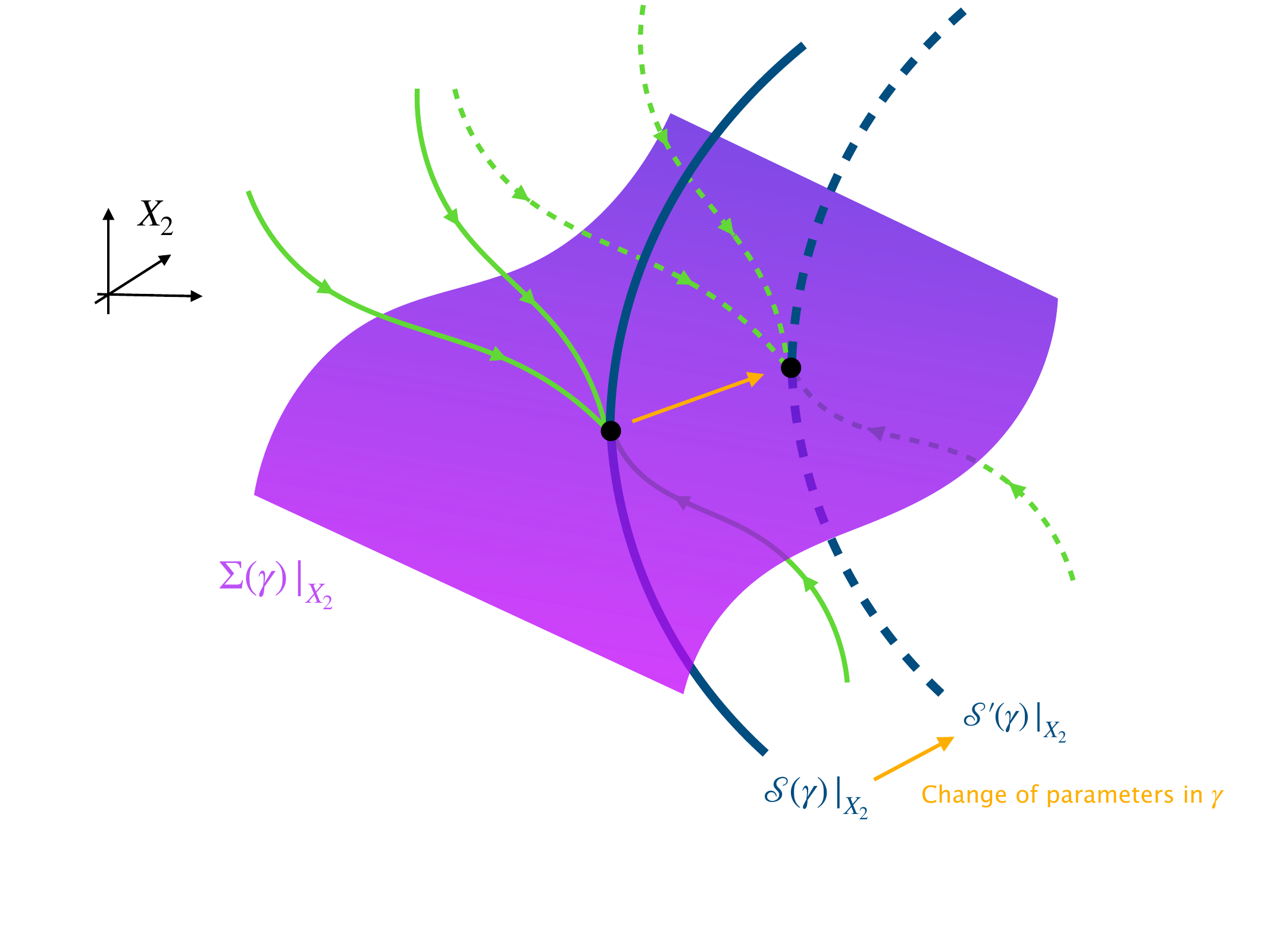}
  \caption{ 
  Schematic illustration of manifold RPA. The steady state of $\bm x_2$ is at the intersection of $\Sigma(\gamma)|_{X_2}$ and another manifold $\mathcal S (\gamma) |_{X_2}$, to which states are attracted.
  When we change the values of parameters in $\gamma$,
  $\mathcal S (\gamma) |_{X_2}$ is deformed to $\mathcal S' (\gamma) |_{X_2}$, while $\Sigma(\gamma)|_{X_2}$ does not change.
  As a result, the steady state still lies within $\Sigma(\gamma)|_{X_2}$ that is invariant under these internal parameters.
  }
  \label{fig:manifold-rpa} 
\end{figure}

We shall call an output-complete subnetwork
as a {\it generalized buffering structure of degree $\lambda(\gamma)$},
which gives rise to manifold RPA to $\lambda(\gamma)$-dimensional submanifold 
for the degrees of freedom outside $\gamma$.

Theorem~\ref{thm:law-of-manifold-localization} is a generalization of the law of localization. Namely, the conventional law of localization corresponds to special cases of Theorem~\ref{thm:law-of-manifold-localization} where the steady-state manifold for the external degrees of freedom is zero-dimensional.

Manifold RPA is very relevant for synthetic biology applications, where it is often important to design biomolecular controllers that robustly maintain some relationship between multiple output variables, without sacrificing the tunability of individual output variables \cite{alexis2022regulation}.

Basing on the formulation developed in the last section, a proof can be described succinctly. 
\begin{proof}
For an output-complete subnetwork $\gamma$, we can perform the same procedure as in Sec.~\ref{sec:gen-int} and obtain a set of steady-state equations~\eqref{eq:sp-r2-s21-c11}~--~\eqref{eq:dl1-dl2-l},
which involve $(\bm x_2, \wt{\bm w})$.
The number of the external variables $(\bm x_2, \wt{\bm w})$ is given by Eq.~\eqref{eq:num-boundary-variables}.
Among Eqs.~\eqref{eq:sp-r2-s21-c11}~--~\eqref{eq:dl1-dl2-l},
if we count the number of independent equations that involve 
$\bm x_2$ or $\wt{\bm w}$ but not $\bm x_1$ and $\bm w$ (note that Eq.~\eqref{eq:dl1-dl2-l} involve both $\bm x_2$ and $\bm x_1$), 
\begin{equation}
\begin{split}
\text{(\# of independent eqs. that involve $(\bm x_2, \wt{\bm w})$ 
but not $(\bm x_1, \bm w$)}
&= 
\underbrace{
|V \setminus V_\gamma| - {\bar d}'(\gamma) 
- d_l (\gamma)
}_{\text{Eq. \eqref{eq:sp-r2-s21-c11}}}
+ 
\underbrace{
{\bar d}'(\gamma) 
}_{\text{Eq. \eqref{eq:d2-x2-l}}}
+
\underbrace{
\wt{d} (\gamma)
}_{\text{Eq. \eqref{eq:c2-dot-r2}}} 
\\
&= |V \setminus V_\gamma| + \wt{d} (\gamma) - d_l (\gamma). 
\end{split}
\end{equation}
The difference between the number of variables $(\bm x_2, \wt{\bm w})$
and the number of independent equations that involve $(\bm x_2, \wt{\bm w})$ 
but not $(\bm x_1, \bm w)$ is
\begin{equation}
\text{(\# of variables)} 
- 
\text{(\# of independent eqs. that involve $(\bm x_2, \wt{\bm w})$
but not $(\bm x_1, \bm w$)}
= \lambda(\gamma) .
\label{eq:num-var-minus-num-eqs-x2-only}
\end{equation}
Thus, in general, the solutions of Eqs.~\eqref{eq:sp-r2-s21-c11} -- \eqref{eq:d2-x2-l} has $\lambda(\gamma)$ undetermined variables, 
and the external concentrations at steady-state are inside a manifold parameterized by them.
Since Eq.~\eqref{eq:sp-r2-s21-c11} -- \eqref{eq:d2-x2-l} do not 
involve parameters in $\gamma$, this manifold is insensitive to them.
This concludes the proof. 
\end{proof}

Let us explain how manifold RPA emerges geometrically (for a fuller explanation, see the next subsection).
We denote the state space of internal and external variables by $X_1$ and $X_2$, i.e.\ $\bm x_1 \in X_1$ and $\bm x_2 \in X_2$. 
As stated in the proof, there are $|V \setminus V_\gamma|+\wt{c}(\gamma)$ variables and $|V \setminus V_\gamma|+\wt{d}(\gamma) - d_l(\gamma)$ constraints for those variables, 
and we have $\lambda(\gamma)$-dimensional steady-state manifold $\Sigma({\gamma})|_{X_2}$ of the external concentrations.
Namely, $\Sigma({\gamma})|_{X_2}$ is a submanifold
$\Sigma({\gamma})|_{X_2} \subset X_2 \subset \mathbb R^{|V \setminus V_\gamma|}$
determined by Eqs.~\eqref{eq:sp-r2-s21-c11}~--~\eqref{eq:d2-x2-l}.
The corresponding integral feedback control
for manifold RPA realized through a generalized buffering structure can be constructed in the same way as Sec.~\ref{sec:integrator-general} (namely, Eqs.~\eqref{eq:d-dx-c2-rw} and \eqref{integrator:gen_buffering} are the integrator equations), 
and the dynamics of $\bm x_2$ is steered toward the submanifold $\Sigma({{\gamma}}) |_{X_2}$ through their action.
The system has a stable steady state by assumption, and this point in the external degrees of freedom can be seen as the intersection of $\Sigma ({\gamma})|_{X_2}$
and another manifold $\mathcal{S}({{\gamma}})|_{X_2} \subset \mathbb R^{|V \setminus V_\gamma|}$ of $(|V\setminus V_\gamma| - \lambda(\gamma))$ dimension.
While $\mathcal{S}({{\gamma}})|_{X_2}$ is dependent on the parameters in $\gamma$, $\Sigma({\gamma})|_{X_2}$ is invariant under the change of parameters in $\gamma$. 
Hence, by perturbing the parameters in $\gamma$, the location of the steady state may change but it will always lie in $\Sigma({{\gamma}}) |_{X_2}$.
We illustrate this phenomenon in Fig. \ref{fig:manifold-rpa}. 
We discuss geometric interpretation of this phenomenon in more detail in the next subsection, and the manifolds $\Sigma (\gamma) |_{X_2}$ and $\mathcal S (\gamma) |_{X_2}$ turn out to be the shadows the {\it RPA manifold} and the {\it susceptible manifold} in the extended total state space.
For a concrete realization, see examples in Sec.~\ref{sec:manifold-rpa-ex-2}.

\subsection{ Geometric interpretation of manifold RPA }\label{sec:geometric}

Here, let us give a geometric interpretation  of the manifold RPA phenomenon realized through the law of manifold localization.
In this subsection, we omit the dependence on $\gamma$ for notational simplicity of quantities such as the number of emergent cycles. 
Namely, 
$\lambda \coloneqq \lambda (\gamma)$, 
$\wt{c}\coloneqq \wt{c}(\gamma)$, 
$d \coloneqq d (\gamma)$, 
$\wt{d} \coloneqq \wt{d} (\gamma)$, 
$d_l \coloneqq d_l (\gamma)$, 
$d^\star \coloneqq d^\star (\gamma)$,
and 
$\bar d' \coloneqq \bar d' (\gamma)$. 
Additionally, we use the following abbreviations, 
\begin{equation}
v \coloneqq |V_\gamma|    , 
\quad 
v' \coloneqq |V \setminus V_\gamma| = |V| - v', 
\quad 
e\coloneqq |E_\gamma|, \quad 
c\coloneqq | (\ker S)_{{\rm supp\,}\gamma}  |, 
\end{equation}
With these notations, the definition of the influence index can be written as 
\begin{equation}
 \lambda = -v + e - c + d + d_l, 
 \label{eq:lambda-original} 
\end{equation}
where we used the relation $|P^0_\gamma (\coker S)| = d + d_l$ (see Eqs.~(167) and (173) of Ref.~\cite{PhysRevResearch.3.043123}). 
The decomposition of the influence index reads 
\begin{equation}
\lambda = \wt{c} + d_l - \wt{d}. 
 \label{eq:lambda-decomposed-short} 
\end{equation}
By taking the difference of 
Eqs.~\eqref{eq:lambda-original} and \eqref{eq:lambda-decomposed-short}, 
we have the following relation, 
\begin{equation}
v + c + \wt{c} 
= e+ d + \wt{d}, 
\label{eq:v-c-c-e-d-d}
\end{equation}
which will be used later. 

In the previous section, we have solved for $\bm r_1$ 
using the Moore-Penrose inverse of $S_{11}$, and this process 
introduces variables to account for the elements in $\ker S_{11}$ (see Eq.~\eqref{eq:r1-equal}).
We have parametrized $\bm c_{11} \in \ker S_{11}$ 
as Eq.~\eqref{eq:c11-param}. 
By introducing the new variables $(\bm w, \wt{\bm w})$, 
we are extending the state space
\begin{equation}
(\bm x_1, \bm x_2)
\in X_1 \times X_2 
\to 
( (\bm x_1, \bm w), (\bm x_2, \wt{\bm w})) 
\in \hat X_1 \times \hat X_2.
\end{equation}
Here, 
we regard 
$(\bm x_1 , \bm w)$ to be in the extended internal state space, 
$(\bm x_1 , \bm w)\in \hat X_1$, 
while $(\bm x_2 , \wt{\bm w})$ is in the extended external state space,
$(\bm x_2 , \wt{\bm w})\in \hat X_2$. 
The number of total variables is expressed as 
\begin{equation}
\text{(\# of variables)} 
= |V| + c + \wt{c}
= 
v + c 
+ 
v' + \wt{c}
\eqqcolon 
n_1 + n_2, 
\end{equation}
where we defined $n_1 \coloneqq v+c$ and $n_2 \coloneqq v'+\wt{c}$.
Together with the introduction of new variables, we have added 
Eq.~\eqref{eq:d-dx-c2-rw} using emergent conserved quantities, which is necessary so that we have sufficient constraints to determine the steady state in the extended description\footnote{As we discussed in the previous section, the extended description is equivalent to the original system dynamically, although we here look at steady-state properties.}.
See Fig.~\ref{fig:extended-system} for the set of variables and steady-state equations.
Indeed, one can check that the number of variables is equal to number of independent equations.
The total number of independent equations is given by 
\begin{equation}
\text{(\# of independent equations)} 
= 
e + v' + d + \wt{d}. 
\end{equation}
The number of variables minus the number of independent equations is 
\begin{equation}
\text{(\# of variables)} 
-
\text{(\# of independent equations)} 
= 
v+c+v'+\wt{c}- e - v' - d - \wt{d} 
=0,
\end{equation}
where we have used Eq.~\eqref{eq:v-c-c-e-d-d}. 

\begin{figure}[tb]
  \centering
  \includegraphics
  [clip, trim=0cm 13cm 0cm 0cm, scale=0.43]
  {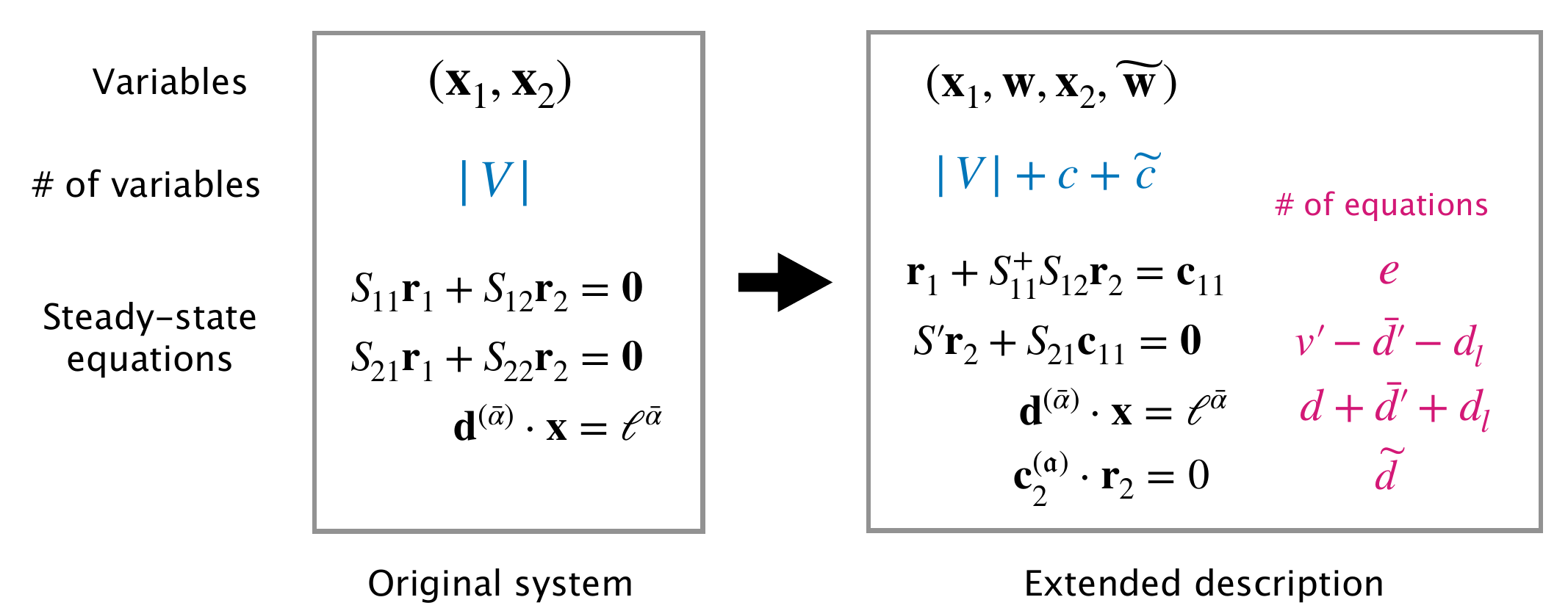}
  \caption{ Rewriting of a reaction system using additional variables. 
  The number of independent equations for each set of equations is denoted for the extended description. 
  }
  \label{fig:extended-system} 
\end{figure}

Let us introduce key objects for describing manifold RPA. 
In Fig.~\ref{fig:rpa-manifold-dimensions}, we summarize the set of equations in the extended description together with the number of independent equations.
We define the {\it RPA manifold} $\Sigma$\footnote{
Note that this object depends on the choice of a subnetwork. We do not explicitly indicate this dependence for notational simplicity. The same comment is true for the susceptible manifold ${\mathcal S}$. 
}
associated with an output-complete subnetwork $\gamma$ 
to be the set of variables $(\bm x_1, \bm w, \bm x_2, \wt{\bm w})$ that satisfy 
the conditions in the solid (light-blue) boxes. 
We define the {\it susceptible manifold} ${\mathcal S}$
to be the set of variables
$(\bm x_1, \bm w, \bm x_2, \wt{\bm w}) \in \hat X_1 \times \hat X_2$
satisfying 
the conditions in dashed (green) boxes. 
A key observation is that the conditions determining $\Sigma$ is independent of parameters in $\gamma$ and they only involve external variables, 
$(\bm x_2, \wt{\bm w}) \in \hat X_2$.
As a result, $\Sigma$ has no dependence on parameters in $\gamma$. 
On the other hand, the susceptible manifold ${\mathcal S}$ is subject to parameters inside and outside $\gamma$. 
The situation is illustrated in Fig~\ref{fig:rpa-manifold-projection}. 

Let us denote the numbers of constraints 
determining $\Sigma$ and ${\mathcal{S}}$  
by $c_{\Sigma}$ and $c_{{\mathcal{S}}}$, respectively.
The dimension of these manifolds are given by 
the total dimension minus the number of constraints, 
\begin{equation}
\dim  \Sigma = n_1 + n_2 - c_{ \Sigma}, 
\quad 
\dim {\mathcal S} = n_1 + n_2 - c_{{\mathcal S}}. 
\end{equation}
As we show in Fig.~\ref{fig:rpa-manifold-dimensions},
by counting the number of constraints, 
we can express the dimensions of $\Sigma$ and ${\mathcal S}$
in terms of the influence index as
\begin{equation}
\dim \Sigma =\lambda + n_1 , \quad 
\dim {\mathcal S} = n_2 - \lambda. 
\end{equation}
These manifolds intersect at a point, 
corresponding to the existence of a steady state. 
Indeed, the sum of the dimensions equal the dimension of the total space, 
\begin{equation}
\dim  \Sigma+ \dim {\mathcal S} = n_1 + n_2. 
\end{equation}

Now, let us consider the projection of $\Sigma$ and ${\mathcal S}$ 
to the space of external variables, $(\bm x_2, \wt{\bm w}) \in \hat X_2$\footnote{We will refer to the projected RPA and susceptible manifolds with the same names, indicating the total space they reside to avoid possible confusions.}.
We denote these projections as $\Sigma |_{\hat X_2}$ and 
${\mathcal S} |_{\hat X_2}$, respectively. 
While the dimension of ${\mathcal S}$ and $\mathcal S|_{\hat X_2}$ are the same (i.e.\ the dimension does not change under the projection),
the dimension of the RPA manifold is decreased by $n_1$,
since $\Sigma$ has no dependence on the variables in $\hat X_1$ (see Fig.~\ref{fig:rpa-manifold-projection}).
Namely, 
\begin{equation}
\dim \Sigma |_{\hat X_2} = \lambda, 
\quad 
\dim \mathcal S |_{\hat X_2} = n_2 - \lambda .
\end{equation}
The projected RPA and susceptible manifolds still intersect at a point in $\hat X_2$,
\begin{equation}
\dim \Sigma |_{\hat X_2} + \dim \mathcal S |_{\hat X_2} = n_2. 
\end{equation}
We can further project $\Sigma |_{\hat X_2}$ 
to $X_2$ by projecting out $\wt{\bm w}$, and obtain $\Sigma |_{X_2}$. 
When the manifold $\Sigma |_{X_2}$ is in a general position, its projection to $X_2$ is also of dimension $\lambda$.
Since $\Sigma |_{\hat X_2}$ does not depend on the parameters inside $\gamma$, its projection to $X_2$ does not either.
Thus, the steady state of $\bm x_2$ is localized within a $\lambda$-dimensional manifold in $X_2$ under the change of parameters in $\gamma$. 

\begin{figure}[tb]
  \centering
  \includegraphics
  [clip, trim=0cm 8cm 0cm 0cm, scale=0.43]
  {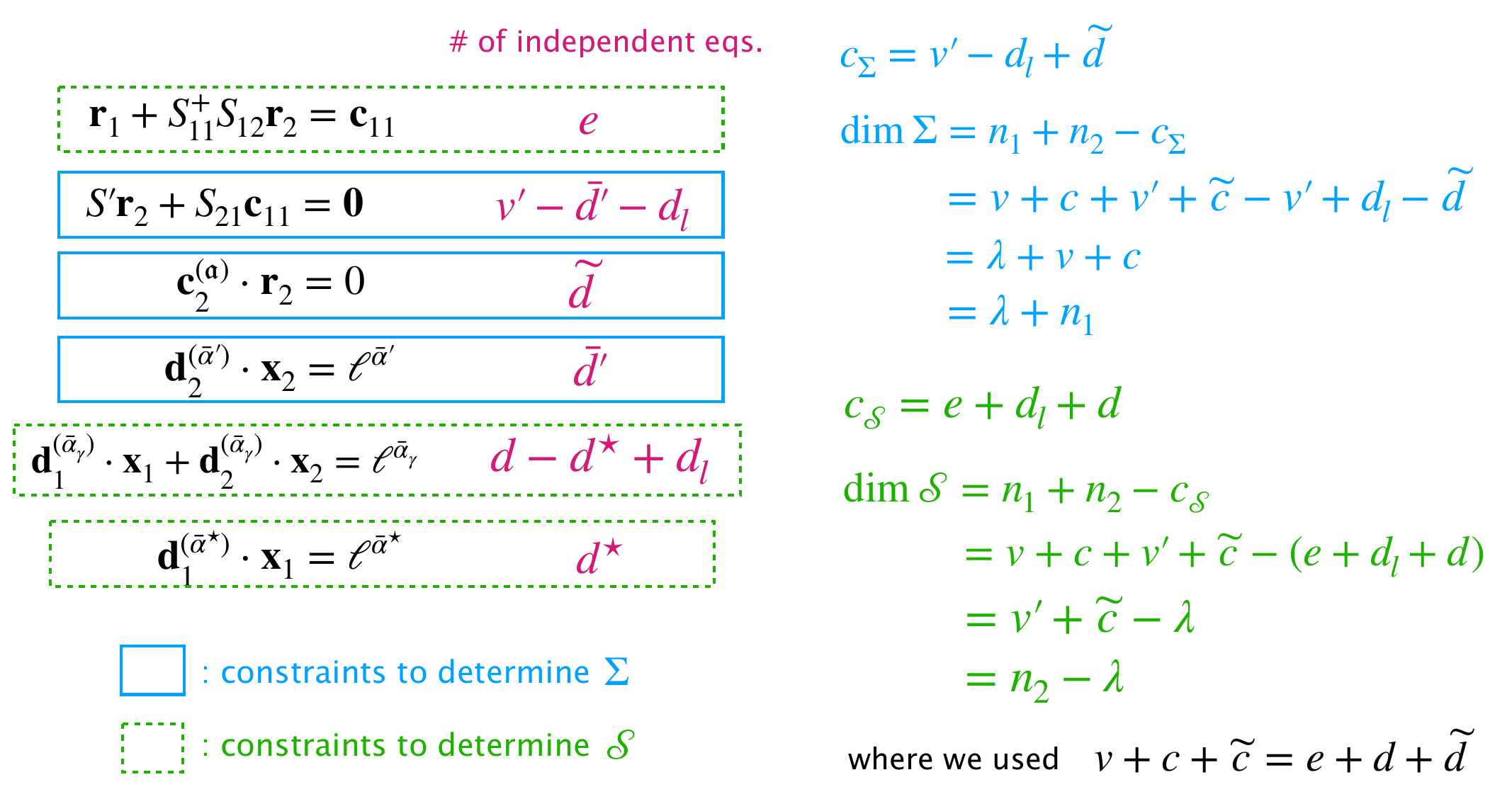}
  \caption{ 
  Illustration of the dimensions of the RPA manifold and the susceptible manifold. We separately wrote the conserved-quantity equations as in Eqs.~\eqref{eq:d2-x2-l} -- \eqref{eq:d1x1-l}.
  }
  \label{fig:rpa-manifold-dimensions}
\end{figure}

\begin{figure}[tb]
  \centering
  \includegraphics
  [clip, trim=0cm 11cm 0cm 0cm, scale=0.4]
  {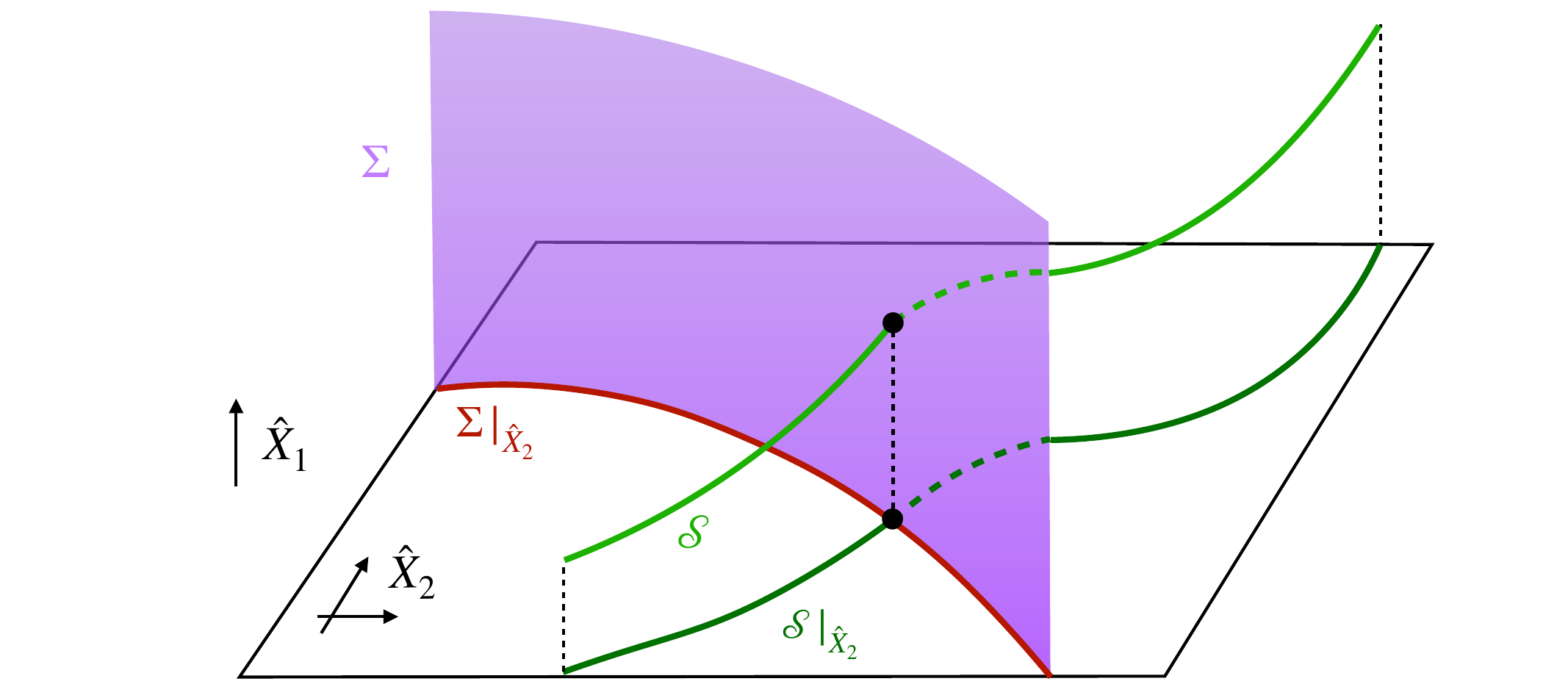}
  \caption{ Schematic of RPA manifold and susceptible manifold in $\hat X_1 \times \hat X_2$ and their projections to $\hat X_2$. }
  \label{fig:rpa-manifold-projection}
\end{figure}

\subsection{Examples} 

Here, we demonstrate the manifold RPA phenomenon and the law of manifold localization using a few examples. 

\subsubsection{ Simple example }

Let us first describe a simple example that consists of two species $\{ x,y \}$ and three reactions,
\begin{align}
e_1&: \emptyset \to y , \nonumber  \\
e_2&: y \to x, \\
e_3&: x \to  \emptyset. \nonumber
\end{align}
The stoichiometric matrix reads 
\begin{equation}
S =
\begin{blockarray}{cccc}
&& \\
\begin{block}{c[ccc]}
{y} \quad\,
& -1 & 1 & 0  \\
{x} \quad\, 
& 1 & 0 & - 1  \\
\end{block}
& e_2 & e_1 & e_3 
\end{blockarray} \,\,,    
\end{equation}
If we use mass-action kinetics, the rate equations are 
\begin{eqnarray}
\dot y &= k_1 - k_2 y ,
\\
\dot x &= k_2 y - k_3 x .
\end{eqnarray}
At the steady state, the concentrations are given by 
$\bar x = k_1 / k_3$ and $\bar y = k_1 / k_2$. 

Let us choose the following output-complete subnetwork, 
\begin{equation}
\gamma = (\{ \}, \{e_2 \}) ,
\end{equation}
which includes no species and one reaction. 
The influence index is nonzero, 
$\lambda (\gamma)=1$, 
and the conventional law of localization does not apply. 
Since the subnetwork is output-complete, the law of manifold localization is applicable. 
The influence index of $\gamma$ is decomposed as
\begin{equation}
\lambda (\gamma)
=
\underset{=\widetilde c(\gamma)}{ 1 }
+
\underset{=d_l(\gamma)}{ 0} 
- 
\underset{=\widetilde d(\gamma)}{ 0 }
= 1. 
\end{equation}
Thus, the subnetwork is a generalized buffering structure of degree $1$.
In the current case, the submatrix $S_{11}$ is $0 \times 1$ dimensional, and the matrices $S'$ and $S_{21}$ are given by 
\begin{equation}
S' =
\begin{blockarray}{ccc}
&& \\
\begin{block}{c[cc]}
{y} \quad\,
&  1 & 0  \\
{x} \quad\, 
&  0 & - 1  \\
\end{block}
 & e_1 & e_3 
\end{blockarray} \,\,,    
\quad 
S_{21}  = 
\begin{bmatrix}
    -1 \\
    1
\end{bmatrix}. 
\end{equation}
The subnetwork $\gamma$ has an emergent cycle, 
and $\ker S_{11}$ is nontrivial,
$
 \ker S_{11} = {\rm span\,} \{
 \begin{bmatrix} 1 \end{bmatrix}
    \} 
$.
We can pick an element of $\ker S_{11}$ as 
$
\bm c_{11}  = \wt{w}(t) 
\begin{bmatrix}
    1
\end{bmatrix}
$, and we have 
\begin{equation}
\begin{bmatrix}
    \dot y \\
    \dot x
\end{bmatrix}    
= 
\begin{bmatrix}
1 & 0 \\
0 & -1 
\end{bmatrix}
\begin{bmatrix}
    k_1  \\
    k_3 x
\end{bmatrix}
+ \wt{w} \begin{bmatrix} -1   \\ 1 \end{bmatrix}.
\label{eq:xy-red-gamma}
\end{equation}
The second term on the RHS is orthogonal to $\begin{bmatrix}
    1 & 1
\end{bmatrix}$, 
and multiplying this vector, we get an integrator equation,
\begin{equation}
\frac{d}{dt} (y+x) = k_1 - k_3 x. 
\end{equation}
This integrator steers the variables $(x,y)$ to the following 
one-dimensional RPA manifold $\Sigma (\gamma) |_{X_2}$,
\begin{equation}
\begin{bmatrix}
    \bar x \\ 
    \bar y 
\end{bmatrix}
= 
\begin{bmatrix}
    k_1 / k_3 \\ 
    s 
\end{bmatrix},
\end{equation}
where $s \in \mathbb R_{\ge 0}$ is a parameter. 
This manifold is independent
of parameters in $\gamma$, which is $k_2$ in the current case.
On the other hand, the susceptible manifold $\mathcal S({\gamma}) |_{X_2}$ is given by the relation
\begin{equation}
k_1 - k_2 \bar y = 0 ,
\end{equation}
which also defines a line in $(\bar x,\bar y)$ plane.
The susceptible manifold $\mathcal S({\gamma}) |_{X_2}$ depends nontrivially on $k_2$.

\subsubsection{ Example with a lost conserved quantity } \label{sec:manifold-rpa-ex-simple-2}

Here we discuss an example with a lost conserved quantity. 
Let us take a reaction network with four species $\{v_1,v_2,v_3,v_4 \}$ and four reactions,
\begin{align}
e_1&: v_1 \to v_2, \nonumber \\ 
e_2&: v_2 \to v_3,  \\
e_3&: v_3 \to v_4, \nonumber \\ 
e_4&: v_4 \to v_1. \nonumber
\end{align}
Under mass-action kinetics, 
the steady-state concentrations are given by 
$\bar x_i = \ell K / k_i$ for $i=1,2,3,,4$ with 
$1/K \coloneqq 1/k_1 + 1/k_2 + 1/k_3 + 1/k_4$, where 
$\ell = \bar x_1 + \bar x_2 + \bar x_3 + \bar x_4$ is a conserved quantity,
and the steady-state rate is 
$\bar r_A = \ell K$ for $A = 1,2,3,4$.

We here take the following output-complete subnetwork,
\begin{equation}
\gamma = ( \{v_1,v_2\}, \{e_1,e_2 \} ) . 
\end{equation}
The subnetwork contains no emergent cycle and no emergent conserved quantity,
and it has one lost conserved quantity. 
The decomposition of the influence index reads 
\begin{equation}
\lambda (\gamma)
=
-2+2-0+1
=
\underset{=\widetilde c(\gamma)}{ 0}
+
\underset{=d_l(\gamma)}{ 1} 
- 
\underset{=\widetilde d(\gamma)}{ 0 }
= 1. 
\end{equation}
Indeed, the projection of the conserved quantity $\bm d = \begin{bmatrix}
    1 & 1 & 1 & 1
\end{bmatrix}^\top \in \coker S$ 
to subnetwork $\gamma$ is not conserved within $\gamma$.
Since there is no emergent conserved quantity, 
the integrator equation is given by Eq.~\eqref{integrator:gen_buffering}, 
\begin{equation}
\frac{d}{dt}
\begin{bmatrix}
x_1 + x_2 + x_3 \\ 
x_4
\end{bmatrix}
= 
\begin{bmatrix}
- k_3 x_3 + k_4 x_4 \\ 
k_3 x_3 - k_4 x_4
\end{bmatrix}. 
\end{equation}
The first and second lines of these equations are equivalent via the conservation law, $x_1 + x_2 + x_3 + x_4 = \ell = {\rm const.}$.
This integrator drives the variables to satisfy 
\begin{equation}
k_4 \bar x_4 - k_3 \bar x_3  = 0, 
\label{eq:l-cons-ex-rate}
\end{equation}
This relation determines the RPA manifold $\Sigma (\gamma)|_{X_2}$, which is one-dimensional in $(\bar x_3, \bar x_4)$. 
Because of the existence of a lost conserved quantity, 
$(\bar x_3,\bar x_4)$ should satisfy an equation that involve $\bar x_1$ and $\bar x_2$ as well, 
\begin{equation}
\bar x_3 + \bar x_4 = \ell - \bar x_1 - \bar x_2. 
\label{eq:x3-p-x4-l-x1-x2}
\end{equation}
The condition \eqref{eq:x3-p-x4-l-x1-x2} also defines a line and this corresponds to the susceptible manifold $\mathcal S ({\gamma}) |_{X_2}$.
While $\mathcal S ({\gamma}) |_{X_2}$ is subject to the parameters in $\gamma$, i.e.\ $k_1$ and $k_2$, the manifold $\Sigma (\gamma)|_{X_2}$ is insensitive to these parameters, as expected.

\subsubsection{ Example with an emergent conserved quantity } \label{sec:manifold-rpa-ex-2}

\begin{figure}[tb]
\centering
\includegraphics
 [keepaspectratio, scale=0.55]
 {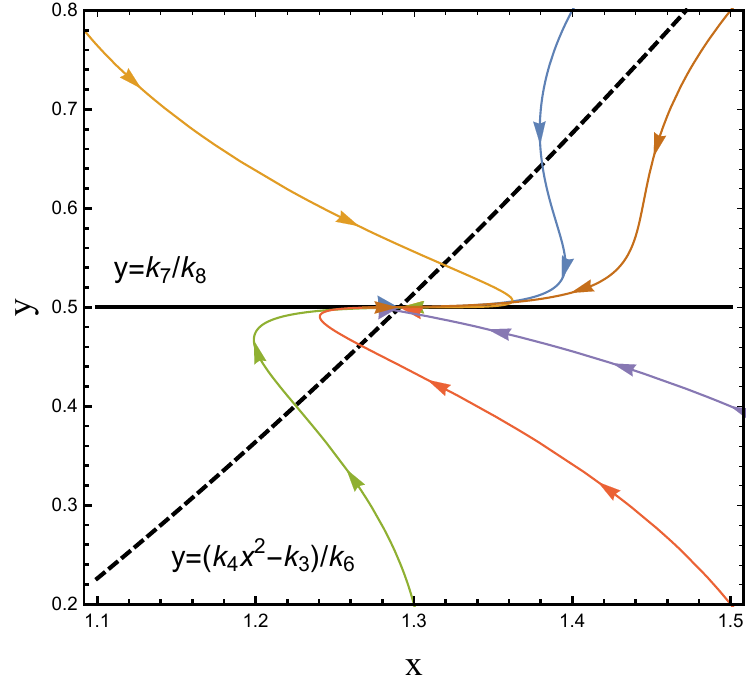} 
\caption{
Approach to the steady state of the concentrations $x(t)$ and $y(t)$. 
The black solid line denotes the steady-state line determined by
the second equation of Eq.~\eqref{eq:xyz1z2-sp}.
The integrator \eqref{eq:xyz1z2-dt} drives the variables $x(t)$ and $y(t)$ to 
the black dashed line $\Sigma (\gamma) |_{X_2}$.
The steady state is at the intersection of the black solid and dashed lines, 
to which the states are absorbed. 
Parameters are taken as
$
k_1 = 1.5, 
k_2 = 0.1,
k_3 = 1,
k_4 = 1.2,
k_5 = 1,
k_6 = 2,
k_7 = 1, 
k_8 = 2.
$
}
\label{fig:flow} 
\end{figure}

As a more nontrivial example, let us revisit the example discussed in Sec.~\ref{sec:ex-manifold-rpa}. 
We choose the following subnetwork
\begin{equation}
\gamma = (\{z_1,z_2\}, \{e_1,e_2,e_5,e_7,e_8\}) . 
\end{equation}
This subnetwork has nonzero influence index, $ \lambda (\gamma)=1$, and cannot be discussed with the conventional law of localization. 
It is output-complete, so we can apply the law of manifold localization. 
The influence index is decomposed as
\begin{equation}
\lambda (\gamma)
=
-2 + 5 - 2 + 0
=
\underset{=\widetilde c(\gamma)}{ 2}
+
\underset{=d_l(\gamma)}{ 0}
- 
\underset{=\widetilde d(\gamma)}{ 1 }
= 1. 
\end{equation}
The spaces $\ker S_{11}$ and $\coker S_{11}$ are spanned by 
\begin{align}
\ker S_{11}  &= 
\text{span}
\{  
\begin{bmatrix}  0 & 1 & 1 & 0 & 0 \end{bmatrix}^\top,
\begin{bmatrix}  0 & 0 & 0 & 1 & 1 \end{bmatrix}^\top,
{
\color{mydarkred} 
\begin{bmatrix}  0 & 0 & 0 & 1 & -1 \end{bmatrix}^\top, 
\begin{bmatrix}  0 & -1 & 1 & 0 & 0 \end{bmatrix}^\top
}
\}, \\
\coker S_{11} 
&= 
\text{span} 
\{ 
{
\color{mydarkred} 
\begin{bmatrix} 1 & -1 \end{bmatrix}^\top 
}
\}, 
\end{align} 
where the components are ordered as 
$(e_2,e_1,e_5,e_8,e_7)$ for $\ker S_{11}$
and $(x,y)$ for $\coker S_{11}$, 
and colored vectors are emergent.
If we set 
$\wt{\bm c}^{(1)}_1 \coloneqq 
\begin{bmatrix}  0 & 0 & 0 & 1 & -1 \end{bmatrix}^\top
$
and 
$\wt{\bm c}^{(2)}_1 \coloneqq
\begin{bmatrix}  0 & -1 & 1 & 0 & 0 \end{bmatrix}^\top
$, 
we can parametrize an element of $\ker S_{11} \cap (\ker S_{21})^\perp$ as
$\bm c_{11} \coloneqq \wt{w}_1 \wt{\bm c}^{(1)}_1 + \wt{w}_2 \wt{\bm c}^{(2)}_2$, 
and we have 
\begin{equation}
S' \bm r_2 + S_{21} \bm c_{11}
= 
\begin{bmatrix}
    -2 \wt{w}_2 \\ 
    -2 \wt{w}_1
\end{bmatrix}.
\end{equation}
Hence, the rate equations do not give us an integrator and 
they just fixes the values of $\wt{w}_1$ and $\wt{w}_2$ to zero at steady state. 
The emergent conserved quantity gives the same integrator equation as Eq.~\eqref{eq:xyz1z2-dt}. 
At steady state, this integrator enforces the relation,
\begin{equation}
k_4 \bar x^2 = k_3 + k_6 \bar y .
\label{eq:man-rpa-ex2-ecq}
\end{equation}
Thus, the RPA manifold for the external concentrations
is given by $(x,y) \in (\mathbb R_{> 0})^2$ subject to constraint~\eqref{eq:man-rpa-ex2-ecq}. 
One can parametrize the one-dimensional manifold as 
\begin{equation}
\begin{bmatrix}
 \bar x \\  \bar y 
\end{bmatrix}    
= 
\begin{bmatrix}
\sqrt{ \frac{1}{k_4}  (k_3 + k_6 s) } 
\\ 
s 
\end{bmatrix},
\end{equation}
where $s \in \mathbb R_{> 0}$.

Figure~\ref{fig:flow} illustrates manifold RPA of this example. 
We plot the time evolution of states in $(x,y)$ space. 
The dashed line is the RPA manifold $\Sigma (\gamma)|_{X_2}$
associated with $\gamma$, that is specified by Eq.~\eqref{eq:man-rpa-ex2-ecq}, 
and the solid line corresponds to $\mathcal S ({\gamma}) |_{X_2}$. 
The states are absorbed to the intersection of $\Sigma (\gamma) |_{X_2}$ 
and $\mathcal S (\gamma) |_{X_2}$.
While he location of $\mathcal S (\gamma) |_{X_2}$ is subject to changes in the parameters in $\gamma$: it depends on $k_7$ and $k_8$, 
the position of $\Sigma (\gamma) |_{X_2}$ is independent of 
all the parameters $(k_1,k_2,k_5,k_7,k_8)$ inside $\gamma$. 
In other words, the relation~\eqref{eq:man-rpa-ex2-ecq} is independent of the parameters $(k_1,k_2,k_5,k_7,k_8)$.

\section{ Robust adaptation is topological } \label{sec:rat}

We have seen that the RPA properties and the corresponding integral feedback control can be identified based on indices determined by the topological characteristics of reaction networks. 
Here, we argue that this feature is natural and generically true
because of the robustness of adaptation in broader classes of physical systems beyond deterministic chemical reaction systems.

Let us consider a class of dynamical systems (which do not necessarily have to be deterministic chemical reaction systems) and think of the set of all such systems. 
Among these, we wish to find the systems that exhibit RPA. 
Suppose also that we shall be able to detect such systems with
a real-valued quantity $\chi (S)$ associated with a system $S$\footnote{
In the case of deterministic chemical reaction systems, 
the specification of a system $S$ should 
include the data of 
a CRN $\Gamma$, 
a subnetwork $\gamma$ of $\Gamma$, 
a set of parameters to be perturbed, 
and a choice of output variables. 
See Table~\ref{table:rat-example}.
}, 
and let us say that system $S$ exhibits RPA if $\chi(S)=0$. 
When the perfect adaptation property is robust, 
the adaptation does not require fine-tuning of system parameters. 
If we have another system $S'$, which only differs from $S$ by the values of parameters, then we should also have $\chi (S') =\chi (S) = 0$ 
due to the robustness requirement.
Namely, the quantity $\chi (S)$ is invariant under 
the change of parameters of the system $S$. 
In this sense, the quantity $\chi$ should be a topological invariant. 
Thus, we argue that the class of systems with an RPA property
should be generically identified by a topological invariant. 
We summarize this feature as \emph{``Robust Adaptation is Topological''(RAT)}, and will refer to this motto as the \emph{RAT principle} for short.

The RAT principle itself is not a single theorem, 
but rather is a template for theorems, similarly to the Internal Model Principle of control theory.
Namely, depending on the technical details such as 
the class of systems under consideration, 
the choice of output variables, 
and 
the class of perturbations under consideration, 
the corresponding topological invariant should be identified 
appropriately (if it exists). 
For example, the law of localization~\cite{PhysRevLett.117.048101,PhysRevE.96.022322}
and the RPA of fluxes~\cite{https://doi.org/10.48550/arxiv.2302.01270} 
provide us with particular incarnations of the RAT principle (see Table~\ref{table:rat-example}).
The one-to-one correspondence of RPA properties and labeled buffering structures indicates that all the RPA properties can be identified topologically. 
The law of manifold localization discussed in Sec.~\ref{sec:manifold-rpa} 
generalizes the regulation problem to nonzero dimension of target manifolds, 
and the condition of output-completeness is still topological, in line with the RAT principle. 
The RAT principle gives us guidance for finding systems with RPA properties:
the question is to identify the appropriate topological invariant for the particular situation under consideration. 
As we showed in Sec.~\ref{sec:rpa-lbs},
buffering structures turned out to be also necessary for RPA to be realized. 

The realizations of the RAT principle mentioned here are 
for deterministic chemical reaction systems. 
It would be highly intriguing to explore the possibility of formulating analogous theorems for various classes of dynamical systems, such as stochastic reaction systems, electrical circuits, and others.
It would also be interesting to examine the known results in control theory about robust control from this perspective: the condition for robust control 
may allow for a topological interpretation. 
As a further exploration, it would be interesting to identify the underlying mathematical object 
behind the topological invariant, similarly to the case of the Euler characteristic, which can be written as the alternating sum of the ranks of homology groups of different degrees~\cite{hatcher2001topology}.

\begingroup
\setlength{\tabcolsep}{10pt} 
\renewcommand{\arraystretch}{1.5} 

\begin{table}[tb]
\centering
\begin{tabular}{| c || c | c| c|} 
 \hline
  & 
  Law of localization~\cite{PhysRevLett.117.048101,PhysRevE.96.022322}
  & 
  RPA of fluxes~\cite{https://doi.org/10.48550/arxiv.2302.01270} 
  &
  Law of manifold localization
  \\ 
 \hline\hline
 Perturbation 
 & 
 Parameters inside $\gamma$ 
 & 
 Parameters inside $\gamma$ 
 &
 Parameters inside $\gamma$ 
 \\ 
 \hline 
 Output variables 
 & 
$x_i$ and $r_A$ outside $\gamma$
 & 
 $r_A$ for all reactions
 &
 $x_i$ and $r_A$ outside $\gamma$
 \\ 
 \hline 
 Dim. of target manifolds
 & 0 & 0 & $\lambda(\gamma)$ 
 \\ 
 \hline 
 Topological conditions
 & 
 $\lambda(\gamma) =0$ \& OC
 & 
 $\lambda_{\rm f}(\gamma)=0$ \& OC
 & 
 OC
 \\
 \hline
\end{tabular}
\caption{
Summary of allowed perturbations, output variables, 
dimension of target manifolds for regulation, 
and topological conditions for three realizations of the RAT principle.
OC in the table means output-completeness.
In all the cases, constant-in-time disturbances of parameters inside 
a subnetwork $\gamma \subset \Gamma$ are considered. 
}
\label{table:rat-example}
\end{table}
\endgroup

\section{ Conclusions } \label{sec:conclusion} 

The ability of cellular biomolecular networks to tightly regulate key quantities is of utmost importance, despite their complexity and various sources of disturbances, since the failure to achieve this precise control can have dire consequences for the cell population and, consequently, the entire organism.
The capacity of a reaction network to exhibit robustness against disturbances, known as Robust Perfect Adaptation (RPA), equips cells with the ability to effectively cope with such disruptions. Hence, unraveling the underlying mechanisms for RPA has been a crucial endeavor in understanding cellular resilience.
The goal of this paper is to develop a systematic method for finding all the kinetics-independent RPA properties and the embedded integral control mechanisms associated with them.
We have shown that every elementary RPA property in a deterministic chemical reaction system can be characterized by a labeled buffering structure, 
which is a subnetwork with distinguished topological characteristics (output-completeness and vanishing of the influence index) (Theorem~\ref{thm:rpa-bs}).
This finding gives us a systematic method for identifying all the RPA properties in any chemical reaction network through the identification of labeled buffering structures, with which any generic RPA property can be obtained. 
We provided an efficient algorithm for the enumeration of all the labeled buffering structures for a given reaction network, and developed a computational package {\bf RPAFinder}~\cite{RPAFinder}.
We also gave a method for finding integral feedback controllers for a generic RPA property, which is represented by a buffering structure.
By doing so, we have shown that integral feedback control is \emph{necessary} (and sufficient) for kinetics-independent RPA in deterministic chemical reaction systems, establishing the Internal Model Principle in these systems for constant-in-time disturbances. Unlike most existing works on this topic, our results apply to generic multi-output situations, where many different shades of the RPA property can naturally emerge. 
These different shades correspond to the dimension of the RPA manifold to which the output-species dynamics is driven by the integral feedback mechanisms. This phenomenon, which we call \emph{manifold RPA}, can be studied naturally using our approach.
Remarkably, every output-complete subnetwork gives rise to manifold RPA and the concentrations of the species external to the subnetwork is driven to a manifold whose dimension is given by the influence index (Theorem~\ref{thm:law-of-manifold-localization}). 
Importantly, this manifold remains unaffected by parameters that influence the species and reactions within the subnetwork, even though these parameters might impact the steady-state concentrations of species outside the subnetwork. 
This observation highlights that perturbations in the subnetwork parameters can influence the external species, but they cannot disrupt the robust steady-state relationship among these species, which is determined by the manifold. Exploring this type of RPA holds crucial significance in systems biology and bears relevance to synthetic biology as well, particularly in the design of biomolecular controllers that can robustly maintain relationships between multiple output variables while retaining the adjustability of individual output variables~\cite{alexis2022regulation}.

In biochemical systems, there is no pre-existing distinction between a ``system'' and a ``controller'' unlike the situations in control engineering, and we have to reverse engineer~\cite{doi:10.1126/science.1069981} 
biological circuits to reveal embedded control mechanisms. 
In the present method, a pair of a controlled subsystem and a controller emerges for each labeled buffering structure.
Typically, a reaction network is endowed with multiple labeled buffering structures and hence it possesses multiple distinct RPA properties, which we can systematically uncover with our method. 
This is useful for understanding network resilience in systems biology, and it provides helpful insights into the design principles for synthetic biology~\cite{del2016control}.
We believe it will also be useful for metabolic engineering~\cite{stephanopoulos1998metabolic,doi:10.1146/annurev-chembioeng-061312-103312,moreno2008metabolic,SHIMIZU2022107887}, 
since we shall be able to narrow down the possible candidate reactions to modify for various objectives, like improving yields or reducing costs.

We have pointed out that the topological characterization of the condition for RPA 
is not limited to the studies of deterministic chemical reaction systems 
and should also be applicable to more generic dynamical systems. 
This idea can be expressed as ``Robust Adaptation is Topological,'' and we call this motto as the RAT principle. 
Results such as the one-to-one correspondence of the RPA property and labeled buffering structures,
the law of localization~\cite{PhysRevLett.117.048101,PhysRevE.96.022322},
and the RPA of reaction fluxes~\cite{PhysRevResearch.3.043123},
can be seen as examples of theorems embodying the RAT principle. 
We believe that the RAT principle provides us with a useful guideline 
as to how to formulate theorems identifying dynamical systems with the RPA property, 
depending on the technical details such as 
the class of systems and the choice of target variables. 
It will be interesting to reexamine existing results on robust control 
from a topological perspective. 

Finally, let us comment on possible further directions.
Labeled buffering structures provides us with a natural biological unit, in the sense that they can be controlled independently. 
It would be interesting to investigate how and when labeled buffering structures change when we embed a network in a larger network. 
Namely, the modularity~\cite{wagner2007road} of labeled buffering structures is to be tested.
The present method of enumeration of buffering structures is specifically for subnetworks with vanishing influence index.
As shown in Theorem~\ref{thm:law-of-manifold-localization}, any output-complete subnetwork can induce regulation to a $\lambda(\gamma)$-dimensional manifold. Developing an efficient method to find subnetworks with nonzero indices is important for uncovering manifold RPA in chemical reaction networks.
In the current work, we have considered deterministic reaction systems. 
Whether a similar result can be established for stochastic reaction systems is a nontrivial and intriguing question.
As a generalization to a different direction, the structural conditions governing the regulation to time-dependent states, such as oscillatory ones~\cite{novak2008design}, are to be explored.
To have a more comprehensive understanding of the behavior of chemical reaction systems, it would be interesting to study the connection of the present results with other approaches based on, for example, information geometry~\cite{PhysRevE.106.044131,PhysRevResearch.4.033208,PhysRevResearch.4.033066,kobayashi2022information}
or non-equilibrium thermodynamics~\cite{ito2015maxwell,10.1063/5.0094849,polettini2014irreversible,PhysRevX.6.041064,PhysRevX.13.021041}.

\begin{acknowledgments}
Y.~H. is grateful to 
Hyukpyo Hong,
Jae Kyoung Kim 
for useful discussions. 
The work of Y.~H. was supported in part by JSPS KAKENHI Grant No. JP22H05111, and in part by an appointment of the JRG Program at the APCTP, which is funded by the Science and Technology Promotion Fund and Lottery Fund of the Korean Government and also by the Korean Local Governments of Gyeongsangbuk-do Province and Pohang City. A.~G. and M.~K. acknowledge funding from ETH Zurich and from the Swiss National Science Foundation (SNSF) under grant 182653.
\end{acknowledgments}

\appendix

\section{ Spectral correspondence }\label{sec:spectral-correspondence} 

In the current paper, we assume that the dynamics is asymptotically stable, i.e., the Jacobian matrix at steady state is non-singular and all its eigenvalues have strictly negative real parts, as discussed in Sec.~\ref{sec:assumptions}.
In the case $\coker S = \bm 0$, this implies that the ${\bf A}$-matrix is invertible due to Theorem 1 in Ref.~\cite{PhysRevE.103.062212}. 
Let us here show that this continues to hold even when $\coker S \neq \bm 0$.

Let us consider the time evolution of 
fluctuations around a steady state, which is governed by 
\begin{equation}
\frac{d}{dt} \delta x_i =\sum_j J_{ij} \delta x_j , 
\end{equation}
where 
$J_{ij}\coloneqq \sum_A S_{iA} r_{A,j}$ is 
the Jacobian matrix. 
To be consistent with the conservation relations, 
the fluctuations should satisfy the constraints 
\begin{equation}
\sum_i d^{(\bar\alpha)}_i \delta x_i = 0. 
\end{equation}
The eigenvalue equations of the Jacobian matrix 
satisfying the constraint from conserved quantities read 
\begin{align}
J \delta \bm x_\lambda &= \lambda \delta \bm x_\lambda, \\ 
D \delta \bm x_\lambda &= \bm 0 ,
\end{align}
where we have used the matrix notation 
and the components of the matrix $D$ is 
given by $D_{\bar\alpha i}\coloneqq d^{(\bar\alpha)}_i$. 
An eigenvector of the Jacobian 
in the presence of conserved quantities 
can be seen as an element of the kernel of the following matrix,
\begin{equation}
{\bf J}_\lambda 
\coloneqq 
\left[ 
\begin{array}{c}
    J  - \lambda \\ 
    D
\end{array}
\right]. 
\end{equation}
We will show that there is a one-to-one correspondence 
between $\ker {\bf J}_\lambda$ and 
the generalized eigenvectors/values of the ${\bf A}$-matrix. 
We here use a matrix notation and express the $\bf A$-matrix~\eqref{eq:mat-a-def} as 
\begin{equation}
{\bf A} 
= 
\left[ 
\begin{array}{c|c}
R & -C \\ \hline 
D & \bm 0 
\end{array} 
\right] . 
\end{equation}
We introduce the following matrix parametrized by $\lambda \in \mathbb C$, 
\begin{equation}
{\bf A}_\lambda 
\coloneqq 
\left[ 
\begin{array}{c|c}
R - \lambda S^+  & -C \\  \hline 
D & \bm 0 
\end{array}
\right] . 
\end{equation}

We will show the following: 
\begin{theorem}
The following map $F$ is a bijection,
\begin{equation}
F : 
\ker  {\bf J}_\lambda
\to 
\ker {\bf A}_\lambda ,
\end{equation}
where $F$ is defined by 
\begin{equation}
\ker {\bf J}_\lambda \ni \delta \bm x_\lambda 
\mapsto 
F (\delta \bm x_\lambda) 
\coloneqq 
\begin{bmatrix}
\delta \bm x_\lambda \\ 
\bm \mu 
\end{bmatrix}, 
\end{equation}
Here, $\bm \mu$ is determined uniquely by 
\begin{equation}
C \bm \mu = (1 - S^+ S) R \delta \bm x_\lambda. 
\label{eq:c-mu-r-dx}
\end{equation}
\end{theorem}

\begin{proof}
    
Let us show that 
indeed $F(\delta \bm x_\lambda) \in \ker {\bf A}_\lambda$. 
An element $\delta \bm x_\lambda \in \ker {\bf J}_\lambda$ satisfies 
\begin{align}
S R \delta \bm x_\lambda &= \lambda \delta \bm x_\lambda ,\label{eq:sr-del-x-x}
\\
D \delta \bm x_\lambda &= \bm 0 \label{eq:d-delta-x}. 
\end{align}
Equation \eqref{eq:d-delta-x} indicates that 
$\delta \bm x_\lambda \in (\coker S)^\perp$. 
Since $S S^+$ is a projection matrix to  $(\coker S)^\perp$, 
we can multiply this to the RHS of 
Eq.~\eqref{eq:sr-del-x-x}, 
and we have 
\begin{equation}
S 
\left( 
R   - \lambda S^+ 
\right) 
\delta \bm x_\lambda 
= \bm 0 . 
\end{equation}
This means that we can write 
\begin{equation}
\left( 
R 
- \lambda S^+ 
\right) 
\delta \bm x_\lambda 
 = C \bm \mu, 
 \label{eq:r-x-lm-s}
\end{equation}
Since 
$S^+ \delta \bm x_\lambda \in (\ker S)^\perp$, 
by multiplying the projection matrix to $\ker S$, we obtain 
Eq.~\eqref{eq:c-mu-r-dx}. 
For a given set of basis vectors $C$ of $\ker S$, 
$\bm \mu$ is determined uniquely. 
Equations \eqref{eq:r-x-lm-s} and \eqref{eq:d-delta-x} 
can be rearranged as 
\begin{equation}
\left[ 
\begin{array}{c|c}
R & -C  \\  \hline 
D & \bm 0 
\end{array}
\right] 
\begin{bmatrix}
    \delta \bm x_\lambda \\ 
    \bm \mu 
\end{bmatrix}
= 
\lambda 
\left[ 
\begin{array}{c|c}
S^+ & \bm 0 \\ \hline 
\bm 0 & \bm 0 
\end{array}
\right]
\begin{bmatrix}
    \delta \bm x_\lambda \\ 
    \bm \mu 
\end{bmatrix}, 
\end{equation}
which means that 
$
\begin{bmatrix}
    \delta \bm x_\lambda \\ 
    \bm \mu 
\end{bmatrix}
\in \ker {\bf A}_\lambda. 
$

The map $F$ is obviously injective. 
Conversely, 
for a given 
$
\begin{bmatrix}
    \delta \bm x_\lambda \\ 
    \bm \mu 
\end{bmatrix}
\in \ker {\bf A}_\lambda,
$
it can be easily checked that 
$\delta \bm x_\lambda \in \ker {\bf J}_\lambda$,
which means that $F$ is surjective. 
Thus, we have shown that $F$ is a bijection. 
This concludes the proof. 

\end{proof}

\begin{cor}\label{cor:invertivility-of-a}
When a chemical reaction system has an asymptotically stable steady state, the corresponding ${\bf A}$-matrix is invertible. 
\end{cor}

\begin{proof}
Because of the asymptotic stability, 
$\ker {\bf J}_\lambda$ is trivial when $\lambda=0$, 
and consequently 
$\ker {\bf A}_\lambda$ is trivial when $\lambda=0$. 
This implies that the square matrix ${\bf A}$ does not have a zero eigenvalue, which means that ${\bf A}$ is invertible. 
\end{proof}

\section{ Formula for second-order responses }\label{sec:second-order-formula}

Here we derive explicit formulas for second-order responses 
of steady-state concentrations and reaction fluxes 
with respect to parameters. 
In the following, we derive the expression 
for $\bar x_{i,AB}$ with $A \neq B$, 
and also for 
\color{magenta}
$\bar y_{\nu,\rho\sigma}$  
\color{black}
with 
\color{magenta}
$\rho \neq \sigma$ 
\color{black}
in parallel.

We start with the formula for the first-order response, 
\begin{equation}
\bar x_{i,A}
= 
- ({\bf A}^{-1})_{iA} \, \p_A r_A. 
\quad \quad 
\color{magenta}
\bar y_{\nu,\rho} 
= 
- ({\bf A}^{-1})_{\nu \rho} \, \p_\rho L_\rho.  
\color{black} 
\end{equation}
By taking another derivative with respect to a parameter, 
\begin{equation}
\bar x_{i,AB}
= 
- 
({\bf A}^{-1})_{iA,B} \, \p_A r_A  
- 
({\bf A}^{-1})_{iA} \, (\p_A r_A)_{,B} . 
\end{equation}
\begin{equation}
\color{magenta}
y_{\nu,\rho \sigma }
= 
- ({\bf A}^{-1})_{\nu\rho,\sigma} \, \p_\rho L_\rho 
- ({\bf A}^{-1})_{\nu\rho} \, (\p_\rho L_\rho)_{,\sigma} . 
\end{equation}
\color{black}
The second term is written as 
\begin{equation}
- ({\bf A}^{-1})_{iA} \, (\p_A r_A)_{,B} 
= 
\bar x_{i,A} (\ln \p_A r_A)_{,B} 
\quad 
\color{magenta}
- ({\bf A}^{-1})_{\nu \rho } \, (\p_\rho L_\rho)_{,\sigma} 
= 
\bar y_{\nu,\rho} (\ln \p_\rho L_\rho)_{,\sigma} .
\end{equation}
\color{black}
Noting that 
\begin{equation}
({\bf A}^{-1})_{iA,B} 
= 
- 
\sum_{C,D}
({\bf A}^{-1})_{i C} 
\p_B {\bf A}_{CD}
({\bf A}^{-1})_{D A} ,
\quad 
\color{magenta}
({\bf A}^{-1})_{\nu \rho,\sigma} 
= 
- 
\sum_{\kappa ,\tau} 
({\bf A}^{-1})_{\nu \kappa} 
\p_\sigma {\bf A}_{\kappa \tau }
({\bf A}^{-1})_{\tau \rho} , 
\end{equation}
\color{black}
and 
\begin{equation}
\p_B {\bf A}_{Cj}    
= 
\delta_{BC} (\p_B r_C)_{,j} 
+ \sum_k r_{C,jk} \, \bar x_{k,B} , 
\quad 
\color{magenta}
\p_\sigma {\bf A}_{C j} 
= 
\delta_{\sigma C} (\p_\sigma r_C)_{,j} 
+ \sum_k r_{C,jk} \, \bar y_{k,\sigma} , 
\end{equation}
\color{black}
the first term is written as 
\begin{equation}
\begin{split}
- ({\bf A}^{-1})_{iA,B} \, \p_A r_A  
&= 
\sum_{C,D}
({\bf A}^{-1})_{i C} 
\p_B {\bf A}_{CD}
({\bf A}^{-1})_{D A} 
 \, \p_A r_A 
\\
&= 
\sum_{C,j}
({\bf A}^{-1})_{i C} 
\left( 
\delta_{BC} (\p_B r_C)_{,j} + \sum_k r_{C,jk} \bar x_{k,B} 
\right)
({\bf A}^{-1})_{j A} 
 \, \p_A r_A 
\\
&= 
\sum_j 
({\bf A}^{-1})_{i B} 
 (\p_B r_B)_{,j} 
({\bf A}^{-1})_{j A} 
 \, \p_A r_A 
+ 
\sum_{C,j,k}
({\bf A}^{-1})_{i C}\, r_{C,jk} \, \bar x_{k,B} 
({\bf A}^{-1})_{j A}  \, \p_A r_A 
 \\
 &= 
\sum_j 
\bar x_{i,B} 
(\ln \p_B r_B)_{,j}
\bar x_{j,A}
+ 
\sum_{C,j,k}
\frac{r_{C,jk} }{\p_C r_C}
\,
\bar x_{i,C} \, 
\bar x_{j,A} \,
\bar x_{k,B} 
\\
&= 
 \bar x_{i,B} (\ln \p_B r_B)_{,A}
+ 
\sum_{C,j,k}
\frac{r_{C,jk} }{\p_C r_C}
\,
\bar x_{i,C} \, 
\bar x_{j,A} \,
\bar x_{k,B} .
\end{split}    
\end{equation}
\color{magenta}
\begin{equation}
\begin{split}
- ({\bf A}^{-1})_{\nu \rho, \sigma} \, \p_\rho L_\rho   
&= 
\sum_{\kappa, \tau }
({\bf A}^{-1})_{\nu \kappa} 
\p_\sigma {\bf A}_{\kappa \tau}
({\bf A}^{-1})_{\tau \rho} 
 \, \p_\rho L_\rho 
\\
&= 
\sum_{C,j}
({\bf A}^{-1})_{\nu C} 
\left( 
\delta_{\sigma C} (\p_\sigma r_C)_{,j} + \sum_k r_{C,jk} \, y_{k,\sigma} 
\right)
({\bf A}^{-1})_{j \rho} 
 \, \p_\rho L_\rho 
\\
&= 
\sum_j 
({\bf A}^{-1})_{\nu \sigma} 
 (\p_\sigma r_\sigma)_{,j} 
({\bf A}^{-1})_{j \rho} 
 \, \p_\rho L_\rho 
+ 
\sum_{C,j,k}
({\bf A}^{-1})_{\nu C}\, r_{C,jk} \, y_{k,\sigma} 
({\bf A}^{-1})_{j \rho}  \, \p_\rho L_\rho 
 \\
 &= 
\sum_j 
\bar y_{\nu,\sigma}  
(\ln \p_\sigma r_\sigma)_{,j}
\, \bar y_{j,\rho} 
+ 
\sum_{C,j,k}
\frac{r_{C,jk} }{\p_C r_C}
\,
\bar y_{\nu,C} \, 
\bar y_{j,\rho} \,
\bar y_{k,\sigma} 
\\
&= 
\bar y_{\nu,\sigma} (\ln \p_\sigma L_\sigma)_{,\rho}
+ 
\sum_{C,j,k}
\frac{r_{C,jk} }{\p_C r_C}
\,
\bar y_{\nu,C} \, 
\bar y_{j,\rho} \,
\bar y_{k,\sigma} . 
\end{split}    
\end{equation}
\color{black}
Thus, we obtain the following formula, 
\begin{equation}
\bar x_{i,AB}
=     
\bar x_{i,A} F_{A,B}
+
\bar x_{i,B} F_{B,A} 
+ 
\sum_{C} 
\bar x_{i,C} 
\sum_{j,k \, \vdash C}
\frac{r_{C,jk} }{r_{C,C}}
\,
\bar x_{j,A} \,
\bar x_{k,B} , 
\label{eq:xiab-formula}
\end{equation}
where $F_A \coloneqq \ln \p_A r_A$. 
\color{magenta}
\begin{equation}
\bar y_{\nu,\rho \sigma}
=     
\bar y_{\nu,\rho} F_{\rho,\sigma} 
+
\bar y_{\nu,\sigma} F_{\sigma, \rho}  
+ 
\sum_{C} 
\bar y_{\nu,C} 
\sum_{j,k \, \vdash C}
\frac{r_{C,jk} }{r_{C,C}}
\,
\bar y_{j,\rho} \,
\bar y_{k,\sigma} , 
\label{eq:y-formula}
\end{equation}
where $F_\sigma \coloneqq \ln \p_\sigma L_\sigma$. 
\color{black}
Note that this expression is manifestly symmetric 
under the exchange $A \leftrightarrow B$ 
\color{magenta} $\rho \leftrightarrow \sigma$\color{black}
. 

We can also derive a similar formula for reaction fluxes. 
Noting that $\bar y_{\alpha} = \mu_\alpha$ and $\bar r_C$ can be written as 
\begin{equation}
\bar r_C = \sum_\alpha c^{(\alpha)}_C 
\bar y_{\alpha} , 
\end{equation}
we find that second-order responses
of reaction fluxes are expressed as
\begin{equation}
\bar r_{C,AB}
=     
\bar r_{C,A} F_{A,B}
+
\bar r_{C,B} F_{B,A} 
+ 
\sum_{D} 
\bar r_{C,D} 
\sum_{j,k \, \vdash D}
\frac{r_{D,jk} }{r_{D,D}}
\,
\bar x_{j,A} \,
\bar x_{k,B} . 
\label{eq:rcab-formula}
\end{equation}
Note that $e_C$ can be taken to be the same as $e_A$ or $e_B$, while we assume $e_A \neq e_B$.
\color{magenta}
Including the conserved-quantity perturbations, we have
\begin{equation}
\bar r_{C,\rho\sigma}
=     
\bar r_{C,\rho} F_{\rho,\sigma}
+
\bar r_{C,\sigma} F_{\sigma,\rho} 
+ 
\sum_{D} 
\bar r_{C,D} 
\sum_{j,k \, \vdash D}
\frac{r_{D,jk} }{r_{D,D}}
\,
\bar y_{j,\rho} \,
\bar y_{k,\sigma} . 
\label{eq:rcab-formula-cons}
\end{equation}
\color{black}

The formulas \eqref{eq:xiab-formula} and \eqref{eq:rcab-formula} (or \eqref{eq:y-formula} and \eqref{eq:rcab-formula-cons})
can be seen as a decomposition of higher-order influence 
in terms of lower-order influence.
We can diagrammatically represent Eq.~\eqref{eq:xiab-formula} as 
\begin{center}
\begin{tikzpicture}[thick]

\path (0,0.7) node(eA) {$e_A$} 
      (0,-0.7) node(eB) {$e_B$}
      (1,0) node(mdl) {}
      (2.6,0) node(eC) {$v_i \,\, =\,\,\sum_{j \vdash A}$}
      (4,0) node(eB1) {$e_B$} 
      (5.6,0) node (ej1) {$v_j \vdash e_A$}  
      (8.5,0) node(vi1) {$v_i \,\,  +\,\, \sum_{k \vdash B} e_A  $} 
      (11.2,0) node(vk2) {$v_k \vdash e_B$} 
      (13.7,0) node(vi2) {$v_i \,\,\,  +\,\,\, \sum_{C,j,k}$} 
      (14.9,0.6)  node(eA3) {$e_A$}      
      (14.9,-0.6) node(eB3) {$e_B$} 
      (16,0.6)  node(vj3) {$v_j$}      
      (16,-0.6) node(vk3) {$v_k$}  
      (16.8,0) node(eC3) {$e_C$} 
      (18,0) node(vi3) {$v_i$}
      (16.4,0.35) node[rotate=-45](vd1) {$\vdash$} 
      (16.4,-0.35) node[rotate=45](vd2) {$\vdash$} 
      ; 

\draw[-, style={decorate, decoration={snake,segment length=2mm, amplitude=0.5mm}}] (eA) -- (1,0);
\draw[-, style={decorate, decoration={snake,segment length=2mm, amplitude=0.5mm}}] (eB) -- (1,0);
\draw[->, style={decorate, decoration={snake,segment length=2mm, amplitude=0.5mm}}] (1,0) -- (eC);

\draw[->, style={decorate, decoration={snake,segment length=2mm, amplitude=0.5mm}}] 
(eB1) -- (ej1); 
\draw[->, style={decorate, decoration={snake,segment length=2mm, amplitude=0.5mm}}] (ej1) -- (vi1); 

\draw[->, style={decorate, decoration={snake,segment length=2mm, amplitude=0.5mm}}] (vi1) -- (vk2); 
\draw[->, style={decorate, decoration={snake,segment length=2mm, amplitude=0.5mm}}] (vk2) -- (vi2);  

\draw[->, style={decorate, decoration={snake,segment length=2mm, amplitude=0.5mm}}] (eA3) -- (vj3);  
\draw[->, style={decorate, decoration={snake,segment length=2mm, amplitude=0.5mm}}] (eB3) -- (vk3);  

\draw[->, style={decorate, decoration={snake,segment length=2mm, amplitude=0.5mm}}] (eC3) -- (vi3);  

\end{tikzpicture}
\end{center}
The LHS represents the influence of both $e_A$ and $e_B$ on $v_i$, 
and there are different contributions. 
The first term on the RHS represents a contribution 
that $e_B$ first influences a reactant of $e_A$ and 
then $e_A$ influences $v_i$, 
and 
in the second term the roles of $e_A$ and $e_B$ are exchanged. 
The third term on the RHS
takes into the contribution where 
$e_A$ and $e_B$ affects 
reactants of some reaction $e_C$, 
and then $e_C$ influences $v_i$.

Let us check the second-order formula with a simple example network, 
\begin{align}
e_1 &: \emptyset \to v_1, \\ 
e_2 &: v_1 \to \emptyset, \\ 
e_3 &: v_1 + v_2 \to v_1, \\ 
e_4 &: \emptyset \to v_2 . 
\end{align}
With mass-action kinetics, 
the rate functions are written as 
$r_1 = k_1$, $r_2 = k_2 x_1$, $r_3 = k_3 x_1 x_2$, $r_4 = k_4$. 
The steady-state concentrations are given by 
\begin{equation}
\bar x_1 = \frac{k_1}{k_2}, 
\quad 
\bar x_2 = \frac{k_2 k_4}{k_1 k_3}. 
\end{equation}
For example, the second-order derivative 
${\bar x}_{2,21} = \p_{k_2} \p_{k_1} \bar x_2$ 
can be explicitly computed as 
\begin{equation}
\bar x_{2,21}  = - \frac{k_4}{(k_1)^2 k_3}. 
\label{eq:second-order-lhs-ex}
\end{equation}
Let us evaluate the RHS of Eq.~\eqref{eq:xiab-formula}. 
The first and second terms are 
\begin{equation}
\bar x_{2,2} F_{2,1}
= 
\frac{k_4}{k_1 k_3}
\cdot 
\frac{(\partial_{k_2} r_2)_{,1}}{\partial_{k_2} r_2}
=
\frac{k_4}{(k_1)^2 k_3}
, 
\quad 
\bar x_{2,1} F_{1,2}
= 
\bar x_{2,1} 
\frac{(\partial_{k_1} r_1)_{,2}}
{\partial_{k_1} r_1}
= 0.     
\label{eq:second-ordcer-first-second-ex}
\end{equation}
The third term is evaluated as 
\begin{equation}
\sum_{C}
\frac{\bar x_{2,C} }
{r_{C,C}}
\sum_{j,k \,\vdash C}
r_{C,jk}
\,
\bar x_{j,2} 
\,
\bar x_{k,1} 
=
\frac{\bar x_{2,3} }{r_{3,3}}
\left( 
r_{3,12}
\,
\bar x_{1,2} 
\,
\bar x_{2,1} 
+
r_{3,21}
\,
\bar x_{2,2} 
\,
\bar x_{1,1} 
\right)
= 
-2 \frac{k_4}{(k_1)^2 k_3}. 
\label{eq:second-ordcer-third-ex}
\end{equation}
The summation of 
Eqs.~\eqref{eq:second-ordcer-first-second-ex} 
and \eqref{eq:second-ordcer-third-ex} 
indeed reproduces Eq.~\eqref{eq:second-order-lhs-ex}.

\section{Extension of the maxRPA characterization result} \label{proof:maxrpageneralizattion}

Here, we prove Theorem~\ref{thm:maxRPA_characterization}, which is an extension of the characterization result for maxRPA networks~\cite{gupta2022universal}. In this extension, we no longer require that $\coker S = \bm 0$ and hence the full network $\Gamma =(V, E)$ may have conserved quantities. The last two reactions, denoted by  $e_{\bar 1}$ and $e_{\bar 2}$ no longer need to have mass-action kinetics, and we will denote their rate functions as $r_{\bar 1}(\bm x; k_{\bar 1})$ and $r_{\bar 2}(\bm x; k_{\bar 2})$ respectively, where  $k_{\bar 1}$ and $k_{\bar 2}$ are parameters.

\begin{proof}
We start by proving the ``if" part of the theorem, namely that if these two conditions are satisfied then the network exhibits maxRPA for the output species $X$. This is straightforward because by picking any $(\bm q , \kappa)$ that satisfies Eq.~\eqref{eq:qs}, we can construct a linear integrator $z \coloneqq \bm q \cdot \bm x$ whose dynamics is given by
\begin{equation}
\dot z = \kappa \, r_{\bar 1}(\bm x;k_{\bar 1} ) - r_{\bar 2 }(\bm x;k_{\bar 2}) = \ r_{\bar 1 }(\bm x;k_{\bar 1}) \ \left( \kappa  -  \Phi(x_M, k_{\bar 1}, k_{\bar 2})\right)  
\label{eq:integrator_dyn_maxrpa}
\end{equation}
where the last relation follows from Eq.~\eqref{max_rpa_ratio_cond}. Setting the RHS to zero, and using the positivity of the fixed point we obtain that at steady-state the concentration of the output species $X$ must satisfy Eq.~\eqref{max_rpa_gen_set_point} and hence it is independent of all other parameters except $k_{\bar 1}$ and $ k_{\bar 2}$. Note that if there exists another pair $(\bm q' , \kappa')$ satisfying Eq.~\eqref{eq:qs}, then we must necessarily have that $\kappa = \kappa'$, or otherwise the condition of uniqueness of the steady-state will be violated. As both $(\bm q' , \kappa)$ and $(\bm q , \kappa)$ satisfy Eq.~\eqref{eq:qs}, subtracting them yields that the vector $(\bm q'- \bm q)$ will be in $\coker S$. This proves that the vector $\bm q$ is unique up to addition of vectors in $\coker S$, or equivalently, we can say that $\bm q$ is unique in the quotient space $\mathbb R^{M} / \coker S$.

We now prove the ``only if" part of the theorem. We start from the steady-state equations Eqs.~\eqref{eq:sr} and \eqref{eq:l-dx}, where $\{\bm d^{(\bar\alpha)}\}_{\bar\alpha=1,\ldots ,|\bar\alpha|}$ is a basis of $\coker S$ and for each $\bar\alpha$, $\ell^{\bar\alpha}$ is the value of the conserved quantity corresponding to $\bm d^{(\bar\alpha)}$. 
Differentiating Eq.~\eqref{eq:sr} with respect to parameter $k_B$ and quantity $\ell^{\bar \beta}$ we arrive at:
\begin{align}
\label{max_rpa_gen1}
\sum_A  S_{iA}  \frac{\p r_A}{\p x_i}  \frac{\p \bar x_i}{\p k_B} = - \sum_A  S_{iA}   \frac{\p r_A}{\p k_B} \qquad \textnormal{and} \qquad \sum_A  S_{iA}  \frac{\p r_A}{\p x_i}  \frac{\p \bar x_i}{\p \ell^{\bar\beta}} &= 0.
\end{align}
As we have seen before, differentiating Eq.~\eqref{eq:l-dx} with respect to $k_B$ and $\ell^{\bar \beta}$ gives us relations, 
\begin{align}
\label{max_rpa_gen2}
\sum_i d_i^{(\bar\alpha)} \frac{\p x_i}{\p k_B} &= 0\qquad \textnormal{and} \qquad
\sum_i  d_i^{(\bar\alpha)} \frac{\p x_i}{\p \ell^{\bar\beta}}=  \delta^{\bar\alpha \bar\beta}. 
\end{align}
We shall write these relations in matrix-form and for this we need to define certain quantities. Recall that ${\bm r}$ is the $N$-dimensional vector of reaction rates and $S$ is the $M \times N$ stoichiometric matrix. 
Let $\nabla_{\bm x} {\bm r} = [\p_i r_A]$ be the $N \times M$ Jacobian matrix of ${\bm r}$ with respect to the explicit dependence on the state $\bm x$, evaluated at the steady-state $\bm x = \bar{\bm x} (\bm k, \bm \ell)$. 
The $M \times M$ Jacobian for the dynamics is given by
\begin{align}
J = S \nabla_{\bm x} {\bm r}. 
\end{align}
Similarly, let $\nabla_{\bm k} {\bm r} = [\p_B {r}_A]$ be the $N \times N$ 
Jacobian matrix of ${\bm r}$ with respect to the explicit dependence on the parameters in $\bm k$,
and let $\nabla_{\bm k, \bm \ell} \bar{\bm x} (\bm k, \bm \ell)$ be the $M \times (N +|\bar\alpha|)$ Jacobian matrix of the steady-state $\bar{\bm x} (\bm k, \bm \ell)$ with respect to both $\bm k$ and $\bm \ell$. Let $D$ be the $|\bar\alpha| \times M$ matrix whose rows are the basis vectors $\{\bm d^{(\bar\alpha)}\}_{\bar\alpha=1,\ldots ,|\bar\alpha|}$ for $\coker S$. Relations \eqref{max_rpa_gen1} and \eqref{max_rpa_gen2} can be succinctly expressed as
\begin{align}
\label{main_maxrpa_reln}
\left[ 
\begin{array}{c}
J \\
D
\end{array}
\right] \nabla_{\bm k, \bm \ell} \bar{\bm x} (\bm k, \bm \ell) = -\left[\begin{array}{cc} 
 S \nabla_{\bm k} {\bm r} &  {\bm 0}_{ M \times |\bar\alpha|} \\
{\bm 0}_{ |\bar\alpha| \times M} & - {\bm 1}_{|\bar\alpha| \times |\bar\alpha| }  
\end{array} \right],
\end{align}
where $\bm 0$ and $\bm 1$ denote zero and identity matrices of the dimensions indicated in the subscript. Note that the matrix $(M + |\bar\alpha| )\times M$ matrix
\begin{align}
B \coloneqq \left[ 
\begin{array}{c}
J \\
D
\end{array}
\right]
\end{align}
has independent columns, because otherwise $\ker B$ would be nontrivial, and this would imply that the steady state is not asymptotically stable. 
Since the columns of $B$ are independent, its Moore-Penrose inverse is explicitly given by
\begin{equation} 
B^{+} = (B^\top \ B)^{-1} B^\top,
\end{equation}
and hence Eq.~\eqref{main_maxrpa_reln} implies that
\begin{align}
\label{main_ss_sens}
\nabla_{\bm k, \bm \ell} \bar{\bm x} (\bm k, \bm \ell) = - (B^\top \ B)^{-1} B^\top \left[\begin{array}{cc} 
 S \nabla_{\bm k} {\bm r} &  {\bm 0}_{ M \times |\bar\alpha|} \\
{\bm 0}_{ |\bar\alpha| \times M} & - {\bm 1}_{|\bar\alpha| \times |\bar\alpha| }  
\end{array} \right].
\end{align}

Define $\tilde{\bm q}$ to be the $(M + |\bar\alpha| )$-dimensional vector given by
\begin{equation}
\tilde{\bm q}^\top = -  {\bm u}^\top_M (B^\top \ B)^{-1} B^\top ,
\end{equation}
where $\bm u_M$ is a $M$-dimensional vector whose all components are zero, except the component at the $M$-th index (i.e. location of the output species $X$) which is 1. We can decompose $\tilde{\bm q}^\top $ as $\tilde{\bm q}^\top = [\tilde{\bm q}_1^\top \  \tilde{\bm q}_2^\top]$, where $\tilde{\bm q}_1^\top$ and $\tilde{\bm q}_2^\top $ are vectors with dimensions $1 \times M$ and $1 \times |\bar\alpha| $ respectively. Multiplying Eq.~\eqref{main_ss_sens} by ${\bm u}^\top_M$ on the left we see that the $1 \times (M + |\bar\alpha|)$ vector of sensitivities of the steady-state $\bar x_M (\bm k, \bm \ell)$ of the output with respect to $\bm k, \bm \ell$ is given by
\begin{align}
\label{maxrpa_output_sens}
\nabla_{\bm k, \bm \ell} \bar x_M (\bm k, \bm \ell) \coloneqq \bm u^\top_M \nabla_{\bm k, \bm \ell}\bar{\bm x} (\bm k, \bm \ell) 
= 
\begin{bmatrix}
\tilde{\bm q}_1^\top & \tilde{\bm q}_2^\top 
\end{bmatrix}
\left[\begin{array}{cc} 
 S \nabla_{\bm k} {\bm r} &  {\bm 0}_{ M \times |\bar\alpha|} \\
{\bm 0}_{ |\bar\alpha| \times M} & -{\bm 1}_{|\bar\alpha| \times |\bar\alpha| }
\end{array} \right] 
= 
\begin{bmatrix}
\tilde{\bm q}_1^\top S \nabla_{\bm k} {\bm r} &  - \tilde{\bm q}_2^\top
\end{bmatrix}.
\end{align}
For maxRPA we require that
\begin{align}
\label{max_rpa_output_func}
\bar x_M (\bm k, \bm \ell) = \phi_{\rm out} \left( k_{\bar 1}, k_{\bar 2} \right),
\end{align}
and hence the steady-state output has zero-sensitivity with respect to all components of $\bm \ell$ and all components of $\bm k$ except $k_{\bar 1}$ and $k_{\bar 2}$. 
Therefore Eq.~\eqref{maxrpa_output_sens} implies that $\tilde{\bm q}^\top_2 = - \nabla_{\bm \ell} \bar x_M (\bm k, \bm \ell) = {\bf 0}$ and 
\begin{align}
\label{maxrpa_output_sens2}
\tilde{\bm q}^\top_1 S \nabla_{\bm k} {\bm r} =
\begin{bmatrix}
{\bf 0} & \,\,
\partial_{k_{\bar 1}}\phi_{\rm out} \left( k_{\bar 1}, k_{\bar 2} \right) & \,\, \partial_{k_{\bar 2}}\phi_{\rm out} \left( k_{\bar 1}, k_{\bar 2} \right) 
\end{bmatrix}.
\end{align}
Only the last two components of the vector on the RHS are non-zero, while the rest are zeros. Since for each reaction $e_A$, the reaction rate $r_A (\bar{\bm x} (\bm k, \bm \ell); k_A)$ has an explicit dependence on an independent $k_A$, 
the $N \times N$ matrix $\nabla_{\bm k} {\bm r} $ is a diagonal matrix.
Therefore, Eq.~\eqref{maxrpa_output_sens2} informs us that for any reaction $A$, which is not one of the last two reactions (i.e.\ $A \neq {\bar 1}, {\bar 2}$) we must have
\begin{align}
\label{max_rpa_charge_preservation}
    \sum_i \tilde{\bm q}_{1 i} S_{i A} = 0,
\end{align}
which says that vector $\tilde{\bm q}_{1}$ is orthogonal to the stoichiometric vector for reaction $e_A$. Now for the last two reactions, Eq.~\eqref{maxrpa_output_sens2} implies that
\begin{align}
    \sum_i \tilde{\bm q}_{1i} S_{i {\bar 1}} = \frac{\partial_{k_{\bar 1}}\phi_{\rm out} \left( k_{\bar 1}, k_{\bar 2} \right)}{\partial_{k_{\bar 1}} r_{\bar 1} (\bar{\bm x} (\bm k, \bm \ell); k_{\bar 1} )}      
    \quad \textnormal{and} \quad     
    \sum_i \tilde{\bm q}_{1i} S_{i {\bar 2}} = \frac{\partial_{k_{\bar 2}}\phi_{\rm out} \left( k_{\bar 1}, k_{\bar 2} \right)}{\partial_{k_{\bar 2}} r_{\bar 2} (\bar{\bm x} (\bm k, \bm \ell); k_{\bar 2} )},
\end{align}
where for $i=1,2$, $\p_{k_{\bar i}} r_{\bar i} (\bar{\bm x} (\bm k, \bm \ell); k_{\bar i} )$ denotes the partial derivative with respect to only the explicit dependence on parameter $k_{\bar i}$ and not the implicit dependence through the steady-state $\bar{\bm x} (\bm k, \bm \ell)$.
Combining these relations with Eq.~\eqref{max_rpa_charge_preservation}, we see that the vector $\tilde{\bm q}_1$ satisfies
\begin{align}
\label{vec_q_first_rel}
\tilde{\bm q}_1^\top S = 
\begin{bmatrix}
{\bm 0},
&
\frac{\partial_{k_{\bar 1}}\phi_{\rm out} \left( k_{\bar 1}, k_{\bar 2} \right)}{\partial_{k_{\bar 1}} r_{\bar 1} (\bar{\bm x} (\bm k, \bm \ell); k_{\bar 1} )}, 
&
\frac{\partial_{k_{\bar 2}}\phi_{\rm out} \left( k_{\bar 1}, k_{\bar 2} \right)}{\partial_{k_{\bar 2}} r_{\bar 2} (\bar{\bm x} (\bm k, \bm \ell); k_{\bar 2} )}  
\end{bmatrix}
.
\end{align}
Now let $\bm c = 
\begin{bmatrix}
c_1 & \cdots & c_N
\end{bmatrix}^\top$ be a vector in $\ker S$ whose last two components are not both zero. Choosing such a $\bm c$ is possible because the network has positive steady-states. Since $S \bm c= {\bm 0}$, we must have $\tilde{\bm q}_1^\top S \bm c = 0$ which implies that
\begin{align}
\label{max_rpa_reln_for_c}
c_{N-1} \frac{\partial_{k_{\bar 1}}\phi_{\rm out} \left( k_{\bar 1}, k_{\bar 2} \right)}{\partial_{k_{\bar 1}} r_{\bar 1} (\bar{\bm x} (\bm k, \bm \ell); k_{\bar 1} )} 
=  
- c_{N} \frac{\partial_{k_{\bar 2}}\phi_{\rm out} \left( k_{\bar 1}, k_{\bar 2} \right)}{\partial_{k_{\bar 2}} r_{\bar 2} (\bar{\bm x} (\bm k, \bm \ell); k_{\bar 2} )}.
\end{align}
Note that for any $\bm c$ fixed, this relation must hold for any $(\bm k, \bm \ell)$. 
We argued before that the matrix $\nabla_{\bm k} {\bm r} $ has full row-rank. This implies that $\coker (S \nabla_{\bm k} {\bm r}) = \coker S$ and hence the row-rank of the product matrix $S \nabla_{\bm k} {\bm r}$ is just $M - |\bar \alpha|$, where $|\bar \alpha| = |\coker S|$. This shows that the matrix
\begin{equation} 
\left[\begin{array}{cc} 
 S \nabla_{\bm k} {\bm r} &  {\bm 0}_{ M \times |\bar\alpha|} \\
{\bm 0}_{ |\bar\alpha| \times M} & - {\bm 1}_{|\bar\alpha| \times |\bar\alpha| }  
\end{array} \right]
\end{equation}
has rank $M$. As matrix $B^\top$ has full row-rank (because columns of $B$ are independent) and matrix $(B^\top B)^{-1}$ has full row-rank (because it is an invertible matrix), from Eq.~\eqref{main_ss_sens} we can conclude that the Jacobian matrix $\nabla_{\bm k, \bm \ell} \bar{\bm x} (\bm k, \bm \ell)$, capturing the sensitivities of the steady-state $\bar{\bm x} (\bm k, \bm \ell)$ with respect to parameters $(\bm k, \bm \ell)$ has rank $M$ which is equal to the number of species. Hence, applying the constant rank theorem (see Theorem 11.1 in \cite{tu2008rank}) we can conclude that by perturbing the parameters in $(\bm k, \bm \ell)$ we can independently perturb the components of the steady-state vector $\bar{\bm x} (\bm k, \bm \ell)$. This will be useful later in the proof. 

Let $\Phi(\bm x,k_{\bar 1}, k_{\bar 2})$ be the ratio of the rate functions for the last two reactions, i.e.\
\begin{align}
 \Phi(\bm x , k_{\bar 1}, k_{\bar 2}) = \frac{r_{\bar 2}(\bm x; k_{\bar 2})}{r_{\bar 1}(\bm x; k_{\bar 1})}. 
 \end{align}
Then by differentiating with respect to both $k_{\bar 1}$ and $k_{\bar 2}$ we obtain
\begin{align}
\label{partial_derivative_last_two_rates}
\partial_{k_{\bar 1}} r_{\bar 1} (\bm x; k_{\bar 1} ) = -r_{\bar 1} (\bm x ; k_{\bar 1} ) \frac{\partial_{k_{\bar 1}}\Phi(\bm x,k_{\bar 1}, k_{\bar 2})}{\Phi(\bm x,k_{\bar 1}, k_{\bar 2})} \qquad \textnormal{and} \qquad \partial_{k_{\bar 2}} r_{\bar 2} (\bm x ; k_{\bar 2} ) = r_{\bar 1} (\bm x ; k_{\bar 1} ) \partial_{k_{\bar 2}}\Phi(\bm x,k_{\bar 1}, k_{\bar 2}).
\end{align}
Setting $\bm x$ to be the steady-state vector $\bm x = \bar{\bm x}(\bm k, \bm \ell)$, and substituting in Eq.~\eqref{max_rpa_reln_for_c} we obtain 
\begin{align}
\label{maxrpa_main_concl}
\Phi( \bar{\bm x}(\bm k, \bm \ell),k_{\bar 1}, k_{\bar 2}) \frac{\p_{k_{\bar 1}}\phi_{\rm out} \left( k_{\bar 1}, k_{\bar 2} \right)}{
\p_{k_{\bar 1}}\Phi(\bar{\bm x}(\bm k, \bm \ell),k_{\bar 1}, k_{\bar 2})
 } 
 =  
 \kappa \frac{\p_{k_{\bar 2}}\phi_{\rm out} \left( k_{\bar 1}, k_{\bar 2} \right)}{\p_{k_{\bar 2}}\Phi( \bar{\bm x} (\bm k, \bm \ell),k_{\bar 1}, k_{\bar 2})},
\end{align}
where 
$\kappa \coloneqq \frac{ c_{N} }{ c_{N-1}}$. Let us now define the vector $\bm q$ as
\begin{align}
\bm q = - {\tilde{\bm q}_1} \frac{r_{\bar 1} (\bar{\bm x}(\bm k, \bm \ell) ; k_{\bar 1} ) \p_{k_{\bar 2}}\Phi( \bar{\bm x} (\bm k, \bm \ell),k_{\bar 1}, k_{\bar 2}) }{\p_{k_{\bar 2}}\phi_{\rm out} \left( k_{\bar 1}, k_{\bar 2} \right)},
\end{align}
then Eqs.~\eqref{vec_q_first_rel}, \eqref{partial_derivative_last_two_rates} and \eqref{maxrpa_main_concl} together imply that pair $(\bm q, \kappa)$ satisfies Eq.~\eqref{eq:qs}. This proves the first condition of the theorem.

To prove the second condition, observe that if we pick a pair $(\bm q, \kappa)$ satisfying Eq.~\eqref{eq:qs}, then for the linear integrator $z \coloneqq \bm q \cdot \bm x$, the dynamics is given by
\begin{equation}
\dot z = \kappa \, r_{\bar 1}(\bm x;k_{\bar 1} ) - r_{\bar 2 }(\bm x;k_{\bar 2}) = \ r_{\bar 1 }(\bm x;k_{\bar 1}) \ \left( \kappa  -  \Phi(\bm x, k_{\bar 1}, k_{\bar 2})\right)  
\end{equation}
which proves that at the steady-state
\begin{align}
\label{max_rpa_gen_set_point_2_2}
 \Phi(\bar{\bm x}(\bm k, \bm \ell) , k_{\bar 1}, k_{\bar 2}) = \kappa.
 \end{align}
The only thing left to show is that the function $ \Phi(\bm x , k_{\bar 1}, k_{\bar 2})$ does not depend on the concentrations $x_1,\dots,x_{M-1}$ of the first $(M-1)$ species, i.e.\ this function can only depend on $x_M$. We shall prove this by showing that the gradient of this function with respect to the state vector $\bm x$ can only have one non-zero component, which is at the $M$-th location (corresponding to $\partial_{x_M} \Phi(\bm x , k_{\bar 1}, k_{\bar 2})$). It suffices to prove this only at the steady-state $\bm x = \bar{\bm x} (\bm k, \bm \ell )$ because as we argued before, the components of $\bar{\bm x} (\bm k, \bm \ell)$ can be independently perturbed by perturbing the parameters in $(\bm k, \bm \ell)$.

 Recall that $\bm u_M$ is the $M$-dimensional vector whose all components are zero, except the component at the $M$-th index (i.e. location of the output species $X$) which is $1$. In order to prove our claim we just need to show that $\bm u_M$ \emph{cannot be independent} of the gradient vector
 \begin{align}
\phi(\bm k, \bm \ell) \coloneqq \nabla_{\bm x}\Phi(\bar{\bm x}(\bm k, \bm \ell) , k_{\bar 1}, k_{\bar 2}).
 \end{align}
We set $\bm q_\Phi$ to be the $(M + |\bar\alpha| )$-dimensional vector
\begin{equation}
{\bm q_\Phi}^\top = -  \phi(\bm k, \bm \ell)^\top (B^\top \ B)^{-1} B^\top.
\end{equation}
and decompose it as ${\bm q_\Phi}^\top = [\bm q_{\Phi,1}^\top \  \bm q_{\Phi,2}^\top]$, where $\bm q_{\Phi,1}^\top$ and $\bm q_{\Phi,2}^\top$ are vectors with dimensions $1 \times M$ and $1 \times |\bar\alpha| $ respectively. Similar to Eq.~\eqref{maxrpa_output_sens} we can derive
\begin{align}
\phi(\bm k, \bm \ell)^\top \nabla_{\bm k, \bm \ell} \bar{\bm x} (\bm k, \bm \ell)  
= 
\begin{bmatrix}
\bm q_{\Phi,1}^\top S \nabla_{\bm k} {\bm r}, &  - \bm q_{\Phi,2}^\top
\end{bmatrix}
.
\end{align}
Eq.~\eqref{max_rpa_gen_set_point_2_2} implies that $\bm q_{\Phi,2}^\top = - \nabla_{\bm \ell}\Phi(\bm x(\bm k, \bm \ell) , k_{\bar 1}, k_{\bar 2}) = \bm 0$ and  
\begin{align}
\label{max_rpa_gen_set_point_3}
\bm q_{\Phi,1}^\top S \nabla_{\bm k} {\bm r}= 
\begin{bmatrix}
\bm 0, 
& 
- \p_{k_{\bar 1}}\Phi(\bar{\bm x}(\bm k, \bm \ell),k_{\bar 1}, k_{\bar 2}), k_{\bar 1}, k_{\bar 2}), 
&
- \partial_{k_{\bar 2}}\Phi(\bar{\bm x}(\bm k, \bm \ell),k_{\bar 1}, k_{\bar 2})
\end{bmatrix}
.
 \end{align}
Due to Eq.~\eqref{maxrpa_main_concl}, the vector on the right is collinear with the vector on the right of Eq.~\eqref{maxrpa_output_sens2}. Hence there exists constants $\zeta_1$ and $\zeta_2$ such that, the vector $\bm z_\Phi = \zeta_1 \  \tilde{\bm q}_1 + \zeta_2 \  \bm q_{\Phi,1}$
satisfies
\begin{align}
\label{max_rpa_gen_set_point_4}
\bm z_\Phi^\top  S \nabla_{\bm k} {\bm r} = \bm 0, 
 \end{align}
which shows that the vector vanishes,
\begin{equation}
\bm z_\Phi^\top = \zeta_1 \  \tilde{\bm q}^\top + \zeta_2 \  \bm q_{\Phi}^\top = -(\zeta_1 \ \bm u_M^\top + \zeta_2 \ \phi(\bm k, \bm \ell)^\top)(B^\top \ B)^{-1} B^\top = \bm 0. 
\end{equation}
However, since the rows of $B^T$ are independent we must have that the vector $\zeta_1 \ \bm u_M^\top + \zeta_2 \  \phi(\bm k, \bm \ell)^\top = \bm 0$, which happen if and only if $\bm u_M$ and $\phi(\bm k, \bm \ell)$ are \emph{dependent} vectors. 
Hence only the last component of $\phi(\bm k, \bm \ell)$ can be nonzero and this proves that $\Phi(\bm x, k_{\bar 1}, k_{\bar 2})$ can only be a function of the output species concentration $x_M$. This completes the proof of this theorem.
\end{proof}

Let us illustrate this result with an example. Consider the minimal reaction network that displays the phenomenon of \emph{Absolute Concentration Robustness (ACR)}, which refers to robustness with respect to a conserved quantity \cite{cappelletti2020hidden}. This network has just two species and reactions:
\begin{align}
    e_1 &: v_1 \to v_2, \\
    e_2 &: v_1 + v_2 \to 2 v_1. 
\end{align}
The rate equations in the mass action kinetics are given by 
\begin{align}
\dot x_1 &=   k_2 x_1 x_2 - k_1 x_1,\\
\dot x_2 &=  -k_2 x_1 x_2 + k_1 x_1.
\end{align}
It is immediate that $\dot x_1 + \dot x_2 = 0$ and so the total concentration $\ell = x_1(t) + x_2(t)$ is constant over time. Fixing the value of $\ell$, the steady-state concentrations are given by
\begin{align}
\bar{x}_1 = \ell - \frac{k_1}{k_2} \qquad \textnormal{and} \qquad \bar{x}_2 = \frac{k_1}{k_2}.
\end{align}
Hence the output species $X = v_2$ is robust to $\ell$, and this network exhibits maxRPA as the steady-state concentration of $X$ is just a function of two mass-action reaction rate constants. The stoichiometric matrix for this network is 
\begin{equation}
S =
\begin{blockarray}{ccc}
&& \\
\begin{block}{c[cc]}
{v_1} \quad\,
& -1 & 1 \\
{v_2} \quad\, 
& 1 & -1\\
\end{block}
& e_1 & e_2 
\end{blockarray} \,\,.     
\end{equation}
Hence $\coker S$ is one-dimensional and spanned by the vector $\bm d = \begin{bmatrix} 1  & 1 \end{bmatrix}^\top$. Since $\coker S$ is non-empty the maxRPA characterisation result in Ref.~\cite{gupta2022universal} does not apply, but we can use the maxRPA result developed here to conclude that this network indeed exhibits maxRPA. Setting $\bm q =\begin{bmatrix} -1 &0 \end{bmatrix}^\top$ and $\kappa =1$, we see that the linear system \eqref{eq:qs} holds and so the first condition of Theorem \ref{thm:maxRPA_characterization} is satisfied. To check the second condition note that under the mass-action kinetics the ratio of the rate functions $r_{1}(x_1,x_2)$ and $r_{2}(x_1,x_2)$ for reactions $e_1$ and $e_2$ satisfy
\begin{equation}
\frac{r_{2}(x_1,x_2)}{r_{1}(x_1,x_2)} = \frac{k_2 x_1 x_2}{k_1 x_1} = \frac{k_2}{k_1}x_2 \eqqcolon \Phi(x_2,k_1,k_2).
\end{equation}
Hence the ratio is only a function of the output species concentration $x_2$ and the set-point is precisely determined by the relation
\begin{equation}
\Phi(x_2,k_1,k_2) = \kappa = 1.    
\end{equation}
Observe that for any real number $\alpha$, $(\bm q +\alpha \bm d, \kappa)$ would also satisfy \eqref{eq:qs} and these are in fact all the solutions to this linear system.

In this example, if we consider the subnetwork $\gamma$ that is formed by excluding the output species and the two reactions from the full network, then $\gamma$ has only one species $v_1$ and no reactions. Hence $\gamma$ is not output-complete which happens because for this exampled maxRPA is not kinetics independent. Indeed, maxRPA does not hold for arbitrary kinetics for the two reactions.

\section{ Derivation of isomorphisms~\eqref{eq:first_iso} and \eqref{eq:second_iso} }\label{sec:iso}

Let us give a proof of the isomorphisms~\eqref{eq:first_iso} and \eqref{eq:second_iso}, which are true when $\wt{c}(\gamma) = 0$ 
and $\wt{d}(\gamma)=0$ hold. 
A proof of the isomorphisms is given in Ref.~\cite{PhysRevResearch.3.043123}. 
We describe it here to make the paper self-contained
together with the explicit construction of the isomorphisms.

We first define the following vector spaces, 
\begin{align}
  C_0 (\Gamma) &\coloneqq 
  \big\{
  \sum_i d_i v_i
  \, |  \, 
  v_i \in V, \, d_i \in \mathbb R 
  \big\} , 
  \\
  C_1 (\Gamma) &\coloneqq 
  \big\{
\sum_A c_A e_A 
  \, |  \,
  e_A \in E, \, c_A \in \mathbb R 
  \big\} .  
\end{align} 
We also define similar spaces for a subnetwork $\gamma = (V_\gamma, E_\gamma) \subset \Gamma$, $C_n (\gamma)$ for $n=0,1$, 
to be those generated by $V_\gamma$ and $E_\gamma$. 
Similarly, we define the spaces for $\Gamma' = \Gamma / \gamma$,
to be the spaces spanned by $(V \setminus V_\gamma, E \setminus E_\gamma)$. 
We will denote the elements of 
$C_1(\gamma)$, 
$C_1(\Gamma)$ 
and 
$C_1(\Gamma')$ 
using a vector,
where each component represents the corresponding coefficient. 
For example, 
\begin{equation}
\bm c_1 \in C_1(\gamma), 
\quad 
\begin{bmatrix}
    \bm c_1 \\
    \bm c_2
\end{bmatrix}
\in C_1(\Gamma), 
\quad 
\bm c_2 \in C_1(\Gamma').
\end{equation}
In the above, the element of $C_1(\Gamma)$ are partitioned 
into those inside/outside $\gamma$. 
Similar notation will be used for $C_0(\gamma)$ and so on. 
To prove the isomorphisms, we consider the following short exact sequence,  
\begin{equation}
  \xymatrix{
    &
    0 
    \ar[d]
    &
    0 \ar[d]
    & 
    0 \ar[d]
    \\
   0 \ar[r] &
   C_1 (\gamma)
   \ar[r]^{\psi_1 }
   \ar[d]^{\p_\gamma } 
   &
   C_1 (\Gamma) 
   \ar[r]^{\varphi_1}
   \ar[d]^{\p }
   & 
   C_1(\Gamma' )
   \ar[r]
   \ar[d]^{\p' }
   & 0
   \\
0
 \ar[r] 
&  
C_0(\gamma)
   \ar[r]^{\psi_0 }
   \ar[d]   
   & 
C_0 (\Gamma) 
  \ar[r]^{\varphi_0}
  \ar[d] 
   &
C_0 (\Gamma' )
   \ar[r]
   \ar[d]
   &
   0
   \\
   & 
   0 
   & 0 
   & 0 
 }
 \label{eq:diag-b}
\end{equation}
where the columns are the chain complexes of $\gamma$, 
$\Gamma$, and $\Gamma'$, respectively. 
The maps $\p_\gamma$, $\p$, and $\p'$ are defined by 
\begin{equation}
  \p_\gamma: \bm c_1 \mapsto S_{11} \bm c_1, 
  \quad 
  \p: \bm c = 
  \begin{bmatrix}
    \bm c_1 \\
    \bm c_2
  \end{bmatrix}
  \mapsto 
  S \bm c, 
  \quad 
  \p': \bm c_2 \mapsto S' \bm c_2. 
\end{equation}
The horizontal maps are given by 
\begin{equation}
  \psi_1: 
  \bm c_1 \mapsto 
  \begin{bmatrix}
    \bm c_1 \\ 
    \bm 0
  \end{bmatrix}, 
  \quad 
  \varphi_1: 
  \begin{bmatrix}
    \bm c_1 \\ 
    \bm c_2
  \end{bmatrix}
  \mapsto \bm c_2 , 
\end{equation}
\begin{equation}
  \psi_0: 
  \bm d_1 \mapsto 
  \begin{bmatrix}
    \bm d_1 \\ 
    S_{21} S^+_{11} \bm d_1
  \end{bmatrix}, 
  \quad 
  \varphi_0: 
  \begin{bmatrix}
    \bm d_1 \\ 
    \bm d_2
  \end{bmatrix}
  \mapsto \bm d_2 - S_{21}S^+_{11} \bm d_1 . 
\end{equation}
One can check that the diagram \eqref{eq:diag-b} commutes 
if and only if the following condition is satisfied: 
\begin{equation}
  S_{21} (1 - S_{11}^+ S_{11}) \bm c_1 = \bm 0, 
  \label{eq:com-1}
\end{equation}
where $\bm c_1 \in C_1(\gamma)$. 
The matrix $ (1 - S_{11}^+ S_{11})$ 
is the projection matrix to $\ker S_{11}$, 
and Eq.~\eqref{eq:com-1} is equivalent to 
\begin{equation}
  \ker S_{11} \subset \ker S_{21}  . 
  \label{eq:com-2}
\end{equation}
The condition \eqref{eq:com-2} is equivalent to 
 $\widetilde c (\gamma) = |\ker S_{11} / (\ker S)_{{\rm supp}\, \gamma}| = 0$,
namely the absence of emergent cycles, 
as shown around Eq.~(182) of Ref.~\cite{PhysRevResearch.3.043123}. 
Thus, the diagram \eqref{eq:diag-b} commutes if and only if $\gamma$ has no emergent cycle.

Applying the snake lemma to Eq.~\eqref{eq:diag-b}, we obtain 
an exact sequence, \\
\begin{equation}
  \xymatrix{
   0 \ar[r] &
   \ker S_{11}
   \ar[r]^{\psi_1} 
   &
   \ker S
   \ar[r]^{\varphi_1}
   &
   \ker S'
   \ar[r]^{\delta_1\quad } &  
   \coker S_{11}
   \ar[r]^{\bar{\psi}_0}
   & 
   \coker S
   \ar[r]^{\bar{\varphi}_0} 
   &
   \coker S'
   \ar[r]
   &
   0 , 
 }  
 \label{eq:les}
\end{equation}
\\ 
where $\bar \psi_0$ and $\bar \varphi_0$ 
are induced maps of $\psi_0$ and $\varphi_0$. 
The map $\delta_1: \ker S' \to \coker S_{11}$ 
is called the connecting map. 
For a given $\bm c_2 \in \ker S'$, 
the connecting map is given by\footnote{
The  map is identified as follows. 
Pick an element $\bm c_2 \in \ker S'$, which 
can be included in $C_1(\Gamma')$. 
Since $\varphi_1$ is surjective, there exists 
$
\bm c = 
\begin{bmatrix}
  \bm c_1 \\
  \bm c_2 
\end{bmatrix}$ 
such that $\varphi_1 (\bm c) = \bm c_2$. 
From the commutativity of the diagram (\ref{eq:diag-b}), 
we have $\varphi_0 (S\bm c) = S' \bm c_2 = \bm 0$. 
Since the rows of Eq.~(\ref{eq:diag-b}) are exact, 
there exists $\bm d_1 \in C_0(\gamma)$ 
such that $\psi_0 (\bm d_1) = S \bm c$. 
We obtain $[\bm d_1]=[S_{11} \bm c_1 + S_{12} \bm c_2]= [S_{12} \bm c_2]\in \coker S_{11}$ by identifying the differences in ${\rm im\,} S_{11}$. 
The mapping $\bm c_2 \mapsto [S_{12} \bm c_2]$ is the connecting map. 
} 
\begin{equation}
\delta_1: 
\bm c_2 \mapsto 
[S_{12} \bm c_2] \in \coker S_{11}, 
\end{equation}
where $[...]$ means to identify the differences in ${\rm im\,} S_{11}$. 

When $\widetilde d(\gamma)=0$ is satisfied, 
the connecting map $\delta_1$ is a zero map, as we show below. 
Since $\widetilde d(\gamma)=0$, for each ${\bm d}_1 \in \coker S_{11}$ there exists a ${\bm d}_2$ 
such that $
\begin{bmatrix}
\bm d_1 \\ \bm d_2    
\end{bmatrix} \in \coker S$. This implies that
\begin{align}
{\bm d}^\top_2 S_{21} &= {\bm 0} ,
\\
{\bm d}^\top_1 S_{12} + {\bm d}^\top_2 S_{22} &= {\bm 0}.
\label{eq:ds12+ds22-0}
\end{align}
Multiplying the second equation on the right by ${\bm c}_2 \in \ker S'$ and using Eq.~\eqref{eq:ds12+ds22-0} and $S'\bm c_2=\bm 0$, 
we obtain
\begin{equation}
{\bm d}^\top_1 S_{12} {\bm c}_2 
= - {\bm d}^\top_2 S_{22} {\bm c}_2 
= - {\bm d}^\top_2  S_{21} S^{+}_{11}S_{12} {\bm c}_2 = {\bm 0},
\end{equation}
where the last equality is because 
${\bm d}^\top_2 S_{21} = {\bm 0}$. 
This shows that ${\bm d}_1$ is orthogonal to $S_{12} {\bm c}_2$ 
and since ${\bm d}_1$ is an arbitrary element of $\coker S_{11}$ we can conclude that $S_{12} {\bm c}_2$ is in the image of $S_{11}$.
Therefore, $\delta_1$ is a zero map.

The exact sequence \eqref{eq:les} and the fact that $\delta_1 = \bm 0$ implies that we have the following exact sequences,
\begin{eqnarray}
&&  \xymatrix{
   0 \ar[r] &
   \ker S_{11}
   \ar[r]
   &
   \ker S
   \ar[r] 
   &
   \ker S'
  \ar[r] 
  & 0
 }  , \\  
 &&
 \xymatrix{
  0 \ar[r] &
  \coker S_{11}
  \ar[r]
  &
  \coker S
  \ar[r] 
  &
  \coker S'
 \ar[r] 
 & 0
}  . 
\end{eqnarray}
These imply the isomorphisms \eqref{eq:first_iso} and \eqref{eq:second_iso}.

Let us explicitly construct the isomorphism \eqref{eq:first_iso}. 
We first define a map $F:\ker S \to \ker S'$ 
by 
\begin{equation}
F: \ker S \ni 
\bm c = 
\begin{bmatrix}
    \bm c_1 \\
    \bm c_2
\end{bmatrix}
\mapsto 
\bm c_2 \in \ker S'. 
\end{equation}
We can see that $\bm c_2$ is indeed an element of $\ker S'$, 
\begin{equation}
S' \bm c_2 = - S_{21} (1 - S^{+}_{11} S_{11}) \bm c_1 = \bm 0, 
\end{equation}
where we have used $S \bm c=\bm 0$ and 
$\ker S_{11} \subset \ker S_{21}$ (since $\wt{c}(\gamma)=0$).

The map $F$ is a surjection as shown in the following. 
Pick any $\bm c_2 \in \ker S'$ and define $\bm c_1 \coloneqq - S^{+}_{11}S_{12} \bm c_2$. 
Then
\begin{align}
\label{iso1_reln1}
S_{11} \bm c_1 + S_{12} \bm c_2 &= (1 - S_{11} S^{+}_{11}) S_{12} \bm c_2 , \\
\label{iso1_reln2}
S_{21} \bm c_1 + S_{22} \bm c_2 & = (- S_{21} S^{+}_{11}S_{12} + S_{22} ) \bm c_2 
= -S' \bm c_2 = {\bm 0}.
\end{align}
The matrix $(1 - S_{11} S^{+}_{11})$ is the projection matrix to $\coker S_{11}$. 
As we showed earlier, $S_{12} \bm c_2$ belongs to the image of $S_{11}$, 
and the RHS of Eq.~\eqref{iso1_reln1} is zero, 
proving that 
$\begin{bmatrix}
\bm c_1 \\ \bm c_2
\end{bmatrix}
\in \ker S$ and so the map $F$ is surjective. 
An element of the kernel of the map $F$ is of the form 
\begin{equation}
\begin{bmatrix}
    \bm c_1 \\ 
    \bm 0
\end{bmatrix} 
\in 
(\ker S)_{{\rm supp}\,\gamma} 
\subset \ker S .
\end{equation}
Thus, the kernel of $F$ is identified as $\ker F = (\ker S)_{{\rm supp}\,\gamma}$. 
Therefore, the induced map 
\begin{equation}
\bar F : \ker S / (\ker S)_{{\rm supp}\, \gamma} \to \ker S, 
\end{equation}
is an isomorphism. 
Noting that $(\ker S)_{{\rm supp}\, \gamma} = \ker S_{11}$ 
since $\wt{c}(\gamma)=0$, we have obtained 
the isomorphism \eqref{eq:first_iso}.

We now prove the isomorphism \eqref{eq:second_iso}. Recall the definitions of subspaces $D_{11}(\gamma)$, $X (\gamma)$ and $\bar{D}'(\gamma)$ from Eqs.~\eqref{defn:d11gamma}, \eqref{defn:xgamma} and \eqref{defn:dprime_gamma}. Since there is no emergent conserved quantities $\wt{d}(\gamma)=0$, we must have $\coker S_{11} = D_{11}(\gamma)$. 
Consider the projection of $\coker S $ to $X(\gamma)$. Clearly the kernel of this surjective map is $\coker S \cap X (\gamma)^{\perp}$, where $X (\gamma)^{\perp}$ denotes the orthogonal complement of $X (\gamma)$. Hence $\coker S$ is isomorphic to the direct sum
\begin{align}
\label{max_rpa_second_iso_pf0}
\coker S \simeq X(\gamma) \oplus \coker S/ X (\gamma).
\end{align}
Now if we consider the projection map from $X(\gamma)$ to $D_{11}(\gamma)$ by 
$$\begin{bmatrix}
\bm d_1 \\ \bm d_2   \end{bmatrix} \mapsto \bm d_1,$$
then the image of this map is exactly $D_{11}(\gamma)$ and the kernel of this map is exactly $\bar{D}'(\gamma)$. Hence $X(\gamma)$ is isomorphic to the direct sum
$$X(\gamma) \simeq D_{11}(\gamma) \oplus  \bar{D}'(\gamma),$$
which together with Eq.~\eqref{max_rpa_second_iso_pf0} shows that
\begin{align*}
\coker S \simeq D_{11}(\gamma) \oplus  \bar{D}'(\gamma) \oplus \coker S / X (\gamma).
\end{align*}
Since $\coker S_{11} = D_{11}(\gamma)$ and the projection map $\bar \varphi_0$ (given by Eq.~\eqref{defn_varphi_0}) is an isomorphism between $\bar{D}'(\gamma)$ and $\coker S' \cap \coker S_{21}$, in order to prove isomorphism \eqref{eq:second_iso} we just need to establish the following isomorphism
\begin{align}
\label{max_rpa_second_iso_pf1}
\coker S / X (\gamma) \simeq \coker S' \cap (\coker S_{21})^{\perp}.
\end{align}
We claim that the projection map $\bar \varphi_0$ also establishes this isomorphism. To see this consider any 
$\begin{bmatrix}
\bm d_1 \\ \bm d_2    
\end{bmatrix} \in \coker S/X (\gamma).$
Without losing generality we can assume that $\bm d_1 \in (\coker S_{11})^{\perp}$ and $\begin{bmatrix}
\bm d_1 \\ \bm d_2    
\end{bmatrix}  \in X(\gamma)^{\perp}$. As 
$\bm d_1^\top S_{11} + \bm d_2^\top S_{21} = 0$
and $\bm d_1 \in (\coker S_{11})^{\perp}$, we can express $\bm d_1^\top$ uniquely as
\begin{align}
\label{max_rpa_iso_2_d_1_defn}
\bm d_1^\top  =-  \bm d_2^\top S_{21}S_{11}^{+}
\end{align}
which upon substitution in relation $\bm d_1^\top S_{12} + \bm d_2^\top S_{22} = 0$ yields
\begin{align}
\label{max_rpa_iso_2_d_1_1}
\bm d_1^\top S_{12} + \bm d_2^\top S_{22} =  \bm d_2^\top (S_{22} - S_{21} S_{11}^+ S_{12}) = \bm d_2^\top S' = 0.
\end{align}
Hence $\bm d_2 \in \coker S'$. Let $\bm d'_2$ be the projection of this vector on $\coker S' \cap \coker S_{21}$. Then certainly $\begin{bmatrix}
\bm 0 \\ \bm d'_2    
\end{bmatrix} \in \coker S \cap X(\gamma)$ which implies that the dot product of this vector with $\begin{bmatrix}
\bm d_1 \\ \bm d_2    
\end{bmatrix}  \in X(\gamma)^{\perp}$
is zero. Therefore $\bm d'_2 = \bm 0$ and so $\bm d_2 \in \coker S' \cap (\coker S_{21})^\perp$. This shows that the image of $\coker S / X (\gamma)$ under the map $\bar \varphi_0$ is in $\coker S' \cap (\coker S_{21})^{\perp}$. Moreover this map is injective on $\coker S / X (\gamma)$ as its kernel is trivial, because if 
$\begin{bmatrix}
\bm d_1 \\ \bm 0   
\end{bmatrix} \in \coker S/X (\gamma)$
then $\bm d_1 \in \coker S_{11}$ and so $\bm d_1 \in (\coker S_{11})^{\perp}$ can only happen when $\bm d_1 = \bm 0$. Next we establish the surjectivity of the map $$\bar \varphi_0 : \coker S / X (\gamma) \to \coker S' \cap (\coker S_{21})^{\perp}.$$
Pick any $\bm d_2 \in \coker S' \cap (\coker S_{21})^{\perp}$. Then setting $\bm d_1^\top$ as in Eq.~\eqref{max_rpa_iso_2_d_1_defn} yields Eq.~\eqref{max_rpa_iso_2_d_1_1}. Moreover 
\begin{align}
\label{max_rpa_second_iso_pf2}
\bm d_1^\top S_{11} + \bm d_2^\top S_{21} =\bm d_2^\top S_{21} (1 - S_{11}^+ S_{11}) = 0
\end{align}
because matrix $(1 - S_{11}^+ S_{11})$ is the projection matrix to $\ker S_{11}$ and $\ker S_{11} \subset \ker S_{21}$ as $\wt{c}(\gamma) =0$. Eqs.~ \eqref{max_rpa_iso_2_d_1_1} and \eqref{max_rpa_second_iso_pf2} show that 
$\begin{bmatrix}
\bm d_1 \\ \bm d_2    
\end{bmatrix} \in \coker S.$ Noting that $\ker S_{11}^+ = \ker S_{11}^\top = \coker S_{11}$, we from Eq.~\eqref{max_rpa_iso_2_d_1_defn} that $\bm d_1 \in (\coker S_{11})^{\perp}$. This also shows that
$\begin{bmatrix}
\bm d_1 \\ \bm d_2    
\end{bmatrix}  \in X(\gamma)^{\perp}$
because for any $\begin{bmatrix}
\wt{\bm d}_1 \\  \wt{\bm d}_2    
\end{bmatrix}  \in X(\gamma)$
we have $\wt{\bm d}_1 \in \coker S_{11}$ (hence orthogonal to $\bm d_1$) and since $\wt{\bm d}_1^\top S_{11} + \wt{\bm d}^\top_2 S_{21} =  \wt{\bm d}^\top_2 S_{21} = 0$, we must have $\wt{\bm d}_2 \in \coker S_{21}$ while $\bm d_2 \in (\coker S_{21})^{\perp}$. This proves the surjectivity of the map $\bar \varphi_0$ from $\coker S / X (\gamma)$ to $\coker S' \cap (\coker S_{21})^{\perp}$, which establishes the isomorphism \eqref{max_rpa_second_iso_pf1}, and therefore proves the isomorphism \eqref{eq:second_iso}.

\section{ Derivation of Eq.~\eqref{eq:rate-eq-indep} }\label{app:derivation-number}

In this Appendix, we show that 
the number of independent equations in Eq.~\eqref{eq:sp-r2-s21-c11} 
is given by 
\begin{equation}
|V \setminus V_\gamma| - \bar d' (\gamma) - d_l (\gamma). 
\end{equation}

We here count the dependent relations in Eq.~\eqref{eq:rate-eq-indep}. 
This amounts to finding independent vectors
$\bm d_2 \in \mathbb R^{|V \setminus V_\gamma|}$ such that
\begin{equation}
\bm d_2^\top S' \bm r_2 = \bm 0, 
\quad 
\bm d_2^\top S_{21} \bm c_{11} = 0,
\end{equation}
for any $\bm c_{11} \in \ker S_{11}$. 
Since $\bm r_2$ and $\bm c_{11} \in \ker S_{11}$ is arbitrary, 
this condition is rephrased as
finding $\bm d_2 \in \mathbb R^{|V \setminus V_\gamma|}$ such that
\begin{align}
\bm d_2^\top S' &= \bm 0 ,
\label{eq:d2sp-0}
\\ 
\bm d_2^\top S_{21} (1 - S^+_{11} S_{11}) &= \bm 0 .
\label{eq:d2s21-s11-0}
\end{align}

We show the following equivalence,
\begin{equation}
\text{$\bm d_2$ satisfies Eqs.~\eqref{eq:d2sp-0} and \eqref{eq:d2s21-s11-0}}  
\quad  \Longleftrightarrow \quad
\bm d_2 \in (\coker S' \cap \coker S_{21})  \oplus  \bar{\varphi}_0 (D_l(\gamma)). 
\label{eq:equiv-d2-d2}
\end{equation}
Let us first show the direction $\Longrightarrow$. 
If $\bm d_2 \in \coker S_{21}$, then
$\bm d_2 \in \coker S' \cap \coker S_{21}$ and the relation holds. 
Suppose $\bm d_2 \notin \coker S_{21}$.
Then, if we define 
$\bm d_1^\top \coloneqq - \bm d_2^\top S_{21} S^+_{11}$, 
Eqs.~\eqref{eq:d2sp-0} and \eqref{eq:d2s21-s11-0} implies that 
\begin{equation}
\begin{bmatrix}
\bm d_1^\top & \bm d_2^\top     
\end{bmatrix}    
\begin{bmatrix}
S_{11} & S_{12} \\ 
S_{21} & S_{22}
\end{bmatrix} 
= \bm 0. 
\end{equation}
Namely, $\begin{bmatrix} \bm d_1 \\ \bm d_2 \end{bmatrix} \in \coker S$.
If we multiply the projection matrix to $\coker S_{11}$ 
on $\bm d_1$, 
\begin{equation}
\bm d_1^\top (1 - S_{11} S^+_{11}) 
= 
-\bm d_{2}^\top S_{21} S_{11}^+ + \bm d_2^\top S_{21} S_{11}^+ S_{11} S^+_{11} 
= \bm 0 ,
\end{equation}
which means that $\bm d_1 \in (\coker S_{11})^\perp$. 
This implies that 
$\begin{bmatrix} \bm d_1 \\ \bm d_2 \end{bmatrix} \in \coker S / X(\gamma) = D_l(\gamma)$ and $\bm d_2 \in \bar{\varphi}_0 (D_l (\gamma))$. 

We now show the opposite direction $\Longleftarrow$.
When $\bm d_2 \in \coker S' \cap \coker S_{21}$, 
$\bm d_2$ trivially satisfies Eqs.~\eqref{eq:d2sp-0} and \eqref{eq:d2s21-s11-0}.
When $\bm d_2 \in \bar{\varphi}_0 (D_l (\gamma))$, 
there exists 
$
\begin{bmatrix}
\bm d_1 \\ \bm d_2    
\end{bmatrix} 
\in \coker S$ such that $\bm d_1 \in (\coker S_{11})^\perp$. 
Then, we have 
\begin{equation}
\begin{split}
\bm d_2^\top S_{21} (1 - S_{11}^+ S_{11}) 
= - \bm d_1^\top S_{11} (1 - S_{11}^+ S_{11})
= \bm 0,
\end{split}
\end{equation}
where we used 
$
\begin{bmatrix}
\bm d_1 \\ \bm d_2    
\end{bmatrix} 
\in \coker S$, and we also have 
\begin{equation}
\bm d_2^\top S' 
=  
\bm d_2^\top (S_{22} - S_{21}S_{11}^+ S_{12})
= 
-
\bm d_1^\top S_{12}
+
\bm d_1^\top S_{11}
S_{11}^+ S_{12})
= 
- \bm d_1^\top (1-S_{11} S_{11}^+) S_{12} 
= \bm 0, 
\end{equation}
because $(1 - S_{11} S_{11}^+)$ is a projection matrix to $\coker S_{11}$. 
This concludes the proof of the equivalence \eqref{eq:equiv-d2-d2}.

Observe also that 
\begin{equation}
\bar \varphi_0 ( \bar D'(\gamma) ) 
= 
(\coker S' ) \cap (\coker S_{21}) . 
\label{bar-d-p-sp-s21}
\end{equation}
Recall that the injective map $\bar{\varphi}_0$ is defined on $\bar D'(\gamma)$ as 
\begin{equation}
\label{defn_varphi_0}
\bar \varphi_0 : 
\bar D' (\gamma) 
\ni 
\begin{bmatrix}
    \bm 0 \\ \bm d_2
\end{bmatrix}
\mapsto \bm d_2 
\in \coker S'. 
\end{equation}
Indeed, suppose that we are given 
$\bm d_2 \in \bar \varphi_0 (\bar D'(\gamma))$. 
Since it satisfies 
\begin{equation}
\begin{bmatrix}
\bm 0 & \bm d_2^\top     
\end{bmatrix} 
\begin{bmatrix}
S_{11}  & S_{12} \\
S_{21}  & S_{22} 
\end{bmatrix}
= 
\begin{bmatrix}
\bm d_2^\top S_{21} & \bm d_2^\top S_{22} 
\end{bmatrix}
= \bm 0 , 
\end{equation}
we have 
$\bm d_2^\top S' = \bm d_2^\top (S_{22} - S_{21}S^+_{11} S_{12}) 
= \bm 0$
and 
$\bm d_2 \in (\coker S') \cap (\coker S_{21})$. 
Conversely, for a given $\bm d_2 \in (\coker S') \cap (\coker S_{21})$, 
we have 
\begin{equation}
\begin{bmatrix}
\bm 0 & \bm d_2^\top     
\end{bmatrix} 
\begin{bmatrix}
S_{11}  & S_{12} \\
S_{21}  & S_{22} 
\end{bmatrix}
= 
\begin{bmatrix}
\bm d_2^\top S_{21} & \bm d_2^\top S_{22} 
\end{bmatrix}
=
\begin{bmatrix}
\bm d_2^\top S_{21} & \bm d_2^\top S_{21} S_{11}^+ S_{12} 
\end{bmatrix}
= \bm 0 , 
\end{equation}
where we have used $\bm d_2^\top S'=\bm 0$, 
and thus $\bm d_2 \in \bar\varphi_0 (\bar D'(\gamma))$. 
This proves Eq.~\eqref{bar-d-p-sp-s21}.

Combining Eq.~\eqref{bar-d-p-sp-s21} and Eq.~\eqref{eq:equiv-d2-d2}, 
the number of independent vectors satisfying 
Eqs.~\eqref{eq:d2sp-0} and \eqref{eq:d2s21-s11-0}
is given by $\bar d'(\gamma)$ and $d_l(\gamma)$. 
Therefore, among $|V \setminus V_\gamma|$ equations of Eq.~\eqref{eq:sp-r2-s21-c11}, the number of independent equations is given by Eq.~\eqref{eq:rate-eq-indep}.

\section{ Finding emergent cycles, emergent conserved quantities, and lost conserved quantities }\label{sec:find-basis}

We here describe how to find emergent cycles, emergent conserved quantities, and lost conserved quantities linear-algebraically.

As for emergent cycles, the space spanned by them can be written as 
\begin{equation}
\begin{split}
{\rm \widetilde C}(\gamma) 
&\coloneqq 
 \ker S_{11}\, / \,
 (\ker S)_{{\rm supp }\gamma} 
 \\ 
 &= 
 \ker S_{11}\, / \,
\left\{ 
\bm c_1 \in \ker S_{11} \, \middle| \, 
\begin{bmatrix}
  \bm c_1 \\
  \bm 0
\end{bmatrix}
\in \ker S
\right\} 
\\ 
&\simeq 
\ker S_{11}\, \cap (\ker S_{21})^\perp ,
\end{split}
\end{equation}
and the last expression can be used to find the basis of ${\rm \wt{C}}(\gamma)$. 
On the other hand, the space of emergent conserved quantities is written as 
\begin{equation}
\widetilde D(\gamma)
 \coloneqq \coker S_{11}  / D_{11} (\gamma) 
\simeq 
 \coker S_{11} \cap \left( D_{11}(\gamma) \right)^\perp .     
 \label{eq:d-tilde-comp}
\end{equation}
The space $D_{11}(\gamma)$ can be written as 
\begin{align}
\begin{split}
D_{11}(\gamma)
&\coloneqq 
  \left\{ 
    \bm d_1 \in \coker S_{11} 
  \, \middle| \,
  {}^{\exists} \bm d_2 {\text{ such that }} 
  \begin{bmatrix}
    \bm d_1 \\
    \bm d_2
  \end{bmatrix}
  \in \coker S 
  \right\} 
\\
&= 
\left\{ 
\bm d_1 
\, \middle| \, 
\begin{bmatrix}
\bm d_1^\top & \bm d_2^\top
\end{bmatrix}
\wt{S} 
= \bm 0
\right\}, 
\end{split}
\label{eq:d11-comp}
\end{align}
where we have defined an extended matrix 
\begin{equation}
\wt{S}
\coloneqq 
\begin{bmatrix}
S_{11} & S_{12} & S_{11} \\ 
S_{21} & S_{22} & \bm 0 
\end{bmatrix}. 
\end{equation}
The last expression of Eq.~\eqref{eq:d11-comp} can be used 
to find the basis of $D_{11} (\gamma)$ and hence that of $\left( D_{11} (\gamma) \right)^\perp$, which can be combined with
Eq.~\eqref{eq:d-tilde-comp} to obtain emergent conserved quantities. 

Similarly, we can also obtain the basis of lost conserved quantities 
by noting that its space is written as 
\begin{equation}
D_l (\gamma) \coloneqq \coker S / X(\gamma) 
\simeq \coker S \cap (X(\gamma))^\perp 
\end{equation}
and 
\begin{equation}
X (\gamma) 
\coloneqq 
  \left\{ 
    \begin{bmatrix}
      \bm d_1 \\
      \bm d_2
    \end{bmatrix}
    \in \coker S 
    \, \middle| \,
    \bm d_1 \in \coker S_{11}
    \right\} 
\simeq 
  \left\{ 
    \begin{bmatrix}
      \bm d_1 \\
      \bm d_2
    \end{bmatrix}
    \, \middle| \,
\begin{bmatrix}
\bm d_1^\top & \bm d_2^\top
\end{bmatrix}
\wt{S} 
= \bm 0
  \right\} .
\end{equation}
We have implemented methods to obtain bases of ${\rm \widetilde C}(\gamma)$, $\widetilde D(\gamma)$, and $D_l(\gamma)$ in {\bf RPAFinder}~\cite{RPAFinder}.

\section{ Details of the example in Sec.~\ref{sec:yeast} } 

\subsection{List of reactions}\label{sec:sc-reactions}

\noindent
1: Glucose  $\rightarrow$  G6P, 

\noindent
2: G6P  $\rightarrow$  F6P, 

\noindent
3: F6P  $\rightarrow$  G6P, 

\noindent
4: F6P  $\rightarrow$  F16P, 

\noindent
5: F16P  $\rightarrow$  DHAP + G3P, 

\noindent
6: DHAP  $\rightarrow$  G3P, 

\noindent
7: G3P  $\rightarrow$  PGP, 

\noindent
8: PGP  $\rightarrow$  3PG, 

\noindent
9: 3PG  $\rightarrow$  PGP, 

\noindent
10: 3PG  $\rightarrow$  2PG, 

\noindent
11: 2PG  $\rightarrow$  3PG, 

\noindent
12: 2PG  $\rightarrow$  PEP, 

\noindent
13: PEP  $\rightarrow$  2PG, 

\noindent
14: PEP  $\rightarrow$  PYR, 

\noindent
15: G6P  $\rightarrow$  PG6, 

\noindent
16: PG6  $\rightarrow$  CO2 + Ru5P, 

\noindent
17: Ru5P  $\rightarrow$  X5P, 

\noindent
18: Ru5P  $\rightarrow$  R5P, 

\noindent
19: R5P + X5P  $\rightarrow$  G3P + S7P, 

\noindent
20: G3P + S7P  $\rightarrow$  R5P + X5P, 

\noindent
21: G3P + S7P  $\rightarrow$  E4P + F6P, 

\noindent
22: E4P + F6P  $\rightarrow$  G3P + S7P, 

\noindent
23: E4P + X5P  $\rightarrow$  F6P + G3P, 

\noindent
24: F6P + G3P  $\rightarrow$  E4P + X5P, 

\noindent
25: PG6  $\rightarrow$  G3P + PYR, 

\noindent
26: PYR  $\rightarrow$  Acetal + CO2, 

\noindent
27: Acetal  $\rightarrow$  Ethanol, 

\noindent
28: Ethanol  $\rightarrow$  Acetal, 

\noindent
29: R5P  $\rightarrow$  $\emptyset$, 

\noindent
30: CO2  $\rightarrow$  $\emptyset$, 

\noindent
31: $\emptyset$  $\rightarrow$  Glucose, 

\noindent
32: Ethanol  $\rightarrow$  $\emptyset$, 

\noindent
33: Acetal  $\rightarrow$  $\emptyset$, 

\noindent
34: PYR  $\rightarrow$  Ala, 

\noindent
35: Ala  $\rightarrow$  PYR, 

\noindent
36: Ala  $\rightarrow$  $\emptyset$.

\subsection{ List of labeled buffering structures }\label{sec:sc-lbs}

\noindent
$\gamma^\ast_{1}=(\{ \rm Glucose\}, \{\} \cup \{1\})$,
 
\noindent
$\gamma_{2,3,4,15,16,17,18,25}=(\{ \rm Acetal, Ala, CO2, DHAP, E4P, Ethanol, F16P, F6P, G3P, G6P, PEP, 2PG, 3PG, PG6, PGP, PYR, 
\\
R5P, Ru5P, S7P, X5P\}, \{2, 3, 4, 5, 6, 7, 8, 9, 10, 11, 12, 13, 14, 15, 16, 17, 18, 19, 20, 21, 22, 23, 24, 25, 26, 27, 28, 29, 30, 32, 33, 34,
\\
35, 36\} \cup \{\})$,
 
\noindent
$\gamma^\ast_{5}=(\{ \rm F16P\}, \{\} \cup \{5\})$,
 
\noindent
$\gamma^\ast_{6}=(\{ \rm DHAP\}, \{\} \cup \{6\})$,
 
\noindent
$\gamma_{7}=(\{ \rm E4P, G3P, S7P, X5P\}, \{19, 20, 21, 22, 23, 24\} \cup \{7\})$,
 
\noindent
$\gamma^\ast_{8}=(\{ \rm PGP\}, \{\} \cup \{8\})$,
 
\noindent
$\gamma_{9}=(\{ \rm PGP\}, \{8, 9\} \cup \{\})$,
 
\noindent
$\gamma_{10}=(\{ \rm 3PG, PGP\}, \{8, 9\} \cup \{10\})$,
 
\noindent
$\gamma_{11}=(\{ \rm 3PG, PGP\}, \{8, 9, 10, 11\} \cup \{\})$,
 
\noindent
$\gamma_{12}=(\{ \rm 2PG, 3PG, PGP\}, \{8, 9, 10, 11\} \cup \{12\})$,
 
\noindent
$\gamma_{13}=(\{ \rm 2PG, 3PG, PGP\}, \{8, 9, 10, 11, 12, 13\} \cup \{\})$,
 
\noindent
$\gamma_{14}=(\{ \rm PEP, 2PG, 3PG, PGP\}, \{8, 9, 10, 11, 12, 13\} \cup \{14\})$,
 
\noindent
$\gamma_{19,20,21,22,23}=(\{ \rm E4P, S7P, X5P\}, \{19, 20, 21, 22\} \cup \{23\})$,
 
\noindent
$\gamma_{24}=(\{ \rm E4P, S7P, X5P\}, \{19, 20, 21, 22, 23, 24\} \cup \{\})$,
 
\noindent
$\gamma_{26,34,35,36}=(\{ \rm Acetal, Ala, CO2, Ethanol, PYR\}, \{26, 27, 28, 30, 32, 33, 34, 35, 36\} \cup \{\})$,
 
\noindent
$\gamma_{27,28,32,33}=(\{ \rm Acetal, Ethanol\}, \{27, 28, 32, 33\} \cup \{\})$,
 
\noindent
$\gamma_{29}=(\{ \rm E4P, R5P, S7P, X5P\}, \{19, 20, 21, 22\} \cup \{23, 29\})$,
 
\noindent
$\gamma^\ast_{30}=(\{ \rm CO2\}, \{\} \cup \{30\})$,
 
\noindent
$\gamma_{31}=(\{ \rm Acetal, Ala, CO2, DHAP, E4P, Ethanol, F16P, F6P, G3P, G6P, Glucose, PEP, 2PG, 3PG, PG6, PGP, PYR,
\\
R5P, Ru5P, S7P, X5P\}, \{1, 2, 3, 4, 5, 6, 7, 8, 9, 10, 11, 12, 13, 14, 15, 16, 17, 18, 19, 20, 21, 22, 23, 24, 25, 26, 27, 28, 29, 30,
\\
31, 32, 33, 34, 35, 36\} \cup \{\})$
\\

Here, $\gamma_{31}$ coincides with the whole network.

\bibliography{refs}

\begin{thebibliography}{88}%
\makeatletter
\providecommand \@ifxundefined [1]{%
 \@ifx{#1\undefined}
}%
\providecommand \@ifnum [1]{%
 \ifnum #1\expandafter \@firstoftwo
 \else \expandafter \@secondoftwo
 \fi
}%
\providecommand \@ifx [1]{%
 \ifx #1\expandafter \@firstoftwo
 \else \expandafter \@secondoftwo
 \fi
}%
\providecommand \natexlab [1]{#1}%
\providecommand \enquote  [1]{``#1''}%
\providecommand \bibnamefont  [1]{#1}%
\providecommand \bibfnamefont [1]{#1}%
\providecommand \citenamefont [1]{#1}%
\providecommand \href@noop [0]{\@secondoftwo}%
\providecommand \href [0]{\begingroup \@sanitize@url \@href}%
\providecommand \@href[1]{\@@startlink{#1}\@@href}%
\providecommand \@@href[1]{\endgroup#1\@@endlink}%
\providecommand \@sanitize@url [0]{\catcode `\\12\catcode `\$12\catcode
  `\&12\catcode `\#12\catcode `\^12\catcode `\_12\catcode `\%12\relax}%
\providecommand \@@startlink[1]{}%
\providecommand \@@endlink[0]{}%
\providecommand \url  [0]{\begingroup\@sanitize@url \@url }%
\providecommand \@url [1]{\endgroup\@href {#1}{\urlprefix }}%
\providecommand \urlprefix  [0]{URL }%
\providecommand \Eprint [0]{\href }%
\providecommand \doibase [0]{https://doi.org/}%
\providecommand \selectlanguage [0]{\@gobble}%
\providecommand \bibinfo  [0]{\@secondoftwo}%
\providecommand \bibfield  [0]{\@secondoftwo}%
\providecommand \translation [1]{[#1]}%
\providecommand \BibitemOpen [0]{}%
\providecommand \bibitemStop [0]{}%
\providecommand \bibitemNoStop [0]{.\EOS\space}%
\providecommand \EOS [0]{\spacefactor3000\relax}%
\providecommand \BibitemShut  [1]{\csname bibitem#1\endcsname}%
\let\auto@bib@innerbib\@empty
\bibitem [{\citenamefont {Stelling}\ \emph {et~al.}(2004)\citenamefont
  {Stelling}, \citenamefont {Sauer}, \citenamefont {Szallasi}, \citenamefont
  {Doyle},\ and\ \citenamefont {Doyle}}]{stelling2004robustness}%
  \BibitemOpen
  \bibfield  {author} {\bibinfo {author} {\bibfnamefont {J.}~\bibnamefont
  {Stelling}}, \bibinfo {author} {\bibfnamefont {U.}~\bibnamefont {Sauer}},
  \bibinfo {author} {\bibfnamefont {Z.}~\bibnamefont {Szallasi}}, \bibinfo
  {author} {\bibfnamefont {F.~J.}\ \bibnamefont {Doyle}},\ and\ \bibinfo
  {author} {\bibfnamefont {J.}~\bibnamefont {Doyle}},\ }\bibfield  {title}
  {\bibinfo {title} {Robustness of cellular functions},\ }\href
  {https://doi.org/https://doi.org/10.1016/j.cell.2004.09.008} {\bibfield
  {journal} {\bibinfo  {journal} {Cell}\ }\textbf {\bibinfo {volume} {118}},\
  \bibinfo {pages} {675} (\bibinfo {year} {2004})}\BibitemShut {NoStop}%
\bibitem [{\citenamefont {Kitano}(2004)}]{kitano2004biological}%
  \BibitemOpen
  \bibfield  {author} {\bibinfo {author} {\bibfnamefont {H.}~\bibnamefont
  {Kitano}},\ }\bibfield  {title} {\bibinfo {title} {Biological robustness},\
  }\href {https://doi.org/https://doi.org/10.1038/nrg1471} {\bibfield
  {journal} {\bibinfo  {journal} {Nature Reviews Genetics}\ }\textbf {\bibinfo
  {volume} {5}},\ \bibinfo {pages} {826} (\bibinfo {year} {2004})}\BibitemShut
  {NoStop}%
\bibitem [{\citenamefont {Kitano}(2007)}]{kitano2007towards}%
  \BibitemOpen
  \bibfield  {author} {\bibinfo {author} {\bibfnamefont {H.}~\bibnamefont
  {Kitano}},\ }\bibfield  {title} {\bibinfo {title} {Towards a theory of
  biological robustness},\ }\href
  {https://doi.org/https://doi.org/10.1038/msb4100179} {\bibfield  {journal}
  {\bibinfo  {journal} {Molecular systems biology}\ }\textbf {\bibinfo {volume}
  {3}},\ \bibinfo {pages} {137} (\bibinfo {year} {2007})}\BibitemShut {NoStop}%
\bibitem [{\citenamefont {Kanehisa}\ and\ \citenamefont
  {Goto}(2000)}]{10.1093/nar/28.1.27}%
  \BibitemOpen
  \bibfield  {author} {\bibinfo {author} {\bibfnamefont {M.}~\bibnamefont
  {Kanehisa}}\ and\ \bibinfo {author} {\bibfnamefont {S.}~\bibnamefont
  {Goto}},\ }\bibfield  {title} {\bibinfo {title} {{KEGG: Kyoto Encyclopedia of
  Genes and Genomes}},\ }\href {https://doi.org/10.1093/nar/28.1.27} {\bibfield
   {journal} {\bibinfo  {journal} {Nucleic Acids Research}\ }\textbf {\bibinfo
  {volume} {28}},\ \bibinfo {pages} {27} (\bibinfo {year} {2000})}\BibitemShut
  {NoStop}%
\bibitem [{\citenamefont {Jeong}\ \emph {et~al.}(2000)\citenamefont {Jeong},
  \citenamefont {Tombor}, \citenamefont {Albert}, \citenamefont {Oltvai},\ and\
  \citenamefont {Barab{\'a}si}}]{jeong2000large}%
  \BibitemOpen
  \bibfield  {author} {\bibinfo {author} {\bibfnamefont {H.}~\bibnamefont
  {Jeong}}, \bibinfo {author} {\bibfnamefont {B.}~\bibnamefont {Tombor}},
  \bibinfo {author} {\bibfnamefont {R.}~\bibnamefont {Albert}}, \bibinfo
  {author} {\bibfnamefont {Z.~N.}\ \bibnamefont {Oltvai}},\ and\ \bibinfo
  {author} {\bibfnamefont {A.-L.}\ \bibnamefont {Barab{\'a}si}},\ }\bibfield
  {title} {\bibinfo {title} {The large-scale organization of metabolic
  networks},\ }\href {https://doi.org/https://doi.org/10.1038/35036627}
  {\bibfield  {journal} {\bibinfo  {journal} {Nature}\ }\textbf {\bibinfo
  {volume} {407}},\ \bibinfo {pages} {651} (\bibinfo {year}
  {2000})}\BibitemShut {NoStop}%
\bibitem [{\citenamefont {Ravasz}\ \emph {et~al.}(2002)\citenamefont {Ravasz},
  \citenamefont {Somera}, \citenamefont {Mongru}, \citenamefont {Oltvai},\ and\
  \citenamefont {Barabási}}]{doi:10.1126/science.1073374}%
  \BibitemOpen
  \bibfield  {author} {\bibinfo {author} {\bibfnamefont {E.}~\bibnamefont
  {Ravasz}}, \bibinfo {author} {\bibfnamefont {A.~L.}\ \bibnamefont {Somera}},
  \bibinfo {author} {\bibfnamefont {D.~A.}\ \bibnamefont {Mongru}}, \bibinfo
  {author} {\bibfnamefont {Z.~N.}\ \bibnamefont {Oltvai}},\ and\ \bibinfo
  {author} {\bibfnamefont {A.-L.}\ \bibnamefont {Barabási}},\ }\bibfield
  {title} {\bibinfo {title} {Hierarchical organization of modularity in
  metabolic networks},\ }\href {https://doi.org/10.1126/science.1073374}
  {\bibfield  {journal} {\bibinfo  {journal} {Science}\ }\textbf {\bibinfo
  {volume} {297}},\ \bibinfo {pages} {1551} (\bibinfo {year}
  {2002})}\BibitemShut {NoStop}%
\bibitem [{\citenamefont {Barkai}\ and\ \citenamefont
  {Leibler}(1997)}]{barkai1997robustness}%
  \BibitemOpen
  \bibfield  {author} {\bibinfo {author} {\bibfnamefont {N.}~\bibnamefont
  {Barkai}}\ and\ \bibinfo {author} {\bibfnamefont {S.}~\bibnamefont
  {Leibler}},\ }\bibfield  {title} {\bibinfo {title} {Robustness in simple
  biochemical networks},\ }\href
  {https://doi.org/https://doi.org/10.1038/43199} {\bibfield  {journal}
  {\bibinfo  {journal} {Nature}\ }\textbf {\bibinfo {volume} {387}},\ \bibinfo
  {pages} {913} (\bibinfo {year} {1997})}\BibitemShut {NoStop}%
\bibitem [{\citenamefont {Levchenko}\ and\ \citenamefont
  {Iglesias}(2002)}]{LEVCHENKO200250}%
  \BibitemOpen
  \bibfield  {author} {\bibinfo {author} {\bibfnamefont {A.}~\bibnamefont
  {Levchenko}}\ and\ \bibinfo {author} {\bibfnamefont {P.~A.}\ \bibnamefont
  {Iglesias}},\ }\bibfield  {title} {\bibinfo {title} {Models of eukaryotic
  gradient sensing: Application to chemotaxis of amoebae and neutrophils},\
  }\href {https://doi.org/https://doi.org/10.1016/S0006-3495(02)75373-3}
  {\bibfield  {journal} {\bibinfo  {journal} {Biophysical Journal}\ }\textbf
  {\bibinfo {volume} {82}},\ \bibinfo {pages} {50} (\bibinfo {year}
  {2002})}\BibitemShut {NoStop}%
\bibitem [{\citenamefont {El-Samad}\ \emph {et~al.}(2002)\citenamefont
  {El-Samad}, \citenamefont {Goff},\ and\ \citenamefont
  {Khammash}}]{el2002calcium}%
  \BibitemOpen
  \bibfield  {author} {\bibinfo {author} {\bibfnamefont {H.}~\bibnamefont
  {El-Samad}}, \bibinfo {author} {\bibfnamefont {J.}~\bibnamefont {Goff}},\
  and\ \bibinfo {author} {\bibfnamefont {M.}~\bibnamefont {Khammash}},\
  }\bibfield  {title} {\bibinfo {title} {Calcium homeostasis and parturient
  hypocalcemia: an integral feedback perspective},\ }\href
  {https://doi.org/https://doi.org/10.1006/jtbi.2001.2422} {\bibfield
  {journal} {\bibinfo  {journal} {Journal of Theoretical Biology}\ }\textbf
  {\bibinfo {volume} {214}},\ \bibinfo {pages} {17} (\bibinfo {year}
  {2002})}\BibitemShut {NoStop}%
\bibitem [{\citenamefont {Tveit}\ \emph {et~al.}(2020)\citenamefont {Tveit},
  \citenamefont {Fjeld}, \citenamefont {Drengstig}, \citenamefont {Filipp},
  \citenamefont {Ruoff},\ and\ \citenamefont
  {Thorsen}}]{Tveit2020.01.02.892729}%
  \BibitemOpen
  \bibfield  {author} {\bibinfo {author} {\bibfnamefont {D.~M.}\ \bibnamefont
  {Tveit}}, \bibinfo {author} {\bibfnamefont {G.}~\bibnamefont {Fjeld}},
  \bibinfo {author} {\bibfnamefont {T.}~\bibnamefont {Drengstig}}, \bibinfo
  {author} {\bibfnamefont {F.~V.}\ \bibnamefont {Filipp}}, \bibinfo {author}
  {\bibfnamefont {P.}~\bibnamefont {Ruoff}},\ and\ \bibinfo {author}
  {\bibfnamefont {K.}~\bibnamefont {Thorsen}},\ }\bibfield  {title} {\bibinfo
  {title} {Exploring mechanisms of glucose uptake regulation and dilution
  resistance in growing cancer cells},\ }\bibfield  {journal} {\bibinfo
  {journal} {bioRxiv}\ }\href {https://doi.org/10.1101/2020.01.02.892729}
  {10.1101/2020.01.02.892729} (\bibinfo {year} {2020})\BibitemShut {NoStop}%
\bibitem [{\citenamefont {Muzzey}\ \emph {et~al.}(2009)\citenamefont {Muzzey},
  \citenamefont {Gómez-Uribe}, \citenamefont {Mettetal},\ and\ \citenamefont
  {{van Oudenaarden}}}]{MUZZEY2009160}%
  \BibitemOpen
  \bibfield  {author} {\bibinfo {author} {\bibfnamefont {D.}~\bibnamefont
  {Muzzey}}, \bibinfo {author} {\bibfnamefont {C.~A.}\ \bibnamefont
  {Gómez-Uribe}}, \bibinfo {author} {\bibfnamefont {J.~T.}\ \bibnamefont
  {Mettetal}},\ and\ \bibinfo {author} {\bibfnamefont {A.}~\bibnamefont {{van
  Oudenaarden}}},\ }\bibfield  {title} {\bibinfo {title} {A systems-level
  analysis of perfect adaptation in yeast osmoregulation},\ }\href
  {https://doi.org/https://doi.org/10.1016/j.cell.2009.04.047} {\bibfield
  {journal} {\bibinfo  {journal} {Cell}\ }\textbf {\bibinfo {volume} {138}},\
  \bibinfo {pages} {160} (\bibinfo {year} {2009})}\BibitemShut {NoStop}%
\bibitem [{\citenamefont {Ferrell}(2016)}]{FERRELL201662}%
  \BibitemOpen
  \bibfield  {author} {\bibinfo {author} {\bibfnamefont {J.~E.}\ \bibnamefont
  {Ferrell}},\ }\bibfield  {title} {\bibinfo {title} {Perfect and near-perfect
  adaptation in cell signaling},\ }\href
  {https://doi.org/https://doi.org/10.1016/j.cels.2016.02.006} {\bibfield
  {journal} {\bibinfo  {journal} {Cell Systems}\ }\textbf {\bibinfo {volume}
  {2}},\ \bibinfo {pages} {62} (\bibinfo {year} {2016})}\BibitemShut {NoStop}%
\bibitem [{\citenamefont {Ben-Zvi}\ and\ \citenamefont
  {Barkai}(2010)}]{ben2010scaling}%
  \BibitemOpen
  \bibfield  {author} {\bibinfo {author} {\bibfnamefont {D.}~\bibnamefont
  {Ben-Zvi}}\ and\ \bibinfo {author} {\bibfnamefont {N.}~\bibnamefont
  {Barkai}},\ }\bibfield  {title} {\bibinfo {title} {Scaling of morphogen
  gradients by an expansion-repression integral feedback control},\ }\href
  {https://doi.org/10.1073/pnas.0912734107} {\bibfield  {journal} {\bibinfo
  {journal} {Proceedings of the National Academy of Sciences}\ }\textbf
  {\bibinfo {volume} {107}},\ \bibinfo {pages} {6924} (\bibinfo {year}
  {2010})}\BibitemShut {NoStop}%
\bibitem [{\citenamefont {Khammash}(2022)}]{9779327}%
  \BibitemOpen
  \bibfield  {author} {\bibinfo {author} {\bibfnamefont {M.~H.}\ \bibnamefont
  {Khammash}},\ }\bibfield  {title} {\bibinfo {title} {Cybergenetics: Theory
  and applications of genetic control systems},\ }\href
  {https://doi.org/10.1109/JPROC.2022.3170599} {\bibfield  {journal} {\bibinfo
  {journal} {Proceedings of the IEEE}\ }\textbf {\bibinfo {volume} {110}},\
  \bibinfo {pages} {631} (\bibinfo {year} {2022})}\BibitemShut {NoStop}%
\bibitem [{\citenamefont {Liu}\ and\ \citenamefont
  {Barab\'asi}(2016)}]{RevModPhys.88.035006}%
  \BibitemOpen
  \bibfield  {author} {\bibinfo {author} {\bibfnamefont {Y.-Y.}\ \bibnamefont
  {Liu}}\ and\ \bibinfo {author} {\bibfnamefont {A.-L.}\ \bibnamefont
  {Barab\'asi}},\ }\bibfield  {title} {\bibinfo {title} {Control principles of
  complex systems},\ }\href {https://doi.org/10.1103/RevModPhys.88.035006}
  {\bibfield  {journal} {\bibinfo  {journal} {Rev. Mod. Phys.}\ }\textbf
  {\bibinfo {volume} {88}},\ \bibinfo {pages} {035006} (\bibinfo {year}
  {2016})}\BibitemShut {NoStop}%
\bibitem [{\citenamefont {Khammash}(2021)}]{KHAMMASH2021509}%
  \BibitemOpen
  \bibfield  {author} {\bibinfo {author} {\bibfnamefont {M.~H.}\ \bibnamefont
  {Khammash}},\ }\bibfield  {title} {\bibinfo {title} {Perfect adaptation in
  biology},\ }\href
  {https://doi.org/https://doi.org/10.1016/j.cels.2021.05.020} {\bibfield
  {journal} {\bibinfo  {journal} {Cell Systems}\ }\textbf {\bibinfo {volume}
  {12}},\ \bibinfo {pages} {509} (\bibinfo {year} {2021})}\BibitemShut
  {NoStop}%
\bibitem [{\citenamefont {Ma}\ \emph {et~al.}(2009)\citenamefont {Ma},
  \citenamefont {Trusina}, \citenamefont {El-Samad}, \citenamefont {Lim},\ and\
  \citenamefont {Tang}}]{ma2009defining}%
  \BibitemOpen
  \bibfield  {author} {\bibinfo {author} {\bibfnamefont {W.}~\bibnamefont
  {Ma}}, \bibinfo {author} {\bibfnamefont {A.}~\bibnamefont {Trusina}},
  \bibinfo {author} {\bibfnamefont {H.}~\bibnamefont {El-Samad}}, \bibinfo
  {author} {\bibfnamefont {W.~A.}\ \bibnamefont {Lim}},\ and\ \bibinfo {author}
  {\bibfnamefont {C.}~\bibnamefont {Tang}},\ }\bibfield  {title} {\bibinfo
  {title} {Defining network topologies that can achieve biochemical
  adaptation},\ }\href
  {https://doi.org/https://doi.org/10.1016/j.cell.2009.06.013} {\bibfield
  {journal} {\bibinfo  {journal} {Cell}\ }\textbf {\bibinfo {volume} {138}},\
  \bibinfo {pages} {760} (\bibinfo {year} {2009})}\BibitemShut {NoStop}%
\bibitem [{\citenamefont {Tang}\ and\ \citenamefont
  {McMillen}(2016)}]{TANG2016274}%
  \BibitemOpen
  \bibfield  {author} {\bibinfo {author} {\bibfnamefont {Z.~F.}\ \bibnamefont
  {Tang}}\ and\ \bibinfo {author} {\bibfnamefont {D.~R.}\ \bibnamefont
  {McMillen}},\ }\bibfield  {title} {\bibinfo {title} {Design principles for
  the analysis and construction of robustly homeostatic biological networks},\
  }\href {https://doi.org/https://doi.org/10.1016/j.jtbi.2016.06.036}
  {\bibfield  {journal} {\bibinfo  {journal} {Journal of Theoretical Biology}\
  }\textbf {\bibinfo {volume} {408}},\ \bibinfo {pages} {274} (\bibinfo {year}
  {2016})}\BibitemShut {NoStop}%
\bibitem [{\citenamefont {Araujo}\ and\ \citenamefont
  {Liotta}(2018)}]{araujo2018topological}%
  \BibitemOpen
  \bibfield  {author} {\bibinfo {author} {\bibfnamefont {R.~P.}\ \bibnamefont
  {Araujo}}\ and\ \bibinfo {author} {\bibfnamefont {L.~A.}\ \bibnamefont
  {Liotta}},\ }\bibfield  {title} {\bibinfo {title} {The topological
  requirements for robust perfect adaptation in networks of any size},\ }\href
  {https://doi.org/https://doi.org/10.1038/s41467-018-04151-6} {\bibfield
  {journal} {\bibinfo  {journal} {Nature communications}\ }\textbf {\bibinfo
  {volume} {9}},\ \bibinfo {pages} {1} (\bibinfo {year} {2018})}\BibitemShut
  {NoStop}%
\bibitem [{\citenamefont {Wang}\ \emph {et~al.}(2021)\citenamefont {Wang},
  \citenamefont {Huang}, \citenamefont {Antoneli},\ and\ \citenamefont
  {Golubitsky}}]{wang2021structure}%
  \BibitemOpen
  \bibfield  {author} {\bibinfo {author} {\bibfnamefont {Y.}~\bibnamefont
  {Wang}}, \bibinfo {author} {\bibfnamefont {Z.}~\bibnamefont {Huang}},
  \bibinfo {author} {\bibfnamefont {F.}~\bibnamefont {Antoneli}},\ and\
  \bibinfo {author} {\bibfnamefont {M.}~\bibnamefont {Golubitsky}},\ }\bibfield
   {title} {\bibinfo {title} {The structure of infinitesimal homeostasis in
  input--output networks},\ }\href
  {https://doi.org/https://doi.org/10.1007/s00285-021-01614-1} {\bibfield
  {journal} {\bibinfo  {journal} {Journal of mathematical biology}\ }\textbf
  {\bibinfo {volume} {82}},\ \bibinfo {pages} {1} (\bibinfo {year}
  {2021})}\BibitemShut {NoStop}%
\bibitem [{\citenamefont {Bin}\ \emph {et~al.}(2022)\citenamefont {Bin},
  \citenamefont {Huang}, \citenamefont {Isidori}, \citenamefont {Marconi},
  \citenamefont {Mischiati},\ and\ \citenamefont {Sontag}}]{bin2022internal}%
  \BibitemOpen
  \bibfield  {author} {\bibinfo {author} {\bibfnamefont {M.}~\bibnamefont
  {Bin}}, \bibinfo {author} {\bibfnamefont {J.}~\bibnamefont {Huang}}, \bibinfo
  {author} {\bibfnamefont {A.}~\bibnamefont {Isidori}}, \bibinfo {author}
  {\bibfnamefont {L.}~\bibnamefont {Marconi}}, \bibinfo {author} {\bibfnamefont
  {M.}~\bibnamefont {Mischiati}},\ and\ \bibinfo {author} {\bibfnamefont
  {E.}~\bibnamefont {Sontag}},\ }\bibfield  {title} {\bibinfo {title} {Internal
  models in control, bioengineering, and neuroscience},\ }\href
  {https://doi.org/https://doi.org/10.1146/annurev-control-042920-102205}
  {\bibfield  {journal} {\bibinfo  {journal} {Annual Review of Control,
  Robotics, and Autonomous Systems}\ }\textbf {\bibinfo {volume} {5}},\
  \bibinfo {pages} {55} (\bibinfo {year} {2022})}\BibitemShut {NoStop}%
\bibitem [{\citenamefont {Sontag}(2003)}]{SONTAG2003119}%
  \BibitemOpen
  \bibfield  {author} {\bibinfo {author} {\bibfnamefont {E.~D.}\ \bibnamefont
  {Sontag}},\ }\bibfield  {title} {\bibinfo {title} {Adaptation and regulation
  with signal detection implies internal model},\ }\href
  {https://doi.org/https://doi.org/10.1016/S0167-6911(03)00136-1} {\bibfield
  {journal} {\bibinfo  {journal} {Systems \& Control Letters}\ }\textbf
  {\bibinfo {volume} {50}},\ \bibinfo {pages} {119} (\bibinfo {year}
  {2003})}\BibitemShut {NoStop}%
\bibitem [{\citenamefont {Alon}\ \emph {et~al.}(1999)\citenamefont {Alon},
  \citenamefont {Surette}, \citenamefont {Barkai},\ and\ \citenamefont
  {Leibler}}]{alon1999robustness}%
  \BibitemOpen
  \bibfield  {author} {\bibinfo {author} {\bibfnamefont {U.}~\bibnamefont
  {Alon}}, \bibinfo {author} {\bibfnamefont {M.~G.}\ \bibnamefont {Surette}},
  \bibinfo {author} {\bibfnamefont {N.}~\bibnamefont {Barkai}},\ and\ \bibinfo
  {author} {\bibfnamefont {S.}~\bibnamefont {Leibler}},\ }\bibfield  {title}
  {\bibinfo {title} {Robustness in bacterial chemotaxis},\ }\href
  {https://doi.org/https://doi.org/10.1038/16483} {\bibfield  {journal}
  {\bibinfo  {journal} {Nature}\ }\textbf {\bibinfo {volume} {397}},\ \bibinfo
  {pages} {168} (\bibinfo {year} {1999})}\BibitemShut {NoStop}%
\bibitem [{\citenamefont {Ang}\ \emph {et~al.}(2010)\citenamefont {Ang},
  \citenamefont {Bagh}, \citenamefont {Ingalls},\ and\ \citenamefont
  {McMillen}}]{ANG2010723}%
  \BibitemOpen
  \bibfield  {author} {\bibinfo {author} {\bibfnamefont {J.}~\bibnamefont
  {Ang}}, \bibinfo {author} {\bibfnamefont {S.}~\bibnamefont {Bagh}}, \bibinfo
  {author} {\bibfnamefont {B.~P.}\ \bibnamefont {Ingalls}},\ and\ \bibinfo
  {author} {\bibfnamefont {D.~R.}\ \bibnamefont {McMillen}},\ }\bibfield
  {title} {\bibinfo {title} {Considerations for using integral feedback control
  to construct a perfectly adapting synthetic gene network},\ }\href
  {https://doi.org/https://doi.org/10.1016/j.jtbi.2010.07.034} {\bibfield
  {journal} {\bibinfo  {journal} {Journal of Theoretical Biology}\ }\textbf
  {\bibinfo {volume} {266}},\ \bibinfo {pages} {723} (\bibinfo {year}
  {2010})}\BibitemShut {NoStop}%
\bibitem [{\citenamefont {Aoki}\ \emph {et~al.}(2019)\citenamefont {Aoki},
  \citenamefont {Lillacci}, \citenamefont {Gupta}, \citenamefont
  {Baumschlager}, \citenamefont {Schweingruber},\ and\ \citenamefont
  {Khammash}}]{aoki2019universal}%
  \BibitemOpen
  \bibfield  {author} {\bibinfo {author} {\bibfnamefont {S.~K.}\ \bibnamefont
  {Aoki}}, \bibinfo {author} {\bibfnamefont {G.}~\bibnamefont {Lillacci}},
  \bibinfo {author} {\bibfnamefont {A.}~\bibnamefont {Gupta}}, \bibinfo
  {author} {\bibfnamefont {A.}~\bibnamefont {Baumschlager}}, \bibinfo {author}
  {\bibfnamefont {D.}~\bibnamefont {Schweingruber}},\ and\ \bibinfo {author}
  {\bibfnamefont {M.}~\bibnamefont {Khammash}},\ }\bibfield  {title} {\bibinfo
  {title} {A universal biomolecular integral feedback controller for robust
  perfect adaptation},\ }\href
  {https://doi.org/https://doi.org/10.1038/s41586-019-1321-1} {\bibfield
  {journal} {\bibinfo  {journal} {Nature}\ }\textbf {\bibinfo {volume} {570}},\
  \bibinfo {pages} {533} (\bibinfo {year} {2019})}\BibitemShut {NoStop}%
\bibitem [{\citenamefont {Xiao}\ and\ \citenamefont {Doyle}(2018)}]{8619101}%
  \BibitemOpen
  \bibfield  {author} {\bibinfo {author} {\bibfnamefont {F.}~\bibnamefont
  {Xiao}}\ and\ \bibinfo {author} {\bibfnamefont {J.~C.}\ \bibnamefont
  {Doyle}},\ }\bibfield  {title} {\bibinfo {title} {Robust perfect adaptation
  in biomolecular reaction networks},\ }in\ \href
  {https://doi.org/10.1109/CDC.2018.8619101} {\emph {\bibinfo {booktitle} {2018
  IEEE Conference on Decision and Control (CDC)}}}\ (\bibinfo {year} {2018})\
  pp.\ \bibinfo {pages} {4345--4352}\BibitemShut {NoStop}%
\bibitem [{\citenamefont {Feinberg}(2019)}]{feinberg2019foundations}%
  \BibitemOpen
  \bibfield  {author} {\bibinfo {author} {\bibfnamefont {M.}~\bibnamefont
  {Feinberg}},\ }\href
  {https://doi.org/https://doi.org/10.1007/978-3-030-03858-8} {\emph {\bibinfo
  {title} {Foundations of chemical reaction network theory}}}\ (\bibinfo
  {publisher} {Springer},\ \bibinfo {year} {2019})\BibitemShut {NoStop}%
\bibitem [{\citenamefont {Shinar}\ and\ \citenamefont
  {Feinberg}(2010)}]{shinar2010structural}%
  \BibitemOpen
  \bibfield  {author} {\bibinfo {author} {\bibfnamefont {G.}~\bibnamefont
  {Shinar}}\ and\ \bibinfo {author} {\bibfnamefont {M.}~\bibnamefont
  {Feinberg}},\ }\bibfield  {title} {\bibinfo {title} {Structural sources of
  robustness in biochemical reaction networks},\ }\href
  {https://doi.org/10.1126/science.1183372} {\bibfield  {journal} {\bibinfo
  {journal} {Science}\ }\textbf {\bibinfo {volume} {327}},\ \bibinfo {pages}
  {1389} (\bibinfo {year} {2010})}\BibitemShut {NoStop}%
\bibitem [{\citenamefont {Cappelletti}\ \emph {et~al.}(2020)\citenamefont
  {Cappelletti}, \citenamefont {Gupta},\ and\ \citenamefont
  {Khammash}}]{cappelletti2020hidden}%
  \BibitemOpen
  \bibfield  {author} {\bibinfo {author} {\bibfnamefont {D.}~\bibnamefont
  {Cappelletti}}, \bibinfo {author} {\bibfnamefont {A.}~\bibnamefont {Gupta}},\
  and\ \bibinfo {author} {\bibfnamefont {M.}~\bibnamefont {Khammash}},\
  }\bibfield  {title} {\bibinfo {title} {A hidden integral structure endows
  absolute concentration robust systems with resilience to dynamical
  concentration disturbances},\ }\href
  {https://doi.org/https://doi.org/10.1098/rsif.2020.0437} {\bibfield
  {journal} {\bibinfo  {journal} {Journal of the Royal Society Interface}\
  }\textbf {\bibinfo {volume} {17}},\ \bibinfo {pages} {20200437} (\bibinfo
  {year} {2020})}\BibitemShut {NoStop}%
\bibitem [{\citenamefont {Ni}\ \emph {et~al.}(2009)\citenamefont {Ni},
  \citenamefont {Drengstig},\ and\ \citenamefont {Ruoff}}]{ni2009control}%
  \BibitemOpen
  \bibfield  {author} {\bibinfo {author} {\bibfnamefont {X.~Y.}\ \bibnamefont
  {Ni}}, \bibinfo {author} {\bibfnamefont {T.}~\bibnamefont {Drengstig}},\ and\
  \bibinfo {author} {\bibfnamefont {P.}~\bibnamefont {Ruoff}},\ }\bibfield
  {title} {\bibinfo {title} {The control of the controller: molecular
  mechanisms for robust perfect adaptation and temperature compensation},\
  }\href {https://doi.org/https://doi.org/10.1016/j.bpj.2009.06.030} {\bibfield
   {journal} {\bibinfo  {journal} {Biophysical journal}\ }\textbf {\bibinfo
  {volume} {97}},\ \bibinfo {pages} {1244} (\bibinfo {year}
  {2009})}\BibitemShut {NoStop}%
\bibitem [{\citenamefont {Toni}\ \emph {et~al.}(2011)\citenamefont {Toni},
  \citenamefont {Jovanovic}, \citenamefont {Huvet}, \citenamefont {Buck},\ and\
  \citenamefont {Stumpf}}]{toni2011qualitative}%
  \BibitemOpen
  \bibfield  {author} {\bibinfo {author} {\bibfnamefont {T.}~\bibnamefont
  {Toni}}, \bibinfo {author} {\bibfnamefont {G.}~\bibnamefont {Jovanovic}},
  \bibinfo {author} {\bibfnamefont {M.}~\bibnamefont {Huvet}}, \bibinfo
  {author} {\bibfnamefont {M.}~\bibnamefont {Buck}},\ and\ \bibinfo {author}
  {\bibfnamefont {M.~P.}\ \bibnamefont {Stumpf}},\ }\bibfield  {title}
  {\bibinfo {title} {From qualitative data to quantitative models: analysis of
  the phage shock protein stress response in escherichia coli},\ }\href
  {https://doi.org/https://doi.org/10.1186/1752-0509-5-69} {\bibfield
  {journal} {\bibinfo  {journal} {BMC systems biology}\ }\textbf {\bibinfo
  {volume} {5}},\ \bibinfo {pages} {1} (\bibinfo {year} {2011})}\BibitemShut
  {NoStop}%
\bibitem [{\citenamefont {Qian}\ \emph {et~al.}(2017)\citenamefont {Qian},
  \citenamefont {Huang}, \citenamefont {Jim{\'e}nez},\ and\ \citenamefont
  {Del~Vecchio}}]{qian2017resource}%
  \BibitemOpen
  \bibfield  {author} {\bibinfo {author} {\bibfnamefont {Y.}~\bibnamefont
  {Qian}}, \bibinfo {author} {\bibfnamefont {H.-H.}\ \bibnamefont {Huang}},
  \bibinfo {author} {\bibfnamefont {J.~I.}\ \bibnamefont {Jim{\'e}nez}},\ and\
  \bibinfo {author} {\bibfnamefont {D.}~\bibnamefont {Del~Vecchio}},\
  }\bibfield  {title} {\bibinfo {title} {Resource competition shapes the
  response of genetic circuits},\ }\href
  {https://doi.org/https://doi.org/10.1021/acssynbio.6b00361} {\bibfield
  {journal} {\bibinfo  {journal} {ACS synthetic biology}\ }\textbf {\bibinfo
  {volume} {6}},\ \bibinfo {pages} {1263} (\bibinfo {year} {2017})}\BibitemShut
  {NoStop}%
\bibitem [{\citenamefont {Gupta}\ and\ \citenamefont
  {Khammash}(2022)}]{gupta2022universal}%
  \BibitemOpen
  \bibfield  {author} {\bibinfo {author} {\bibfnamefont {A.}~\bibnamefont
  {Gupta}}\ and\ \bibinfo {author} {\bibfnamefont {M.}~\bibnamefont
  {Khammash}},\ }\bibfield  {title} {\bibinfo {title} {Universal structural
  requirements for maximal robust perfect adaptation in biomolecular
  networks},\ }\href {https://doi.org/10.1073/pnas.2207802119} {\bibfield
  {journal} {\bibinfo  {journal} {Proceedings of the National Academy of
  Sciences}\ }\textbf {\bibinfo {volume} {119}},\ \bibinfo {pages}
  {e2207802119} (\bibinfo {year} {2022})}\BibitemShut {NoStop}%
\bibitem [{\citenamefont {Alexis}\ \emph {et~al.}(2022)\citenamefont {Alexis},
  \citenamefont {Schulte}, \citenamefont {Cardelli},\ and\ \citenamefont
  {Papachristodoulou}}]{alexis2022regulation}%
  \BibitemOpen
  \bibfield  {author} {\bibinfo {author} {\bibfnamefont {E.}~\bibnamefont
  {Alexis}}, \bibinfo {author} {\bibfnamefont {C.~C.}\ \bibnamefont {Schulte}},
  \bibinfo {author} {\bibfnamefont {L.}~\bibnamefont {Cardelli}},\ and\
  \bibinfo {author} {\bibfnamefont {A.}~\bibnamefont {Papachristodoulou}},\
  }\bibfield  {title} {\bibinfo {title} {Regulation strategies for two-output
  biomolecular networks},\ }\href
  {https://doi.org/https://doi.org/10.1101/2022.02.28.482258} {\bibfield
  {journal} {\bibinfo  {journal} {bioRxiv}\ ,\ \bibinfo {pages} {2022}}
  (\bibinfo {year} {2022})}\BibitemShut {NoStop}%
\bibitem [{\citenamefont {Araujo}\ and\ \citenamefont
  {Liotta}(2023)}]{araujo2023universal}%
  \BibitemOpen
  \bibfield  {author} {\bibinfo {author} {\bibfnamefont {R.~P.}\ \bibnamefont
  {Araujo}}\ and\ \bibinfo {author} {\bibfnamefont {L.~A.}\ \bibnamefont
  {Liotta}},\ }\bibfield  {title} {\bibinfo {title} {Universal structures for
  adaptation in biochemical reaction networks},\ }\href
  {https://doi.org/https://doi.org/10.1038/s41467-023-38011-9} {\bibfield
  {journal} {\bibinfo  {journal} {Nature Communications}\ }\textbf {\bibinfo
  {volume} {14}},\ \bibinfo {pages} {2251} (\bibinfo {year}
  {2023})}\BibitemShut {NoStop}%
\bibitem [{\citenamefont {Guldberg}\ and\ \citenamefont
  {Waage}(1864)}]{guldberg1864studies}%
  \BibitemOpen
  \bibfield  {author} {\bibinfo {author} {\bibfnamefont {C.~M.}\ \bibnamefont
  {Guldberg}}\ and\ \bibinfo {author} {\bibfnamefont {P.}~\bibnamefont
  {Waage}},\ }\bibfield  {title} {\bibinfo {title} {Studies concerning
  affinity},\ }\href@noop {} {\bibfield  {journal} {\bibinfo  {journal} {CM
  Forhandlinger: Videnskabs-Selskabet i Christiana}\ }\textbf {\bibinfo
  {volume} {35}},\ \bibinfo {pages} {1864} (\bibinfo {year}
  {1864})}\BibitemShut {NoStop}%
\bibitem [{\citenamefont {Schnell}\ and\ \citenamefont
  {Turner}(2004)}]{schnell2004reaction}%
  \BibitemOpen
  \bibfield  {author} {\bibinfo {author} {\bibfnamefont {S.}~\bibnamefont
  {Schnell}}\ and\ \bibinfo {author} {\bibfnamefont {T.}~\bibnamefont
  {Turner}},\ }\bibfield  {title} {\bibinfo {title} {Reaction kinetics in
  intracellular environments with macromolecular crowding: simulations and rate
  laws},\ }\href
  {https://doi.org/https://doi.org/10.1016/j.pbiomolbio.2004.01.012} {\bibfield
   {journal} {\bibinfo  {journal} {Progress in biophysics and molecular
  biology}\ }\textbf {\bibinfo {volume} {85}},\ \bibinfo {pages} {235}
  (\bibinfo {year} {2004})}\BibitemShut {NoStop}%
\bibitem [{\citenamefont {Hall}\ and\ \citenamefont
  {Minton}(2003)}]{hall2003macromolecular}%
  \BibitemOpen
  \bibfield  {author} {\bibinfo {author} {\bibfnamefont {D.}~\bibnamefont
  {Hall}}\ and\ \bibinfo {author} {\bibfnamefont {A.~P.}\ \bibnamefont
  {Minton}},\ }\bibfield  {title} {\bibinfo {title} {Macromolecular crowding:
  qualitative and semiquantitative successes, quantitative challenges},\ }\href
  {https://doi.org/https://doi.org/10.1016/S1570-9639(03)00167-5} {\bibfield
  {journal} {\bibinfo  {journal} {Biochimica et Biophysica Acta (BBA)-Proteins
  and Proteomics}\ }\textbf {\bibinfo {volume} {1649}},\ \bibinfo {pages} {127}
  (\bibinfo {year} {2003})}\BibitemShut {NoStop}%
\bibitem [{\citenamefont {Bauermann}\ \emph {et~al.}(2022)\citenamefont
  {Bauermann}, \citenamefont {Laha}, \citenamefont {McCall}, \citenamefont
  {J{\"u}licher},\ and\ \citenamefont {Weber}}]{bauermann2022chemical}%
  \BibitemOpen
  \bibfield  {author} {\bibinfo {author} {\bibfnamefont {J.}~\bibnamefont
  {Bauermann}}, \bibinfo {author} {\bibfnamefont {S.}~\bibnamefont {Laha}},
  \bibinfo {author} {\bibfnamefont {P.~M.}\ \bibnamefont {McCall}}, \bibinfo
  {author} {\bibfnamefont {F.}~\bibnamefont {J{\"u}licher}},\ and\ \bibinfo
  {author} {\bibfnamefont {C.~A.}\ \bibnamefont {Weber}},\ }\bibfield  {title}
  {\bibinfo {title} {Chemical kinetics and mass action in coexisting phases},\
  }\href {https://doi.org/https://doi.org/10.1021/jacs.2c06265} {\bibfield
  {journal} {\bibinfo  {journal} {Journal of the American Chemical Society}\
  }\textbf {\bibinfo {volume} {144}},\ \bibinfo {pages} {19294} (\bibinfo
  {year} {2022})}\BibitemShut {NoStop}%
\bibitem [{\citenamefont {Elowitz}\ \emph {et~al.}(1999)\citenamefont
  {Elowitz}, \citenamefont {Surette}, \citenamefont {Wolf}, \citenamefont
  {Stock},\ and\ \citenamefont {Leibler}}]{elowitz1999protein}%
  \BibitemOpen
  \bibfield  {author} {\bibinfo {author} {\bibfnamefont {M.~B.}\ \bibnamefont
  {Elowitz}}, \bibinfo {author} {\bibfnamefont {M.~G.}\ \bibnamefont
  {Surette}}, \bibinfo {author} {\bibfnamefont {P.-E.}\ \bibnamefont {Wolf}},
  \bibinfo {author} {\bibfnamefont {J.~B.}\ \bibnamefont {Stock}},\ and\
  \bibinfo {author} {\bibfnamefont {S.}~\bibnamefont {Leibler}},\ }\bibfield
  {title} {\bibinfo {title} {Protein mobility in the cytoplasm of escherichia
  coli},\ }\href
  {https://doi.org/https://doi.org/10.1128/jb.181.1.197-203.1999} {\bibfield
  {journal} {\bibinfo  {journal} {Journal of bacteriology}\ }\textbf {\bibinfo
  {volume} {181}},\ \bibinfo {pages} {197} (\bibinfo {year}
  {1999})}\BibitemShut {NoStop}%
\bibitem [{\citenamefont {Bu}\ and\ \citenamefont
  {Callaway}(2011)}]{bu2011proteins}%
  \BibitemOpen
  \bibfield  {author} {\bibinfo {author} {\bibfnamefont {Z.}~\bibnamefont
  {Bu}}\ and\ \bibinfo {author} {\bibfnamefont {D.~J.}\ \bibnamefont
  {Callaway}},\ }\bibfield  {title} {\bibinfo {title} {Proteins move! protein
  dynamics and long-range allostery in cell signaling},\ }\href
  {https://doi.org/https://doi.org/10.1016/B978-0-12-381262-9.00005-7}
  {\bibfield  {journal} {\bibinfo  {journal} {Advances in protein chemistry and
  structural biology}\ }\textbf {\bibinfo {volume} {83}},\ \bibinfo {pages}
  {163} (\bibinfo {year} {2011})}\BibitemShut {NoStop}%
\bibitem [{\citenamefont {Perica}\ \emph {et~al.}(2021)\citenamefont {Perica},
  \citenamefont {Mathy}, \citenamefont {Xu}, \citenamefont {Jang},
  \citenamefont {Zhang}, \citenamefont {Kaake}, \citenamefont {Ollikainen},
  \citenamefont {Braberg}, \citenamefont {Swaney}, \citenamefont {Lambright}
  \emph {et~al.}}]{perica2021systems}%
  \BibitemOpen
  \bibfield  {author} {\bibinfo {author} {\bibfnamefont {T.}~\bibnamefont
  {Perica}}, \bibinfo {author} {\bibfnamefont {C.~J.}\ \bibnamefont {Mathy}},
  \bibinfo {author} {\bibfnamefont {J.}~\bibnamefont {Xu}}, \bibinfo {author}
  {\bibfnamefont {G.~M.}\ \bibnamefont {Jang}}, \bibinfo {author}
  {\bibfnamefont {Y.}~\bibnamefont {Zhang}}, \bibinfo {author} {\bibfnamefont
  {R.}~\bibnamefont {Kaake}}, \bibinfo {author} {\bibfnamefont
  {N.}~\bibnamefont {Ollikainen}}, \bibinfo {author} {\bibfnamefont
  {H.}~\bibnamefont {Braberg}}, \bibinfo {author} {\bibfnamefont {D.~L.}\
  \bibnamefont {Swaney}}, \bibinfo {author} {\bibfnamefont {D.~G.}\
  \bibnamefont {Lambright}}, \emph {et~al.},\ }\bibfield  {title} {\bibinfo
  {title} {Systems-level effects of allosteric perturbations to a model
  molecular switch},\ }\href
  {https://doi.org/https://doi.org/10.1038/s41586-021-03982-6} {\bibfield
  {journal} {\bibinfo  {journal} {Nature}\ }\textbf {\bibinfo {volume} {599}},\
  \bibinfo {pages} {152} (\bibinfo {year} {2021})}\BibitemShut {NoStop}%
\bibitem [{\citenamefont {Cardinale}\ and\ \citenamefont
  {Arkin}(2012)}]{cardinale2012contextualizing}%
  \BibitemOpen
  \bibfield  {author} {\bibinfo {author} {\bibfnamefont {S.}~\bibnamefont
  {Cardinale}}\ and\ \bibinfo {author} {\bibfnamefont {A.~P.}\ \bibnamefont
  {Arkin}},\ }\bibfield  {title} {\bibinfo {title} {Contextualizing context for
  synthetic biology--identifying causes of failure of synthetic biological
  systems},\ }\href {https://doi.org/https://doi.org/10.1002/biot.201200085}
  {\bibfield  {journal} {\bibinfo  {journal} {Biotechnology journal}\ }\textbf
  {\bibinfo {volume} {7}},\ \bibinfo {pages} {856} (\bibinfo {year}
  {2012})}\BibitemShut {NoStop}%
\bibitem [{\citenamefont {Minton}(2001)}]{minton2001influence}%
  \BibitemOpen
  \bibfield  {author} {\bibinfo {author} {\bibfnamefont {A.~P.}\ \bibnamefont
  {Minton}},\ }\bibfield  {title} {\bibinfo {title} {The influence of
  macromolecular crowding and macromolecular confinement on biochemical
  reactions in physiological media},\ }\href
  {https://doi.org/https://doi.org/10.1074/jbc.R100005200} {\bibfield
  {journal} {\bibinfo  {journal} {Journal of biological chemistry}\ }\textbf
  {\bibinfo {volume} {276}},\ \bibinfo {pages} {10577} (\bibinfo {year}
  {2001})}\BibitemShut {NoStop}%
\bibitem [{\citenamefont {Okada}\ and\ \citenamefont
  {Mochizuki}(2016)}]{PhysRevLett.117.048101}%
  \BibitemOpen
  \bibfield  {author} {\bibinfo {author} {\bibfnamefont {T.}~\bibnamefont
  {Okada}}\ and\ \bibinfo {author} {\bibfnamefont {A.}~\bibnamefont
  {Mochizuki}},\ }\bibfield  {title} {\bibinfo {title} {Law of localization in
  chemical reaction networks},\ }\href
  {https://doi.org/10.1103/PhysRevLett.117.048101} {\bibfield  {journal}
  {\bibinfo  {journal} {Phys. Rev. Lett.}\ }\textbf {\bibinfo {volume} {117}},\
  \bibinfo {pages} {048101} (\bibinfo {year} {2016})}\BibitemShut {NoStop}%
\bibitem [{\citenamefont {Okada}\ and\ \citenamefont
  {Mochizuki}(2017)}]{PhysRevE.96.022322}%
  \BibitemOpen
  \bibfield  {author} {\bibinfo {author} {\bibfnamefont {T.}~\bibnamefont
  {Okada}}\ and\ \bibinfo {author} {\bibfnamefont {A.}~\bibnamefont
  {Mochizuki}},\ }\bibfield  {title} {\bibinfo {title} {Sensitivity and network
  topology in chemical reaction systems},\ }\href
  {https://doi.org/10.1103/PhysRevE.96.022322} {\bibfield  {journal} {\bibinfo
  {journal} {Phys. Rev. E}\ }\textbf {\bibinfo {volume} {96}},\ \bibinfo
  {pages} {022322} (\bibinfo {year} {2017})}\BibitemShut {NoStop}%
\bibitem [{\citenamefont {Okada}\ \emph {et~al.}(2018)\citenamefont {Okada},
  \citenamefont {Tsai},\ and\ \citenamefont {Mochizuki}}]{PhysRevE.98.012417}%
  \BibitemOpen
  \bibfield  {author} {\bibinfo {author} {\bibfnamefont {T.}~\bibnamefont
  {Okada}}, \bibinfo {author} {\bibfnamefont {J.-C.}\ \bibnamefont {Tsai}},\
  and\ \bibinfo {author} {\bibfnamefont {A.}~\bibnamefont {Mochizuki}},\
  }\bibfield  {title} {\bibinfo {title} {Structural bifurcation analysis in
  chemical reaction networks},\ }\href
  {https://doi.org/10.1103/PhysRevE.98.012417} {\bibfield  {journal} {\bibinfo
  {journal} {Phys. Rev. E}\ }\textbf {\bibinfo {volume} {98}},\ \bibinfo
  {pages} {012417} (\bibinfo {year} {2018})}\BibitemShut {NoStop}%
\bibitem [{\citenamefont {Hirono}\ \emph {et~al.}(2021)\citenamefont {Hirono},
  \citenamefont {Okada}, \citenamefont {Miyazaki},\ and\ \citenamefont
  {Hidaka}}]{PhysRevResearch.3.043123}%
  \BibitemOpen
  \bibfield  {author} {\bibinfo {author} {\bibfnamefont {Y.}~\bibnamefont
  {Hirono}}, \bibinfo {author} {\bibfnamefont {T.}~\bibnamefont {Okada}},
  \bibinfo {author} {\bibfnamefont {H.}~\bibnamefont {Miyazaki}},\ and\
  \bibinfo {author} {\bibfnamefont {Y.}~\bibnamefont {Hidaka}},\ }\bibfield
  {title} {\bibinfo {title} {Structural reduction of chemical reaction networks
  based on topology},\ }\href
  {https://doi.org/10.1103/PhysRevResearch.3.043123} {\bibfield  {journal}
  {\bibinfo  {journal} {Phys. Rev. Research}\ }\textbf {\bibinfo {volume}
  {3}},\ \bibinfo {pages} {043123} (\bibinfo {year} {2021})}\BibitemShut
  {NoStop}%
\bibitem [{\citenamefont {Hirono}\ \emph {et~al.}(2023)\citenamefont {Hirono},
  \citenamefont {Hong},\ and\ \citenamefont
  {Kim}}]{https://doi.org/10.48550/arxiv.2302.01270}%
  \BibitemOpen
  \bibfield  {author} {\bibinfo {author} {\bibfnamefont {Y.}~\bibnamefont
  {Hirono}}, \bibinfo {author} {\bibfnamefont {H.}~\bibnamefont {Hong}},\ and\
  \bibinfo {author} {\bibfnamefont {J.~K.}\ \bibnamefont {Kim}},\ }\href
  {https://doi.org/10.48550/ARXIV.2302.01270} {\bibinfo {title} {Robust perfect
  adaptation of reaction fluxes ensured by network topology}} (\bibinfo {year}
  {2023})\BibitemShut {NoStop}%
\bibitem [{RPA()}]{RPAFinder}%
  \BibitemOpen
  \href@noop {} {\bibinfo {title} {{GitHub repository for RPAFinder}}},\
  \bibinfo {howpublished}
  {\url{https://github.com/yhirono/RPAFinder}}\BibitemShut {NoStop}%
\bibitem [{\citenamefont {Gupta}\ and\ \citenamefont
  {Khammash}(2023{\natexlab{a}})}]{gupta2023internal}%
  \BibitemOpen
  \bibfield  {author} {\bibinfo {author} {\bibfnamefont {A.}~\bibnamefont
  {Gupta}}\ and\ \bibinfo {author} {\bibfnamefont {M.}~\bibnamefont
  {Khammash}},\ }\bibfield  {title} {\bibinfo {title} {The internal model
  principle for biomolecular control theory},\ }\href
  {https://doi.org/10.1109/OJCSYS.2023.3244089} {\bibfield  {journal} {\bibinfo
   {journal} {IEEE Open Journal of Control Systems}\ }\textbf {\bibinfo
  {volume} {2}},\ \bibinfo {pages} {63} (\bibinfo {year}
  {2023}{\natexlab{a}})}\BibitemShut {NoStop}%
\bibitem [{\citenamefont {Anderson}\ and\ \citenamefont
  {Kurtz}(2015)}]{anderson2015stochastic}%
  \BibitemOpen
  \bibfield  {author} {\bibinfo {author} {\bibfnamefont {D.~F.}\ \bibnamefont
  {Anderson}}\ and\ \bibinfo {author} {\bibfnamefont {T.~G.}\ \bibnamefont
  {Kurtz}},\ }\href {https://doi.org/10.1007/978-3-319-16895-1} {\emph
  {\bibinfo {title} {Stochastic analysis of biochemical systems}}},\ Vol.\
  \bibinfo {volume} {674}\ (\bibinfo  {publisher} {Springer},\ \bibinfo {year}
  {2015})\BibitemShut {NoStop}%
\bibitem [{\citenamefont {Mochizuki}\ and\ \citenamefont
  {Fiedler}(2015)}]{MOCHIZUKI2015189}%
  \BibitemOpen
  \bibfield  {author} {\bibinfo {author} {\bibfnamefont {A.}~\bibnamefont
  {Mochizuki}}\ and\ \bibinfo {author} {\bibfnamefont {B.}~\bibnamefont
  {Fiedler}},\ }\bibfield  {title} {\bibinfo {title} {Sensitivity of chemical
  reaction networks: A structural approach. 1. examples and the carbon
  metabolic network},\ }\href
  {https://doi.org/https://doi.org/10.1016/j.jtbi.2014.10.025} {\bibfield
  {journal} {\bibinfo  {journal} {Journal of Theoretical Biology}\ }\textbf
  {\bibinfo {volume} {367}},\ \bibinfo {pages} {189 } (\bibinfo {year}
  {2015})}\BibitemShut {NoStop}%
\bibitem [{\citenamefont {Francis}\ and\ \citenamefont
  {Wonham}(1975)}]{francis1975internal}%
  \BibitemOpen
  \bibfield  {author} {\bibinfo {author} {\bibfnamefont {B.~A.}\ \bibnamefont
  {Francis}}\ and\ \bibinfo {author} {\bibfnamefont {W.~M.}\ \bibnamefont
  {Wonham}},\ }\bibfield  {title} {\bibinfo {title} {The internal model
  principle for linear multivariable regulators},\ }\href
  {https://doi.org/https://doi.org/10.1007/BF01447855} {\bibfield  {journal}
  {\bibinfo  {journal} {Applied mathematics and optimization}\ }\textbf
  {\bibinfo {volume} {2}},\ \bibinfo {pages} {170} (\bibinfo {year}
  {1975})}\BibitemShut {NoStop}%
\bibitem [{\citenamefont {Francis}\ and\ \citenamefont
  {Wonham}(1976)}]{francis1976internal}%
  \BibitemOpen
  \bibfield  {author} {\bibinfo {author} {\bibfnamefont {B.~A.}\ \bibnamefont
  {Francis}}\ and\ \bibinfo {author} {\bibfnamefont {W.~M.}\ \bibnamefont
  {Wonham}},\ }\bibfield  {title} {\bibinfo {title} {The internal model
  principle of control theory},\ }\href
  {https://doi.org/https://doi.org/10.1016/0005-1098(76)90006-6} {\bibfield
  {journal} {\bibinfo  {journal} {Automatica}\ }\textbf {\bibinfo {volume}
  {12}},\ \bibinfo {pages} {457} (\bibinfo {year} {1976})}\BibitemShut
  {NoStop}%
\bibitem [{\citenamefont {Fiedler}\ and\ \citenamefont
  {Mochizuki}(2015)}]{doi:10.1002/mma.3436}%
  \BibitemOpen
  \bibfield  {author} {\bibinfo {author} {\bibfnamefont {B.}~\bibnamefont
  {Fiedler}}\ and\ \bibinfo {author} {\bibfnamefont {A.}~\bibnamefont
  {Mochizuki}},\ }\bibfield  {title} {\bibinfo {title} {Sensitivity of chemical
  reaction networks: a structural approach. 2. regular monomolecular systems},\
  }\href {https://doi.org/10.1002/mma.3436} {\bibfield  {journal} {\bibinfo
  {journal} {Mathematical Methods in the Applied Sciences}\ }\textbf {\bibinfo
  {volume} {38}},\ \bibinfo {pages} {3519} (\bibinfo {year}
  {2015})}\BibitemShut {NoStop}%
\bibitem [{\citenamefont {Vassena}\ and\ \citenamefont
  {Matano}(2017)}]{https://doi.org/10.1002/mma.4557}%
  \BibitemOpen
  \bibfield  {author} {\bibinfo {author} {\bibfnamefont {N.}~\bibnamefont
  {Vassena}}\ and\ \bibinfo {author} {\bibfnamefont {H.}~\bibnamefont
  {Matano}},\ }\bibfield  {title} {\bibinfo {title} {Monomolecular reaction
  networks: Flux-influenced sets and balloons},\ }\href
  {https://doi.org/https://doi.org/10.1002/mma.4557} {\bibfield  {journal}
  {\bibinfo  {journal} {Mathematical Methods in the Applied Sciences}\ }\textbf
  {\bibinfo {volume} {40}},\ \bibinfo {pages} {7722} (\bibinfo {year}
  {2017})}\BibitemShut {NoStop}%
\bibitem [{\citenamefont {Brehm}\ and\ \citenamefont
  {Fiedler}(2018)}]{doi:10.1002/mma.4668}%
  \BibitemOpen
  \bibfield  {author} {\bibinfo {author} {\bibfnamefont {B.}~\bibnamefont
  {Brehm}}\ and\ \bibinfo {author} {\bibfnamefont {B.}~\bibnamefont
  {Fiedler}},\ }\bibfield  {title} {\bibinfo {title} {Sensitivity of chemical
  reaction networks: A structural approach. 3. regular multimolecular
  systems},\ }\href {https://doi.org/10.1002/mma.4668} {\bibfield  {journal}
  {\bibinfo  {journal} {Mathematical Methods in the Applied Sciences}\ }\textbf
  {\bibinfo {volume} {41}},\ \bibinfo {pages} {1344} (\bibinfo {year}
  {2018})}\BibitemShut {NoStop}%
\bibitem [{\citenamefont {Giordano}\ \emph {et~al.}(2016)\citenamefont
  {Giordano}, \citenamefont {Cuba~Samaniego}, \citenamefont {Franco},\ and\
  \citenamefont {Blanchini}}]{giordano2016computing}%
  \BibitemOpen
  \bibfield  {author} {\bibinfo {author} {\bibfnamefont {G.}~\bibnamefont
  {Giordano}}, \bibinfo {author} {\bibfnamefont {C.}~\bibnamefont
  {Cuba~Samaniego}}, \bibinfo {author} {\bibfnamefont {E.}~\bibnamefont
  {Franco}},\ and\ \bibinfo {author} {\bibfnamefont {F.}~\bibnamefont
  {Blanchini}},\ }\bibfield  {title} {\bibinfo {title} {Computing the
  structural influence matrix for biological systems},\ }\href@noop {}
  {\bibfield  {journal} {\bibinfo  {journal} {Journal of Mathematical Biology}\
  }\textbf {\bibinfo {volume} {72}},\ \bibinfo {pages} {1927} (\bibinfo {year}
  {2016})}\BibitemShut {NoStop}%
\bibitem [{\citenamefont {Okada}\ \emph {et~al.}(2021)\citenamefont {Okada},
  \citenamefont {Mochizuki}, \citenamefont {Furuta},\ and\ \citenamefont
  {Tsai}}]{PhysRevE.103.062212}%
  \BibitemOpen
  \bibfield  {author} {\bibinfo {author} {\bibfnamefont {T.}~\bibnamefont
  {Okada}}, \bibinfo {author} {\bibfnamefont {A.}~\bibnamefont {Mochizuki}},
  \bibinfo {author} {\bibfnamefont {M.}~\bibnamefont {Furuta}},\ and\ \bibinfo
  {author} {\bibfnamefont {J.-C.}\ \bibnamefont {Tsai}},\ }\bibfield  {title}
  {\bibinfo {title} {Flux-augmented bifurcation analysis in chemical reaction
  network systems},\ }\href {https://doi.org/10.1103/PhysRevE.103.062212}
  {\bibfield  {journal} {\bibinfo  {journal} {Phys. Rev. E}\ }\textbf {\bibinfo
  {volume} {103}},\ \bibinfo {pages} {062212} (\bibinfo {year}
  {2021})}\BibitemShut {NoStop}%
\bibitem [{\citenamefont {Berg}\ and\ \citenamefont
  {Brown}(1972)}]{berg1972chemotaxis}%
  \BibitemOpen
  \bibfield  {author} {\bibinfo {author} {\bibfnamefont {H.~C.}\ \bibnamefont
  {Berg}}\ and\ \bibinfo {author} {\bibfnamefont {D.~A.}\ \bibnamefont
  {Brown}},\ }\bibfield  {title} {\bibinfo {title} {Chemotaxis in escherichia
  coli analysed by three-dimensional tracking},\ }\href
  {https://doi.org/https://doi.org/10.1038/239500a0} {\bibfield  {journal}
  {\bibinfo  {journal} {Nature}\ }\textbf {\bibinfo {volume} {239}},\ \bibinfo
  {pages} {500} (\bibinfo {year} {1972})}\BibitemShut {NoStop}%
\bibitem [{\citenamefont {Alon}(2006)}]{alon2006introduction}%
  \BibitemOpen
  \bibfield  {author} {\bibinfo {author} {\bibfnamefont {U.}~\bibnamefont
  {Alon}},\ }\href {https://doi.org/https://doi.org/10.1201/9781420011432}
  {\emph {\bibinfo {title} {An Introduction to Systems Biology: Design
  Principles of Biological Circuits}}}\ (\bibinfo  {publisher} {CRC Press},\
  \bibinfo {year} {2006})\BibitemShut {NoStop}%
\bibitem [{\citenamefont {Yi}\ \emph {et~al.}(2000)\citenamefont {Yi},
  \citenamefont {Huang}, \citenamefont {Simon},\ and\ \citenamefont
  {Doyle}}]{yi2000robust}%
  \BibitemOpen
  \bibfield  {author} {\bibinfo {author} {\bibfnamefont {T.-M.}\ \bibnamefont
  {Yi}}, \bibinfo {author} {\bibfnamefont {Y.}~\bibnamefont {Huang}}, \bibinfo
  {author} {\bibfnamefont {M.~I.}\ \bibnamefont {Simon}},\ and\ \bibinfo
  {author} {\bibfnamefont {J.}~\bibnamefont {Doyle}},\ }\bibfield  {title}
  {\bibinfo {title} {Robust perfect adaptation in bacterial chemotaxis through
  integral feedback control},\ }\href
  {https://doi.org/https://doi.org/10.1073/pnas.97.9.4649} {\bibfield
  {journal} {\bibinfo  {journal} {Proceedings of the National Academy of
  Sciences}\ }\textbf {\bibinfo {volume} {97}},\ \bibinfo {pages} {4649}
  (\bibinfo {year} {2000})}\BibitemShut {NoStop}%
\bibitem [{\citenamefont {Bleris}\ \emph {et~al.}(2011)\citenamefont {Bleris},
  \citenamefont {Xie}, \citenamefont {Glass}, \citenamefont {Adadey},
  \citenamefont {Sontag},\ and\ \citenamefont
  {Benenson}}]{https://doi.org/10.1038/msb.2011.49}%
  \BibitemOpen
  \bibfield  {author} {\bibinfo {author} {\bibfnamefont {L.}~\bibnamefont
  {Bleris}}, \bibinfo {author} {\bibfnamefont {Z.}~\bibnamefont {Xie}},
  \bibinfo {author} {\bibfnamefont {D.}~\bibnamefont {Glass}}, \bibinfo
  {author} {\bibfnamefont {A.}~\bibnamefont {Adadey}}, \bibinfo {author}
  {\bibfnamefont {E.}~\bibnamefont {Sontag}},\ and\ \bibinfo {author}
  {\bibfnamefont {Y.}~\bibnamefont {Benenson}},\ }\bibfield  {title} {\bibinfo
  {title} {Synthetic incoherent feedforward circuits show adaptation to the
  amount of their genetic template},\ }\href
  {https://doi.org/https://doi.org/10.1038/msb.2011.49} {\bibfield  {journal}
  {\bibinfo  {journal} {Molecular Systems Biology}\ }\textbf {\bibinfo {volume}
  {7}},\ \bibinfo {pages} {519} (\bibinfo {year} {2011})}\BibitemShut {NoStop}%
\bibitem [{\citenamefont {Tyson}\ \emph {et~al.}(2003)\citenamefont {Tyson},
  \citenamefont {Chen},\ and\ \citenamefont {Novak}}]{TYSON2003221}%
  \BibitemOpen
  \bibfield  {author} {\bibinfo {author} {\bibfnamefont {J.~J.}\ \bibnamefont
  {Tyson}}, \bibinfo {author} {\bibfnamefont {K.~C.}\ \bibnamefont {Chen}},\
  and\ \bibinfo {author} {\bibfnamefont {B.}~\bibnamefont {Novak}},\ }\bibfield
   {title} {\bibinfo {title} {Sniffers, buzzers, toggles and blinkers: dynamics
  of regulatory and signaling pathways in the cell},\ }\href
  {https://doi.org/https://doi.org/10.1016/S0955-0674(03)00017-6} {\bibfield
  {journal} {\bibinfo  {journal} {Current Opinion in Cell Biology}\ }\textbf
  {\bibinfo {volume} {15}},\ \bibinfo {pages} {221} (\bibinfo {year}
  {2003})}\BibitemShut {NoStop}%
\bibitem [{\citenamefont {Gupta}\ and\ \citenamefont
  {Khammash}(2023{\natexlab{b}})}]{10041993}%
  \BibitemOpen
  \bibfield  {author} {\bibinfo {author} {\bibfnamefont {A.}~\bibnamefont
  {Gupta}}\ and\ \bibinfo {author} {\bibfnamefont {M.}~\bibnamefont
  {Khammash}},\ }\bibfield  {title} {\bibinfo {title} {The internal model
  principle for biomolecular control theory},\ }\href
  {https://doi.org/10.1109/OJCSYS.2023.3244089} {\bibfield  {journal} {\bibinfo
   {journal} {IEEE Open Journal of Control Systems}\ }\textbf {\bibinfo
  {volume} {2}},\ \bibinfo {pages} {63} (\bibinfo {year}
  {2023}{\natexlab{b}})}\BibitemShut {NoStop}%
\bibitem [{\citenamefont {Ostergaard}\ \emph {et~al.}(2000)\citenamefont
  {Ostergaard}, \citenamefont {Olsson},\ and\ \citenamefont
  {Nielsen}}]{doi:10.1128/MMBR.64.1.34-50.2000}%
  \BibitemOpen
  \bibfield  {author} {\bibinfo {author} {\bibfnamefont {S.}~\bibnamefont
  {Ostergaard}}, \bibinfo {author} {\bibfnamefont {L.}~\bibnamefont {Olsson}},\
  and\ \bibinfo {author} {\bibfnamefont {J.}~\bibnamefont {Nielsen}},\
  }\bibfield  {title} {\bibinfo {title} {Metabolic engineering of saccharomyces
  cerevisiae},\ }\href {https://doi.org/10.1128/MMBR.64.1.34-50.2000}
  {\bibfield  {journal} {\bibinfo  {journal} {Microbiology and Molecular
  Biology Reviews}\ }\textbf {\bibinfo {volume} {64}},\ \bibinfo {pages} {34}
  (\bibinfo {year} {2000})}\BibitemShut {NoStop}%
\bibitem [{\citenamefont {Schaaff}\ \emph {et~al.}(1989)\citenamefont
  {Schaaff}, \citenamefont {Heinisch},\ and\ \citenamefont
  {Zimmermann}}]{https://doi.org/10.1002/yea.320050408}%
  \BibitemOpen
  \bibfield  {author} {\bibinfo {author} {\bibfnamefont {I.}~\bibnamefont
  {Schaaff}}, \bibinfo {author} {\bibfnamefont {J.}~\bibnamefont {Heinisch}},\
  and\ \bibinfo {author} {\bibfnamefont {F.~K.}\ \bibnamefont {Zimmermann}},\
  }\bibfield  {title} {\bibinfo {title} {Overproduction of glycolytic enzymes
  in yeast},\ }\href {https://doi.org/https://doi.org/10.1002/yea.320050408}
  {\bibfield  {journal} {\bibinfo  {journal} {Yeast}\ }\textbf {\bibinfo
  {volume} {5}},\ \bibinfo {pages} {285} (\bibinfo {year} {1989})}\BibitemShut
  {NoStop}%
\bibitem [{\citenamefont {Davies}\ and\ \citenamefont
  {Brindle}(1992)}]{davies1992effects}%
  \BibitemOpen
  \bibfield  {author} {\bibinfo {author} {\bibfnamefont {S.~E.}\ \bibnamefont
  {Davies}}\ and\ \bibinfo {author} {\bibfnamefont {K.~M.}\ \bibnamefont
  {Brindle}},\ }\bibfield  {title} {\bibinfo {title} {Effects of overexpression
  of phosphofructokinase on glycolysis in the yeast saccharomyces cerevisiae},\
  }\href {https://doi.org/10.1021/bi00134a028} {\bibfield  {journal} {\bibinfo
  {journal} {Biochemistry}\ }\textbf {\bibinfo {volume} {31}},\ \bibinfo
  {pages} {4729} (\bibinfo {year} {1992})}\BibitemShut {NoStop}%
\bibitem [{\citenamefont {Hatcher}(2001)}]{hatcher2001topology}%
  \BibitemOpen
  \bibfield  {author} {\bibinfo {author} {\bibfnamefont {A.}~\bibnamefont
  {Hatcher}},\ }\href@noop {} {\emph {\bibinfo {title} {Algebraic Topology}}}\
  (\bibinfo  {publisher} {Cambridge University Press},\ \bibinfo {year}
  {2001})\BibitemShut {NoStop}%
\bibitem [{\citenamefont {Csete}\ and\ \citenamefont
  {Doyle}(2002)}]{doi:10.1126/science.1069981}%
  \BibitemOpen
  \bibfield  {author} {\bibinfo {author} {\bibfnamefont {M.~E.}\ \bibnamefont
  {Csete}}\ and\ \bibinfo {author} {\bibfnamefont {J.~C.}\ \bibnamefont
  {Doyle}},\ }\bibfield  {title} {\bibinfo {title} {Reverse engineering of
  biological complexity},\ }\href {https://doi.org/10.1126/science.1069981}
  {\bibfield  {journal} {\bibinfo  {journal} {Science}\ }\textbf {\bibinfo
  {volume} {295}},\ \bibinfo {pages} {1664} (\bibinfo {year}
  {2002})}\BibitemShut {NoStop}%
\bibitem [{\citenamefont {Del~Vecchio}\ \emph {et~al.}(2016)\citenamefont
  {Del~Vecchio}, \citenamefont {Dy},\ and\ \citenamefont
  {Qian}}]{del2016control}%
  \BibitemOpen
  \bibfield  {author} {\bibinfo {author} {\bibfnamefont {D.}~\bibnamefont
  {Del~Vecchio}}, \bibinfo {author} {\bibfnamefont {A.~J.}\ \bibnamefont
  {Dy}},\ and\ \bibinfo {author} {\bibfnamefont {Y.}~\bibnamefont {Qian}},\
  }\bibfield  {title} {\bibinfo {title} {Control theory meets synthetic
  biology},\ }\href {https://doi.org/https://doi.org/10.1098/rsif.2016.0380}
  {\bibfield  {journal} {\bibinfo  {journal} {Journal of The Royal Society
  Interface}\ }\textbf {\bibinfo {volume} {13}},\ \bibinfo {pages} {20160380}
  (\bibinfo {year} {2016})}\BibitemShut {NoStop}%
\bibitem [{\citenamefont {Stephanopoulos}\ \emph {et~al.}(1998)\citenamefont
  {Stephanopoulos}, \citenamefont {Aristidou},\ and\ \citenamefont
  {Nielsen}}]{stephanopoulos1998metabolic}%
  \BibitemOpen
  \bibfield  {author} {\bibinfo {author} {\bibfnamefont {G.}~\bibnamefont
  {Stephanopoulos}}, \bibinfo {author} {\bibfnamefont {A.~A.}\ \bibnamefont
  {Aristidou}},\ and\ \bibinfo {author} {\bibfnamefont {J.}~\bibnamefont
  {Nielsen}},\ }\href
  {https://doi.org/https://doi.org/10.1016/B978-0-12-666260-3.X5000-6} {\emph
  {\bibinfo {title} {Metabolic engineering: principles and methodologies}}}\
  (\bibinfo  {publisher} {Elsevier},\ \bibinfo {year} {1998})\BibitemShut
  {NoStop}%
\bibitem [{\citenamefont {Woolston}\ \emph {et~al.}(2013)\citenamefont
  {Woolston}, \citenamefont {Edgar},\ and\ \citenamefont
  {Stephanopoulos}}]{doi:10.1146/annurev-chembioeng-061312-103312}%
  \BibitemOpen
  \bibfield  {author} {\bibinfo {author} {\bibfnamefont {B.~M.}\ \bibnamefont
  {Woolston}}, \bibinfo {author} {\bibfnamefont {S.}~\bibnamefont {Edgar}},\
  and\ \bibinfo {author} {\bibfnamefont {G.}~\bibnamefont {Stephanopoulos}},\
  }\bibfield  {title} {\bibinfo {title} {Metabolic engineering: Past and
  future},\ }\href {https://doi.org/10.1146/annurev-chembioeng-061312-103312}
  {\bibfield  {journal} {\bibinfo  {journal} {Annual Review of Chemical and
  Biomolecular Engineering}\ }\textbf {\bibinfo {volume} {4}},\ \bibinfo
  {pages} {259} (\bibinfo {year} {2013})}\BibitemShut {NoStop}%
\bibitem [{\citenamefont {Moreno-S{\'a}nchez}\ \emph
  {et~al.}(2008)\citenamefont {Moreno-S{\'a}nchez}, \citenamefont {Saavedra},
  \citenamefont {Rodr{\'\i}guez-Enr{\'\i}quez}, \citenamefont
  {Ol{\'\i}n-Sandoval} \emph {et~al.}}]{moreno2008metabolic}%
  \BibitemOpen
  \bibfield  {author} {\bibinfo {author} {\bibfnamefont {R.}~\bibnamefont
  {Moreno-S{\'a}nchez}}, \bibinfo {author} {\bibfnamefont {E.}~\bibnamefont
  {Saavedra}}, \bibinfo {author} {\bibfnamefont {S.}~\bibnamefont
  {Rodr{\'\i}guez-Enr{\'\i}quez}}, \bibinfo {author} {\bibfnamefont
  {V.}~\bibnamefont {Ol{\'\i}n-Sandoval}}, \emph {et~al.},\ }\bibfield  {title}
  {\bibinfo {title} {Metabolic control analysis: a tool for designing
  strategies to manipulate metabolic pathways},\ }\href
  {https://doi.org/10.1155/2008/597913} {\bibfield  {journal} {\bibinfo
  {journal} {BioMed Research International}\ }\textbf {\bibinfo {volume}
  {2008}} (\bibinfo {year} {2008})}\BibitemShut {NoStop}%
\bibitem [{\citenamefont {Shimizu}\ and\ \citenamefont
  {Matsuoka}(2022)}]{SHIMIZU2022107887}%
  \BibitemOpen
  \bibfield  {author} {\bibinfo {author} {\bibfnamefont {K.}~\bibnamefont
  {Shimizu}}\ and\ \bibinfo {author} {\bibfnamefont {Y.}~\bibnamefont
  {Matsuoka}},\ }\bibfield  {title} {\bibinfo {title} {Feedback regulation and
  coordination of the main metabolism for bacterial growth and metabolic
  engineering for amino acid fermentation},\ }\href
  {https://doi.org/https://doi.org/10.1016/j.biotechadv.2021.107887} {\bibfield
   {journal} {\bibinfo  {journal} {Biotechnology Advances}\ }\textbf {\bibinfo
  {volume} {55}},\ \bibinfo {pages} {107887} (\bibinfo {year}
  {2022})}\BibitemShut {NoStop}%
\bibitem [{\citenamefont {Wagner}\ \emph {et~al.}(2007)\citenamefont {Wagner},
  \citenamefont {Pavlicev},\ and\ \citenamefont {Cheverud}}]{wagner2007road}%
  \BibitemOpen
  \bibfield  {author} {\bibinfo {author} {\bibfnamefont {G.~P.}\ \bibnamefont
  {Wagner}}, \bibinfo {author} {\bibfnamefont {M.}~\bibnamefont {Pavlicev}},\
  and\ \bibinfo {author} {\bibfnamefont {J.~M.}\ \bibnamefont {Cheverud}},\
  }\bibfield  {title} {\bibinfo {title} {The road to modularity},\ }\href
  {https://doi.org/https://doi.org/10.1038/nrg2267} {\bibfield  {journal}
  {\bibinfo  {journal} {Nature Reviews Genetics}\ }\textbf {\bibinfo {volume}
  {8}},\ \bibinfo {pages} {921} (\bibinfo {year} {2007})}\BibitemShut {NoStop}%
\bibitem [{\citenamefont {Nov{\'a}k}\ and\ \citenamefont
  {Tyson}(2008)}]{novak2008design}%
  \BibitemOpen
  \bibfield  {author} {\bibinfo {author} {\bibfnamefont {B.}~\bibnamefont
  {Nov{\'a}k}}\ and\ \bibinfo {author} {\bibfnamefont {J.~J.}\ \bibnamefont
  {Tyson}},\ }\bibfield  {title} {\bibinfo {title} {Design principles of
  biochemical oscillators},\ }\href
  {https://doi.org/https://doi.org/10.1038/nrm2530} {\bibfield  {journal}
  {\bibinfo  {journal} {Nature reviews Molecular cell biology}\ }\textbf
  {\bibinfo {volume} {9}},\ \bibinfo {pages} {981} (\bibinfo {year}
  {2008})}\BibitemShut {NoStop}%
\bibitem [{\citenamefont {Ohga}\ and\ \citenamefont
  {Ito}(2022)}]{PhysRevE.106.044131}%
  \BibitemOpen
  \bibfield  {author} {\bibinfo {author} {\bibfnamefont {N.}~\bibnamefont
  {Ohga}}\ and\ \bibinfo {author} {\bibfnamefont {S.}~\bibnamefont {Ito}},\
  }\bibfield  {title} {\bibinfo {title} {Information-geometric structure for
  chemical thermodynamics: An explicit construction of dual affine
  coordinates},\ }\href {https://doi.org/10.1103/PhysRevE.106.044131}
  {\bibfield  {journal} {\bibinfo  {journal} {Phys. Rev. E}\ }\textbf {\bibinfo
  {volume} {106}},\ \bibinfo {pages} {044131} (\bibinfo {year}
  {2022})}\BibitemShut {NoStop}%
\bibitem [{\citenamefont {Kobayashi}\ \emph
  {et~al.}(2022{\natexlab{a}})\citenamefont {Kobayashi}, \citenamefont
  {Loutchko}, \citenamefont {Kamimura},\ and\ \citenamefont
  {Sughiyama}}]{PhysRevResearch.4.033208}%
  \BibitemOpen
  \bibfield  {author} {\bibinfo {author} {\bibfnamefont {T.~J.}\ \bibnamefont
  {Kobayashi}}, \bibinfo {author} {\bibfnamefont {D.}~\bibnamefont {Loutchko}},
  \bibinfo {author} {\bibfnamefont {A.}~\bibnamefont {Kamimura}},\ and\
  \bibinfo {author} {\bibfnamefont {Y.}~\bibnamefont {Sughiyama}},\ }\bibfield
  {title} {\bibinfo {title} {Hessian geometry of nonequilibrium chemical
  reaction networks and entropy production decompositions},\ }\href
  {https://doi.org/10.1103/PhysRevResearch.4.033208} {\bibfield  {journal}
  {\bibinfo  {journal} {Phys. Rev. Res.}\ }\textbf {\bibinfo {volume} {4}},\
  \bibinfo {pages} {033208} (\bibinfo {year} {2022}{\natexlab{a}})}\BibitemShut
  {NoStop}%
\bibitem [{\citenamefont {Kobayashi}\ \emph
  {et~al.}(2022{\natexlab{b}})\citenamefont {Kobayashi}, \citenamefont
  {Loutchko}, \citenamefont {Kamimura},\ and\ \citenamefont
  {Sughiyama}}]{PhysRevResearch.4.033066}%
  \BibitemOpen
  \bibfield  {author} {\bibinfo {author} {\bibfnamefont {T.~J.}\ \bibnamefont
  {Kobayashi}}, \bibinfo {author} {\bibfnamefont {D.}~\bibnamefont {Loutchko}},
  \bibinfo {author} {\bibfnamefont {A.}~\bibnamefont {Kamimura}},\ and\
  \bibinfo {author} {\bibfnamefont {Y.}~\bibnamefont {Sughiyama}},\ }\bibfield
  {title} {\bibinfo {title} {Kinetic derivation of the hessian geometric
  structure in chemical reaction networks},\ }\href
  {https://doi.org/10.1103/PhysRevResearch.4.033066} {\bibfield  {journal}
  {\bibinfo  {journal} {Phys. Rev. Res.}\ }\textbf {\bibinfo {volume} {4}},\
  \bibinfo {pages} {033066} (\bibinfo {year} {2022}{\natexlab{b}})}\BibitemShut
  {NoStop}%
\bibitem [{\citenamefont {Kobayashi}\ \emph
  {et~al.}(2022{\natexlab{c}})\citenamefont {Kobayashi}, \citenamefont
  {Loutchko}, \citenamefont {Kamimura}, \citenamefont {Horiguchi},\ and\
  \citenamefont {Sughiyama}}]{kobayashi2022information}%
  \BibitemOpen
  \bibfield  {author} {\bibinfo {author} {\bibfnamefont {T.~J.}\ \bibnamefont
  {Kobayashi}}, \bibinfo {author} {\bibfnamefont {D.}~\bibnamefont {Loutchko}},
  \bibinfo {author} {\bibfnamefont {A.}~\bibnamefont {Kamimura}}, \bibinfo
  {author} {\bibfnamefont {S.}~\bibnamefont {Horiguchi}},\ and\ \bibinfo
  {author} {\bibfnamefont {Y.}~\bibnamefont {Sughiyama}},\ }\href
  {https://doi.org/https://doi.org/10.48550/arXiv.2211.14455} {\bibinfo {title}
  {Information geometry of dynamics on graphs and hypergraphs}} (\bibinfo
  {year} {2022}{\natexlab{c}}),\ \Eprint {https://arxiv.org/abs/2211.14455}
  {arXiv:2211.14455 [cs.IT]} \BibitemShut {NoStop}%
\bibitem [{\citenamefont {Ito}\ and\ \citenamefont
  {Sagawa}(2015)}]{ito2015maxwell}%
  \BibitemOpen
  \bibfield  {author} {\bibinfo {author} {\bibfnamefont {S.}~\bibnamefont
  {Ito}}\ and\ \bibinfo {author} {\bibfnamefont {T.}~\bibnamefont {Sagawa}},\
  }\bibfield  {title} {\bibinfo {title} {Maxwell’s demon in biochemical
  signal transduction with feedback loop},\ }\href
  {https://doi.org/https://doi.org/10.1038/ncomms8498} {\bibfield  {journal}
  {\bibinfo  {journal} {Nature communications}\ }\textbf {\bibinfo {volume}
  {6}},\ \bibinfo {pages} {7498} (\bibinfo {year} {2015})}\BibitemShut
  {NoStop}%
\bibitem [{\citenamefont {Penocchio}\ \emph {et~al.}(2022)\citenamefont
  {Penocchio}, \citenamefont {Avanzini},\ and\ \citenamefont
  {Esposito}}]{10.1063/5.0094849}%
  \BibitemOpen
  \bibfield  {author} {\bibinfo {author} {\bibfnamefont {E.}~\bibnamefont
  {Penocchio}}, \bibinfo {author} {\bibfnamefont {F.}~\bibnamefont
  {Avanzini}},\ and\ \bibinfo {author} {\bibfnamefont {M.}~\bibnamefont
  {Esposito}},\ }\bibfield  {title} {\bibinfo {title} {{Information
  thermodynamics for deterministic chemical reaction networks}},\ }\href
  {https://doi.org/10.1063/5.0094849} {\bibfield  {journal} {\bibinfo
  {journal} {The Journal of Chemical Physics}\ }\textbf {\bibinfo {volume}
  {157}},\ \bibinfo {pages} {034110} (\bibinfo {year} {2022})}\BibitemShut
  {NoStop}%
\bibitem [{\citenamefont {Polettini}\ and\ \citenamefont
  {Esposito}(2014)}]{polettini2014irreversible}%
  \BibitemOpen
  \bibfield  {author} {\bibinfo {author} {\bibfnamefont {M.}~\bibnamefont
  {Polettini}}\ and\ \bibinfo {author} {\bibfnamefont {M.}~\bibnamefont
  {Esposito}},\ }\bibfield  {title} {\bibinfo {title} {Irreversible
  thermodynamics of open chemical networks. i. emergent cycles and broken
  conservation laws},\ }\href
  {https://doi.org/https://doi.org/10.1063/1.4886396} {\bibfield  {journal}
  {\bibinfo  {journal} {The Journal of chemical physics}\ }\textbf {\bibinfo
  {volume} {141}},\ \bibinfo {pages} {07B610\_1} (\bibinfo {year}
  {2014})}\BibitemShut {NoStop}%
\bibitem [{\citenamefont {Rao}\ and\ \citenamefont
  {Esposito}(2016)}]{PhysRevX.6.041064}%
  \BibitemOpen
  \bibfield  {author} {\bibinfo {author} {\bibfnamefont {R.}~\bibnamefont
  {Rao}}\ and\ \bibinfo {author} {\bibfnamefont {M.}~\bibnamefont {Esposito}},\
  }\bibfield  {title} {\bibinfo {title} {Nonequilibrium thermodynamics of
  chemical reaction networks: Wisdom from stochastic thermodynamics},\ }\href
  {https://doi.org/10.1103/PhysRevX.6.041064} {\bibfield  {journal} {\bibinfo
  {journal} {Phys. Rev. X}\ }\textbf {\bibinfo {volume} {6}},\ \bibinfo {pages}
  {041064} (\bibinfo {year} {2016})}\BibitemShut {NoStop}%
\bibitem [{\citenamefont {Avanzini}\ \emph {et~al.}(2023)\citenamefont
  {Avanzini}, \citenamefont {Freitas},\ and\ \citenamefont
  {Esposito}}]{PhysRevX.13.021041}%
  \BibitemOpen
  \bibfield  {author} {\bibinfo {author} {\bibfnamefont {F.}~\bibnamefont
  {Avanzini}}, \bibinfo {author} {\bibfnamefont {N.}~\bibnamefont {Freitas}},\
  and\ \bibinfo {author} {\bibfnamefont {M.}~\bibnamefont {Esposito}},\
  }\bibfield  {title} {\bibinfo {title} {Circuit theory for chemical reaction
  networks},\ }\href {https://doi.org/10.1103/PhysRevX.13.021041} {\bibfield
  {journal} {\bibinfo  {journal} {Phys. Rev. X}\ }\textbf {\bibinfo {volume}
  {13}},\ \bibinfo {pages} {021041} (\bibinfo {year} {2023})}\BibitemShut
  {NoStop}%
\bibitem [{\citenamefont {Tu}(2008)}]{tu2008rank}%
  \BibitemOpen
  \bibfield  {author} {\bibinfo {author} {\bibfnamefont {L.~W.}\ \bibnamefont
  {Tu}},\ }\bibfield  {title} {\bibinfo {title} {The rank of a smooth map},\
  }\href {https://doi.org/https://doi.org/10.1007/978-1-4419-7400-6} {\bibfield
   {journal} {\bibinfo  {journal} {An Introduction to Manifolds}\ ,\ \bibinfo
  {pages} {105}} (\bibinfo {year} {2008})}\BibitemShut {NoStop}%
\end{thebibliography}%

\end{document}